\documentclass[prodmode,acmtecs]{acmsmall} 

\usepackage[ruled]{algorithm2e}
\usepackage{graphicx}
\usepackage{amssymb,amsfonts,amsmath}
\usepackage{subfigure}
\usepackage{float}
\usepackage{xspace}
\usepackage{times}
\usepackage{color}

\newcommand{\tc}[2]{\textcolor{#1}{#2}}
\newcommand{\pedro}[1]{\tc{black}{#1}}

\newcommand{\ied}{IED\xspace}
\newcommand{\ieds}{IEDs\xspace}

\newcommand{\tll}{LLG\xspace}

\newcommand{\dppzero}{SFP\xspace}
\newcommand{\dpp}{SFP\xspace}

\newcommand{\groupsfp}{\textit{MetaDist}\xspace}

\newcommand{\Parsimony}{Parsimony\xspace}

\newcommand{\fallacyiid}{i.i.d. fallacy\xspace}
\newcommand{\Fallacyiid}{I.i.d. fallacy\xspace}

\newcommand{\topconcavity}{top-concavity\xspace}

\newcommand{\beq}{\begin{equation}}
\newcommand{\eeq}{\end{equation}}
\newcommand{\bit}{\begin{itemize}}
\newcommand{\eit}{\end{itemize}}

\newcommand{\hide}[1]{}
\newtheorem{model}{Model}
\newtheorem{universal}{Universal Pattern}

\renewcommand{\algorithmcfname}{ALGORITHM}
\SetAlFnt{\small}
\SetAlCapFnt{\small}
\SetAlCapNameFnt{\small}
\SetAlCapHSkip{0pt}
\IncMargin{-\parindent}

\acmVolume{0}
\acmNumber{0}
\acmArticle{0}
\acmYear{2014}
\acmMonth{3}

\begin{document}


\title{Universal and Distinct Properties of Communication Dynamics: How to Generate Realistic Inter-event Times}
\author{Pedro O.S. Vaz de Melo
\affil{Universidade Federal de Minas Gerais}
Christos Faloutsos
\affil{Carnegie Mellon University}
Renato Assun\c{c}\~{a}o
\affil{Universidade Federal de Minas Gerais}
Rodrigo Alves
\affil{Universidade Federal de Minas Gerais}
Antonio A.F. Loureiro
\affil{Universidade Federal de Minas Gerais}}

\begin{abstract}
With the advancement of information systems, means of communications are becoming cheaper, faster and more available. Today, millions of people carrying smart-phones or tablets are able to communicate at practically any time and anywhere they want. Among others, they can access their e-mails, comment on weblogs, watch and post comments on videos, make phone calls or text messages almost ubiquitously. Given this scenario, in this paper we tackle a fundamental aspect of this new era of communication: how the time intervals between communication events behave for different technologies and means of communications? Are there universal patterns for the inter-event time distribution (\ied)? In which ways inter-event times behave differently among particular technologies? To answer these questions, we analyze eight different datasets from real and modern communication data and we found four well defined patterns that are seen in all the eight datasets. Moreover, we propose the use of the Self-Feeding Process (\dpp{}) to generate inter-event times between communications. The \dpp is extremely \textit{parsimonious} point process that requires at most two parameters and is able to generate inter-event times with all the universal properties we observed in the data. We show the potential application of \dpp by proposing a framework to generate a synthetic dataset containing realistic communication events of any one of the analyzed means of communications (e.g. phone calls, e-mails, comments on blogs) and an algorithm to detect anomalies.
\end{abstract}





\maketitle

\section{Introduction}

A popular saying that came with the advancement of information systems is that the distance among people is decreasing over the years. It is well known that the main reason for that is the fact that means of communications are becoming cheaper, faster and more available. Today, millions of people carrying smart-phones or tablets are able to communicate at practically any time and anywhere they want. Among others, they can access their e-mails, comment on weblogs, watch and post comments on videos, make phone calls or text messages almost ubiquitously. It is fascinating that the growing accessibility, reach and speed of these means of communications are making them more and more homogeneous and similar. For instance, consider a smart-phone user with a permanent Internet connection. What is the fastest way to reach this person? By a phone call, by a SMS message, by e-mail or by an instant messaging subscription service (e.g. WhatsApp)? Nowadays, maybe all of these are equally or similarly effective.

Given this scenario, in this paper we tackle a fundamental aspect of this new era of communication: how the time intervals between communication events behave for different technologies and means of communications? Are there universal patterns for the inter-event time distribution (\ied)? In which ways inter-event times behave differently among particular technologies? To answer these questions, we analyze eight different datasets from real and modern communication data, that can be divided into two groups. The first group contains five datasets extracted from Web applications in which several users comment on a given topic. The datasets are extracted from five popular websites: Youtube, MetaFilter, MetaTalk, Ask MetaFilter and Digg. The second group contains three datasets in which individuals perform and receive communication events. In this group we have a Short Message Service (SMS), a mobile phone-call and a public e-mail dataset. These datasets comprise a set of different types of interactions that are common and routine in most human lives. 

As the first contribution of this paper, we found four well defined patterns that are seen in all the eight datasets. First, we show that the marginal distribution of the time intervals between communications follows an odds ratio power law. Second, we show that the slope of this power law is approximately 1 for the majority of the data analyzed. Third, unlike previous studies, we analyze the temporal correlations between inter-event times, illustrating the ``\fallacyiid'' that has been routinely ignored until recently~\cite{karsai:nature:2012}.  We show that, unlike the PP that generates independent and identically distributed (i.i.d.) inter-event times, individual sequences of communications tend to show a high dependence between consecutive inter-arrival times. Finally, we show that the collection of individual \ieds of all systems is very well modeled by a Bivariate Gaussian Distribution. Moreover, besides these four universal properties, we also identified features that differentiate one system from the other and that naturally come from the idiosyncrasies of each system.

As the second contribution of this paper, we propose the use of the Self-Feeding Process (\dpp{}) to generate inter-event times between communications of an individual or in blog posts or videos.  The \dpp is extremely \textit{parsimonious} point process that requires at most two parameters.  We show that it is able to generate inter-event times with all the universal properties we observed in the data and also reconciles existing and contrasting theories in human communication dynamics\cite{barabasi:2005,malmgren:2008}. Moreover, we show how the \dpp can be easily modified to also encompass the particularities seen in each of the analyzed systems.

Finally, as the third contribution of this paper, we show two possible applications of the findings described in this paper. First, through the use of the \dpp, we propose a framework to generate a synthetic dataset containing realistic communication events of any one of the analyzed means of communications (e.g. phone calls, e-mails, comments on blogs). This framework considers all the universal properties and the particularities of each system. Second, we show how to detect anomalies in the systems we investigated through the use of this framework. Among the regular individuals, we were able to identify, for instance, a SMS automated service, blog posts that were deleted by the moderators because of their content and a polemic Youtube video populated by flaming\footnote{hostile and insulting interaction between Internet users} discussions.

The rest of the paper is organized as follows. Section~\ref{sec:related} provides a brief survey of the related work that analyzed inter-event times between communications. Section~\ref{sec:data} describes the eight datasets used in this work. Section~\ref{sec:ied} shows the \ied of individuals from these datasets and that the Odds Ratio function of their \ied{s} is well modeled by a power law. Section~\ref{sec:correlation} shows that the typical behavior of inter-event sequences shows a positive correlation between consecutive inter-event times. Section~\ref{sec:sfp} describes the \dpp{} model, which provides an intuitive and simple explanation for the observed data. Section~\ref{sec:unifying} shows that the \dpp{} model also unifies existing theories on communication dynamics. Section~\ref{sec:collectiveBehavior} describes a model to represent the collective behavior of users in the analyzed systems. Section~\ref{sec:anomaly} shows a method to spot anomalies. Finally, we show the conclusions and future research directions in Section~\ref{sec:conclusion}.

\section{Related Work}
\label{sec:related}

The study of the time interval in which events occur in human activity is not new in the literature. The most primitive model is the classic Poisson process~\cite{haight:1967}. Although the most recent approaches have among themselves significant differences, they all agree that the timing of individuals systematically deviates from this classical approach. The Poisson process predicts that the time interval $\Delta_t$ between two consecutive events by the same individual follows an exponential distribution with expected value $\beta$ and rate $\lambda = 1/\beta$, where
\begin{equation}
\begin{array}{rcl}
\Delta_t &=& -\beta \times \ln(U(0,1)),
\end{array}
\label{eq:pp}
\end{equation}
where $U(0,1)$ is a uniformly random distributed number between $[0,1]$. While in a Poisson process consecutive events follow each other at a relatively regular time, real data shows that humans have very long periods of inactivity and also bursts of intense activity~\cite{barabasi:2005}.

Moreover, recent analysis on the time interval between communication activities shows apparent conflicting ideas among them. First,
Barab\'{a}si~\cite{barabasi:2005} proposed that bursts and heavy-tails in human activities are a consequence of a decision-based queuing process, when tasks are executed according to some perceived priority. In this way, most of the tasks are rapidly executed and some of then may take a very long time. The queuing models proposed in~\cite{barabasi:2005} generates power law~\cite{faloutsos:1999} distributions $p(X=x) \approx x^{-\alpha}$ with slopes $\alpha \approx 1$ or $\alpha \approx 1.5$. In the literature, there are examples that are approximated by the universality class model in e-mail records~\cite{eckmann:2004,vazques:2006}, web surfing~\cite{dezso:2006,vazques:2006}, library visitation, letters correspondence and stock broker's activities~\cite{vazques:2006}, arrival times of requests to print in a student laboratory~\cite{harder:2006} and in short-messages~\cite{wei:2009}, most of them reporting slopes from $1$ to $1.5$ and, in the case of~\cite{wei:2009}, also slopes higher than $1.5$. Although a power law visually fits well the tail of the IED, it usually can not explain the whole distribution~\cite{malmgren:2008}.

Second, other work in literature propose that the IED is well explained by variations of the PP, such as the Interrupted Poisson~\cite{kuczura:1973} (IPP), Non-Homogeneous Poisson Process~\cite{malmgren:2008}, Kleinberg's burst model~\cite{kleinberg:2002} and others. For instance, Malmgreen et\ {al.}~\cite{malmgren:2008} proposed a non-homogeneous Poisson process to explain the inter-event times distribution. The model is based on the circadian and weekly cycles and coupled to the cascading activity and has a varying rate $\lambda(t)$ that depends on time $t$ in a periodic manner. This process generates active intervals according to $\lambda(t)$. Each active interval initiates a homogeneous Poisson process with a determined rate $\lambda_a$. In order to generate the active intervals, the model needs (i) the average number of active intervals per week, and (ii) the probabilities of starting an active interval at a particular time of day and (iii) week. Malmgreen et\ {al.} estimated these parameters empirically and they showed that the model accurately fits the real data. However, this model explains the data at the cost of requiring several parameters and careful data analysis, being impractical for synthetic data generators, for instance. Later, the authors adapted this model to a more parsimonious version~\cite{malmgren:2009b}, but it still has $9$ parameters.

\section{Data Description}
\label{sec:data}

In this work we analyze eight datasets that can be divided into two groups. The first group contains five datasets extracted from Web applications in which several users comment on a given topic. The datasets are extracted from five popular websites: Youtube, MetaFilter, MetaTalk, Ask MetaFilter and Digg. The second group contains three datasets in which individuals perform and receive communication events. In this group we have a Short Message Service (SMS), a mobile phone-call and a public e-mail dataset. For simplicity, we use the term ``individual'' to refer both to topics of the first group and users of the second group. 

In the first group, we analyze a public online news dataset, containing a set of stories and comments over each story. More specifically, the data is from the popular social media site Digg and has 1,485 stories and over 7 million comments~\cite{DeChoudhury:2009}. The Digg dataset is public for research interests and can be downloaded at \scriptsize\texttt{http://www.infochimps.com/datasets/diggcom-data-set}\normalsize. We also analyze three publicly available datasets from the \textit{Metafilter Infodump Project}\footnote{downloaded on September 22nd from http://stuff.metafilter.com/infodump/}, extracted from three discussion forums: MetaFilter\footnote{http://www.metafilter.com/} (Mefi), MetaTalk\footnote{http://metatalk.metafilter.com/} (Meta) and Ask MetaFilter~\footnote{http://ask.metafilter.com/} (Askme). After disregarding topics which received less than 30 comments, the Mefi dataset has 8,384 topics and 1,471,153 comments, the Meta dataset has 2,484 topics and 503,644 comments and the Askme dataset has 498 topics and 65,950 comments. 

Our final dataset from the first group was collected from the Youtube website using the Google's Youtube API\footnote{https://developers.google.com/youtube/}. We collected all the comments posted on the videos classified as \textit{trending} by the API\footnote{https://gdata.youtube.com/feeds/api/standardfeeds/on\_the\_web} from 22/Aug/2012 to 25/Sep/2012. We collected a total of 1,221,390 comments on 989 videos, but we use in our dataset only those videos with more than 30 comments and which the comments span for more than one week, a total of 610 videos and 1,008,511 comments. The full dataset can be downloaded at \scriptsize\texttt{www.dcc.ufmg.br/\~{}olmo/youtube.zip}\normalsize. 

In the second group, the mobile phone calls dataset contains more than 3.1 million customers of a large mobile operator of a large city, with more than 263.6 million phone call records registered during \textit{one month}.  From this same operator, we also have a SMS dataset of 300,000 users spanning six months of data, for a total of 8,784,101 records. These datasets from the mobile operator is under Non-Disclosure Agreement (NDA) and belong to the iLab Research at the Heinz College at CMU, but was already used in several papers~\cite{vazdemelo:2010,vazdemelo:2011b,akoglu:2012}. We also analyze the public Enron e-mail dataset, consisting of 200,399 messages belonging to 158 users with an average of 757 messages per user~\cite{klimt:2004}. The data is public and can be downloaded at \scriptsize\texttt{http://www.cs.cmu.edu/~enron/}\normalsize.

\section{Marginal Distribution}
\label{sec:ied}

In this work, we are first interested on the inter-event time distribution {\ied} of the random variable $\Delta_k$ representing the time $\Delta_k$ between the $k-th$ and the $(k-1)-th$ communication events on a given topic (first group) or of an user (second group). For simplicity, we use the term ``individual'' to refer both to topics (videos, blog posts, news) of the first group and users of the second group. 

\subsection{Odds Ratio Using the Cumulative Distribution Function}

In Figure~\ref{fig:superuser}, we show the distribution of the time intervals $\Delta_k$ between communication events for a typical active user of the SMS dataset, with 44785 SMS messages sent or received. The histogram is showed in Figure~\ref{fig:superuser}-a and, as we observe, this user had a significantly high number of events separated by small periods of time and also long periods of inactivity. Moreover, both the power law fitting, which in the best fit has an exponent of $~-2$, and the exponential fitting, which is generated by a PP, deviates from the real data. The method we use to fit the power law is based on the Maximum likelihood estimation (MLE) described in~\cite{Clauset:2009}.

\begin{figure*}[!htb]
\centering
\subfigure[Histogram]
  {\includegraphics[width=.23\textwidth]{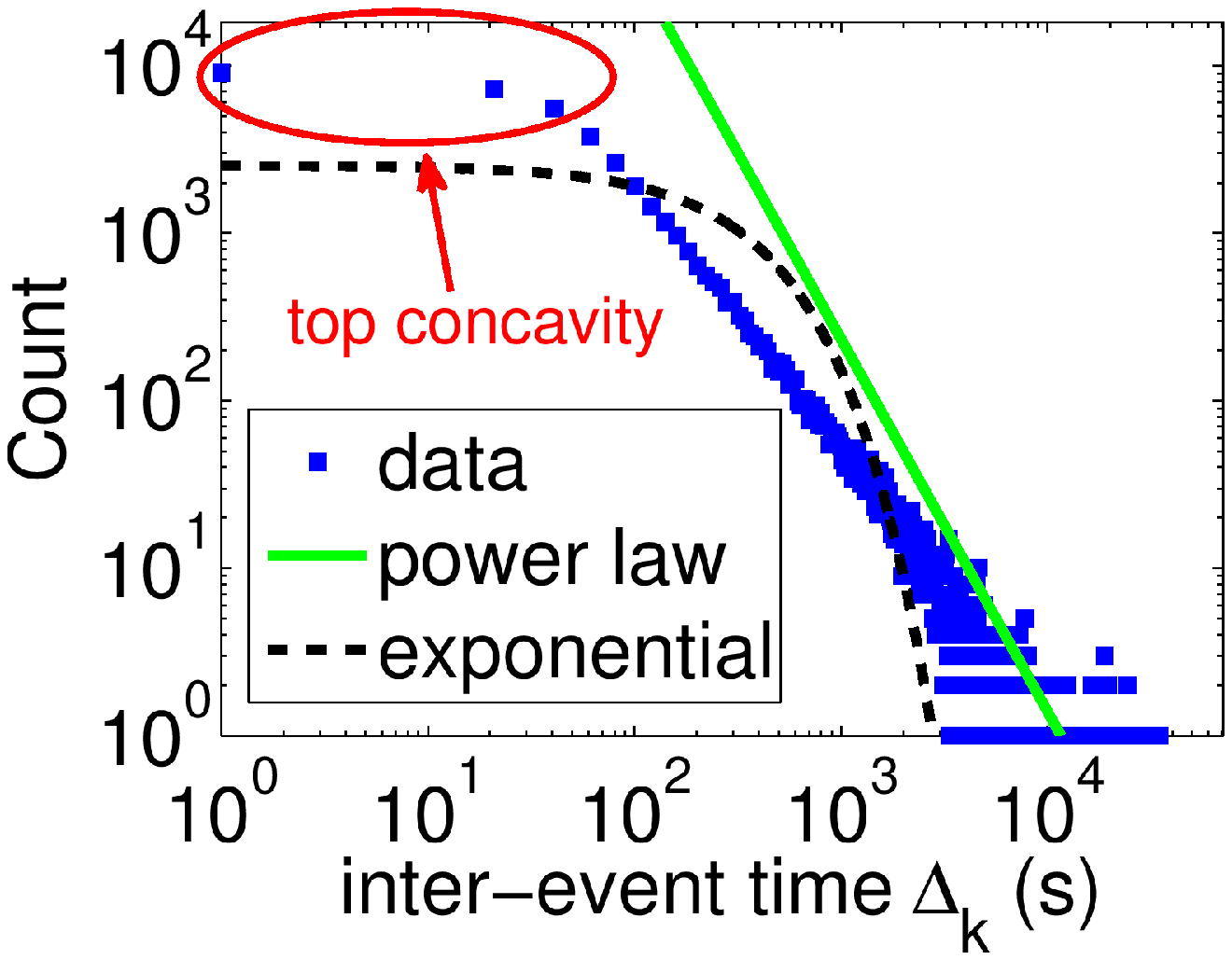}}
\subfigure[CDF]
  {\includegraphics[width=.23\textwidth]{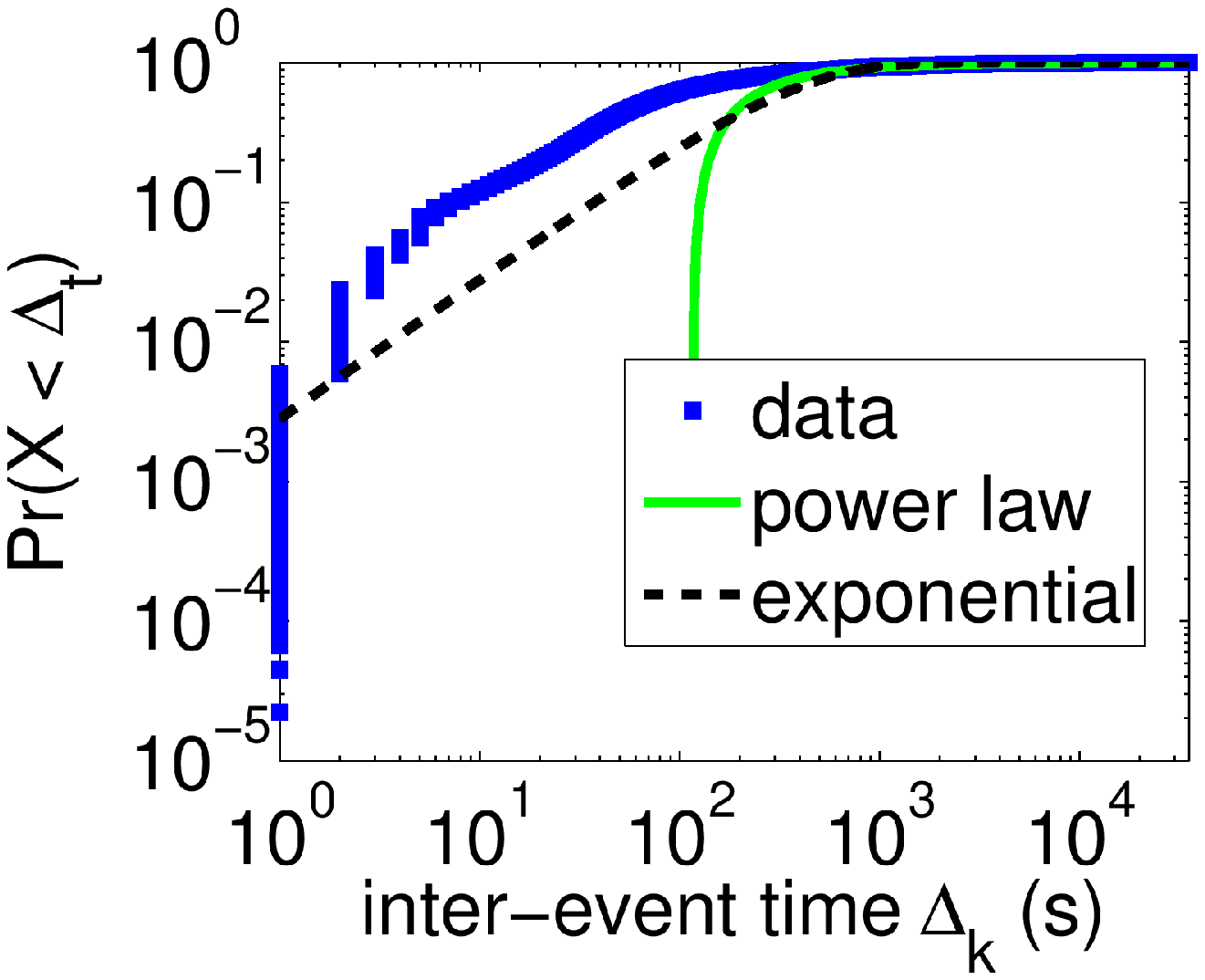}}
\subfigure[CCDF]
  {\includegraphics[width=.23\textwidth]{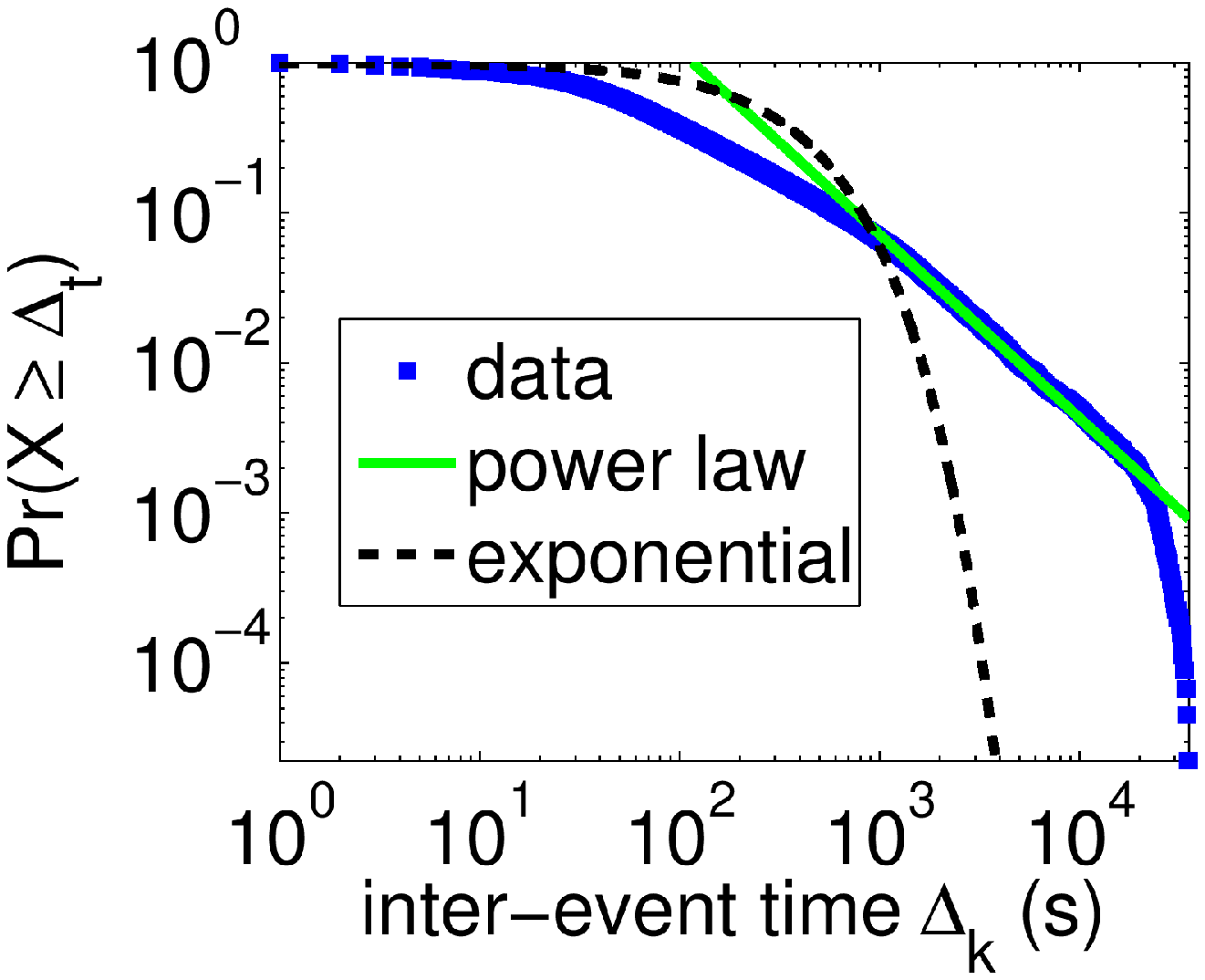}}
\subfigure[Odds Ratio]
  {\includegraphics[width=.23\textwidth]{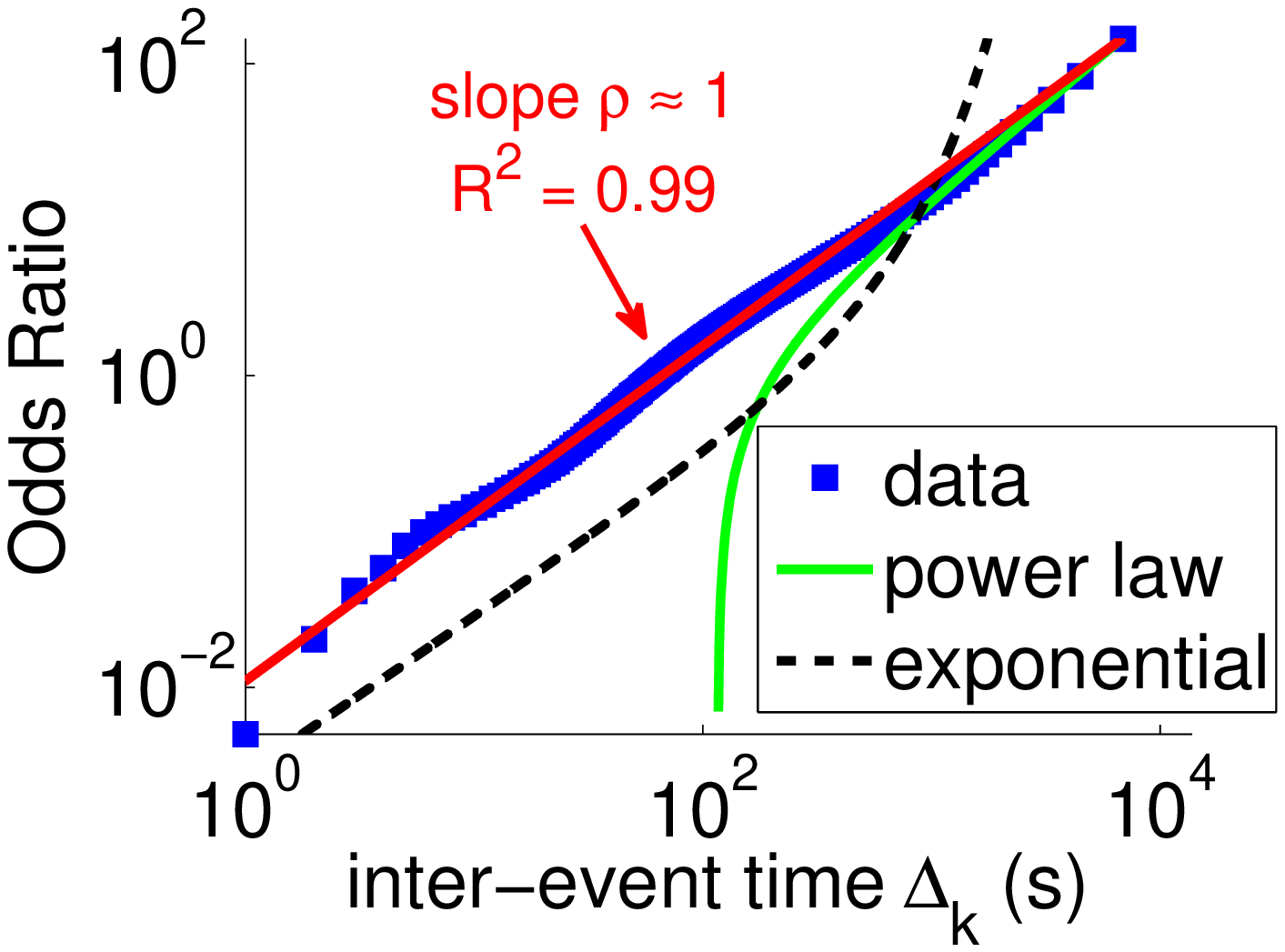}}
  \caption{The inter-event times distribution of the most talkative user of our four datasets, with 44785 SMS messages sent and received. We observe that both the power law fitting (PL fitting) with exponent $\approx 2$ and the exponential fitting, generated by a PP, deviate from the real data. We also observe that the OR is very well fitted by a straight line with slope $\approx 1$.}
  \label{fig:superuser}
\end{figure*}

In empirical data that spans for several orders of magnitude, which is the case of the \ied{}s, it is very difficult to identify statistical patterns in the histograms, since the distribution is considerably noisy at its tail~\cite{barabasi:2005,malmgren:2008}. A possible option 
is to move away from the histogram and analyze the cumulative distributions, i.e., cumulative density function (CDF) and complementary cumulative density function (CCDF), which veil the data sparsity. However, by using the CDF, as we observe in Figure~\ref{fig:superuser}-b,we lose information in the tail of the distribution and, on the other hand, by using the CCDF, as we observe in Figure~\ref{fig:superuser}-c, we lose information in the head of the distribution.

In order to escape from these drawbacks, we propose the use of the Odds Ratio (OR) function combined with the CDF as it allows for a clean visualization of the 
distribution behavior either in the head or in the tail. This $OR(k)$ function is commonly used in the survival analysis~\cite{bennet:1983,mahmood:2000} and measures the ratio between the number of individuals who have not survived by time $t$ and the ones that survived. Its formula is given by:
\begin{eqnarray}
 Odds~Ratio(t) = OR(t) &=& \frac{CDF(t)}{1-CDF(t)}. \label{eq:or}
\end{eqnarray}

In this paper, for a set of $n$ inter-event times $\{\Delta_1, \Delta_2, ..., \Delta_n\}$, we calculate the odds ratio for each percentile $P_1, P_2, ..., P_{100}$ of the data. This avoids that minor deviations in the data harms the goodness of fit test we perform, which we explain in Section~\ref{sec:gof}.

Thus, in Figure~\ref{fig:superuser}-d, we plot the OR for the selected user. From the OR plot, we can clearly see the cumulative behavior in the head and in the tail of the distribution. Also, observe again that both the exponential and the power law significantly deviate from the real data. Moreover, we can also observe that the OR of the inter-event times seems to entirely follow a linear behavior in logarithmic scales. That is, $\log(OR(t))$ is a linear function of $\log(t)$ with slope $\rho \approx 1$ and we say that we have a OR power law behavior. If the approximation is turned into an equality, this implies that the inter-event times follow a log-logistic distribution (see Appendix~\ref{sec:loglogistic}).

In Figure~\ref{fig:pdfUser}, we plot the OR of a typical individual of each dataset. As in Figure~\ref{fig:pdfUser}-d, the OR plots show a clear and almost perfect linear relationship between $\log(OR(t))$ and $\log(t)$ in all the examples considered. This implies that the \ied follows a log-logistic distribution, as we explain in Appendix~\ref{sec:loglogistic}.  Also, we can observe that the OR of the inter-event times seems to follow entirely the same linear behavior in logarithmic scales, having,
then, an OR power law behavior. This implies that the marginal distribution of the \ied{s} is approximately equal to a log-logistic distribution~\cite{fisk:1961}, since it also shows a OR power law behavior. \pedro{For a larger sample, please see the Appendix~\ref{sec:sample}.}

\begin{figure*}[!htb]
\centering
\subfigure[Youtube]
  {\includegraphics[width=.23\textwidth]{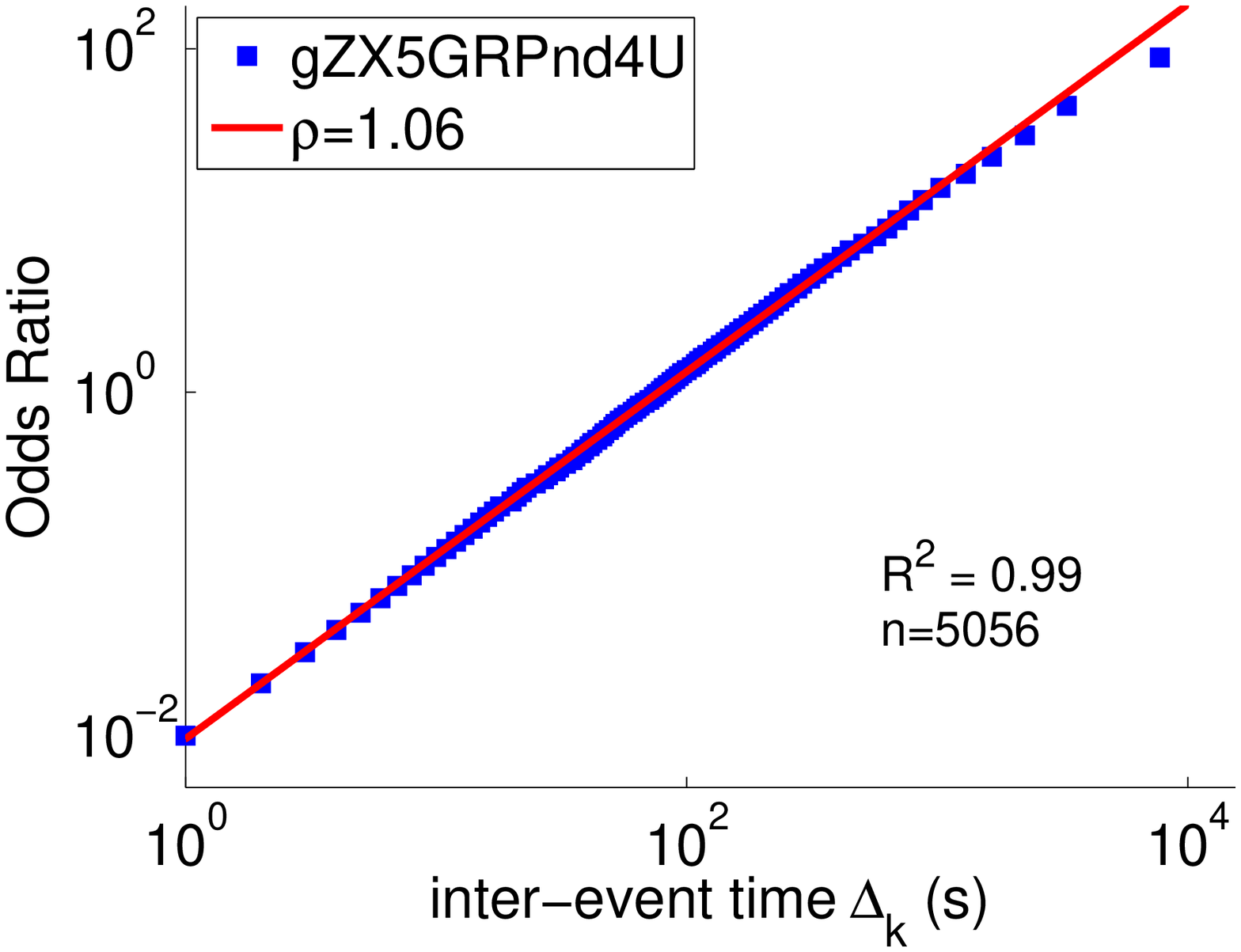}}  
\subfigure[MetaFilter]
  {\includegraphics[width=.23\textwidth]{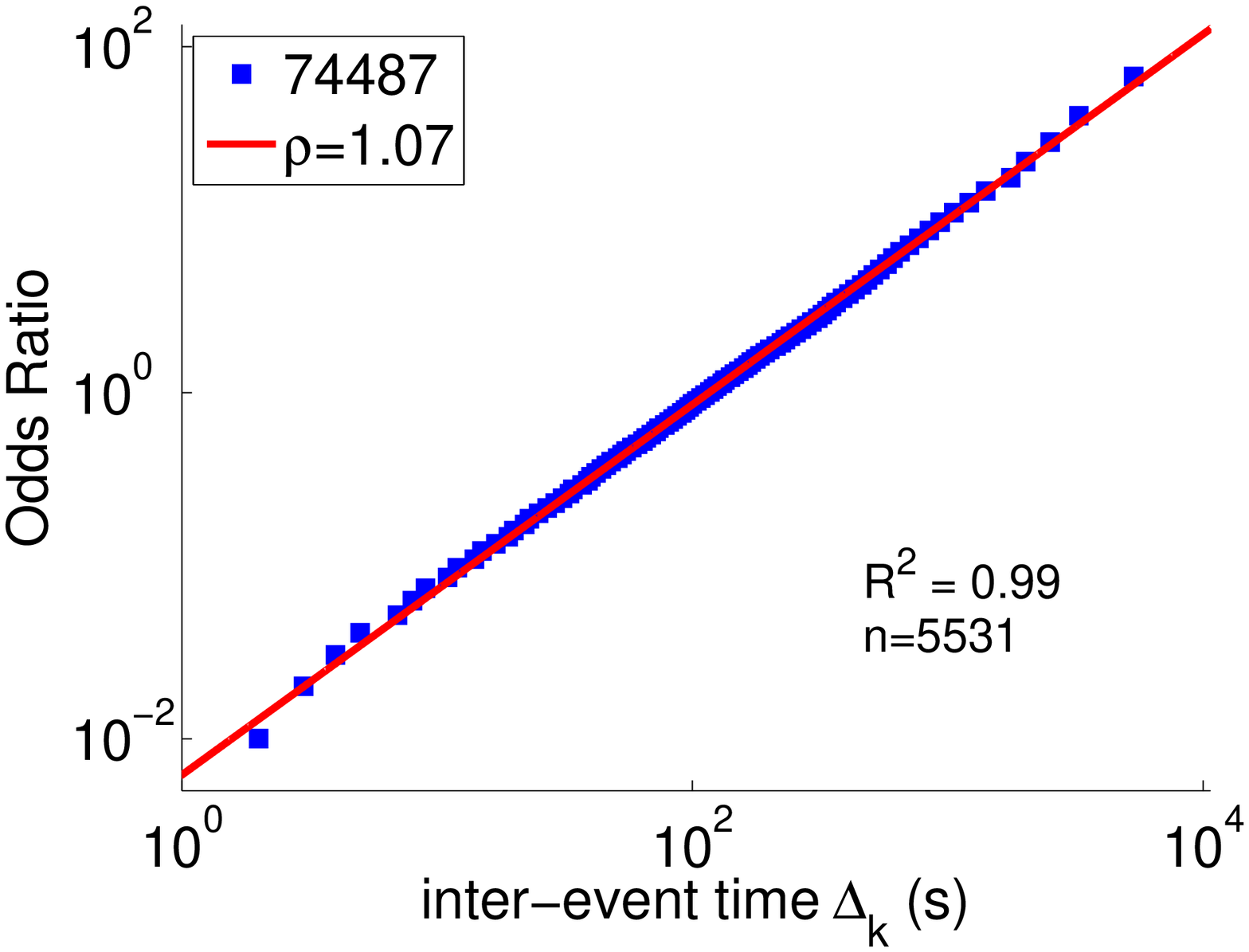}}  
\subfigure[MetaTalk]
  {\includegraphics[width=.23\textwidth]{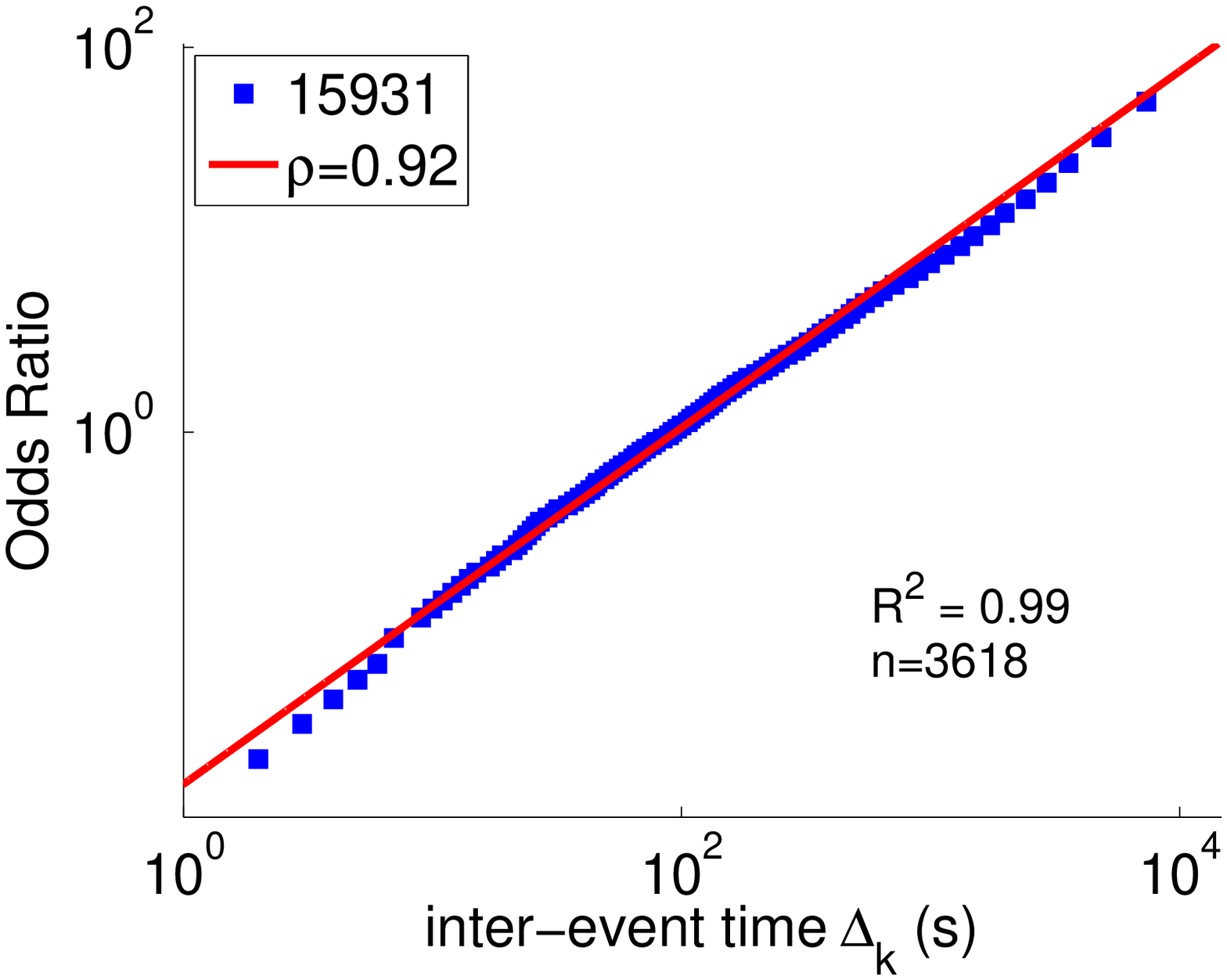}}  
\subfigure[Ask MetaFilter]
  {\includegraphics[width=.23\textwidth]{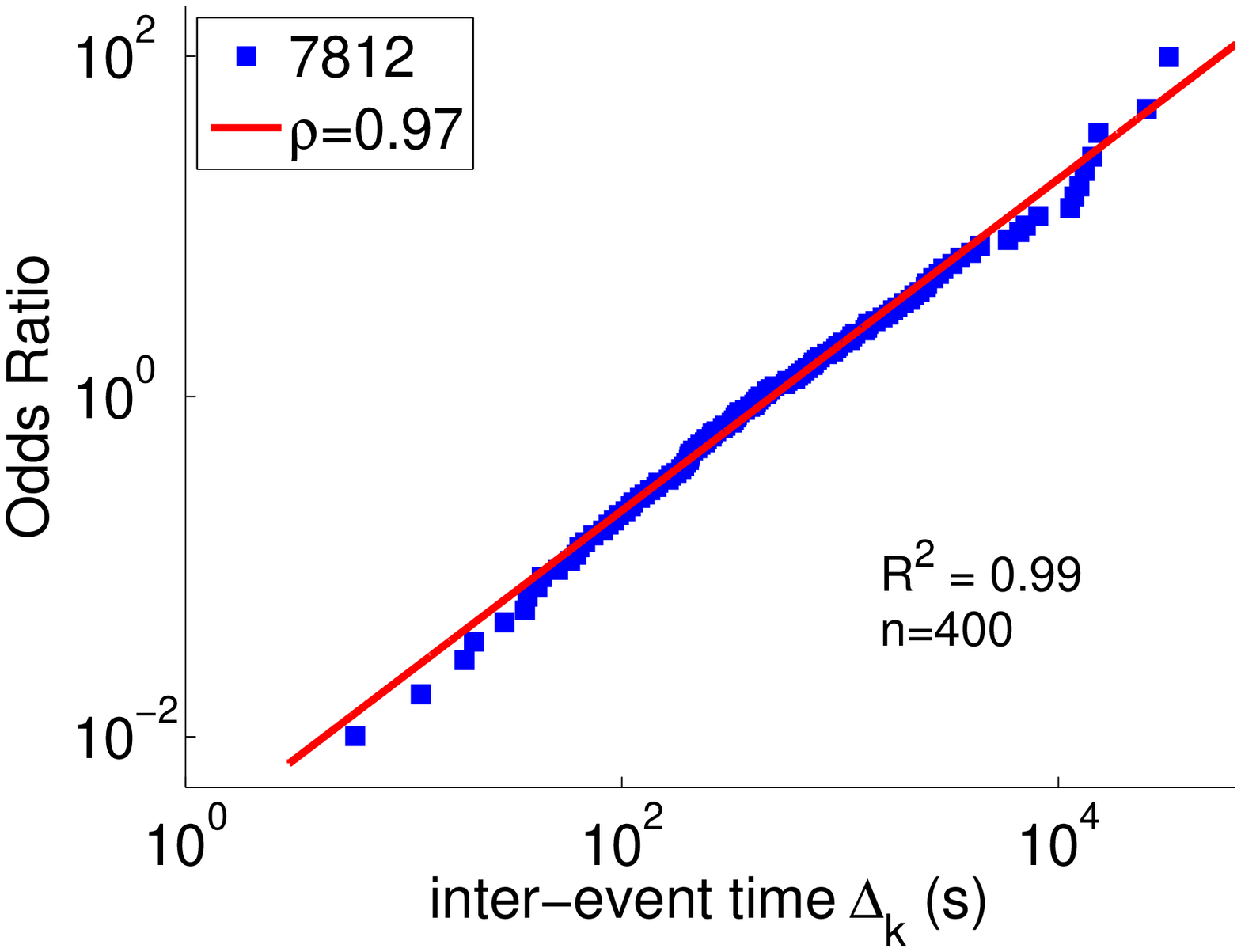}}  
\subfigure[Digg]
  {\includegraphics[width=.23\textwidth]{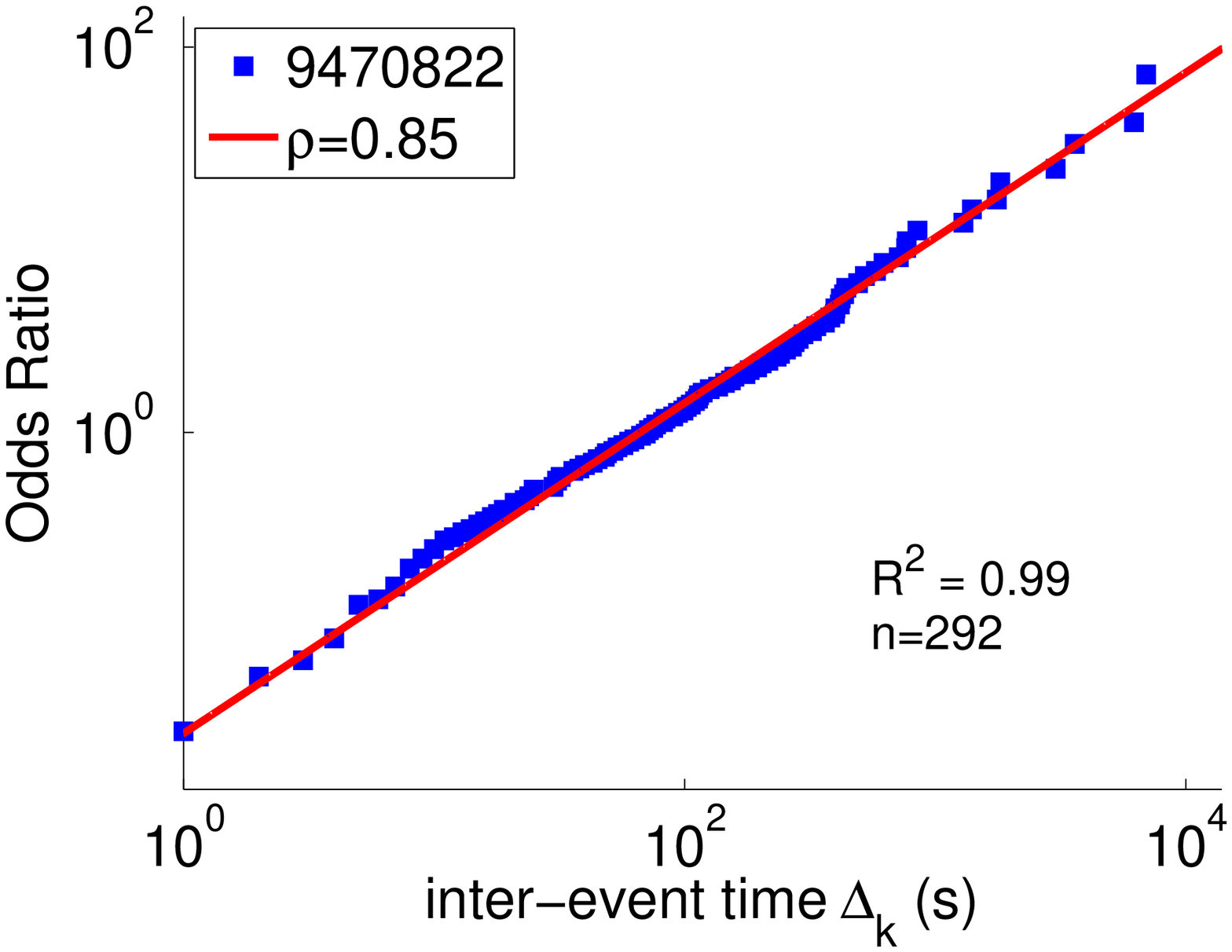}}  
\subfigure[SMS]
  {\includegraphics[width=.23\textwidth]{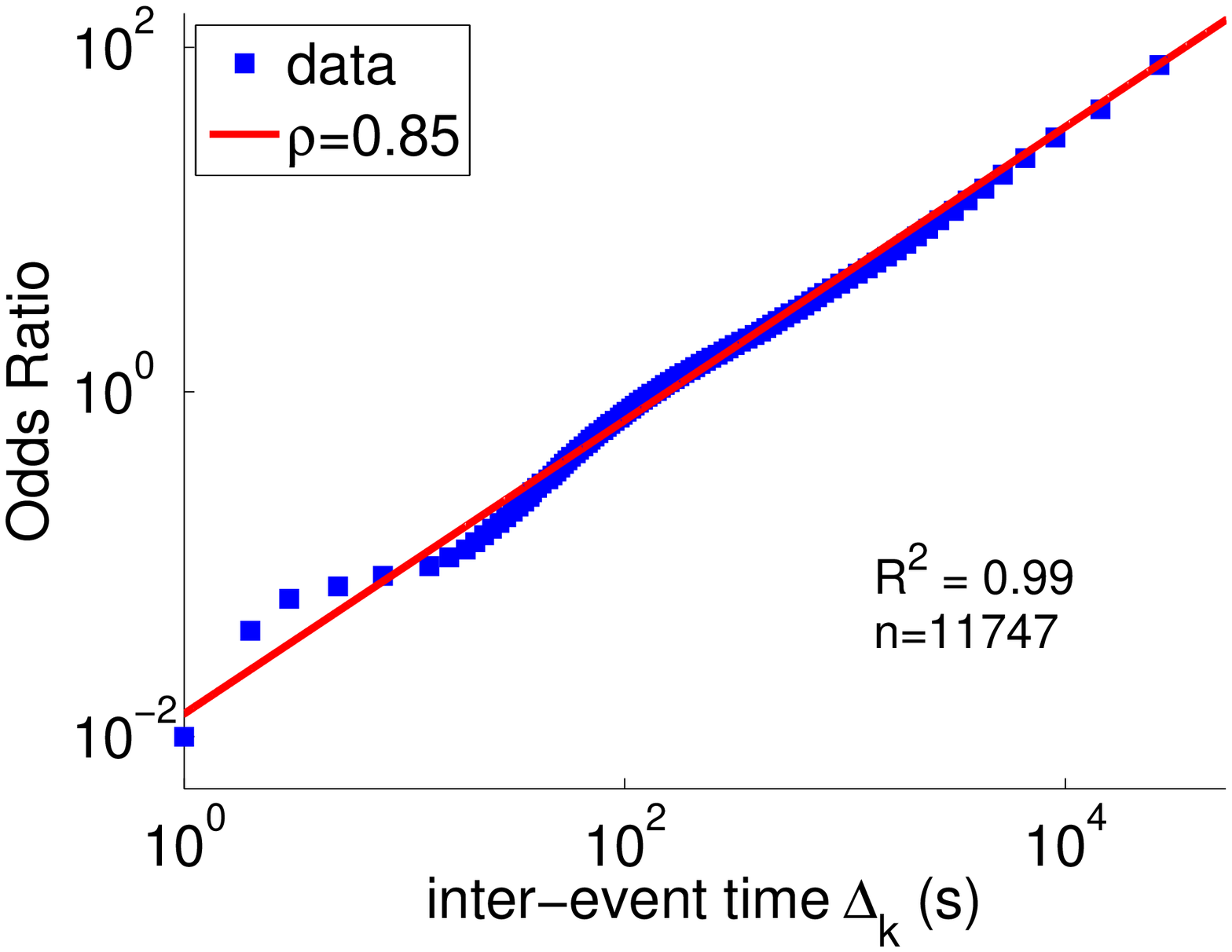}}  
\subfigure[Phone]
  {\includegraphics[width=.23\textwidth]{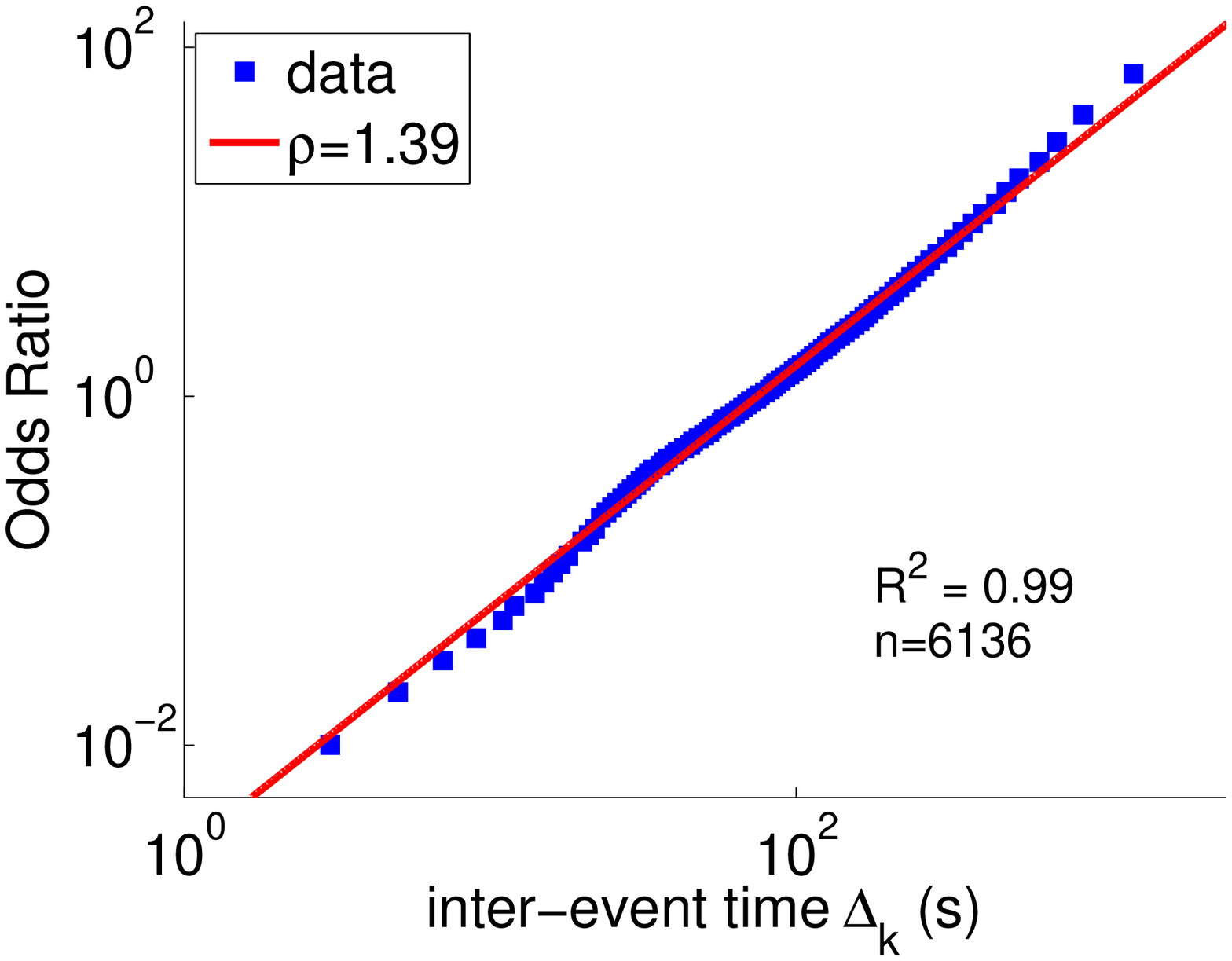}}  
\subfigure[E-mail]
  {\includegraphics[width=.23\textwidth]{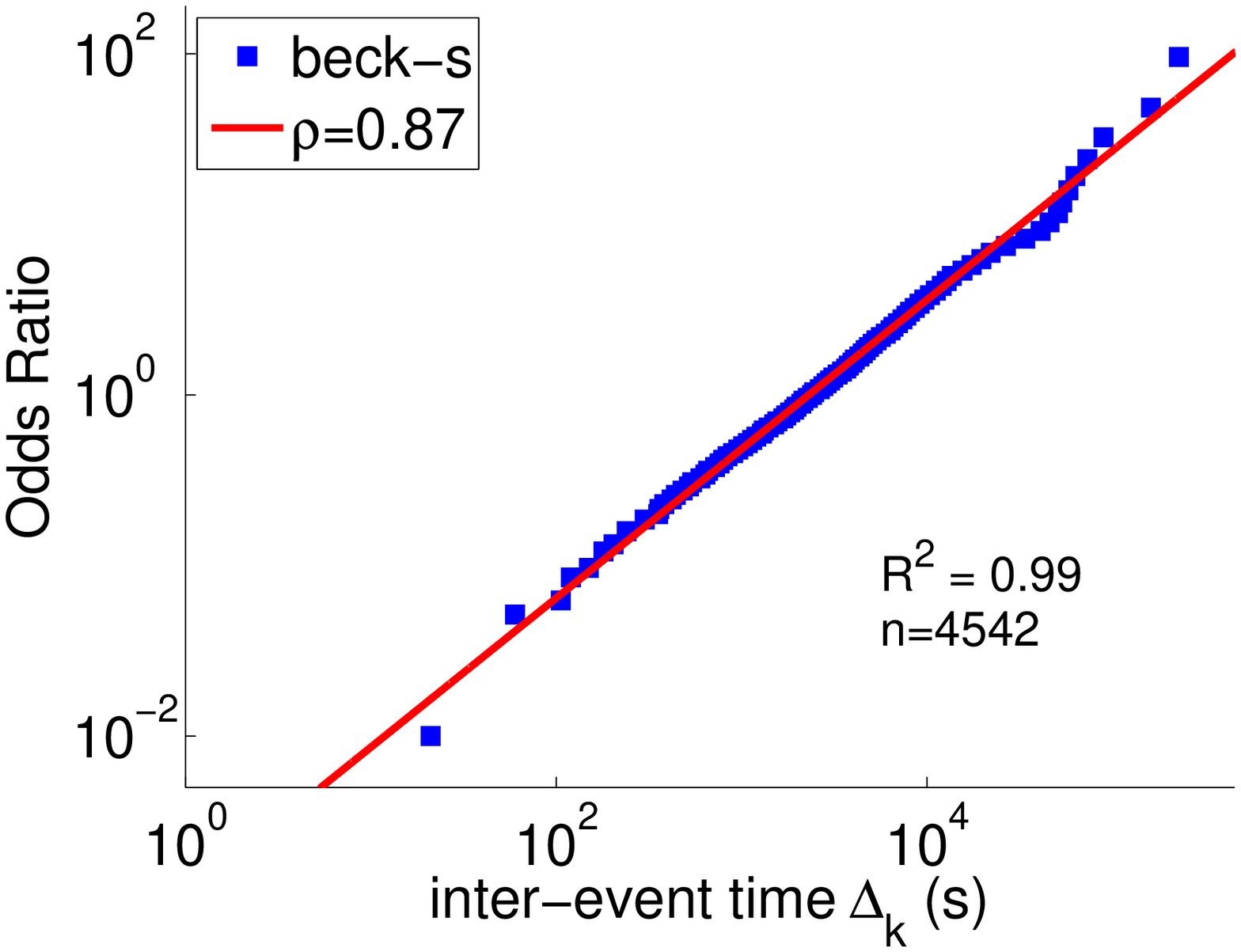}}    
  \caption{The Odds Ratio plot for one typical active individual of each dataset. Observe that an odds ratio power law, represented by a straight line with slope $\rho$ in a log-log scale, is an appropriate fit for all individuals.}
  \label{fig:pdfUser}
\end{figure*}

\subsection{Goodness of Fit}
\label{sec:gof}

In this section, we check whether the OR of the \ied{}s of all individuals of our datasets can be explained by a power law. We perform a linear regression using least squares fitting on the OR of the \ied{}s of all individuals. Since we consider every percentile and the OR is based on the CDF, a cumulative distribution, the linear regression can be used to measure the goodness of fit. We performed a Kolmogorov-Smirnov goodness of fit test, but because of digitalization errors and other deviations in the data
, this test is only approximated..

Figure~\ref{fig:gof} shows the histogram of the determination coefficient $R^2$ of the performed linear regressions. The determination coefficient $R^2$ is a statistical measure of how well the regression line approximates to the real data points. An $R^2 = 1.0$ indicates that the regression line perfectly fits the data. We observe that for the vast majority of individuals of our eight datasets, the $R^2$ is very close to $1.0$. More specifically, in the first group, the $R^2$ averages $0.99$ for the phone dataset, $0.96$ for the SMS dataset and $0.97$ for the e-mail dataset. For the second group, the $R^2$ averages $0.97$ for the Youtube, Askme and Digg datasets and $0.98$ for the Mefi and Meta datasets. This allows us to state the following universal pattern:

\begin{universal}
 The Odds Ratio of the inter-event time distribution of communication events is well fitted by a power law.
\end{universal}

 \begin{figure}[!hbtp]
  \centering
\subfigure[First group]
  {\includegraphics[width=.40\textwidth]{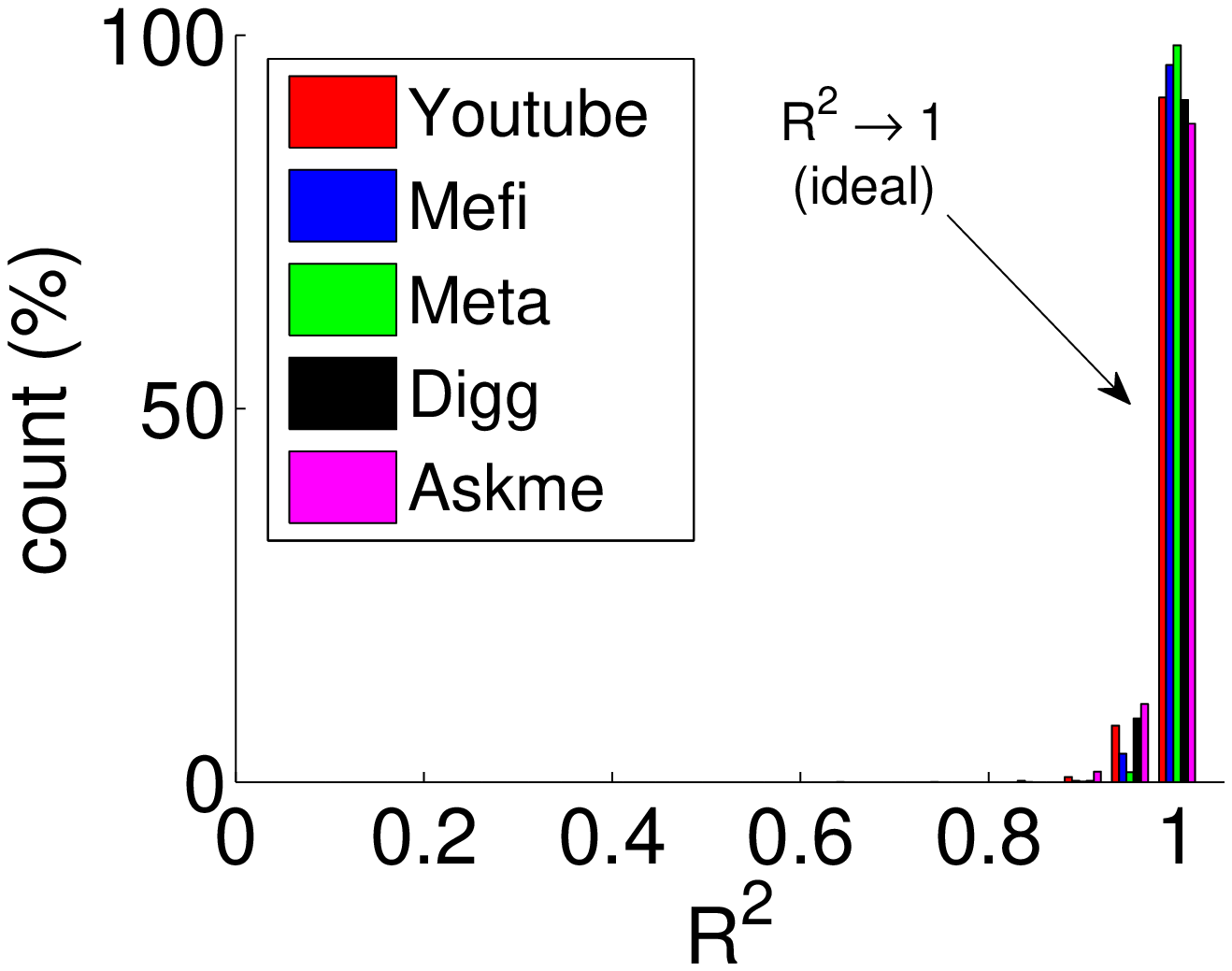}}  
\subfigure[Second group]
  {\includegraphics[width=.40\textwidth]{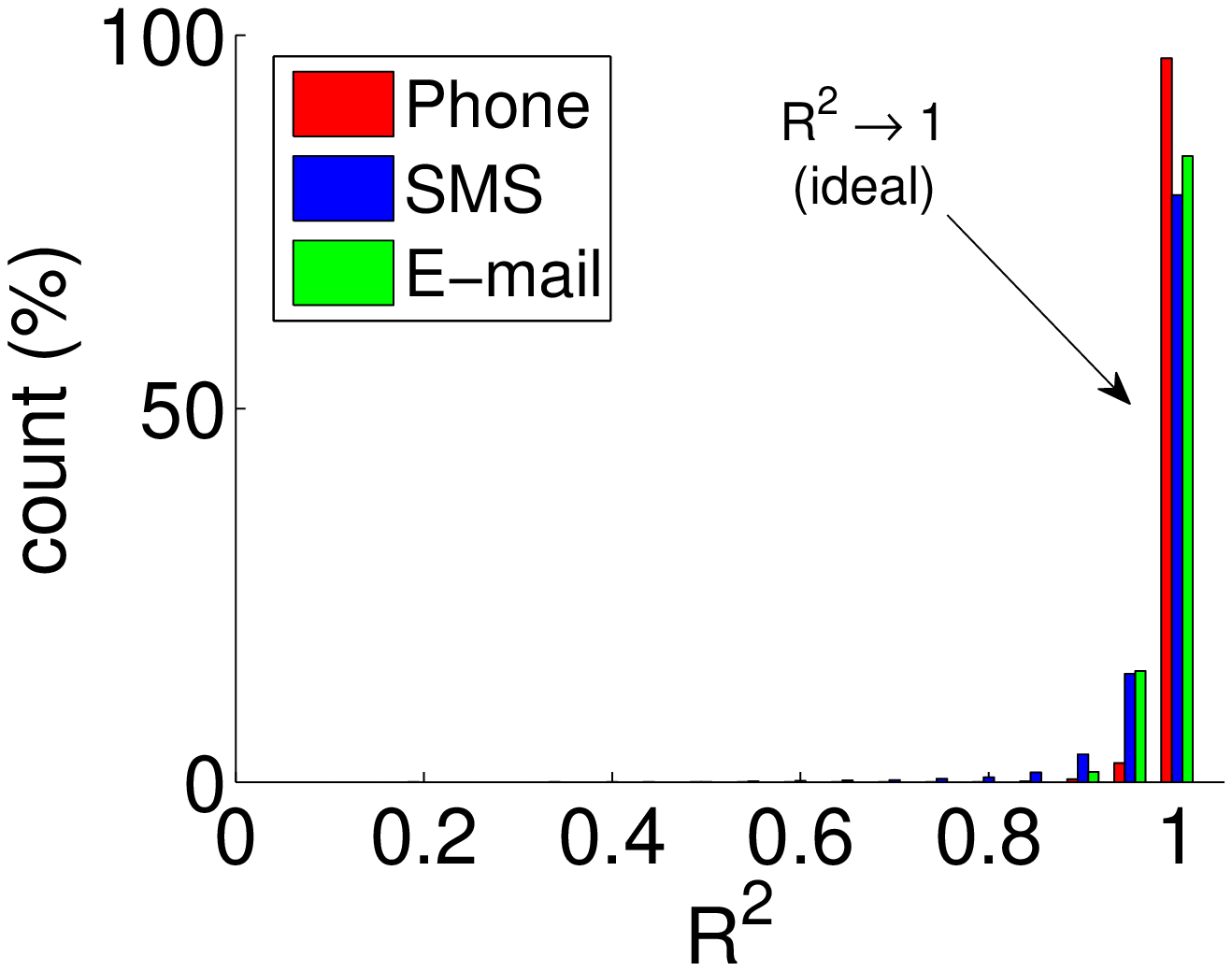}}  
  \caption{The goodness of fit of our proposed model. We show the histograms of the $R^2$s measured for every user in the eight datasets. These histograms consider bins of size $0.05$. Thus, observe that the $R^2$ value for the great majority of individuals is located in the last bin, from $0.95$ to $1$.}
  \label{fig:gof}
\end{figure}

\subsection{Odds Ratio of Well Known Distributions}

\pedro{We have seen that the \ied of the majority of the individuals of our datasets is well modeled by an odds ratio power law. 
In Figure~\ref{fig:ordists}, observe that this odds ratio power law behavior is also seen in log-logistically distributed data and cannot be seen in other well known distributions. This is the first indication that the marginal distribution of time intervals between communications of individuals is well modeled by a log-logistic distribution.}

\begin{figure*}[tpb]
\centering
\subfigure[log-logistic]
  {\includegraphics[width=.23\textwidth]{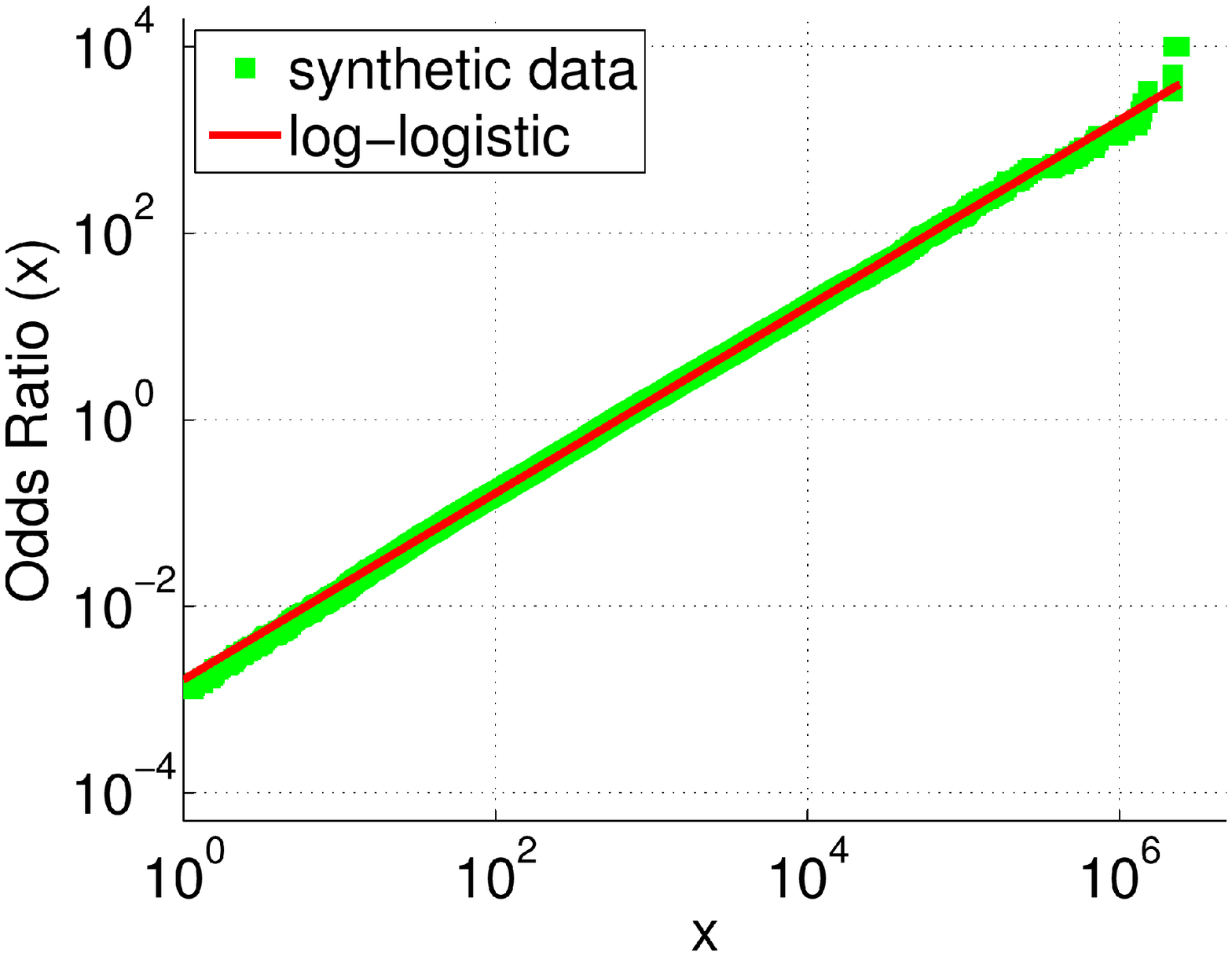}}
\subfigure[log-normal]
  {\includegraphics[width=.23\textwidth]{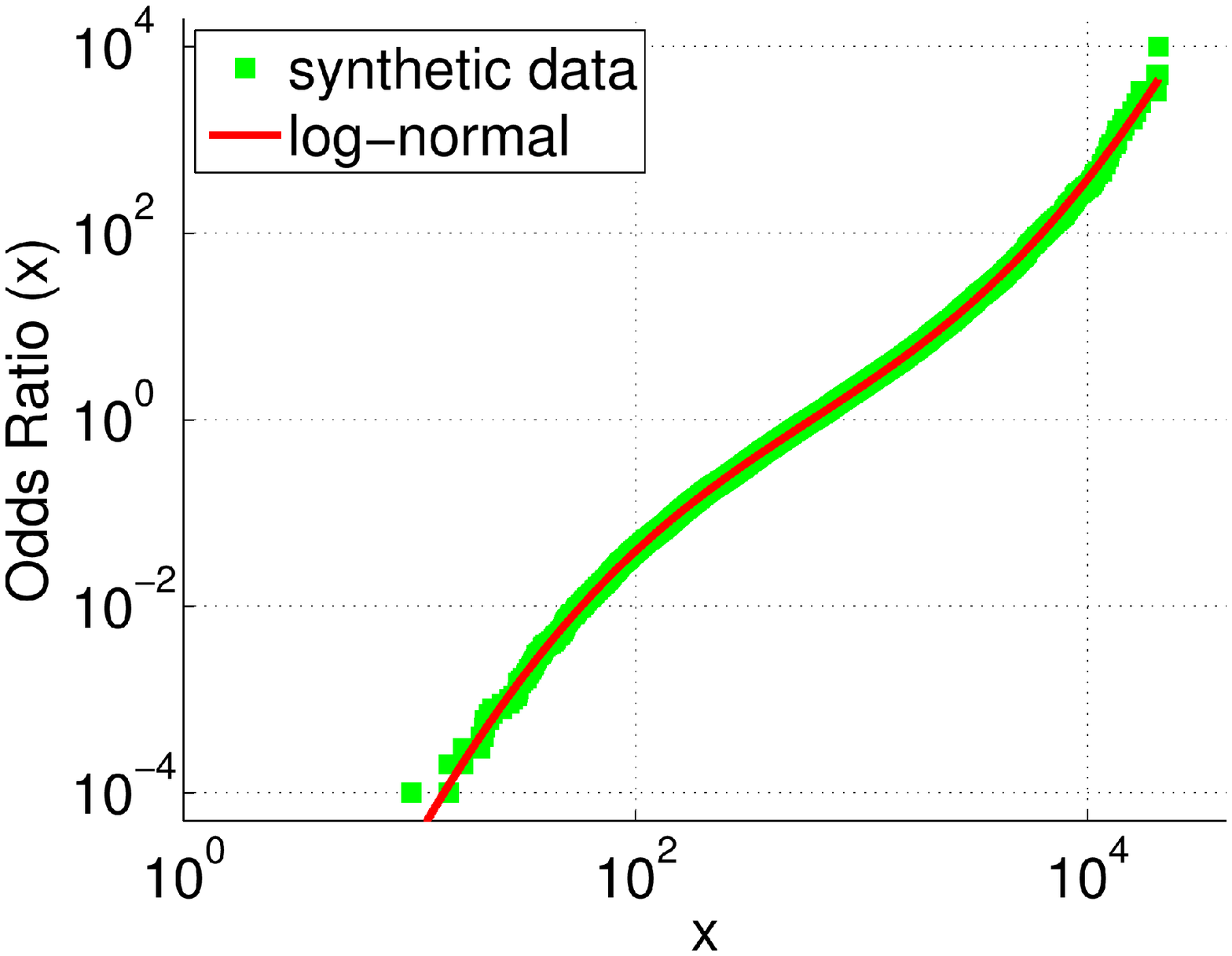}}
\subfigure[exponential]
  {\includegraphics[width=.23\textwidth]{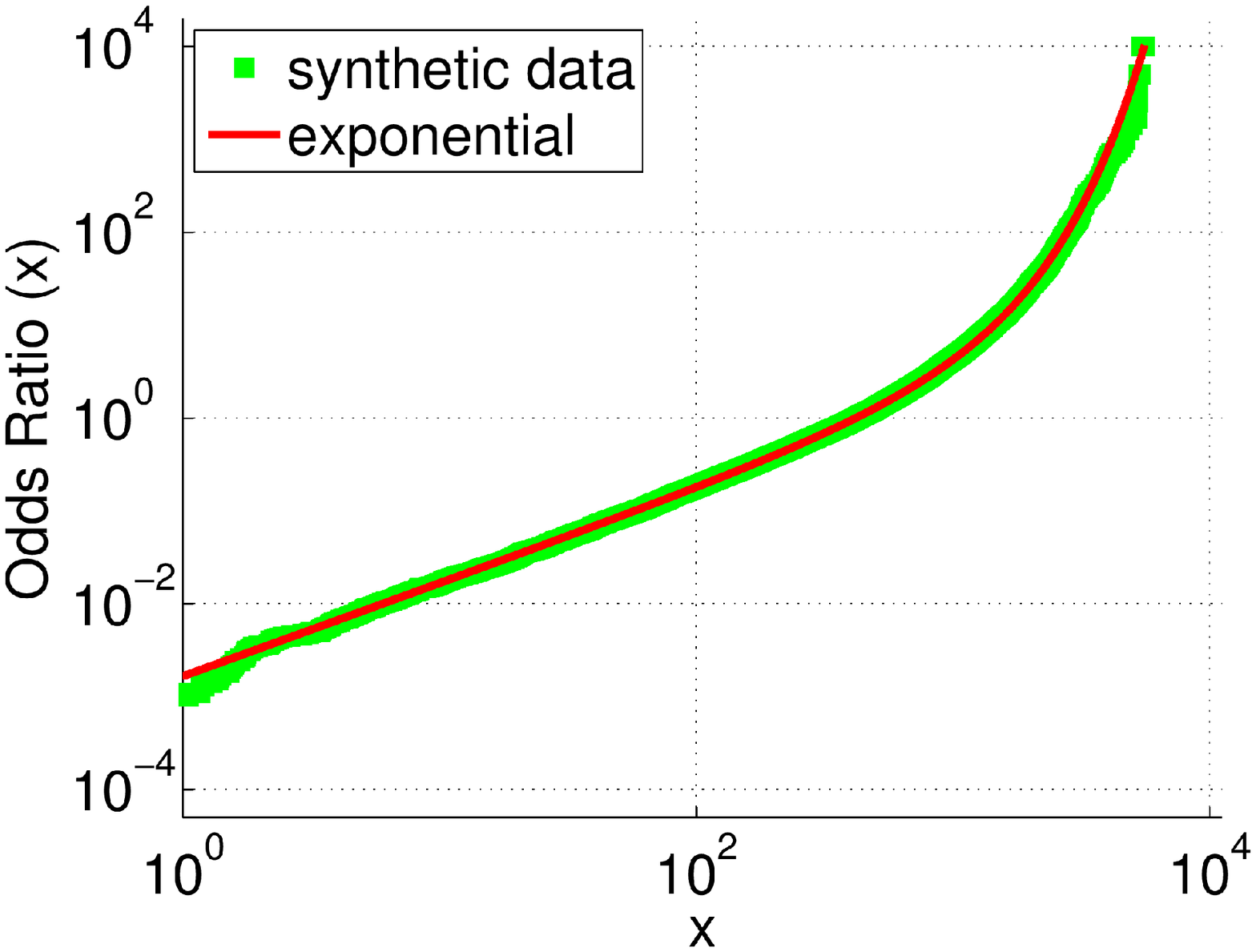}}
\subfigure[Weibull]
  {\includegraphics[width=.23\textwidth]{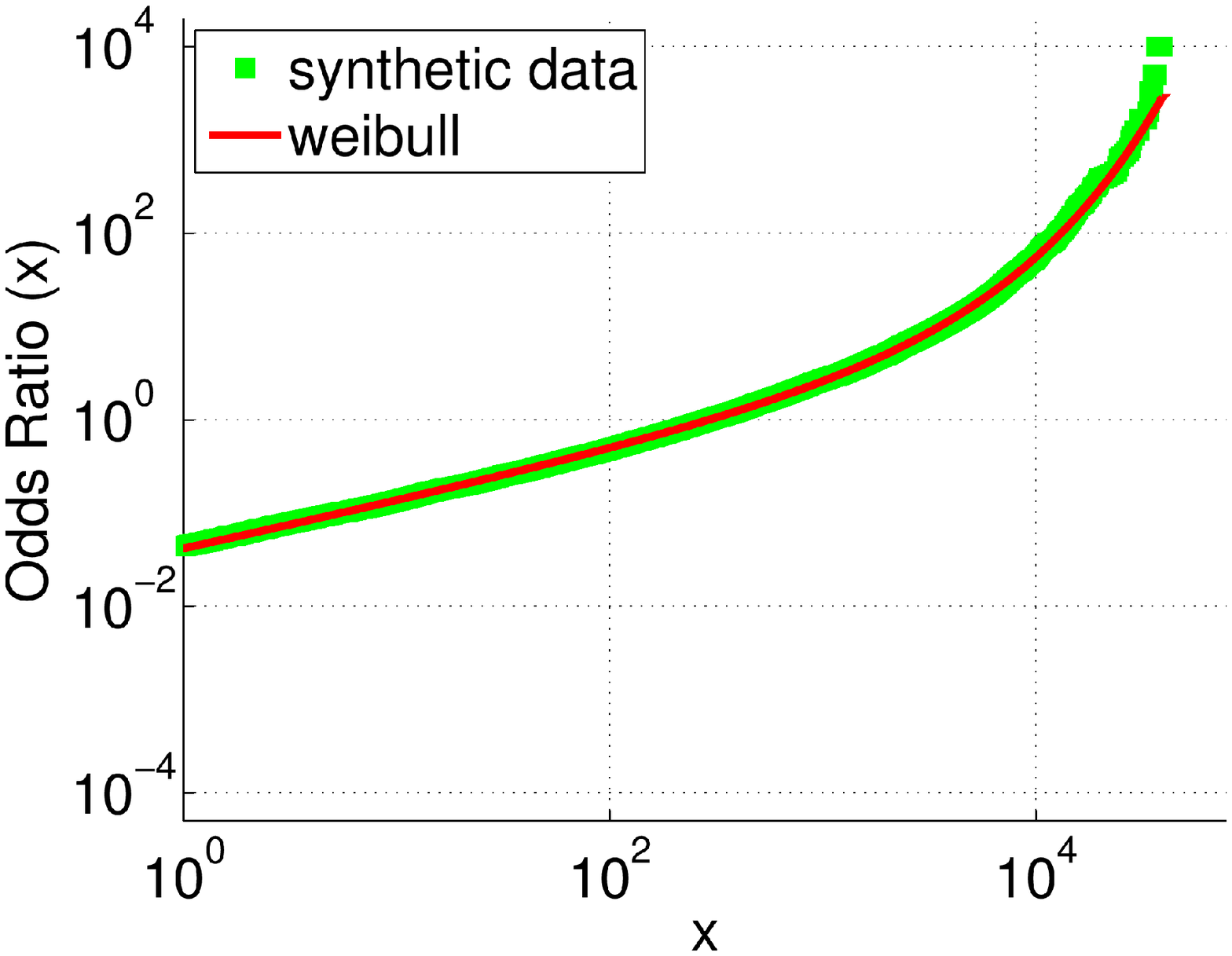}}
\subfigure[Rayleigh]
  {\includegraphics[width=.23\textwidth]{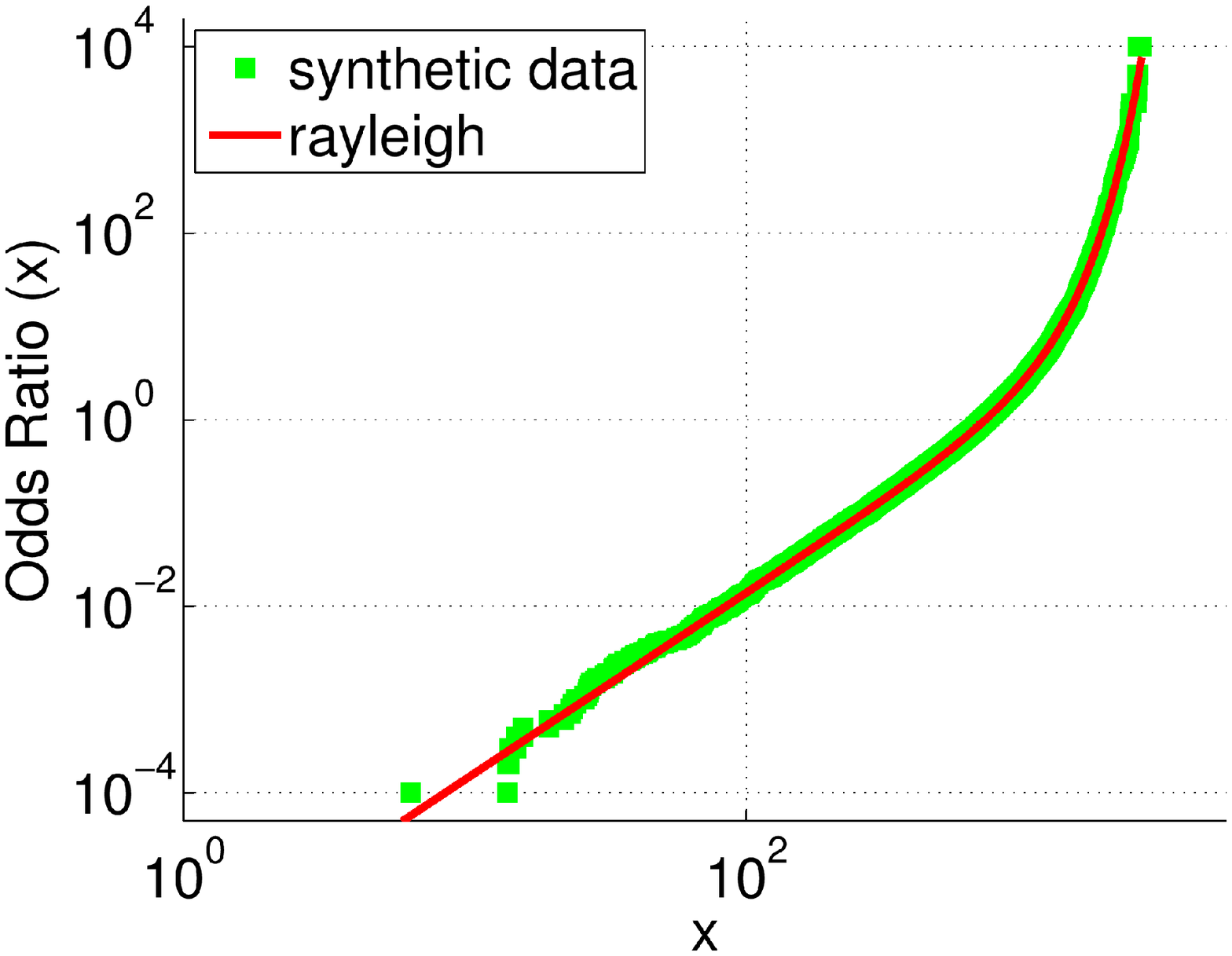}}
\subfigure[gamma]
  {\includegraphics[width=.23\textwidth]{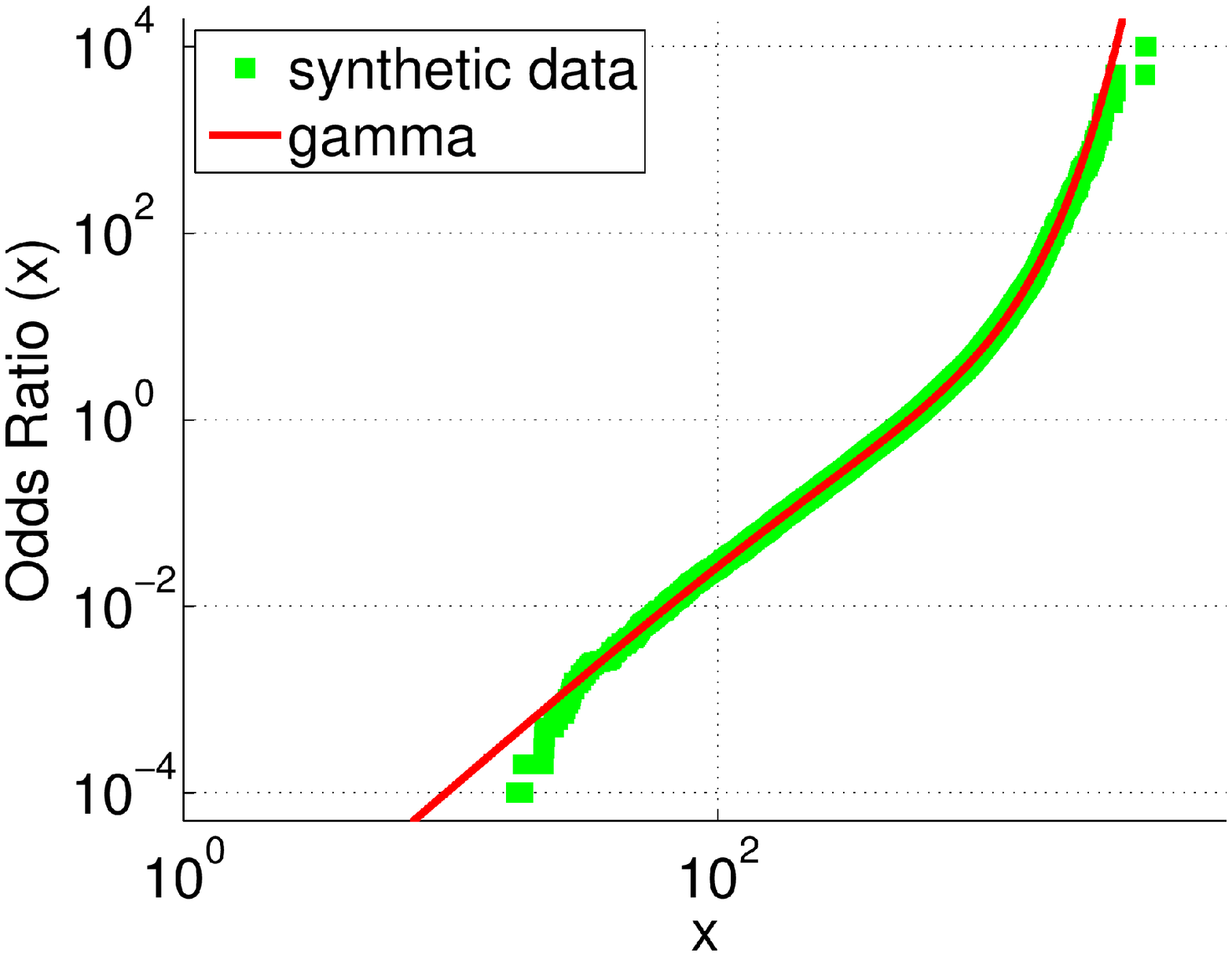}}
\subfigure[normal]
  {\includegraphics[width=.23\textwidth]{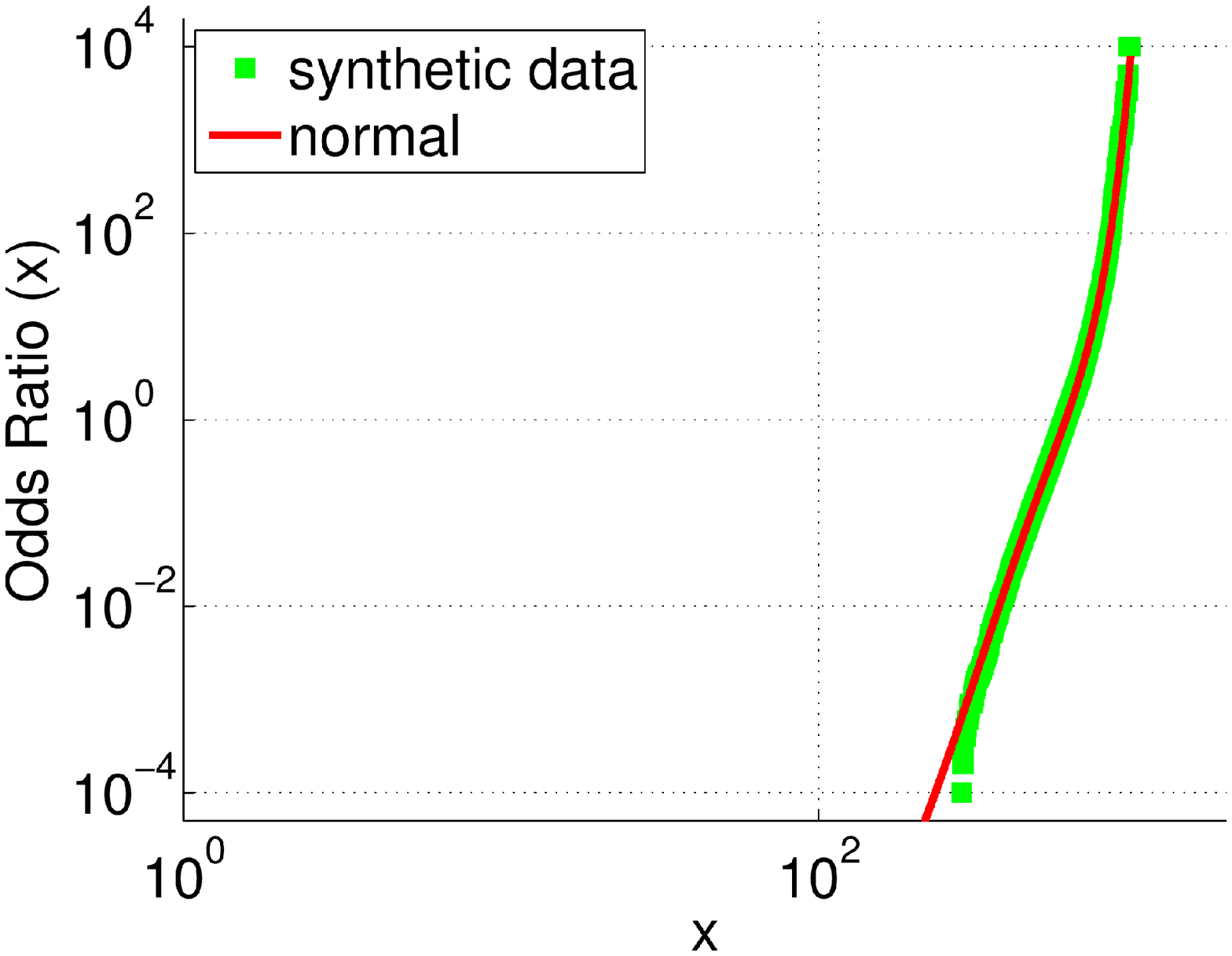}}
\subfigure[power law]
  {\includegraphics[width=.23\textwidth]{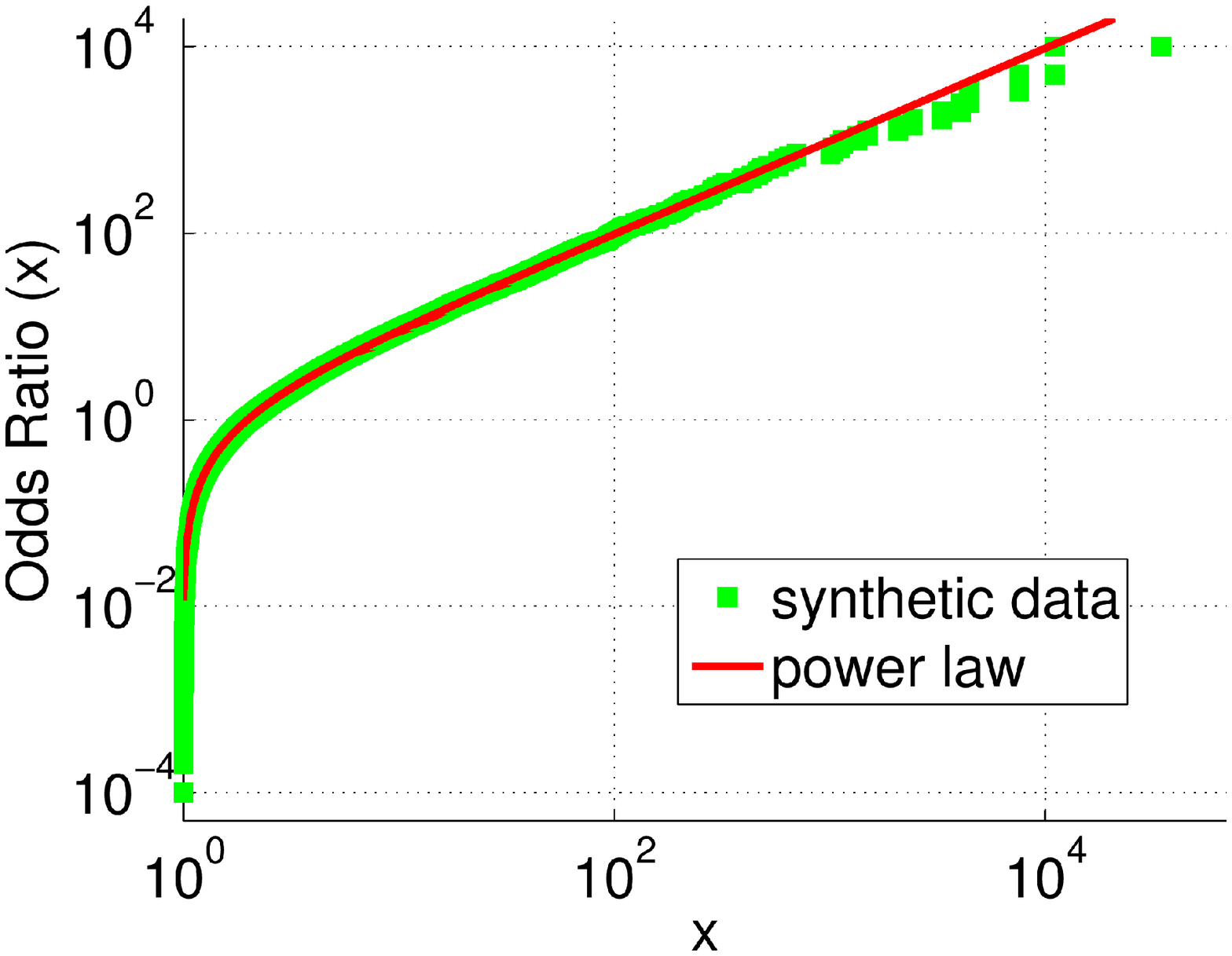}}  
  \caption{The Odds Ratio function for eight well known distributions.}
  \label{fig:ordists}
\end{figure*}

\section{Temporal Correlation}
\label{sec:correlation}

Although most previous analysis focus solely on the marginal \ied, 
a subtle point is the {\em correlation} between successive inter-event times 
($\Delta_{k-1}$ and $\Delta_k$).
What we illustrate here is that the independence between $\Delta_k$ and $\Delta_{k-1}$ does \textbf{not} hold for the eight datasets we analyzed in this work.

In Figure~\ref{fig:acf}, we plot, for the same typical users of Figure~\ref{fig:pdfUser}, all the pairs of consecutive inter-event times $(\Delta_{k-1},\Delta_{k})$. We also show the regression of the data points using the LOWESS smoother~\cite{Cleveland:JASS:1979}.
While the PP, as for any other renewal process, the regression is a flat line with slope $0$, for the eight typical users $\Delta_{k}$ tends to grow
with  $\Delta_{k-1}$. This means that if I called you five years ago, my next phone call will be in about five years later. In short, there is a strong, positive dependency 
between the current inter-event time ($\Delta_k$) and the previous one ($\Delta_{k-1}$), clearly contradicting the independence assumption.

\begin{figure*}[tpb]
\centering
\subfigure[Youtube]
  {\includegraphics[width=.23\textwidth]{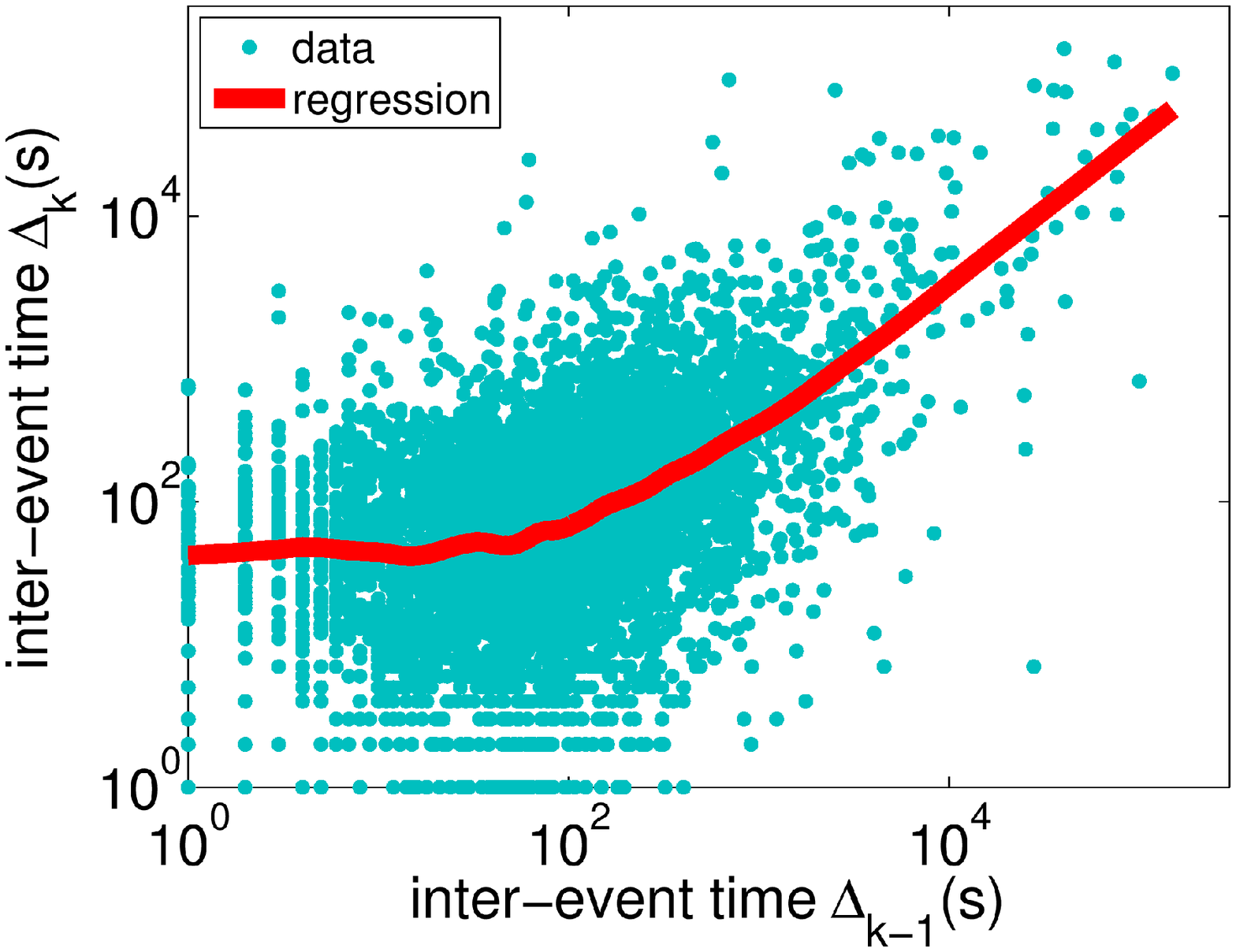}}  
\subfigure[MetaFilter]
  {\includegraphics[width=.23\textwidth]{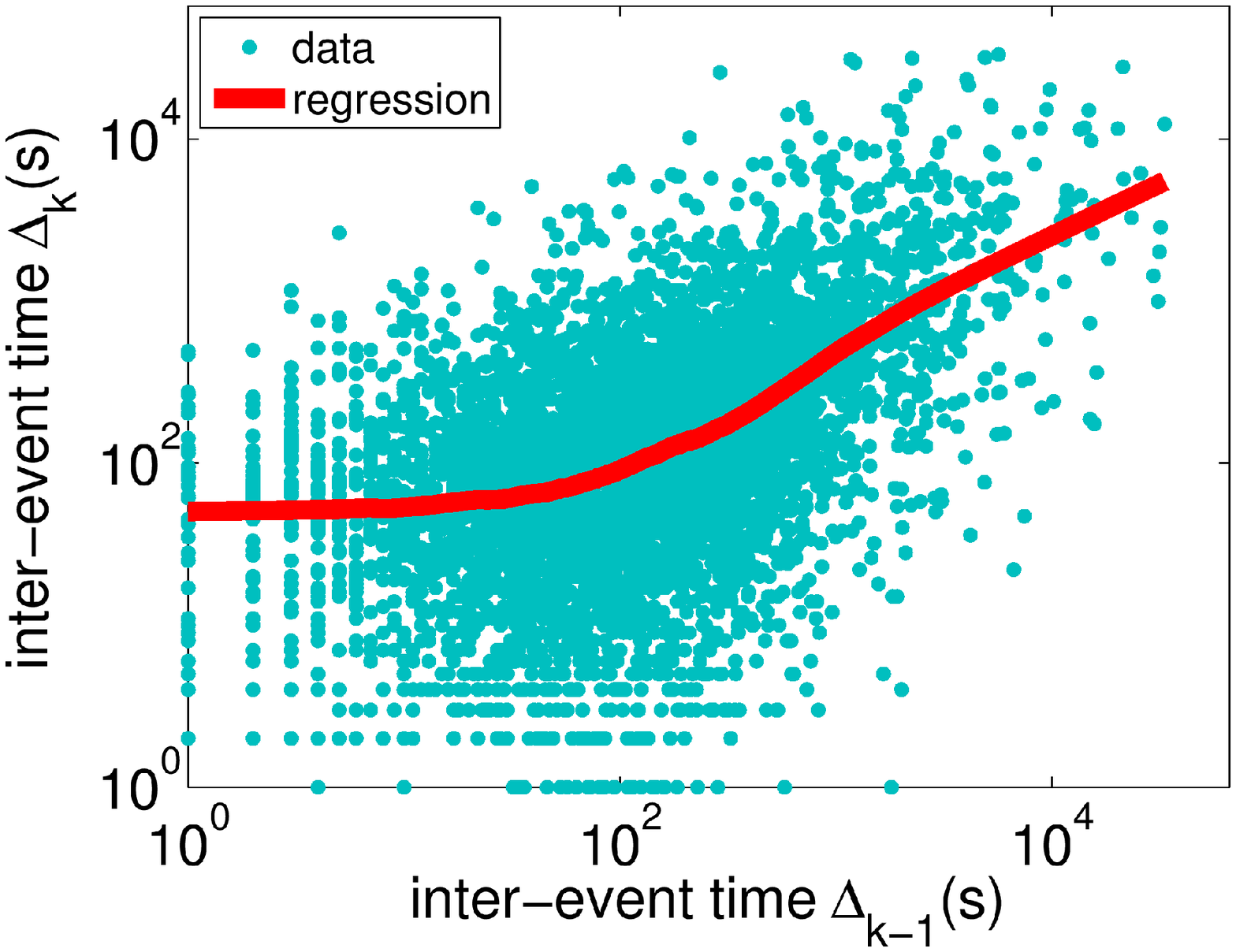}}  
\subfigure[MetaTalk]
  {\includegraphics[width=.23\textwidth]{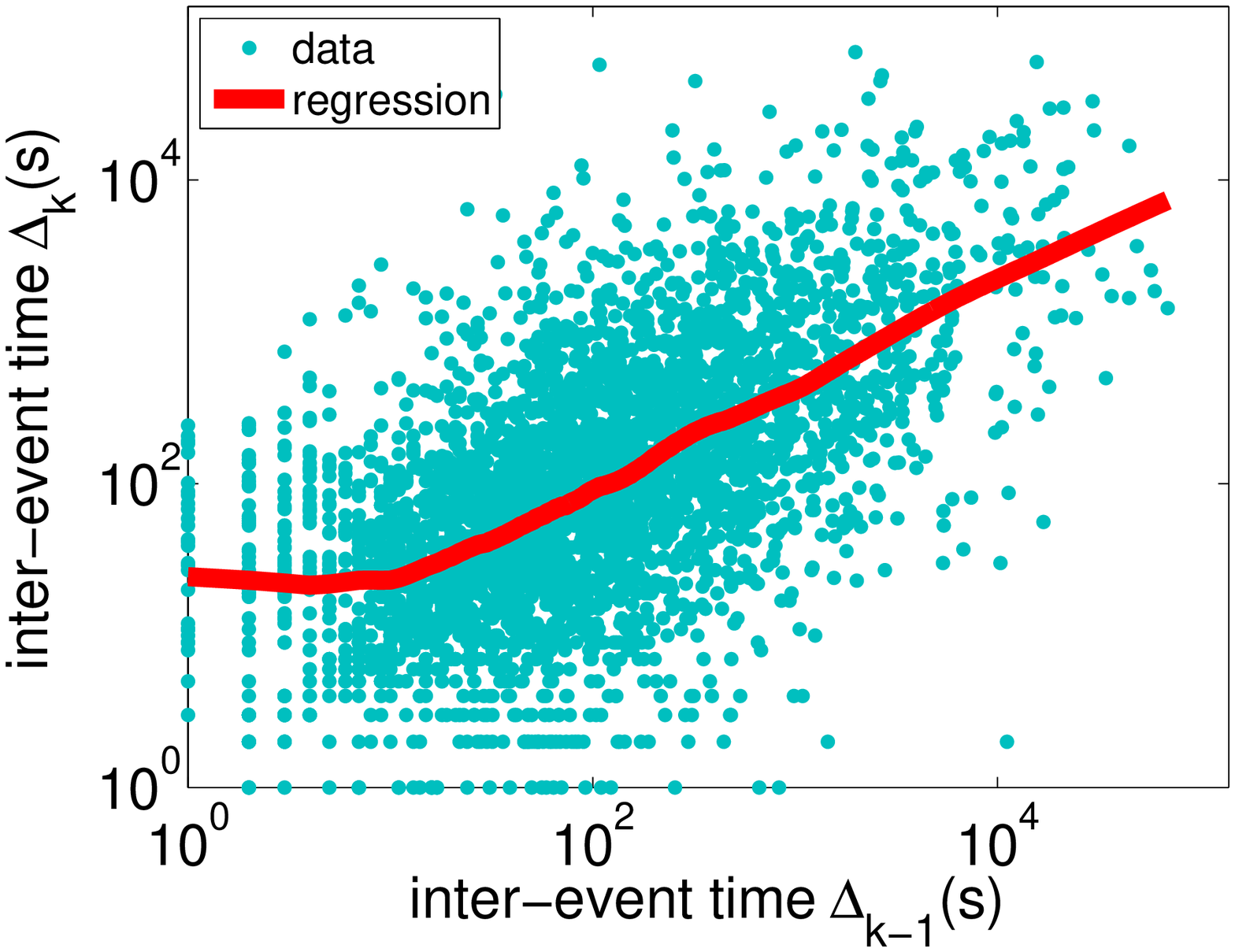}}  
\subfigure[Ask MetaFilter]
  {\includegraphics[width=.23\textwidth]{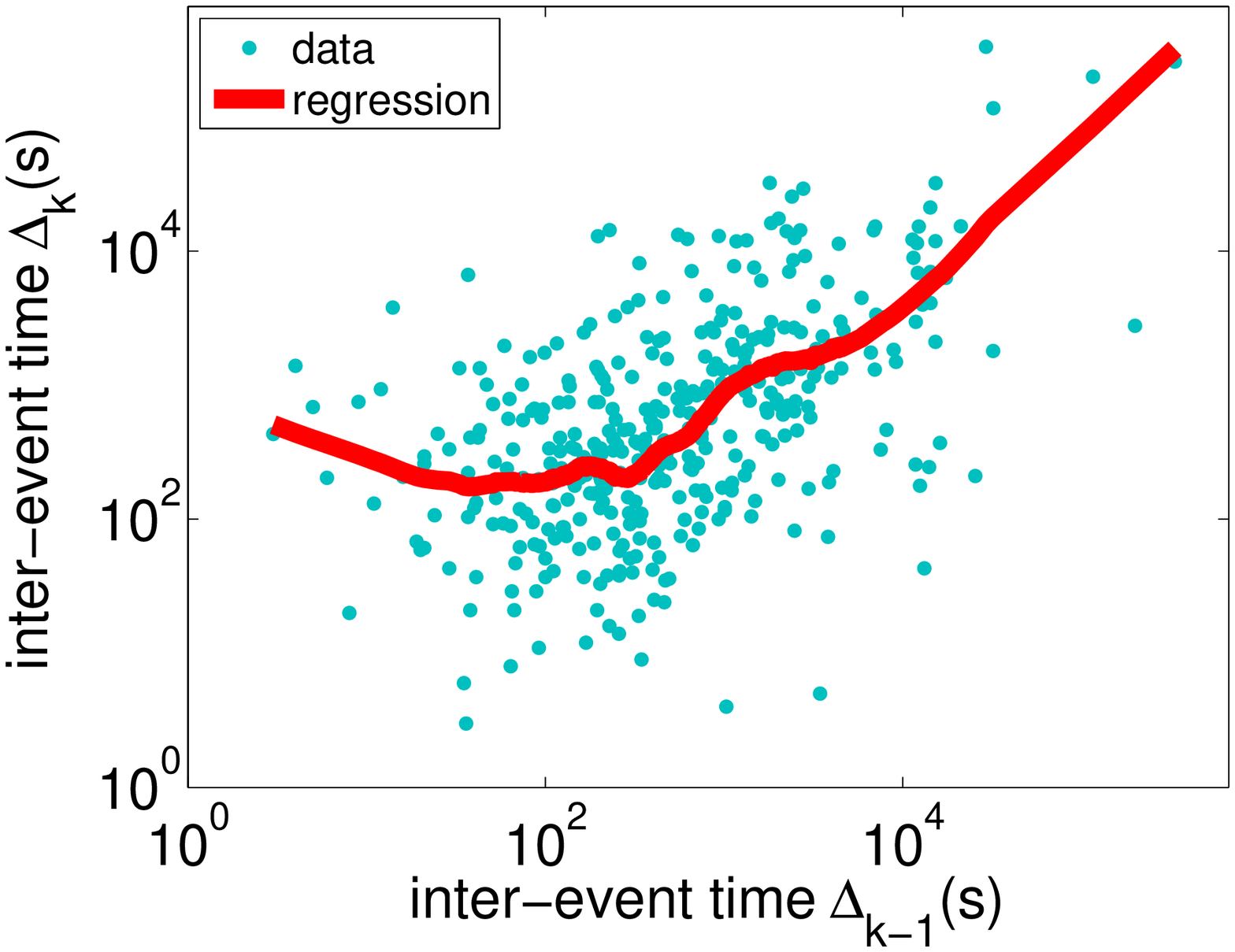}}  
\subfigure[Digg]
  {\includegraphics[width=.23\textwidth]{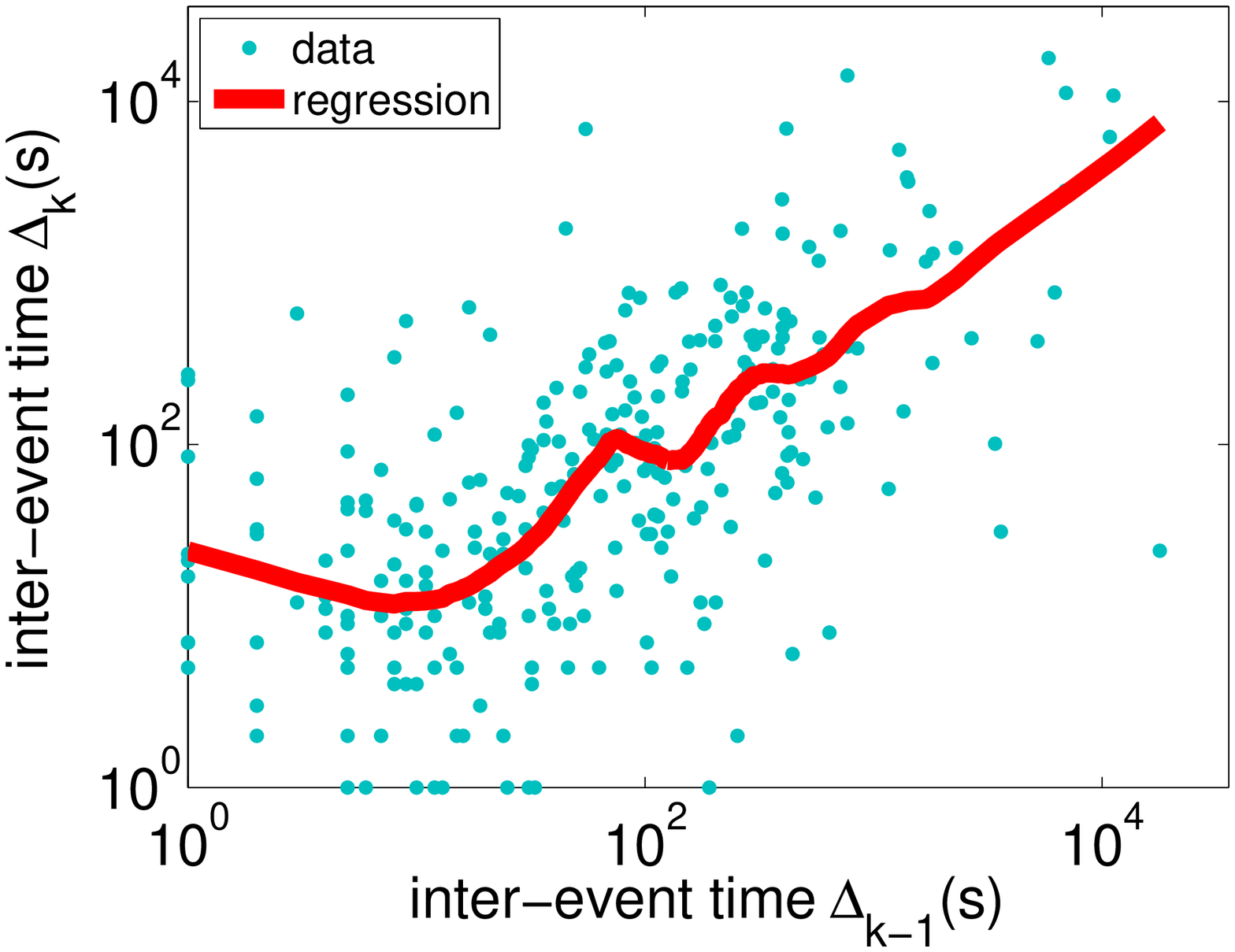}}  
\subfigure[SMS]
  {\includegraphics[width=.23\textwidth]{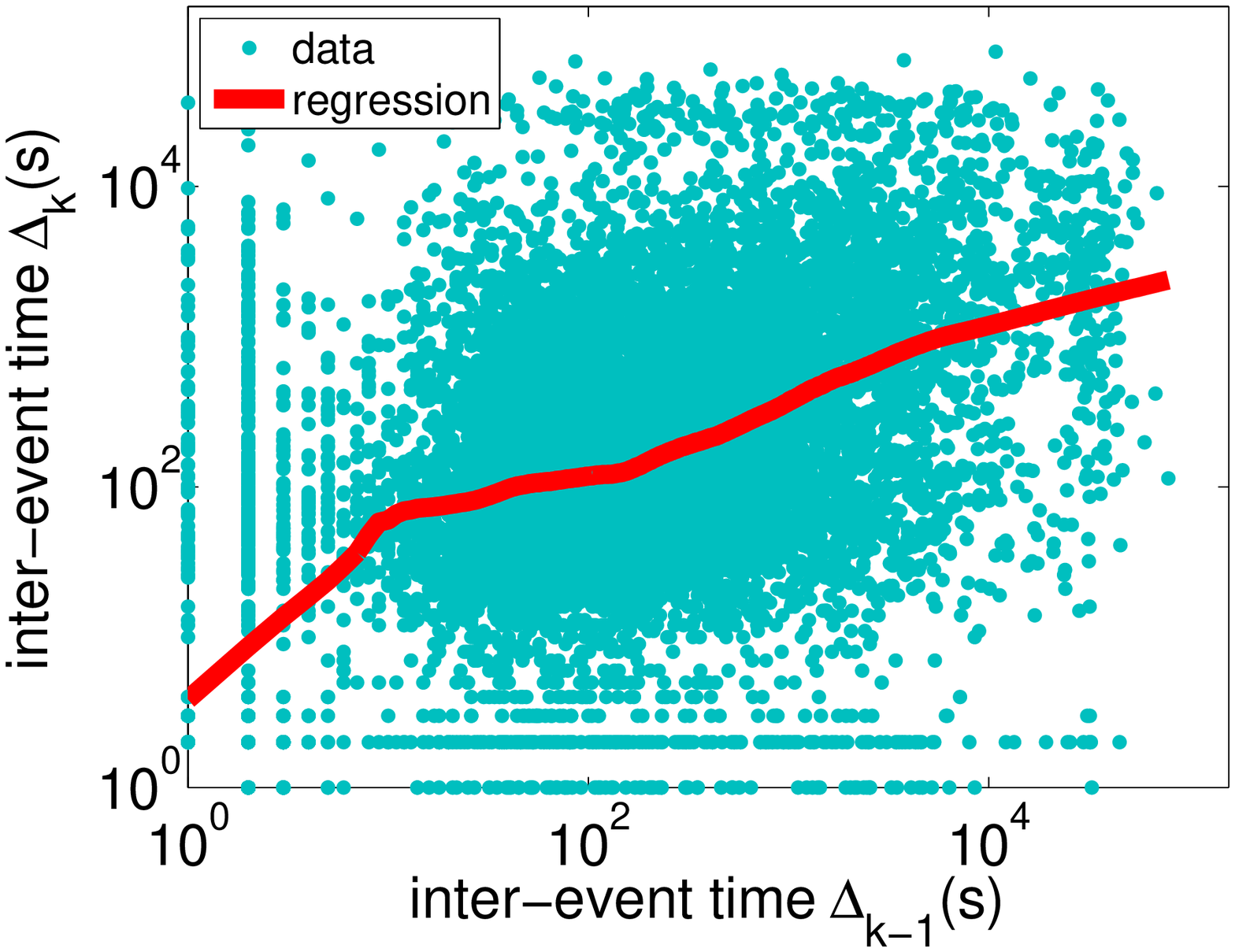}}  
\subfigure[Phone]
  {\includegraphics[width=.23\textwidth]{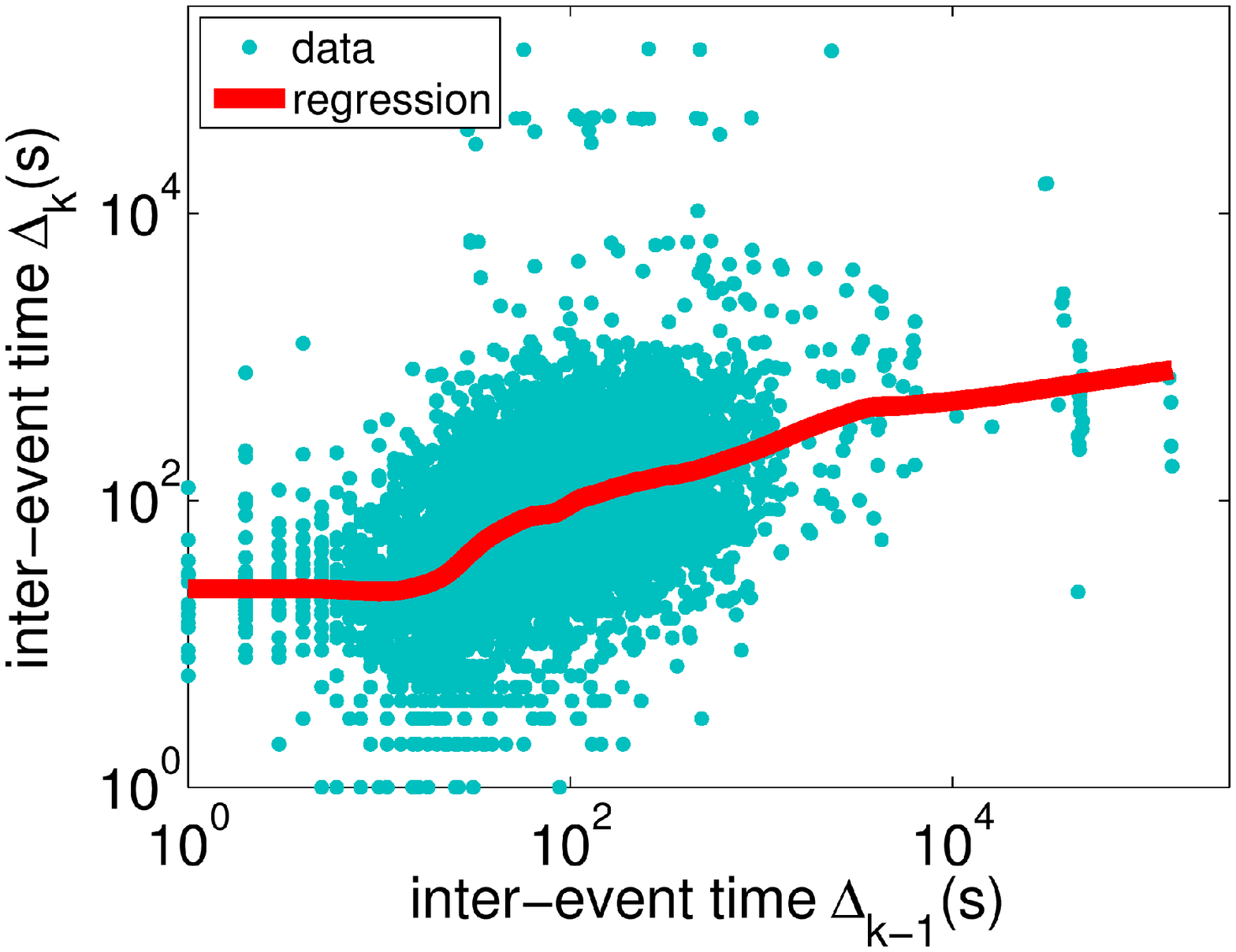}}  
\subfigure[E-mail]
  {\includegraphics[width=.23\textwidth]{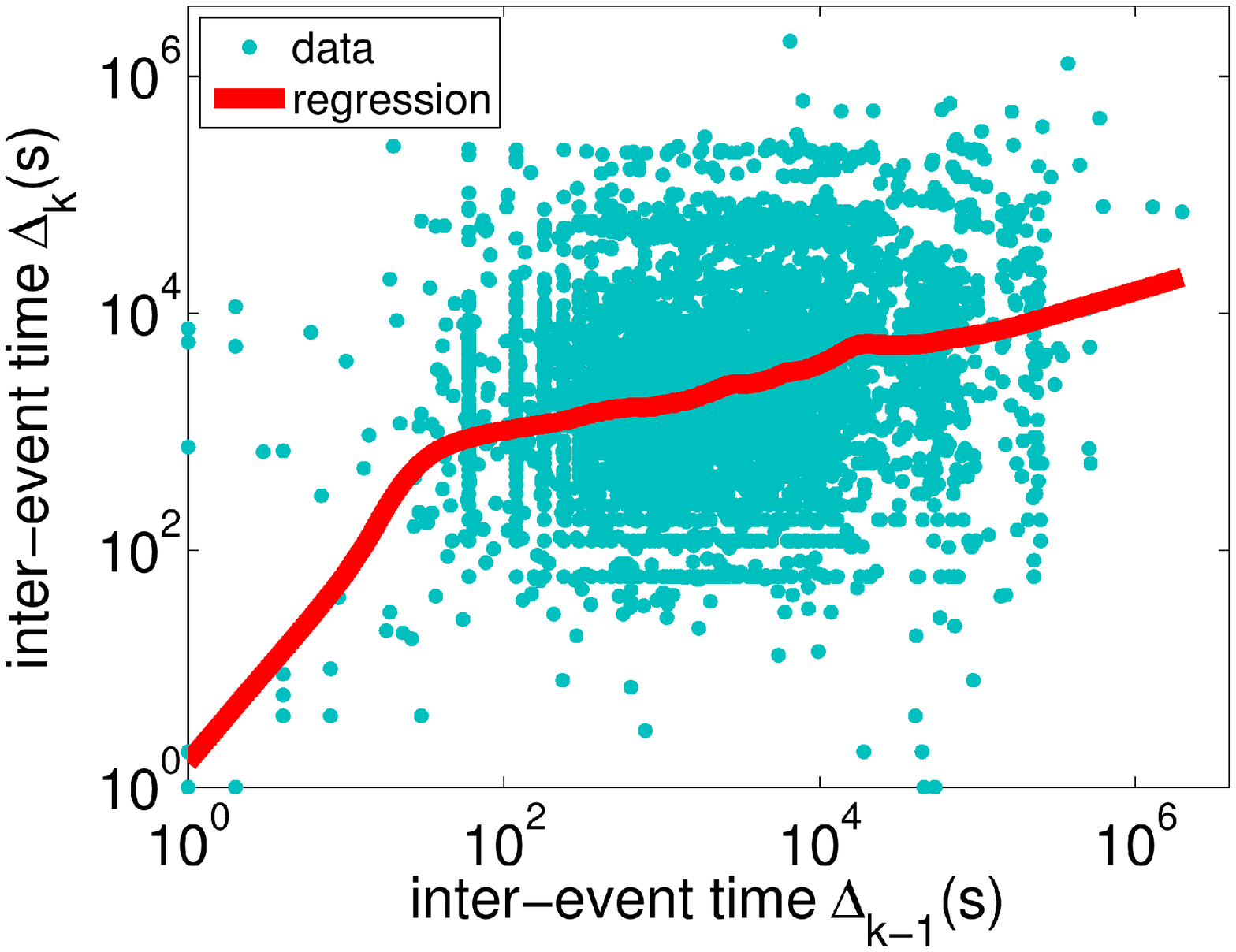}}   
  \caption{\Fallacyiid: dependence between $\Delta_k$ and $\Delta_{k-1}$. Each point represents a pair of consecutive inter-event times $(\Delta_{k-1},\Delta_{k})$ registered for a typical active individual of each dataset. The red line is a regression of the data points using the LOWESS smoother.}
  \label{fig:acf}
\end{figure*}

We formally investigate if two consecutive inter-event times are correlated analyzing the autocorrelation~\cite{box:1994} of all the time series involving the inter-event times $\Delta_k$ of the individuals of our datasets. Autocorrelation refers to the correlation of a time series with its own past and future values. A positive autocorrelation, which is suggested by Figure~\ref{fig:acf}, might be considered a specific form of ``persistence'', i.e., a tendency for a system to remain in the same state from one observation to the next.

We test if all the $\Delta_k$ time series of every individual of our datasets are random or autocorrelated. For this, we define the hypothesis test $H_0$ that a series $S = \{\Delta_0, \Delta_1, ..., \Delta_n\}$ of inter-event times is random. If $S$ is random, then its empirical autocorrelation coefficient $AC_l \approx 0$ for all lags $l>0$, where a lag $l$ is used to compare, in this case, values of $\Delta_k$ and $\Delta_{k-l}$. More formally, if $AC_l$ is within the $95\%$ confidence interval for $S$ to be random, then we accept $H_0$ that $S$ is random. As we show in Figure~\ref{fig:acf1}, we reject the null hypothesis $H_0$ that the inter-event times of the individual of Figure~\ref{fig:superuser} is random, since all $AC_l, 1<l\leq10$ are outside the confidence interval.

\begin{figure}[!hbt]
\centering
{\includegraphics[width=.40\textwidth]{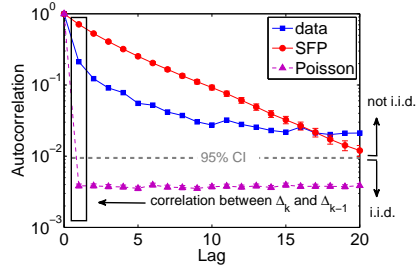}}  
  \caption{The sample autocorrelation for the same individual of Figure~\ref{fig:superuser} and for synthetic data generated by our proposed \dpp model (see Section~\ref{sec:sfp}) and a PP with the same number of communication events and median.}
  \label{fig:acf1}
\end{figure}

Since we are interested only in the case where the lag $l=1$, we propose an alternative hypothesis test $H_1$ that the first-order autocorrelation coefficient $AC_1$ is greater than $0$. If $AC_1$ is greater than the confidence interval for randomness, then we accept $H_1$ that the series is not random, i.e., there is a dependence between $\Delta_k$ and $\Delta_{k-1}$. In Figure~\ref{fig:acf2}, we show the proportion $P(H_1)$ of individuals in our data to which $H_1$ is true grouped by their number of events $n$. As we observe, as the number of communication events $n$ grows and becomes significant, $P(H_1)$ increases rapidly. This strongly suggests that, on the contrary of what happens with the i.i.d. inter-event times distribution generated by the Poisson Process or simply sampling from a log-logistic distribution (LLG-iid), in real data there is a dependence between $\Delta_k$ and $\Delta_{k-1}$. This also agrees with a recent work~\cite{owczarczuk:memory}, which reports that daily series of number of calls made by a customer exhibits strong autocorrelation. Thus, in summary we can state that 
\begin{equation}
E(\Delta_k) \neq E(\Delta_k | \Delta_{k-1})
\end{equation}
or
\
\begin{equation}
E(\Delta_k|\Delta_{k-1}) = f(\Delta_{k-1})
\end{equation}
where $f$ is a function that describes the dependency between $\Delta_k$ and $\Delta_{k-1}$. Moreover, we can state the following universal pattern:

\begin{universal}
There is a significant positive correlation between two consecutive inter-event times.
\end{universal}

\begin{figure}[!hbt]
\centering
\subfigure[First group]
  {\includegraphics[width=.40\textwidth]{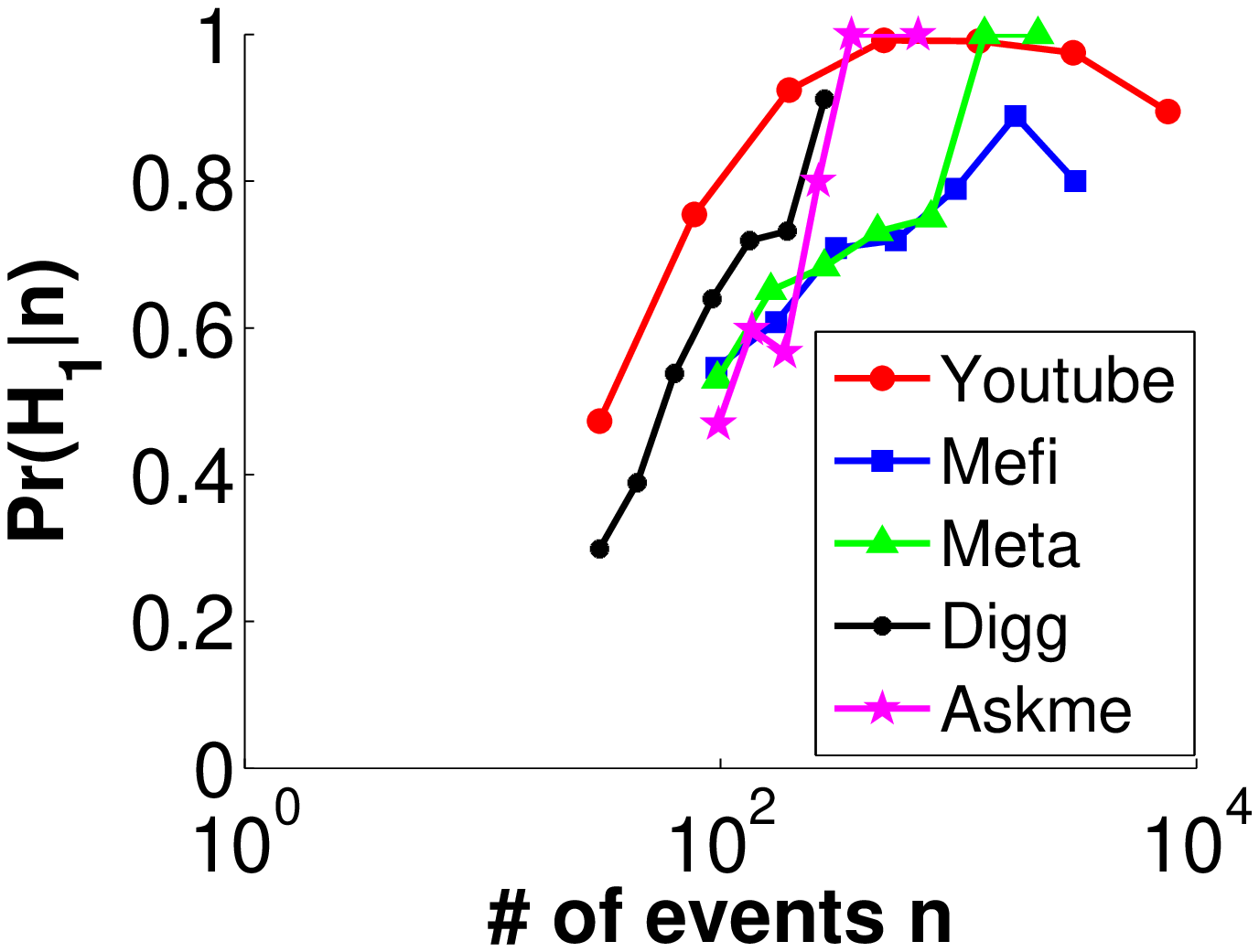}}  
\subfigure[Second group]
  {\includegraphics[width=.40\textwidth]{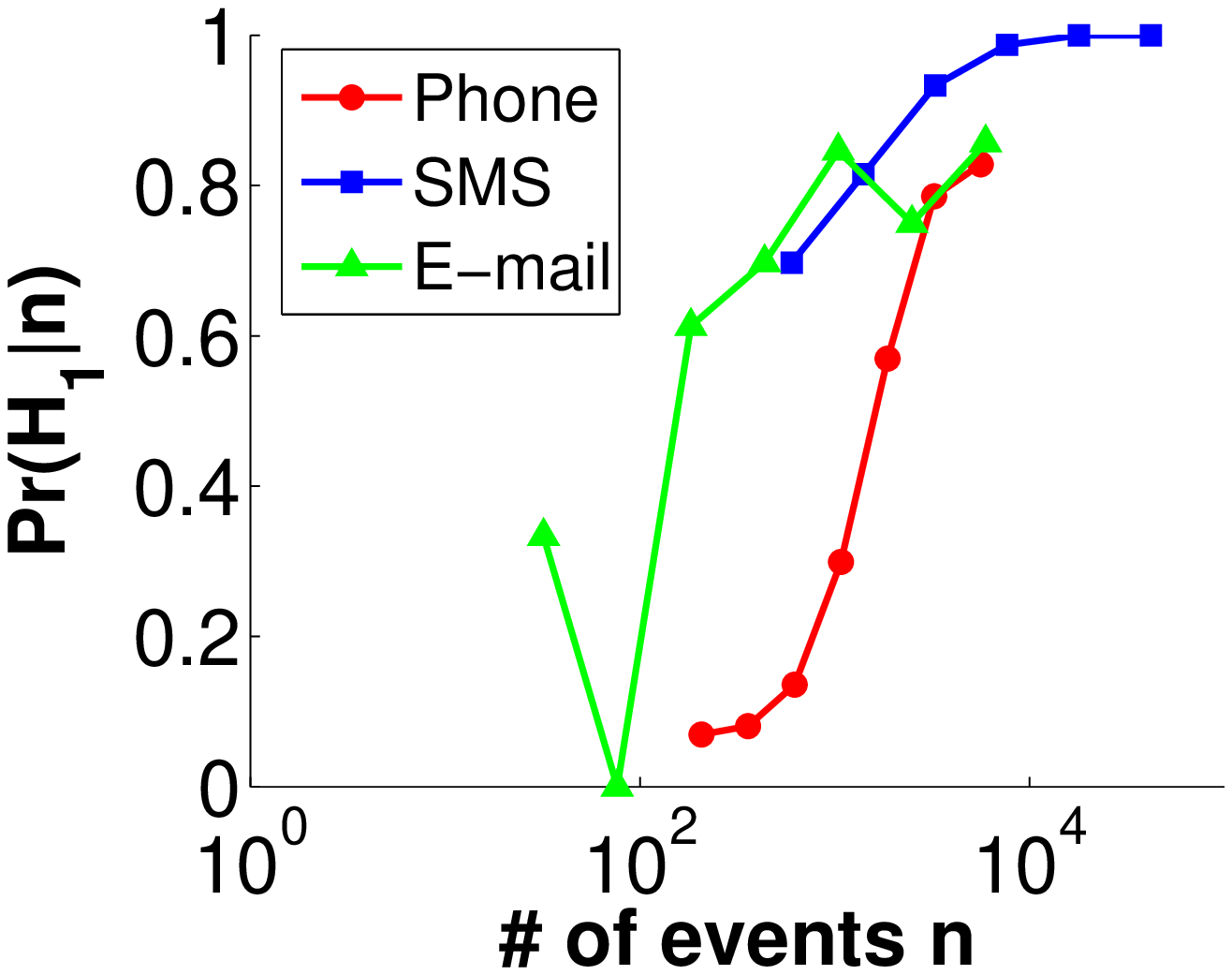}}
\caption{The proportion $P(H_1)$ of individuals in our data to which $H_1$ is true grouped by their number of events. Note that as the number of events grows, the proportion of individuals that have autocorrelated series increases rapidly for the eight datasets.}
  \label{fig:acf2}
\end{figure}

\section{The Self-feeding Process}
\label{sec:sfp}

Given all the above evidence (OR power law; \fallacyiid) and all the previous
evidence (power law tails by Barab\'{a}si; short-term regular behavior as the PP),
the question is whether we can design a generator which will match all these properties?
Our requirements for the ideal generator are the following:
\begin{description}
	\item[R1: Realism -- marginals] \label{req:llg} The model should generate OR power law marginal \ied; 
	\item[R2: Realism -- locally-Poisson:] The model should behave as a Poisson Process within a short window of time; \label{req:poisson}
	\item[R3: Avoid the \fallacyiid] Two consecutive inter-event times should be correlated; \label{req:noniid}
	\item[R4: \Parsimony] It should need only few  parameters, and ideally, just one or two. \label{req:parsimony}
\end{description}

\subsection{Candidate Parameters}
\label{sec:typicalb}

Since the \ied of the majority of individuals is well modeled by an odds ratio power law, which implies a log-logistic distribution, we can characterize their behavior by the two parameters, the slope $\rho$ and the median $\mu$, from the linear relationship in the OR plot. Observe in Figure~\ref{fig:rhoall} the PDF of the slopes $\rho_i$ measured for every individual $i$ of our eight datasets. Except the SMS dataset, the typical $\rho_i$ for the majority of individuals is approximately $1$. This surprising result allows us to state the following universal pattern:

\begin{universal}
The typical slope of the odds ratio power law that best fits the inter-event time distribution of an individual is $1$.
\end{universal}

 Moreover, observe in Figure~\ref{fig:muall} the PDF of the medians $\mu_i$ measured for every individual $i$ of our eight datasets. Observe that, while  the typical $\mu_i$ is around $1$ hour for the first group, for the second group it varies from 3 to 8 minutes. Thus, in Section~\ref{sec:sdpp} we propose a simplified one-parameter model that generates \ieds with slopes $\rho=1$ and varied medians. Then, in Section~\ref{sec:gdpp}, we propose a generalized two-parameter model that generates \ieds with varied slopes and medians.
%
\begin{figure}[!hbt]
\centering
\subfigure[First group]
  {\includegraphics[width=.40\textwidth]{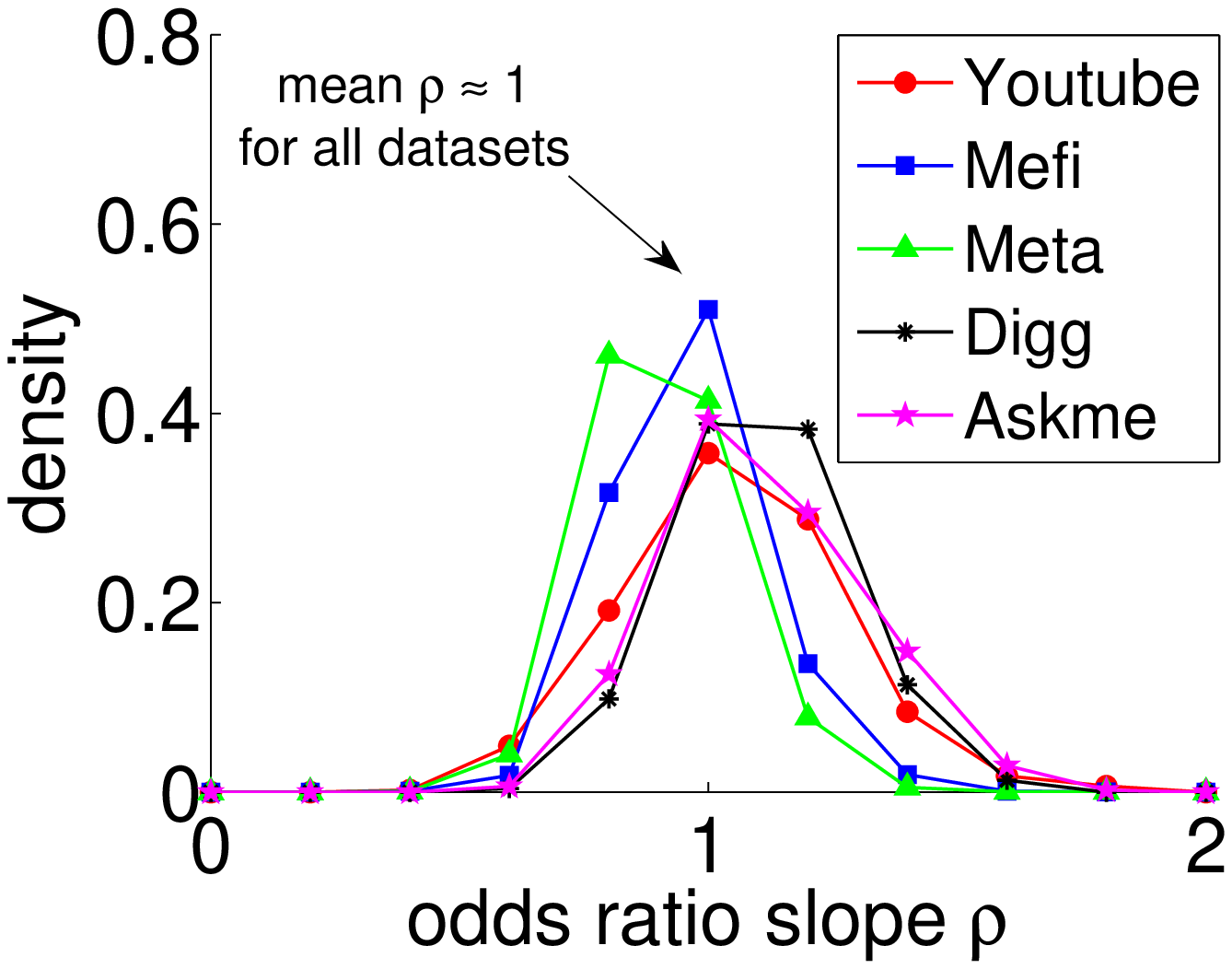}}  
\subfigure[Second group]
  {\includegraphics[width=.40\textwidth]{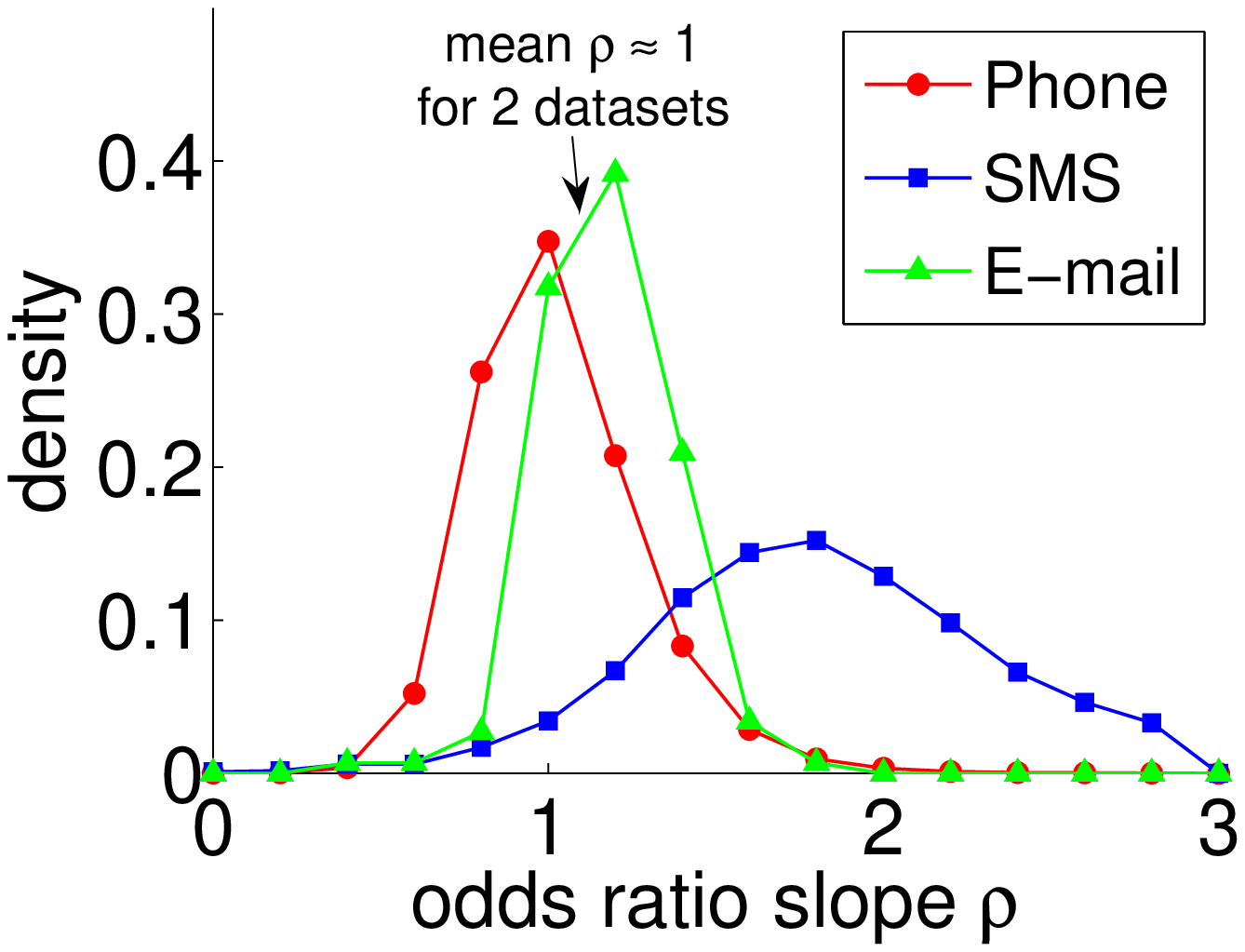}}
  \caption{The PDF of the slopes $\rho_i$ measured for every user $u_i$ of our eight datasets. Except the SMS dataset,  the typical $\rho_i$ for the majority of individuals is approximately $1$.}
  \label{fig:rhoall}
\end{figure}

\begin{figure}[!hbt]
\centering
\subfigure[First group]
  {\includegraphics[width=.40\textwidth]{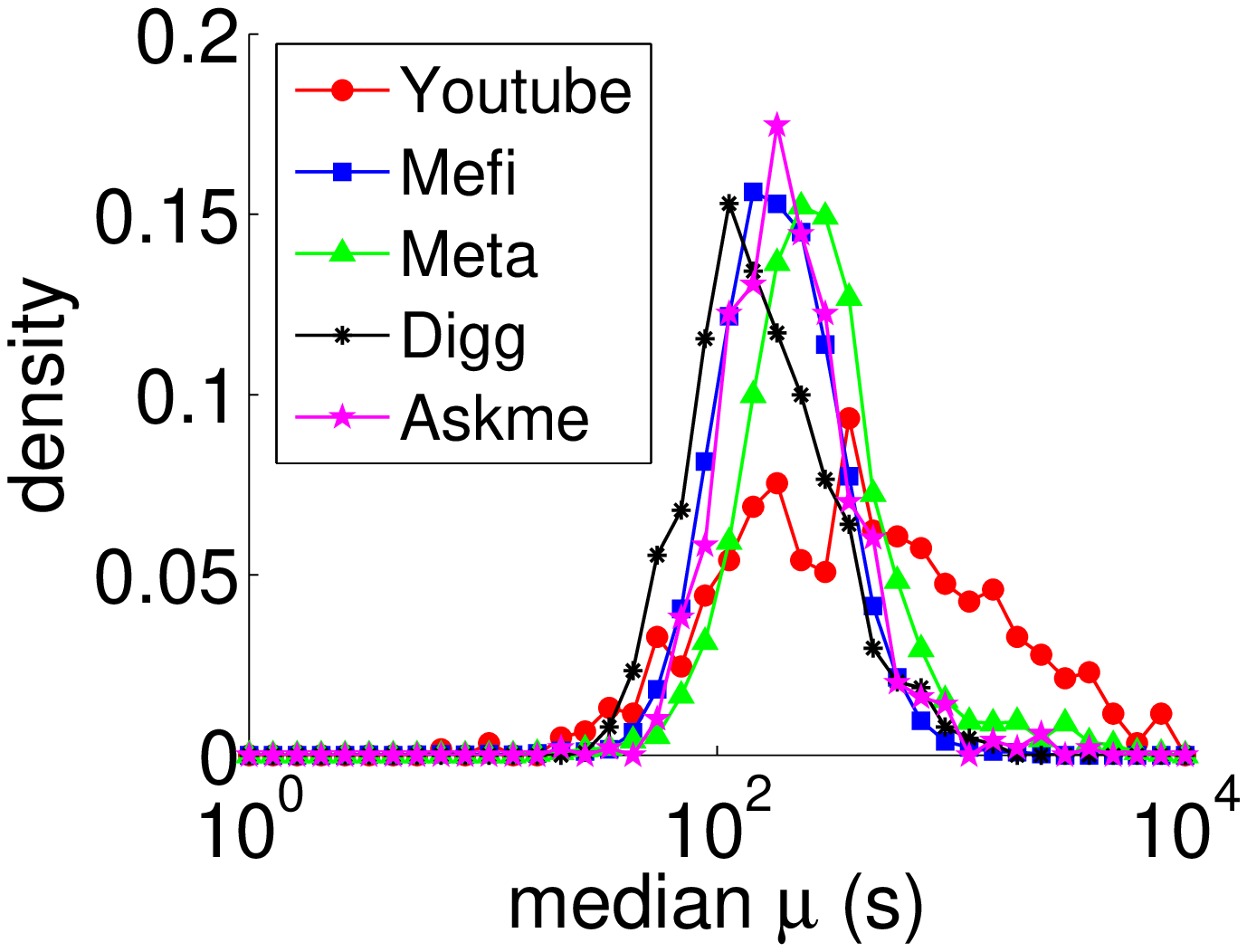}}   
\subfigure[Second group] 
  {\includegraphics[width=.40\textwidth]{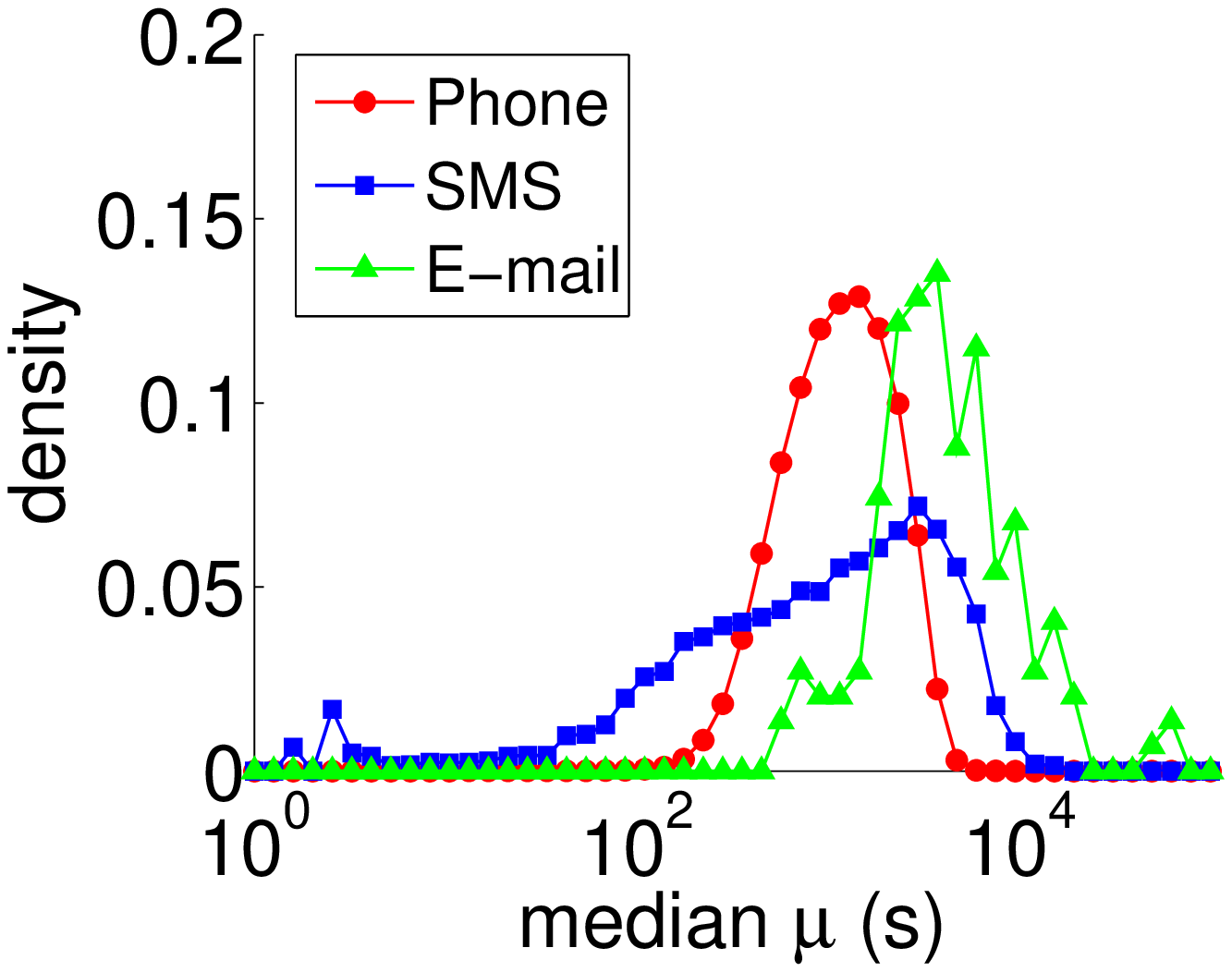}}
  \caption{The PDF of the medians $\mu_i$ measured for every user of our eight datasets. Observe that, while  the typical $\mu_i$ is around 3 and 8 minutes for the first group, and around $1$ hour for the second group.}
  \label{fig:muall}
\end{figure}

\subsection{The Simplified \dpp Model}
\label{sec:sdpp}

At a high level, our proposal is  that the next inter-arrival time will be an exponential random variable, with rate that {\em depends on the previous} inter-arrival time. It is subtle, but in this way our generator behaves like Poisson in the short term, gives power-law tails in the long term, generates OR power law marginals and is extremely parsimonious: just one parameter, the median $\mu$ of the \ied.
We call this model the \textit{Self-Feeding Process} (\dppzero{}).

We propose the  generator as follows
\begin{model}
Self-Feeding Process \dppzero($\mu$). 

//$\mu$ is the desired median of the marginal PDF
\begin{equation*}
\boxed{
\begin{array}{rcl}
\Delta_1 &\leftarrow& \mu \\
\Delta_k &\leftarrow & \mbox{Exponential  }( \mbox{mean~} \beta = \Delta_{k-1} + \mu/e) \\
\end{array}}
\end{equation*}
\label{eq:dppzero}
\end {model}  
where $\mu$ is the only parameter of the model, being the desired median of the \ied. The part $\mu/e$ must be greater than $0$ to avoid $\Delta_k$ to converge to $0$ and has to be divided by the Euler's number \textit{e} to make the median of the generated \ied around the target median $\mu$ (more details in the Appendix~\ref{sec:parameters}). This type of model is not new in the literature~\cite{wold:1948,cox:1955} but they have not been extensively studied, perhaps due to the lack of empirical data fitting the implied distribution.

In Figures~\ref{fig:model1}-a and~\ref{fig:model1}-b we compare, respectively, the histogram and the OR of the inter-event times generated by the \dppzero model, all values rounded up, with the inter-event times of the individual of Figure~\ref{fig:superuser}. Notice that the distributions are very similar
and both are well fitted by a log-logistic distribution, which looks like a hyperbola, thus addressing both the power-law tail, as well as the ``\topconcavity'' that real data exhibits.

\begin{figure}[hbtp]
\centering
\subfigure[Histogram]
  {\includegraphics[width=.40\textwidth]{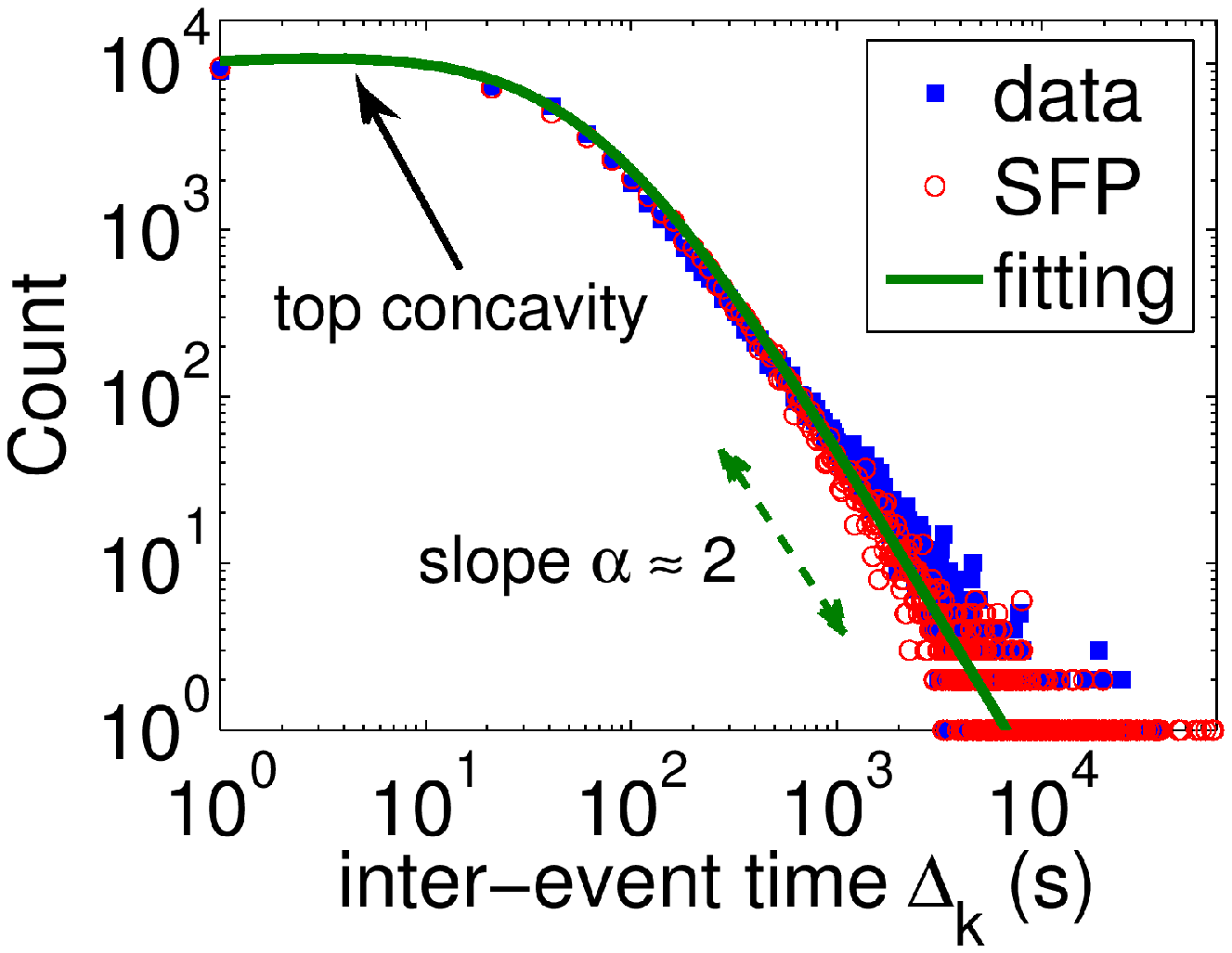}}
\subfigure[Odds ratio]
  {\includegraphics[width=.40\textwidth]{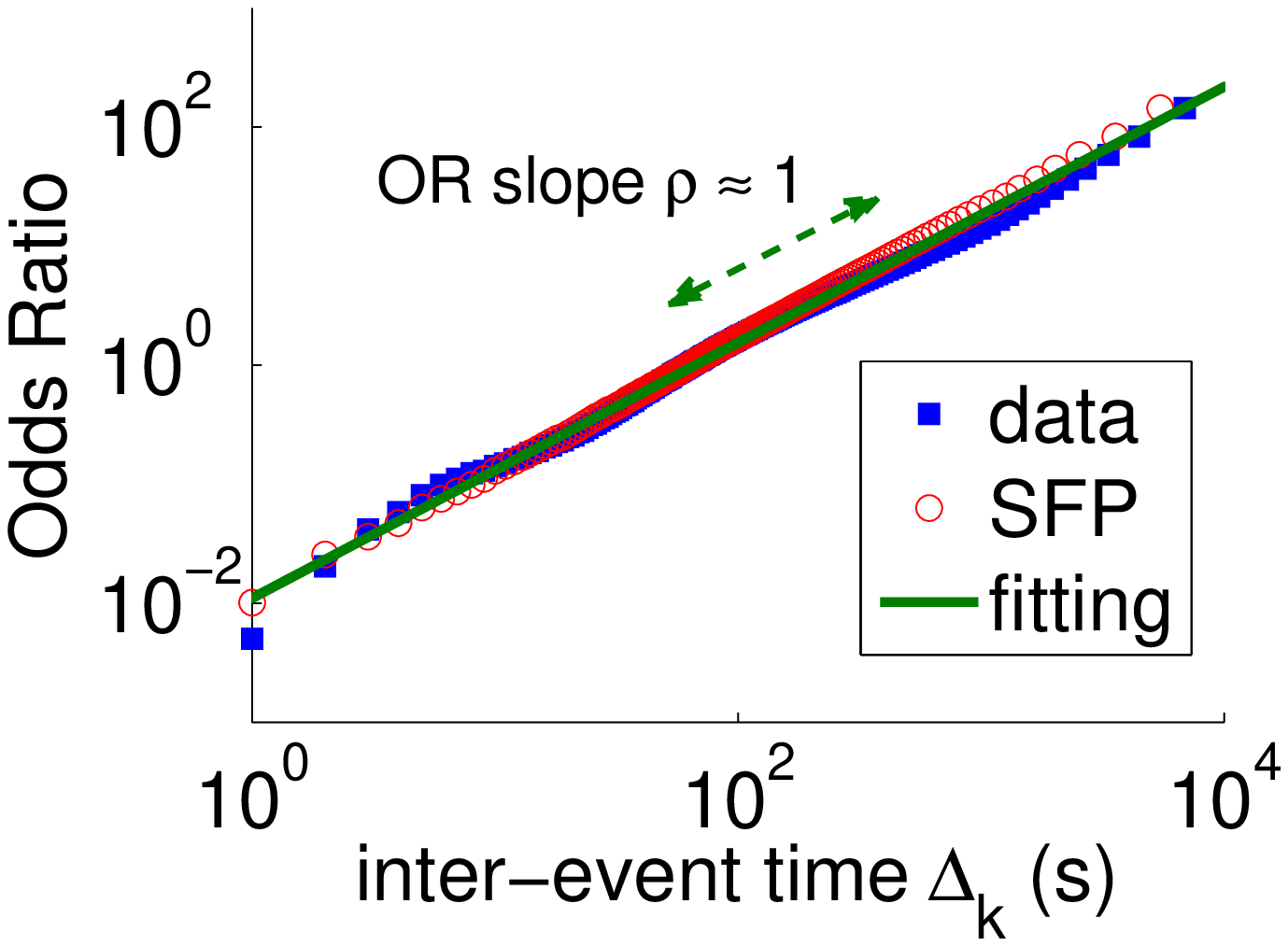}}  
  \caption{Comparison of the marginal distribution of the inter-event times generated by the \dppzero{} model with the inter-event times of the most active user of Figure~\ref{fig:pdfUser}. Observe that both the histogram (a) and the OR (b) are almost identical.}
  \label{fig:model1}
\end{figure}

The \dpp{} model naturally generates an odds ratio power law for IED with slope $\rho=1$, which is the slope that characterizes the majority of the users of our datasets (see Figure~\ref{fig:rhoall}). To the best of our knowledge, this is the first work that studies the \ied{} of human communications using such a varied, modern and large collection of data. Despite the fact that the means of communications are intrinsically different, having their own idiosyncrasies, we have observed that the \ied{} of most individuals of these systems have the same characteristics, i.e., they follow an odds ratio power law behavior. Moreover, when the OR slope $\rho=1$, the power law exponent of the PDF is $\alpha = -2$ (see the Appendix~\ref{sec:powerlaw} for details). This is the same \ied slope $\alpha$ reported in~\cite{hidalgo:2006,vazques:2006} as a result of fluctuations in the execution rate and in particular periodic changes. It has been argued that seasonality can only robustly give rise to heavy-tailed IEDs with exponent $\alpha=2$. 

\subsection{The Generalized \dpp Model}
\label{sec:gdpp}

In Figure~\ref{fig:rhoall} we showed the slopes $\rho$ of the OR fitting for the \ied{}s of all individuals of our datasets. It is fascinating that the typical $\rho_i$ for the individuals of seven of our datasets is approximately $1$, the same slope generated by the simplified \dppzero{} model. 
Several individuals though, mainly from the SMS dataset,
have a much higher value of $\rho$, close to $\rho \approx 2$.
To accommodate that and all the variance seen in the data, we introduce our Generalized \dpp model,
which needs just one parameter more, $\rho$.
Thus, we have:

\begin{model}
Generalized Self-Feeding Process $\dpp(\mu,\rho)$.
\begin{equation*}
\boxed{
\begin{array}{rcl}
       \delta_1 &\leftarrow& \mu \\
       \delta_t &\leftarrow& \mbox{Exponential }(\mbox{mean:~} \beta = \delta_{t-1} + \mu^\rho/e) \\
       \Delta_k &\leftarrow& \delta_t^{1/\rho}.
\end{array}
}
\end{equation*}
\label{eq:dpp}
\end {model}

Note the auxiliary variable $\delta_t$, 
which stores the inter-event times without the influence of $\rho$. For more details about the \dpp parameters, please see Appendix~\ref{sec:parameters}.

\section{The Unifying Power of the \dpp}
\label{sec:unifying}

In this section, we emphasize the unifying power of the \dpp. Several works~\cite{Karagiannis:2004,malmgren:2008,malmgren:2009,malmgren:2009b,kuczura:1973,kleinberg:2002} claim that in the short term, real data behave as regular as a PP. Our model also captures that, since successive inter-event times are exponentially distributed, with similar (but not identical) rates. Thus, one of the major contributions of this work is the unification of the two seemingly-conflicting viewpoints we mentioned earlier. The proposed \dpp model unifies both theories by generating Poisson-like traffic in the short term, with smoothly varying rate, like the second viewpoint, and also generates a power-law tail distribution (see the Appendix~\ref{sec:powerlaw}), even matching  the \topconcavity that power laws can not match, like the first modern approach of Barab\'{a}si~\cite{barabasi:2005}.



In Figure~\ref{fig:timeseries}, we explicitly show the \dpp's unifying power. We compare synthetic data generated by the \dpp{} model using the same odds ratio slope $\rho$, median $\mu$ and number of events of the user of Figure~\ref{fig:superuser}-a with the real data from this user. Notice the bursts of activity and also the long periods of inactivity, in the first two columns of Figure~\ref{fig:timeseries}. Also notice that both synthetic and real traffic significantly deviate from Poisson (sloping lines in Figures~\ref{fig:timeseries}-b and~\ref{fig:timeseries}-f) but are similar between themselves. However, in the short term, both real and synthetic data behave like Poisson, being practically on top of the black dashed lines of Figures~\ref{fig:timeseries}-d and~\ref{fig:timeseries}-h. 


\begin{figure*}
\centering
\begin{tabular}{|cc|cc|}
\hline
\multicolumn{2}{|c|}{\textbf{Long term behavior}} & \multicolumn{2}{|c|}{\textbf{Short term behavior}}\\ \hline
\multicolumn{4}{|c|}{\textbf{Real data}} \\ \hline
\subfigure[]
  {\includegraphics[width=.23\textwidth]{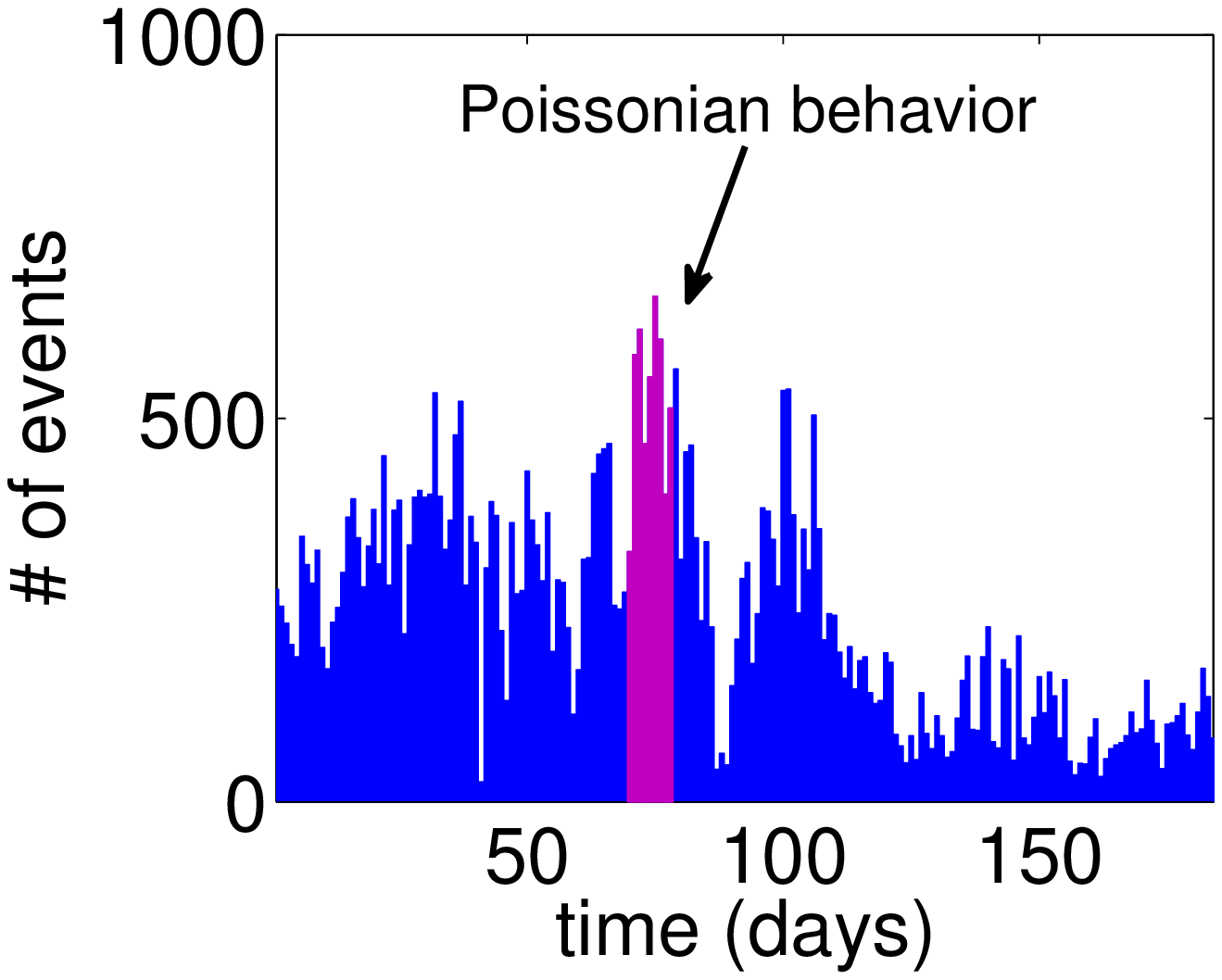}} &
\subfigure[]
  {\includegraphics[width=.23\textwidth]{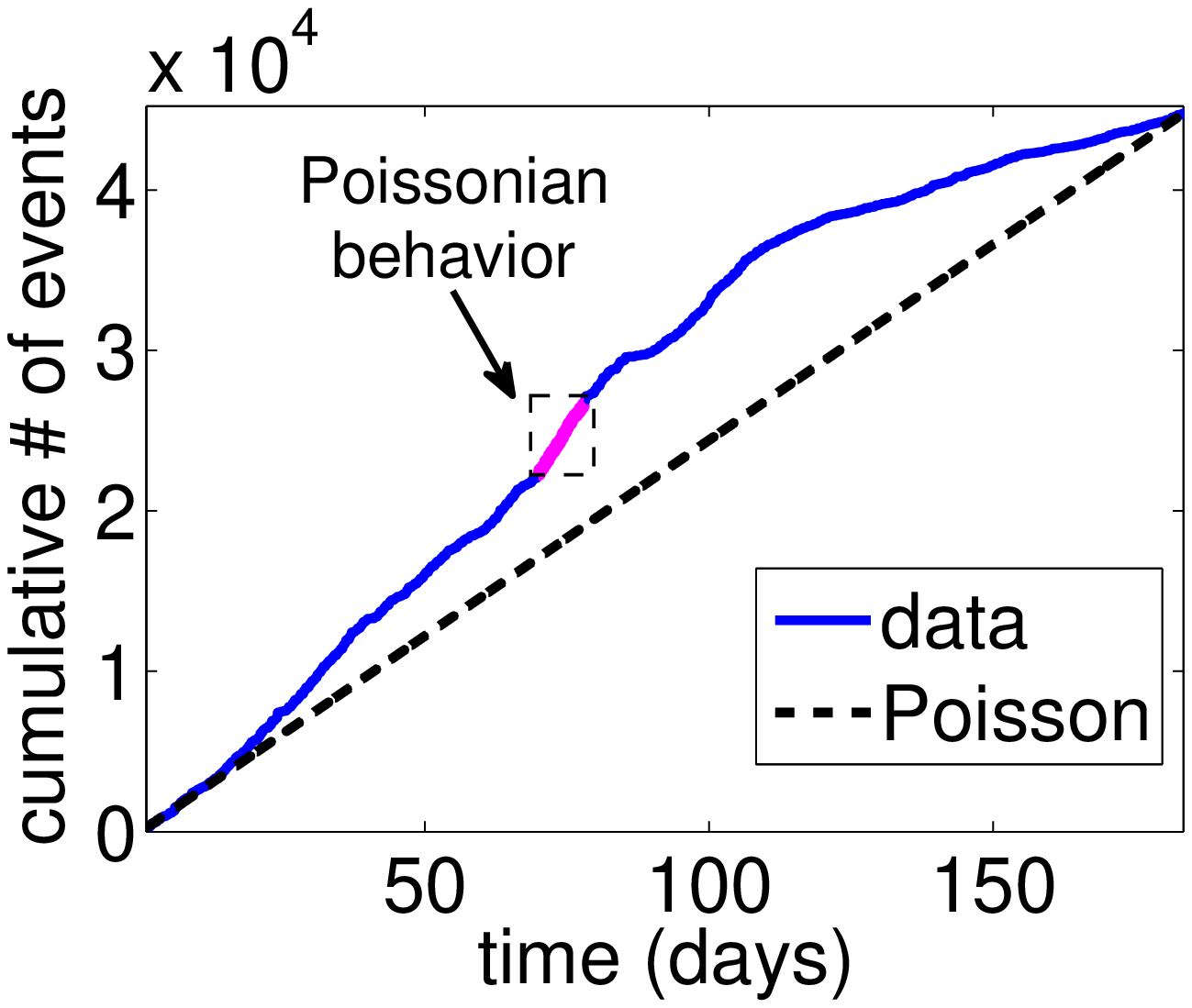}} &
\subfigure[]
  {\includegraphics[width=.23\textwidth]{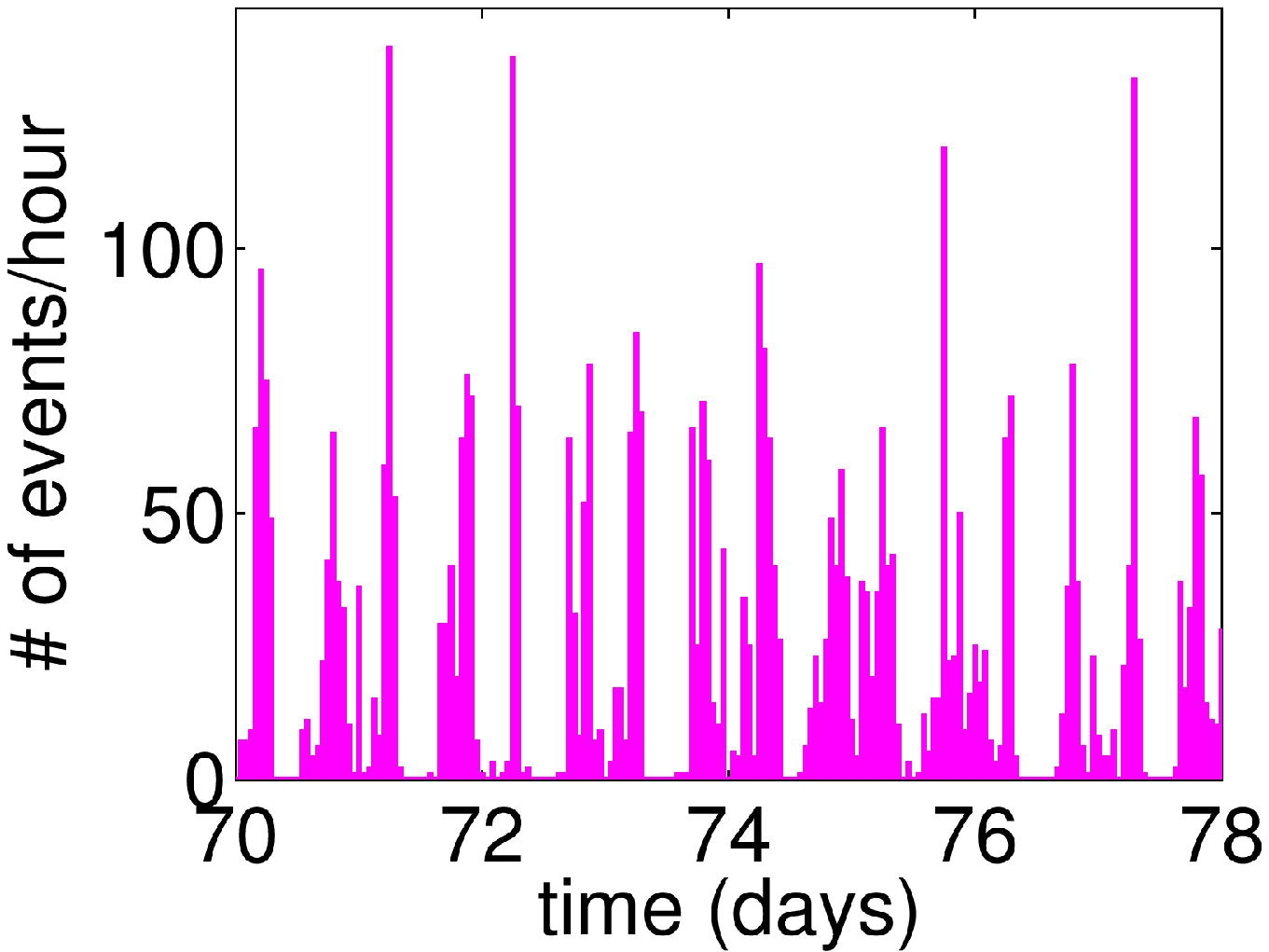}} &
\subfigure[]
  {\includegraphics[width=.23\textwidth]{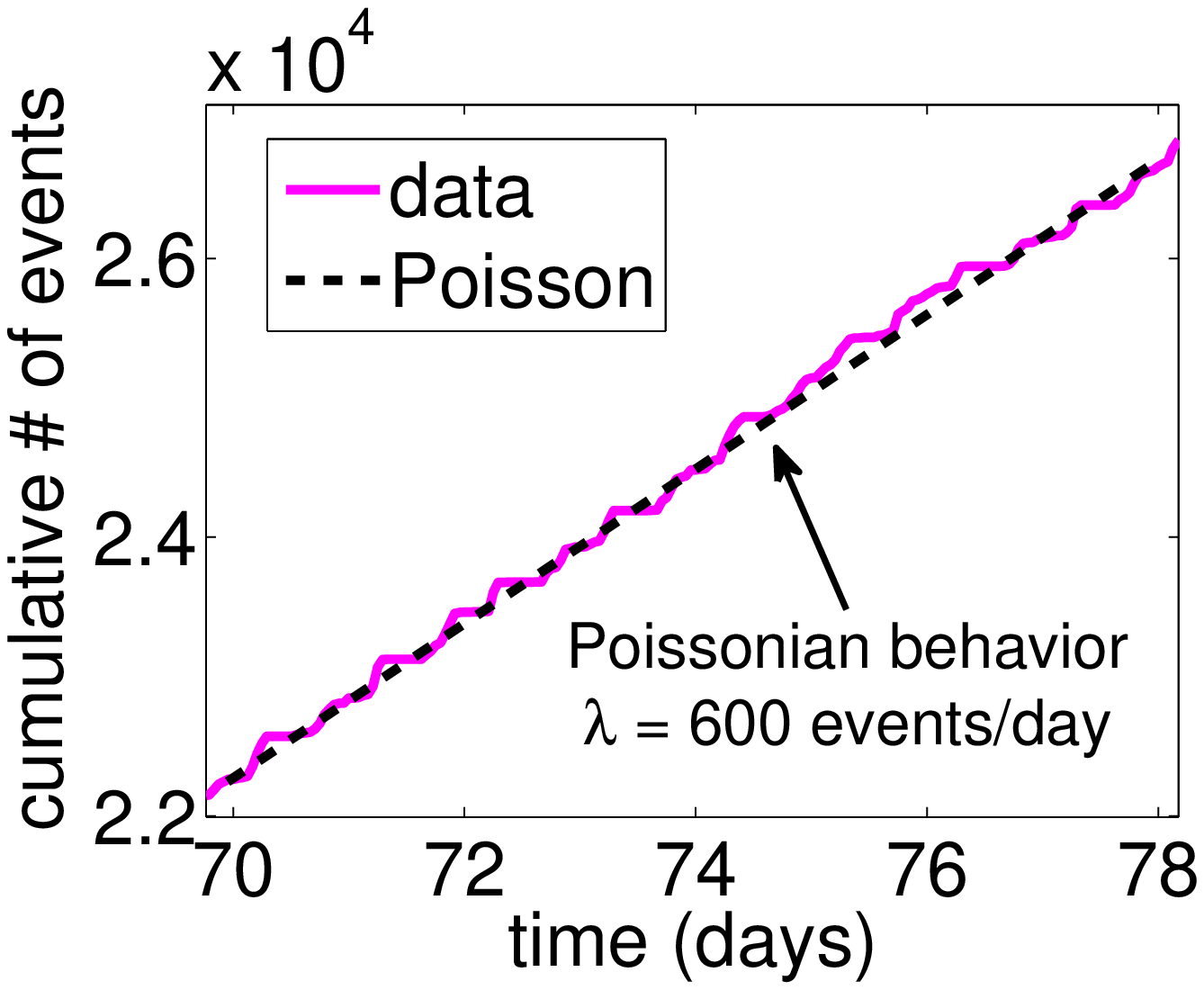}} \\ \hline
\multicolumn{4}{|c|}{\textbf{SFP data}} \\ \hline
\subfigure[]
  {\includegraphics[width=.23\textwidth]{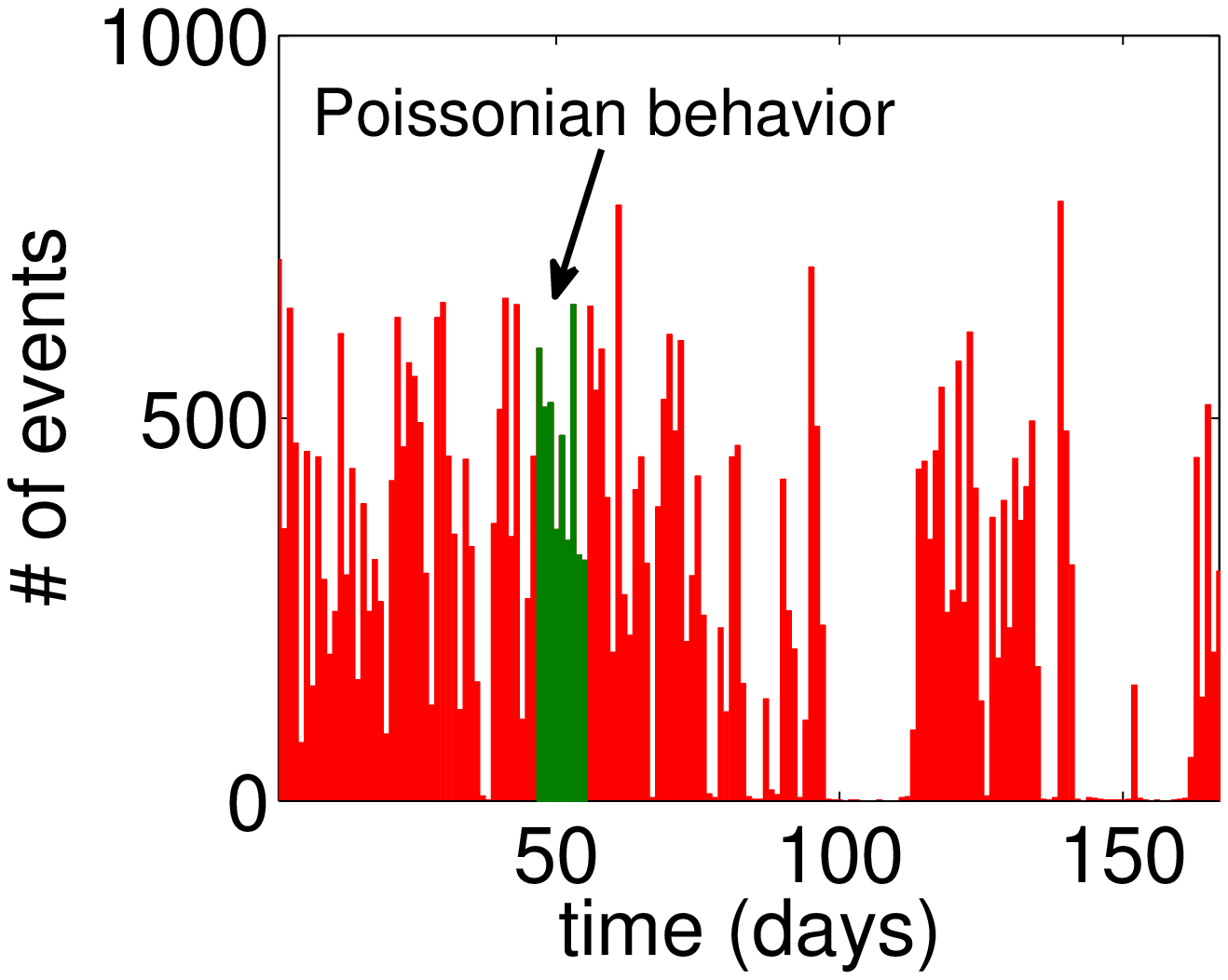}} &
\subfigure[]
  {\includegraphics[width=.23\textwidth]{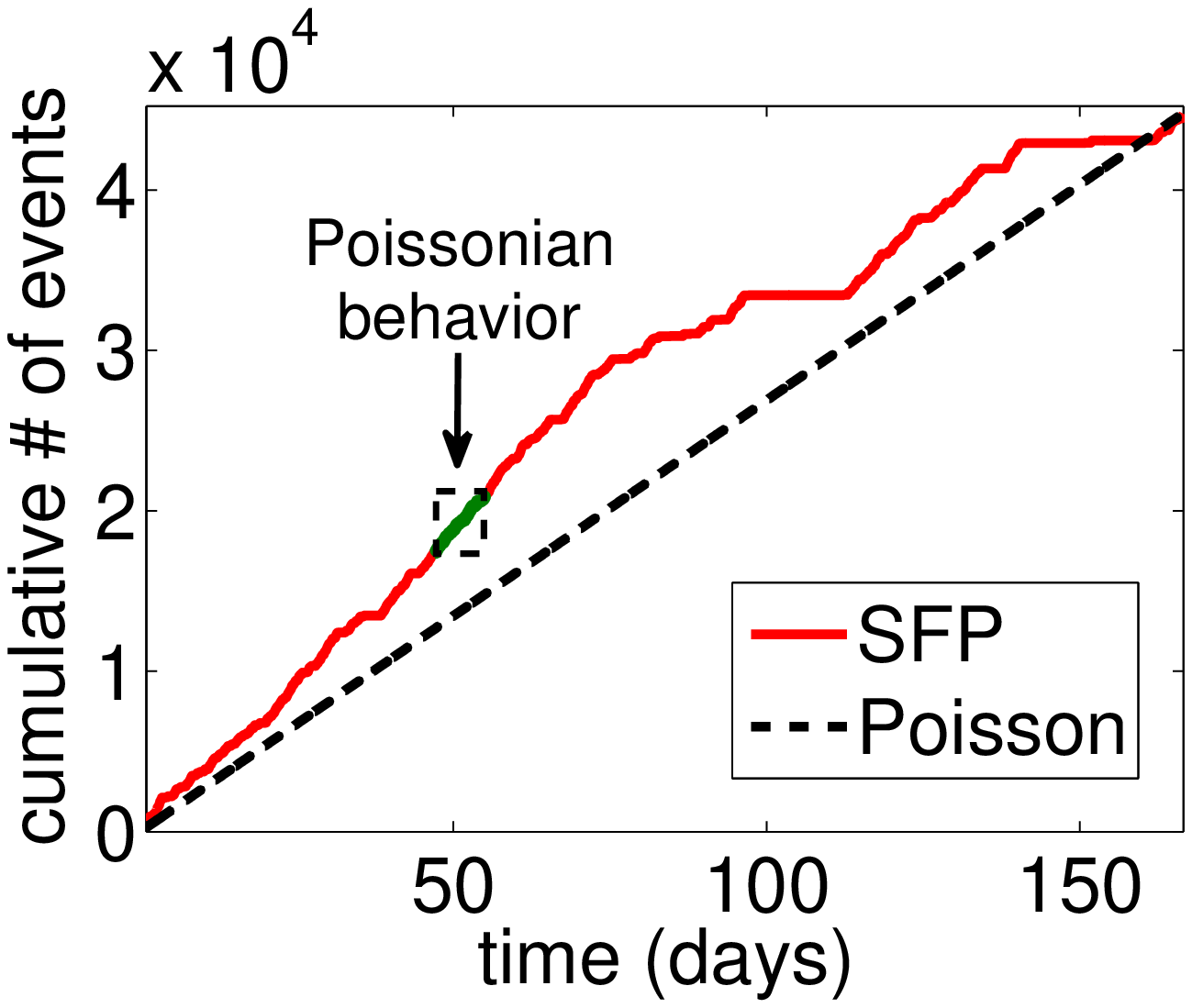}} &
\subfigure[]
  {\includegraphics[width=.23\textwidth]{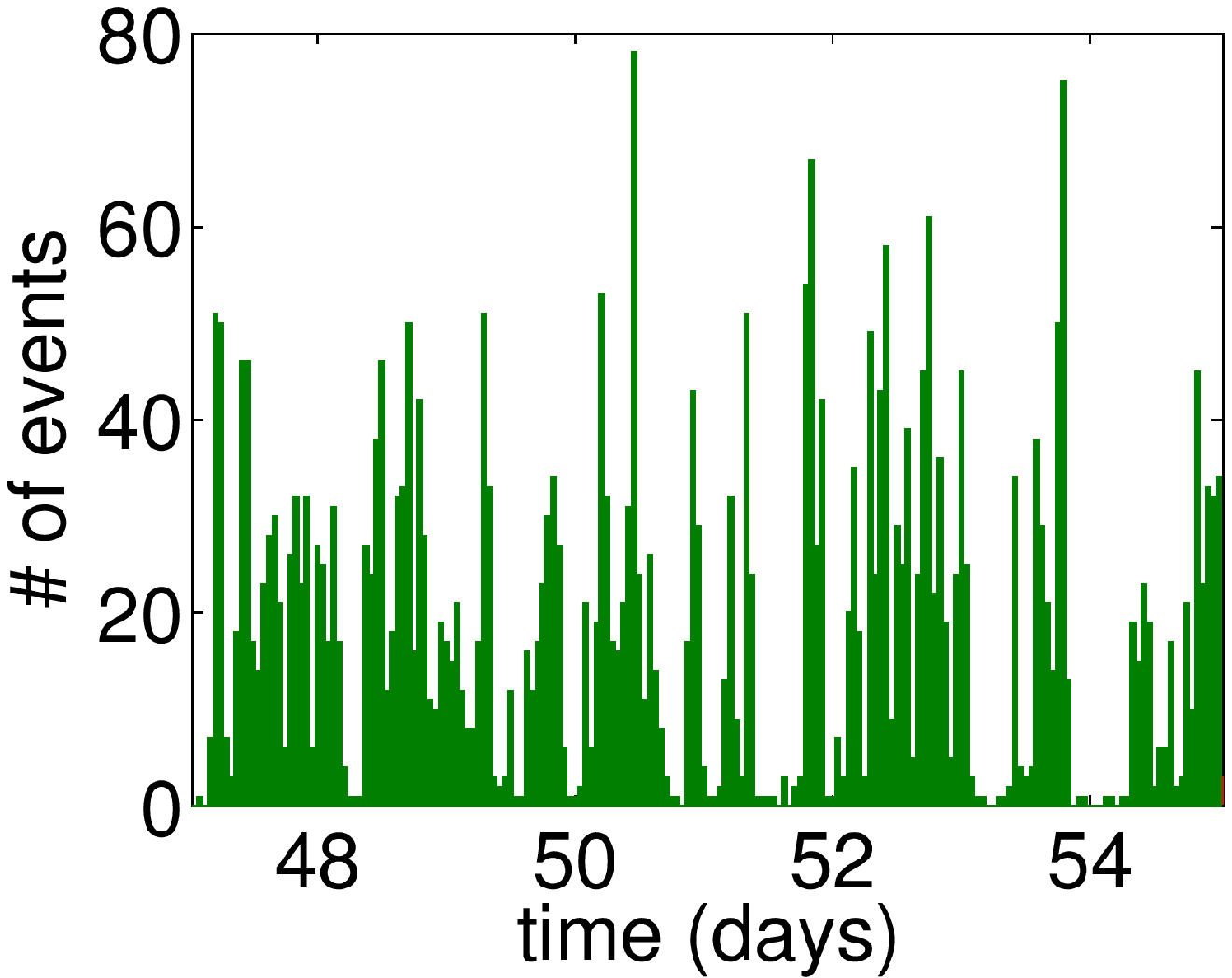}} &
\subfigure[]
  {\includegraphics[width=.23\textwidth]{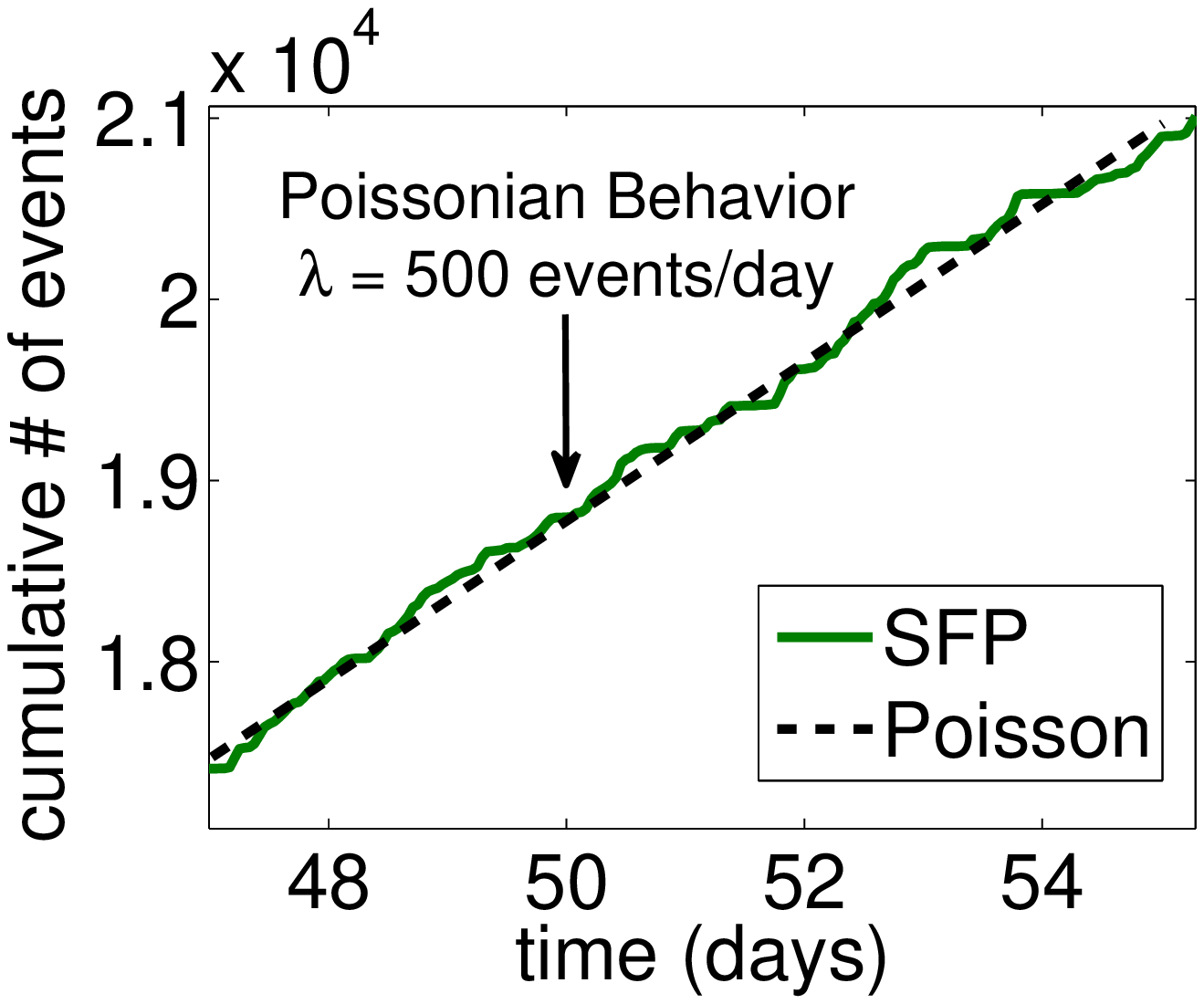}} \\ \hline
  \end{tabular}
  \caption{Unification Power of \dpp{}:
  non-Poisson/bursty in the long term,
  but Poisson in the short term.
  Real data:
  Traffic of the user of Figure~\ref{fig:pdfUser}-1, showing
  event-count per unit time (a and c)
  and respective cumulative event-count (b and d).
  \dpp data:
  synthetic traffic generated by the \dpp model 
  (with matching $\mu$, $\rho$ and event-count).
  Observe that (1) both time series are visually similar;
  (2) both are bursty in the long run (spikes; inactivity)
  (3) both are Poisson-like in the short term (last two columns)}
  \label{fig:timeseries}
\end{figure*}

\section{Collective Behavior}
\label{sec:collectiveBehavior}

Since we know that the great majority of users' \ied can be modeled by the \dpp model, we can figure out how each individual $i$ is distributed in its population according to their parameters $\rho_i$ and $\mu_i$ of the \dpp model. If the meta-distribution of the parameters $\rho_i$ and $\mu_i$ is well defined, then we can model the collective  behavior  of the individuals, which may serve for various applications, such as synthetic generators, anomaly detection, among others. From now on, we will call the meta-distribution of the parameters $\rho_i$ and $\mu_i$ the \groupsfp distribution.

In Figure~\ref{fig:collBehComp}-a, we show the scatter plot of the fitted parameters $\rho_i$ and $\log(\mu_i)$ of every MetaTalk individual $i$. It is difficult to visualize the patterns due to the large number of overlapping data points but, however, we can spot outliers. Moreover, by plotting  the $\rho_i$ and $\log(\mu_i)$ parameters using isocontours, as shown in Figure~\ref{fig:collBehComp}-b, we automatically smooth the visualization by disconsidering low populated regions. While darker color  mean a higher concentration of pairs $\rho_i$ and $\log(\mu_i)$, white color mean that there is small probability of observing a user with these values of $\rho_i$ and $\log(\mu_i)$. We use $\log(\mu_i)$ instead of $\mu_i$ because, as we see in Figure~\ref{fig:muall}, the logarithm of the medians can be approximated by a normal distribution for all datasets.

Surprisingly, we observe that the isocontours of Figure~\ref{fig:collBehComp}-b are very similar to the ones of a bivariate Gaussian. In order to verify this, we extracted from the \groupsfp{} distribution the means $P$ and $B$ of the parameters $\rho_i$ and $\log(\mu_i)$, respectively, and also the covariance matrix $\Sigma$. We use these values to generate the isocontours of a bivariate Gaussian distribution and we plotted it in Figure~\ref{fig:collBehComp}-c. We observe that the isocontours of the generated bivariate Gaussian distribution are very similar to the ones from the \groupsfp distribution. Thus a bivariate Gaussian distribution fits the real data of fitted $\rho_i$s and $\log(\mu_i)$s and hence, it is a good model to represent the population of individuals whose \ied can be modeled by the \dpp.


\begin{figure*}[hbtp]
\centering
\subfigure[MetaTalk data plot]
  {\includegraphics[width=.33\textwidth]{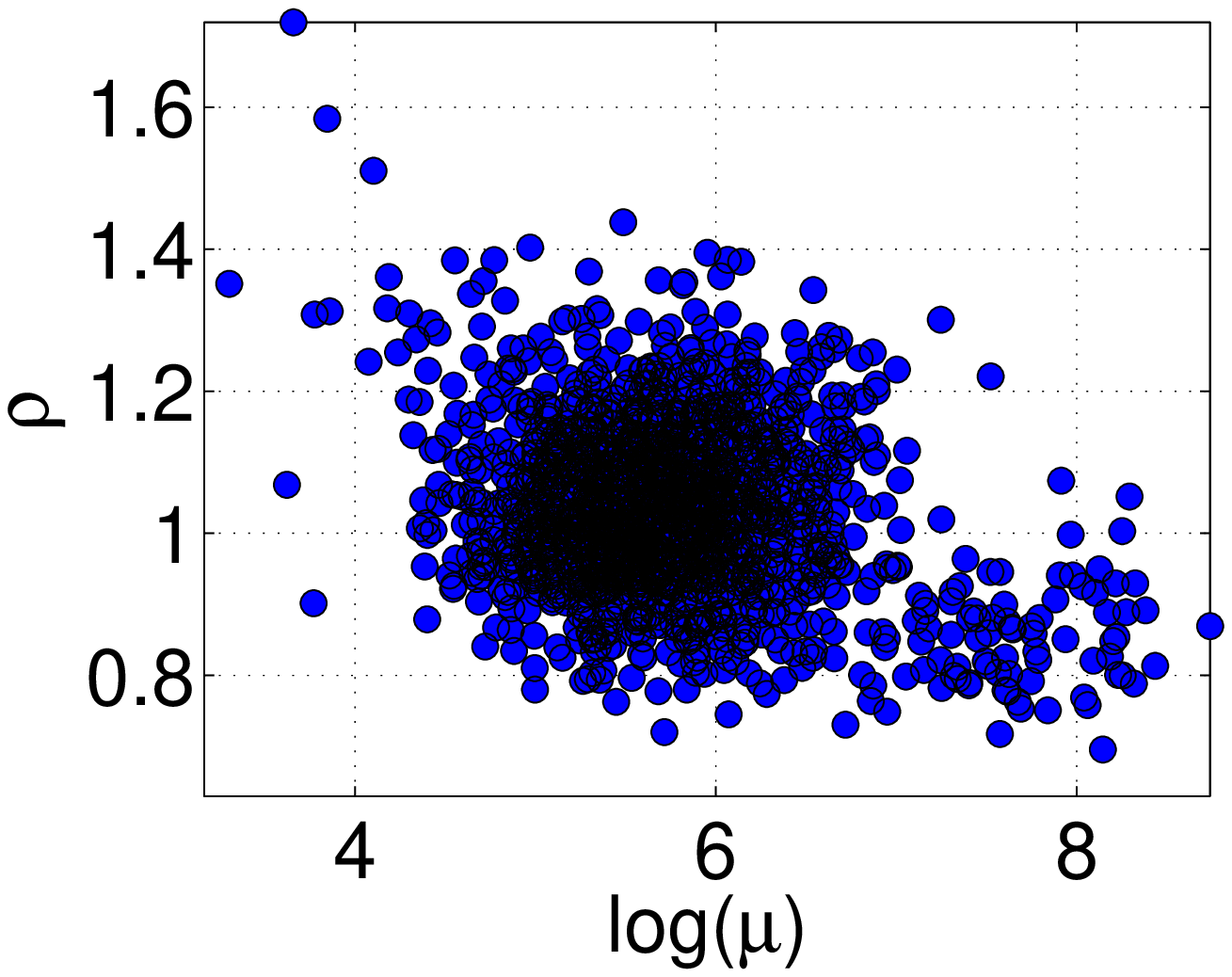}}
\subfigure[MetaTalk data isocontours]
  {\includegraphics[width=.33\textwidth]{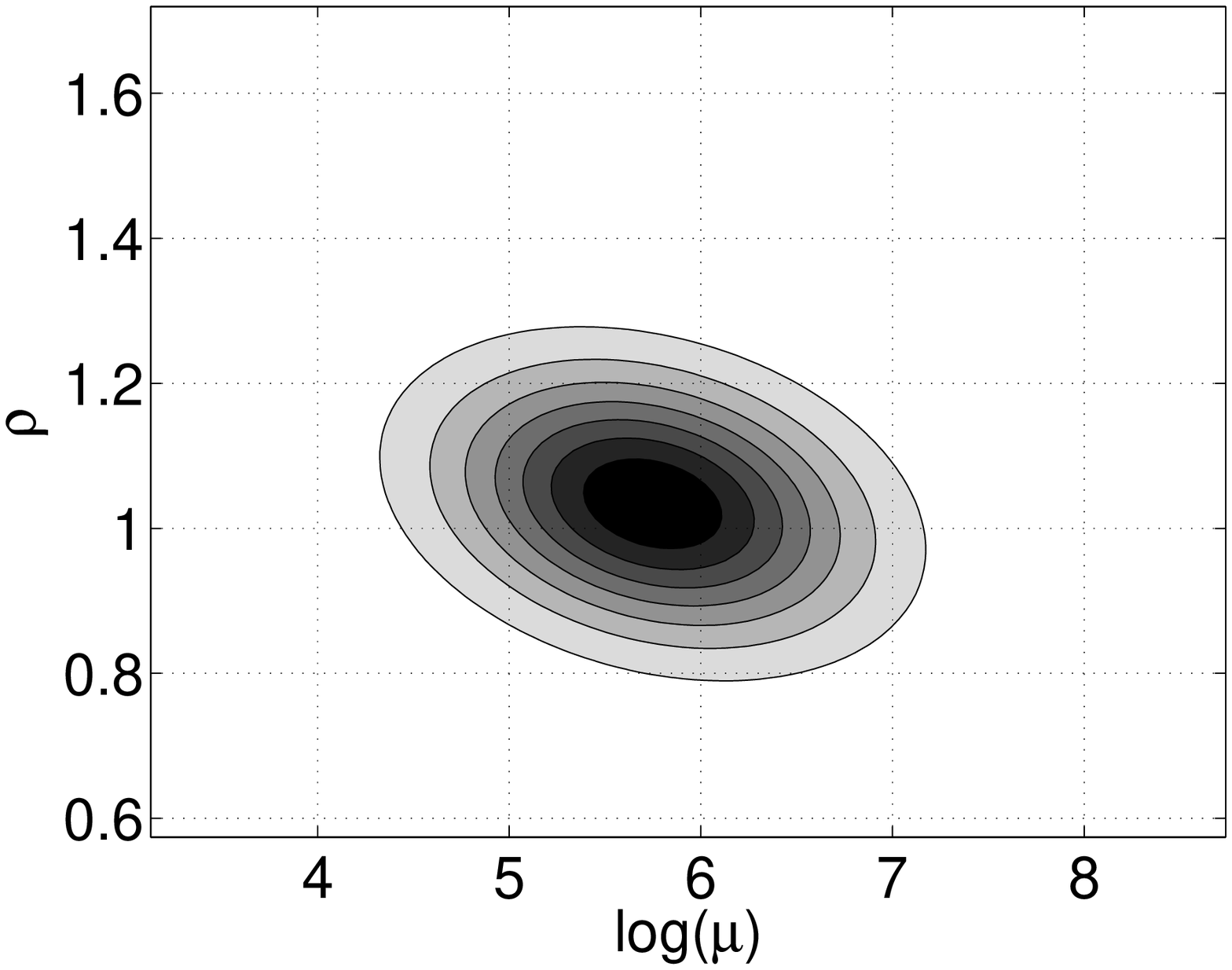}}  
\subfigure[Bivariate Gaussian]
  {\includegraphics[width=.33\textwidth]{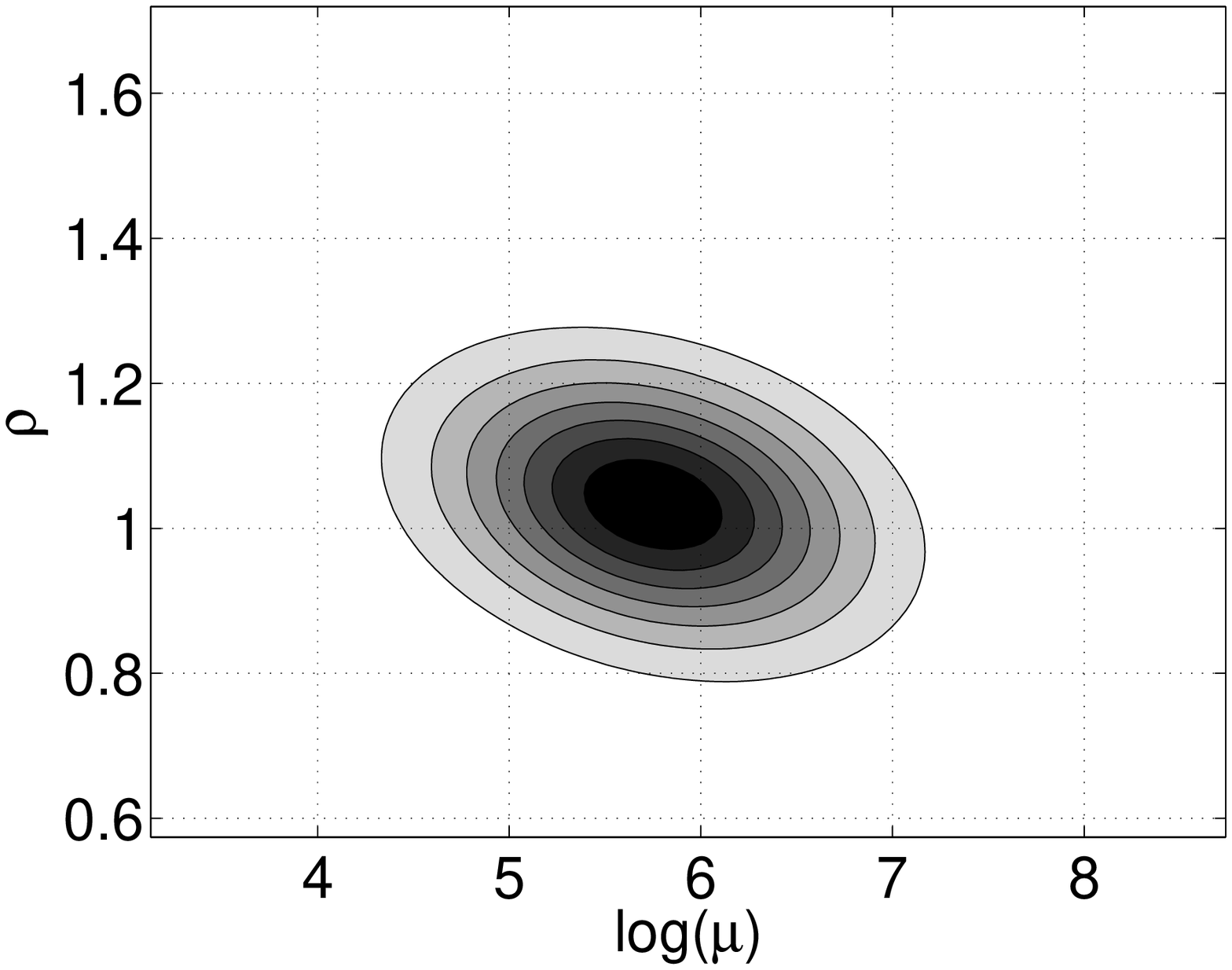}}  
	\caption{Plots of the parameters $\rho_i$ and $\log(\mu_i)$ of every MetaTalk individual $i$. In (a) we can not see any particular pattern, but we can spot outliers. By plotting the isocontours (b), we can observe how well a bivariate Gaussian (c) fits the real distribution of the pairs ($\rho_i$,$\log(\mu_i)$) ('meta-fitting').}
  \label{fig:collBehComp}
\end{figure*}

In order to verify if this pattern replicates in the other datasets, we show in Figure~\ref{fig:metadists} the comparison between the real data and the synthetic data generated from the 'meta-fitting' of the parameters $\rho_i$ and $\log(\mu_i)$. Observe that the bivariate Gaussian distribution can also model the collective behavior for all the other seven datasets, allowing us to state the following universal pattern:

\begin{universal}
The joint distribution of the parameters $\rho_i$ and $\log(\mu_i)$ associated with individual $i$ of a particular communication system follows a bivariate Gaussian distribution.
\end{universal}

This result is very useful, since it allow us to easily generate a synthetic dataset for a particular system. In order to do that, we simply have to perform the following steps:

\begin{enumerate}
	\item Select from Table~\ref{tab:gausstable} the system which you would like to generate the synthetic dataset;
	\item Create a bivariate Gaussian sampler using the correspondent parameters, which are shown in Table~\ref{tab:gausstable};
	\item Sample the $n$ individuals from the bivariate Gaussian sampler;
	\item Select the duration window $T$ of the dataset, e.g. $T = $1 month;
	\item For each individual $i$, generate $N_i$ inter-event times $\Delta_0, \Delta_1, ..., \Delta_{N_i}$, where $N_i$ is the highest natural number that satisfies $\sum_{k}^{N_i}{\Delta_k} < T$. 
\end{enumerate}

\begin{figure*}[!hbtp]
\centering
\subfigure[Askme: Real]
  {\includegraphics[width=.23\textwidth]{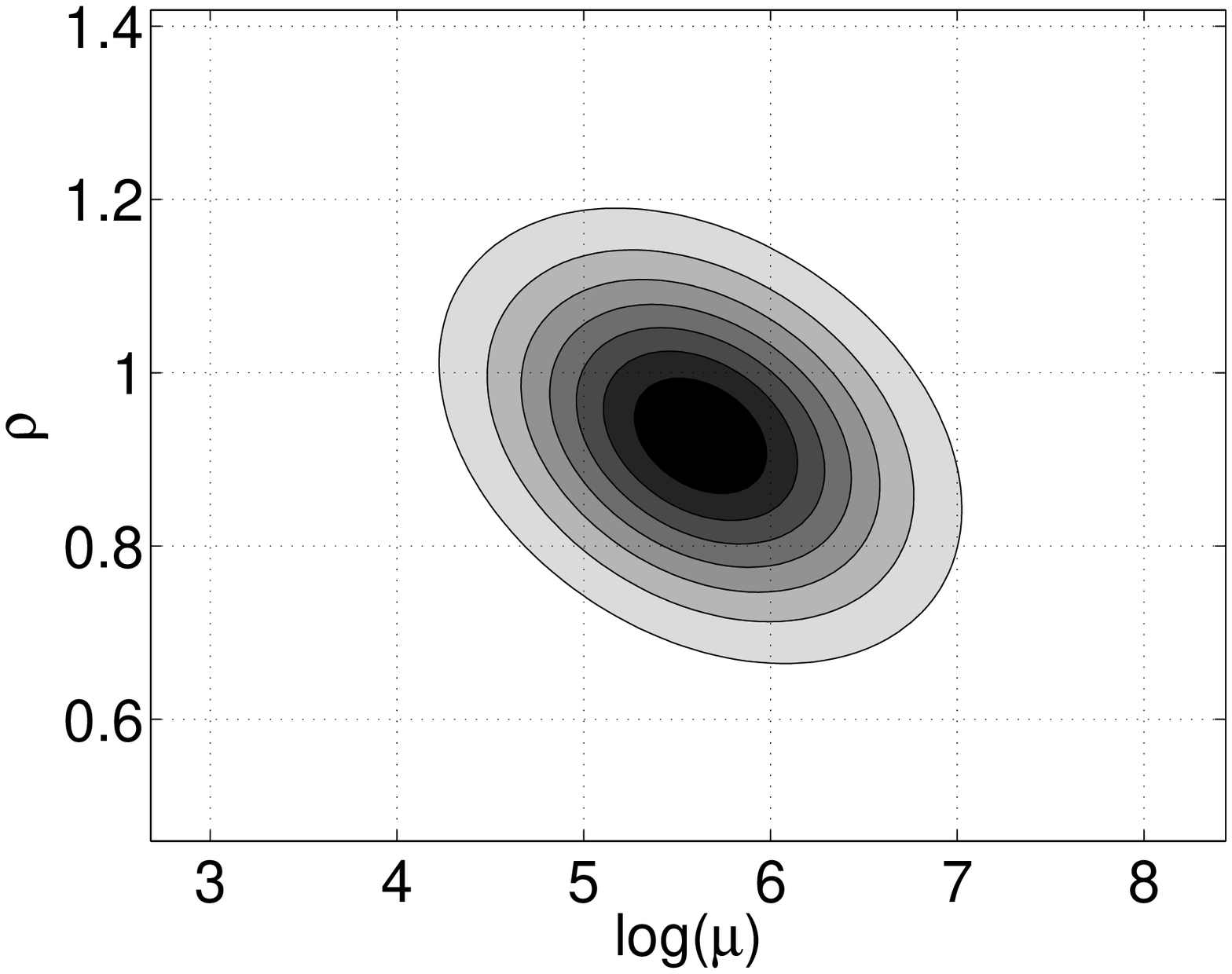}}  
\subfigure[Askme: Synthetic]
  {\includegraphics[width=.23\textwidth]{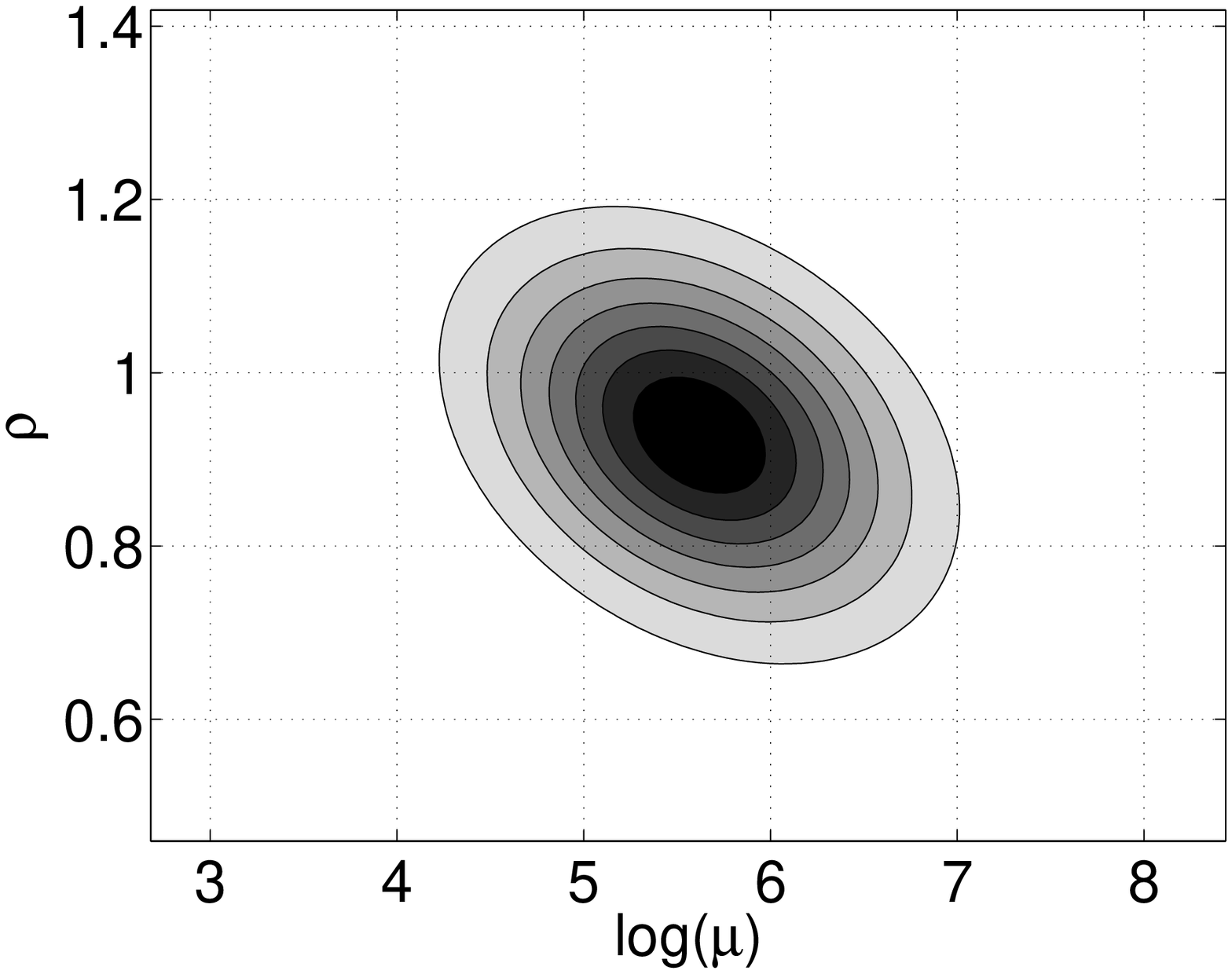}}  
\subfigure[Digg: Real]
  {\includegraphics[width=.23\textwidth]{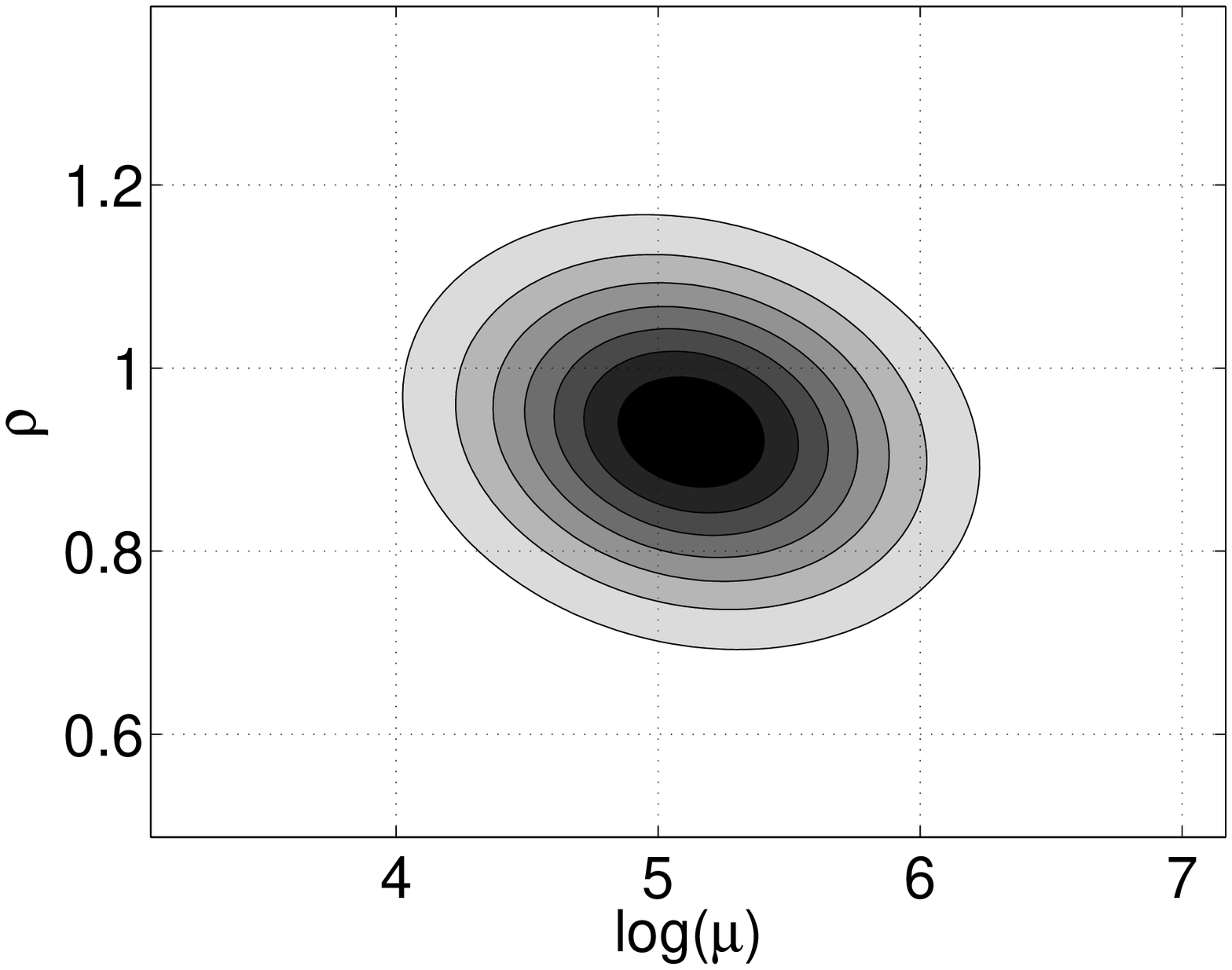}}  
\subfigure[Digg: Synthetic]
  {\includegraphics[width=.23\textwidth]{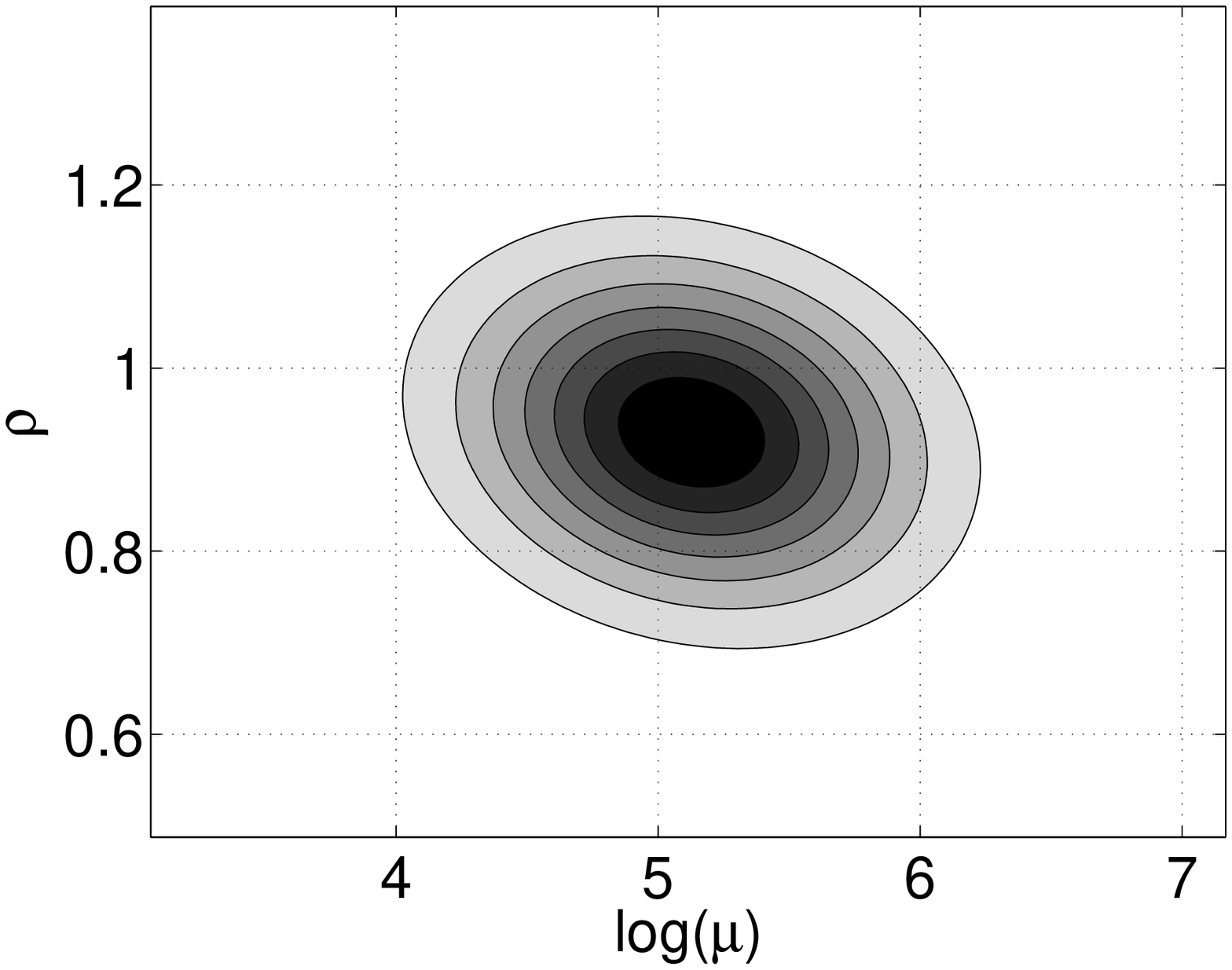}}  
\subfigure[Enron: Real]
  {\includegraphics[width=.23\textwidth]{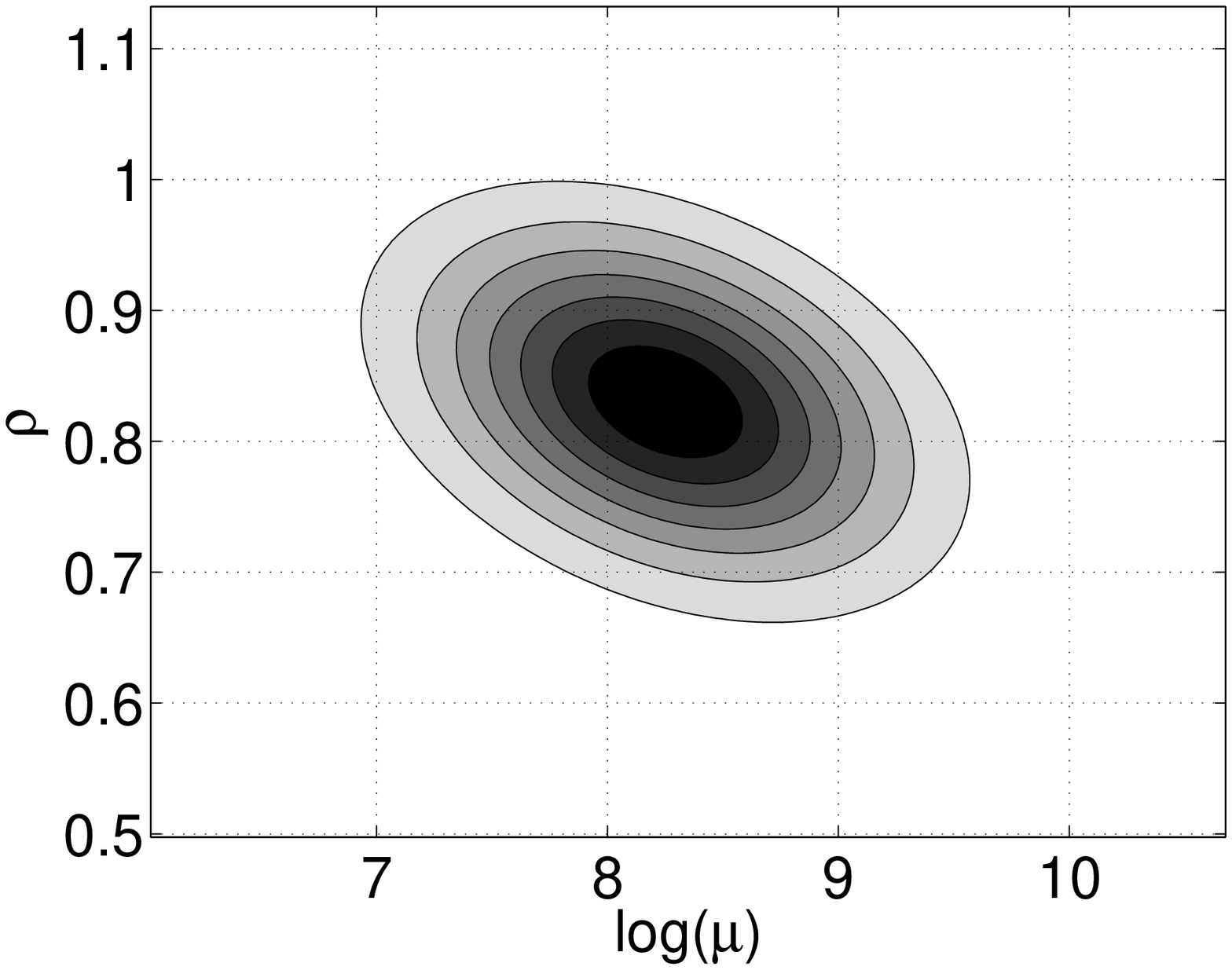}}  
\subfigure[Enron: Synthetic]
  {\includegraphics[width=.23\textwidth]{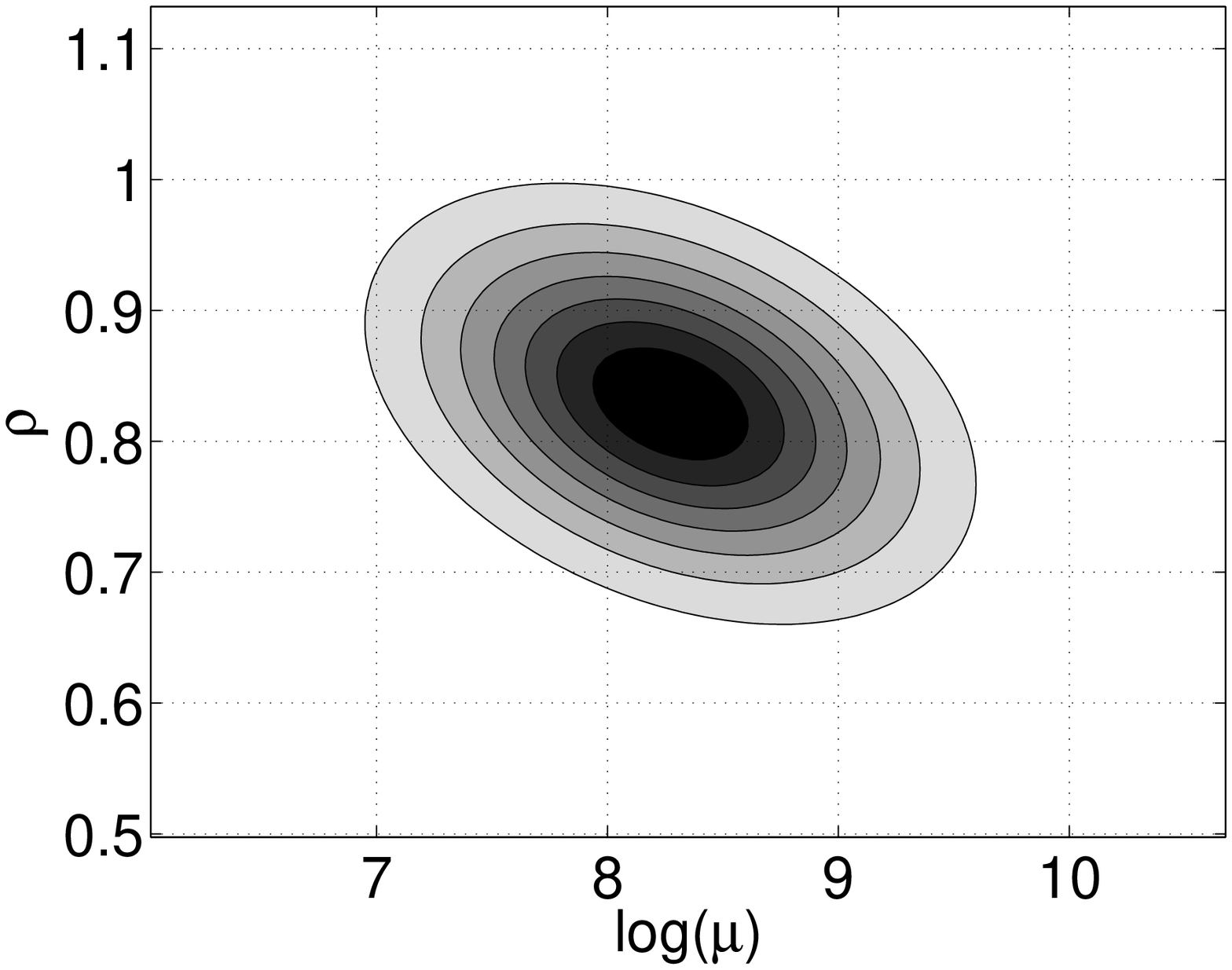}}  
\subfigure[Mefi: Real]
  {\includegraphics[width=.23\textwidth]{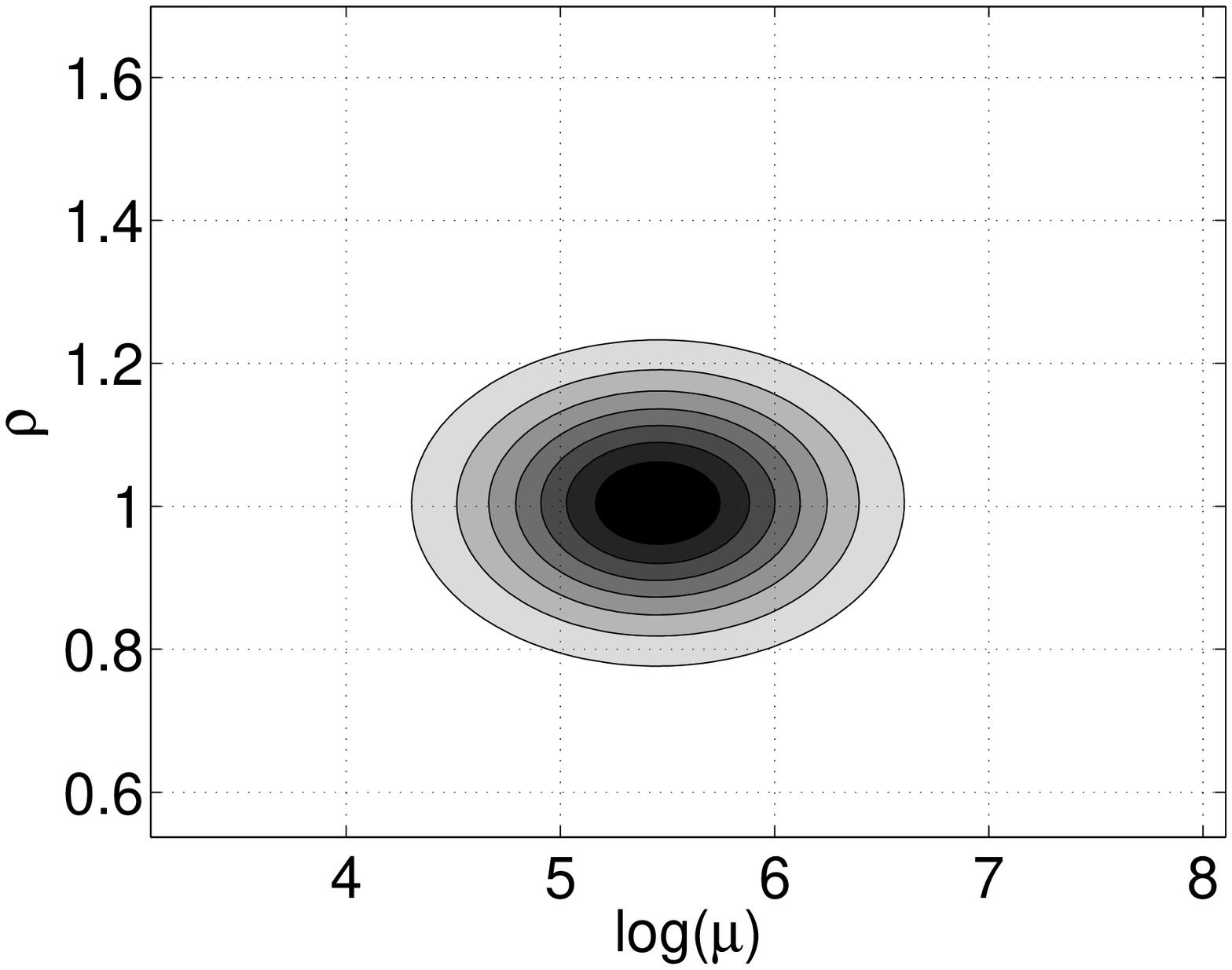}}  
\subfigure[Mefi: Synthetic]
  {\includegraphics[width=.23\textwidth]{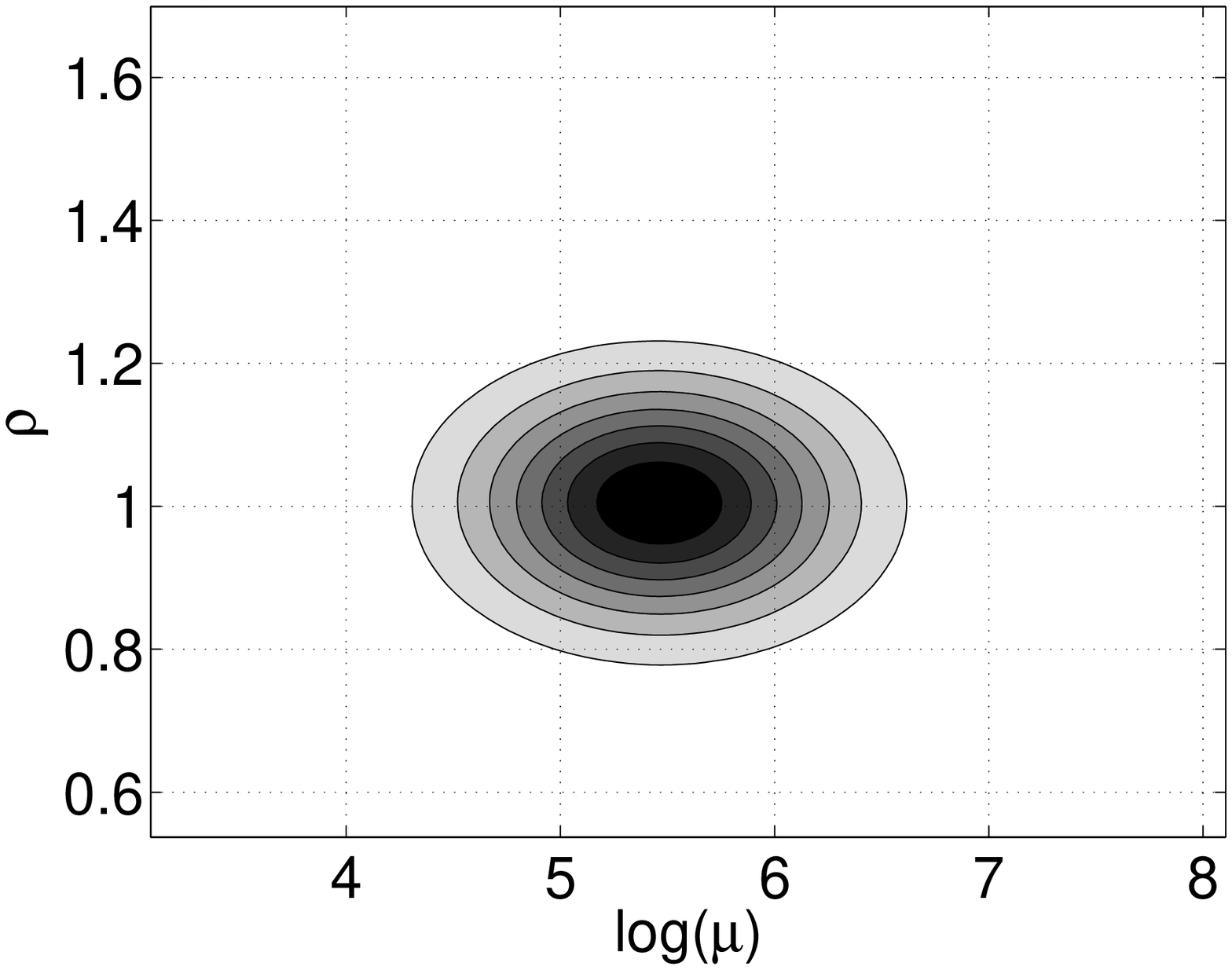}}  
\subfigure[Phone: Real]
  {\includegraphics[width=.23\textwidth]{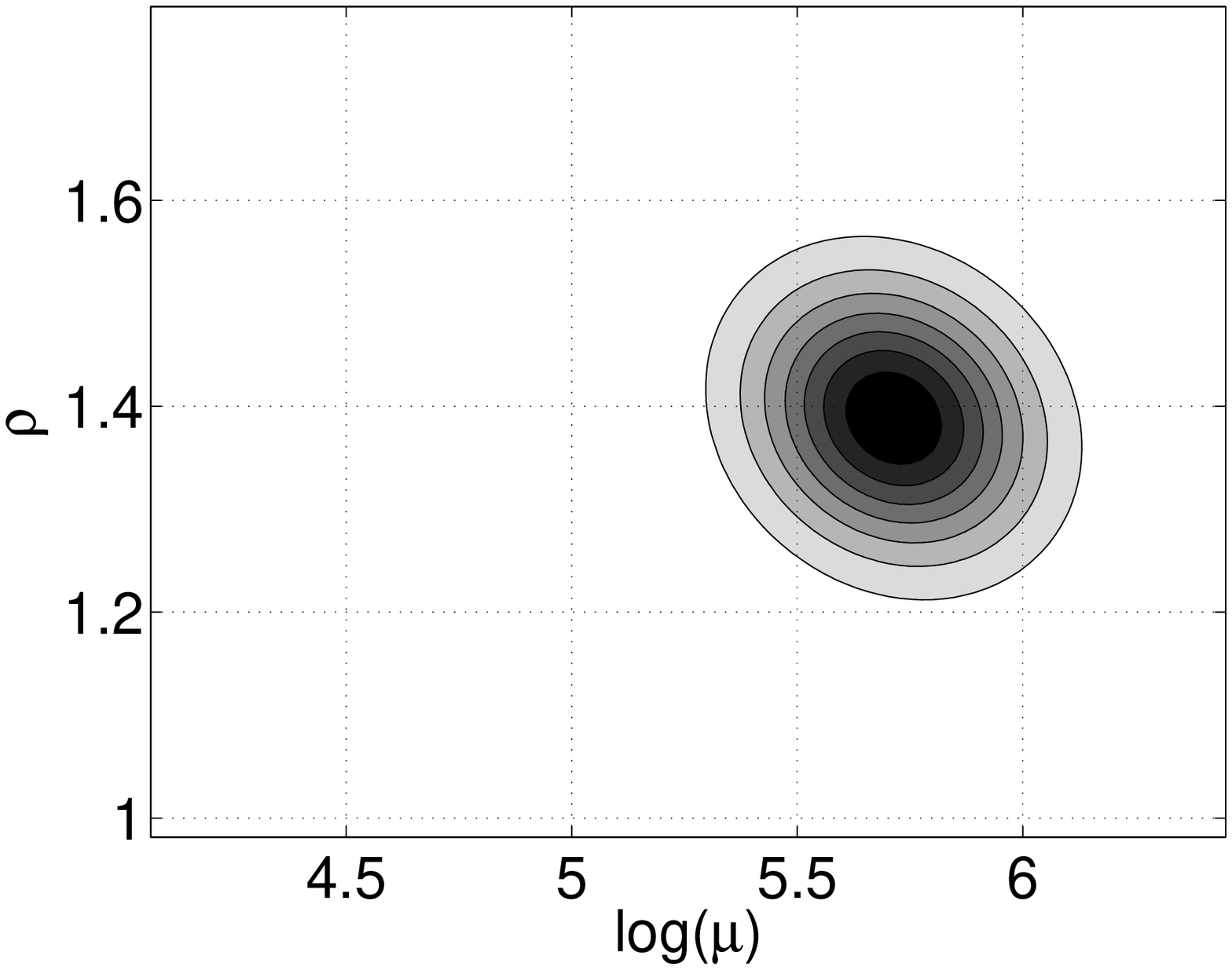}}  
\subfigure[Phone: Synthetic]
  {\includegraphics[width=.23\textwidth]{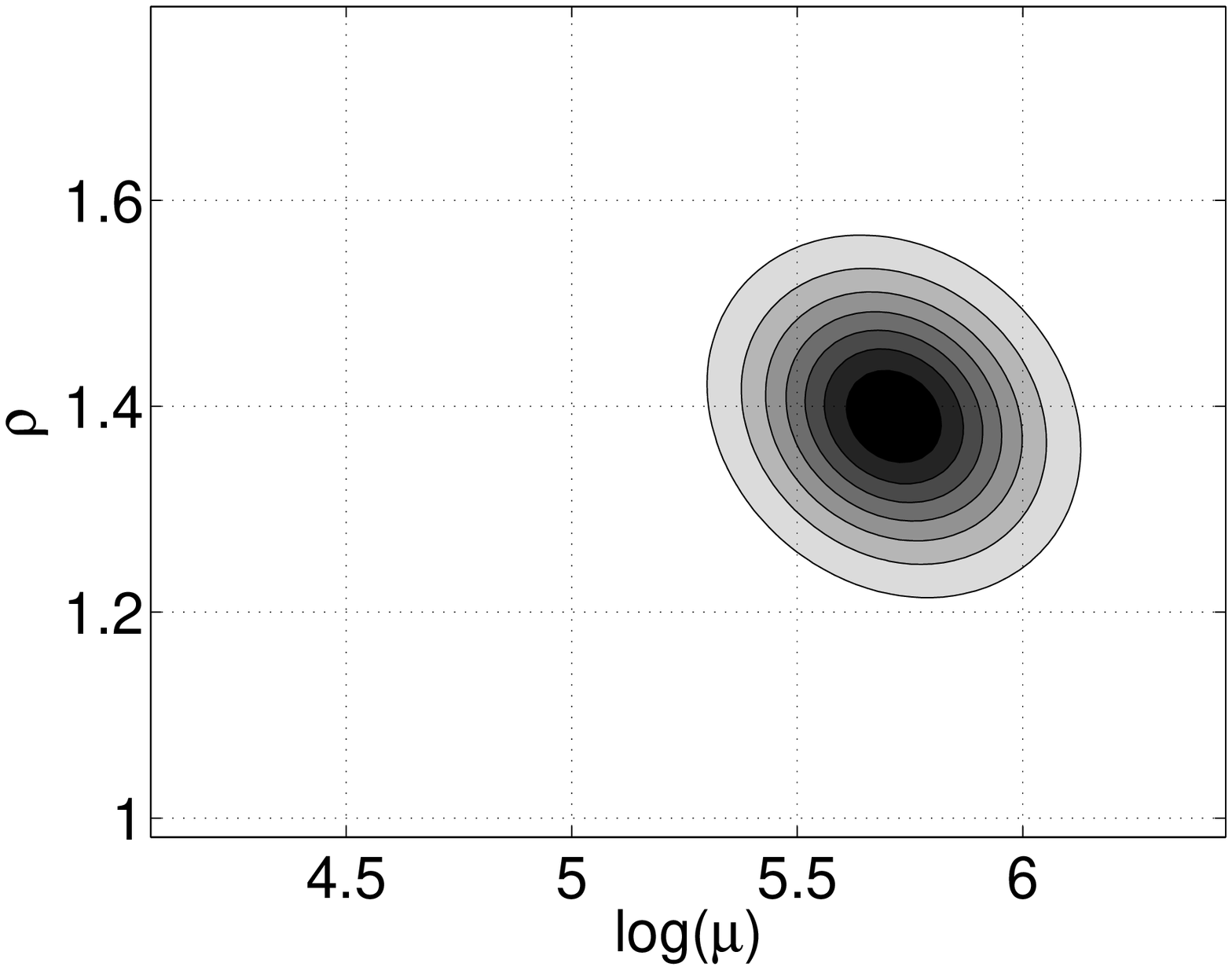}}  
\subfigure[SMS: Real]
  {\includegraphics[width=.23\textwidth]{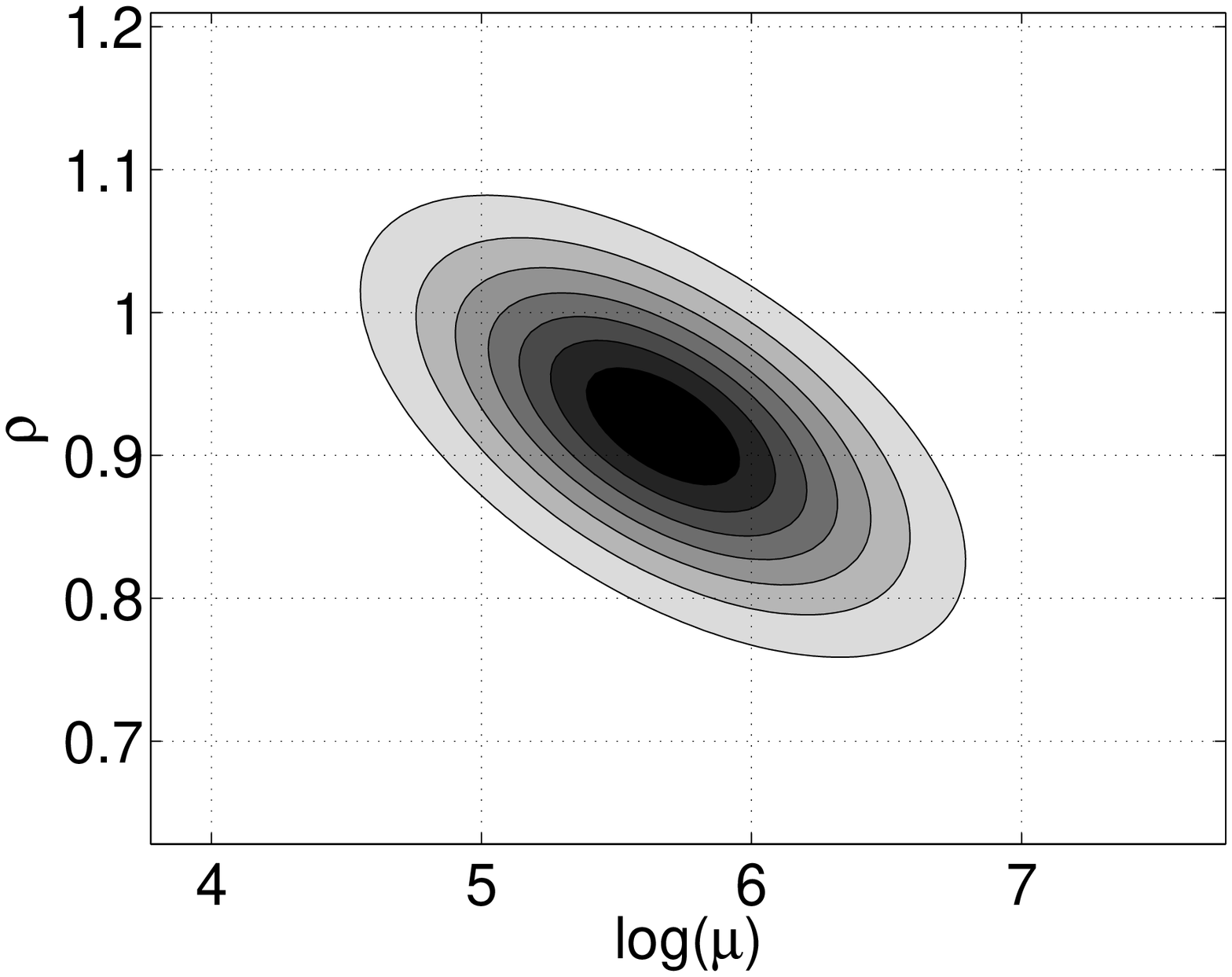}}  
\subfigure[SMS: Synthetic]
  {\includegraphics[width=.23\textwidth]{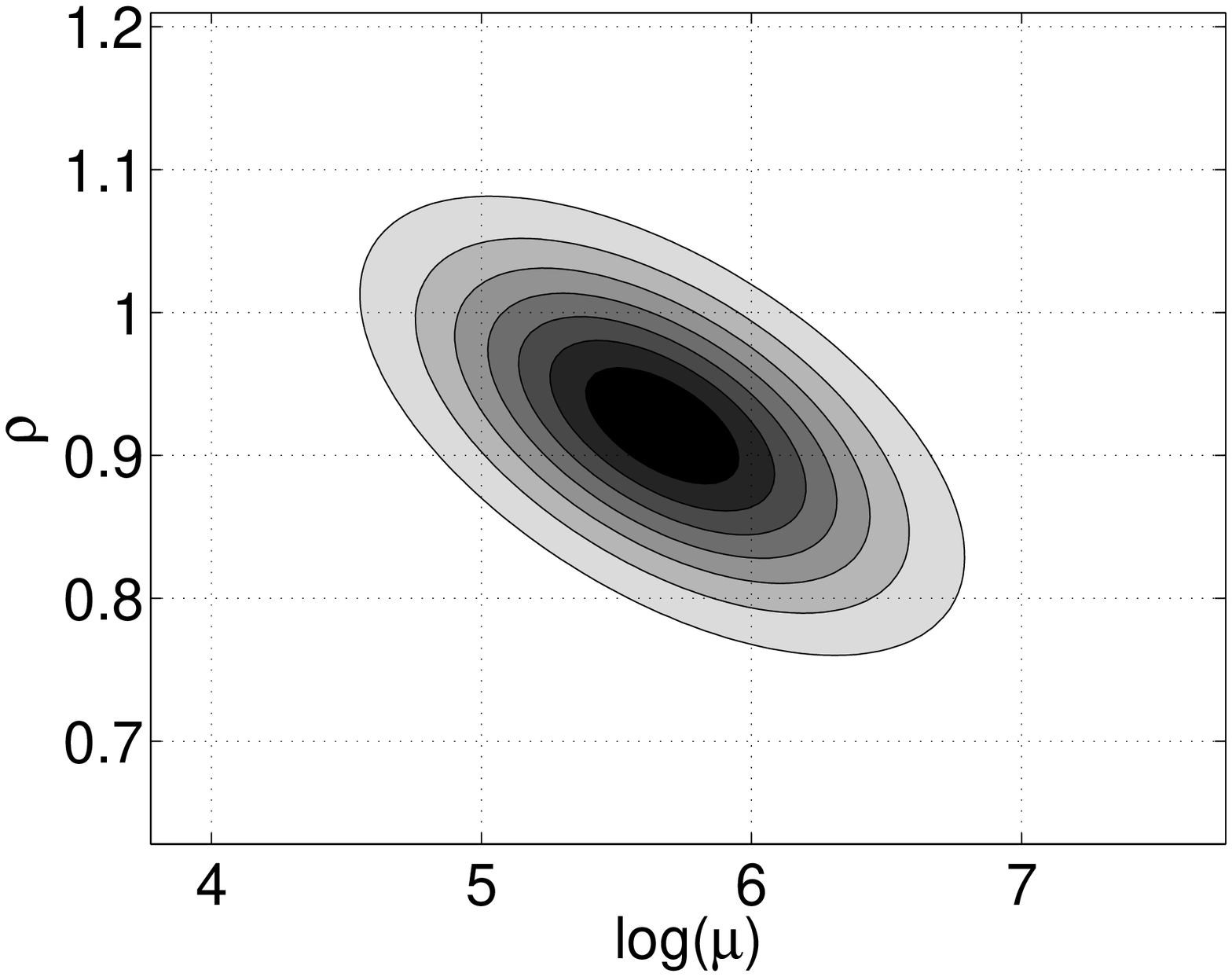}}  
\subfigure[Youtube: Real]
  {\includegraphics[width=.23\textwidth]{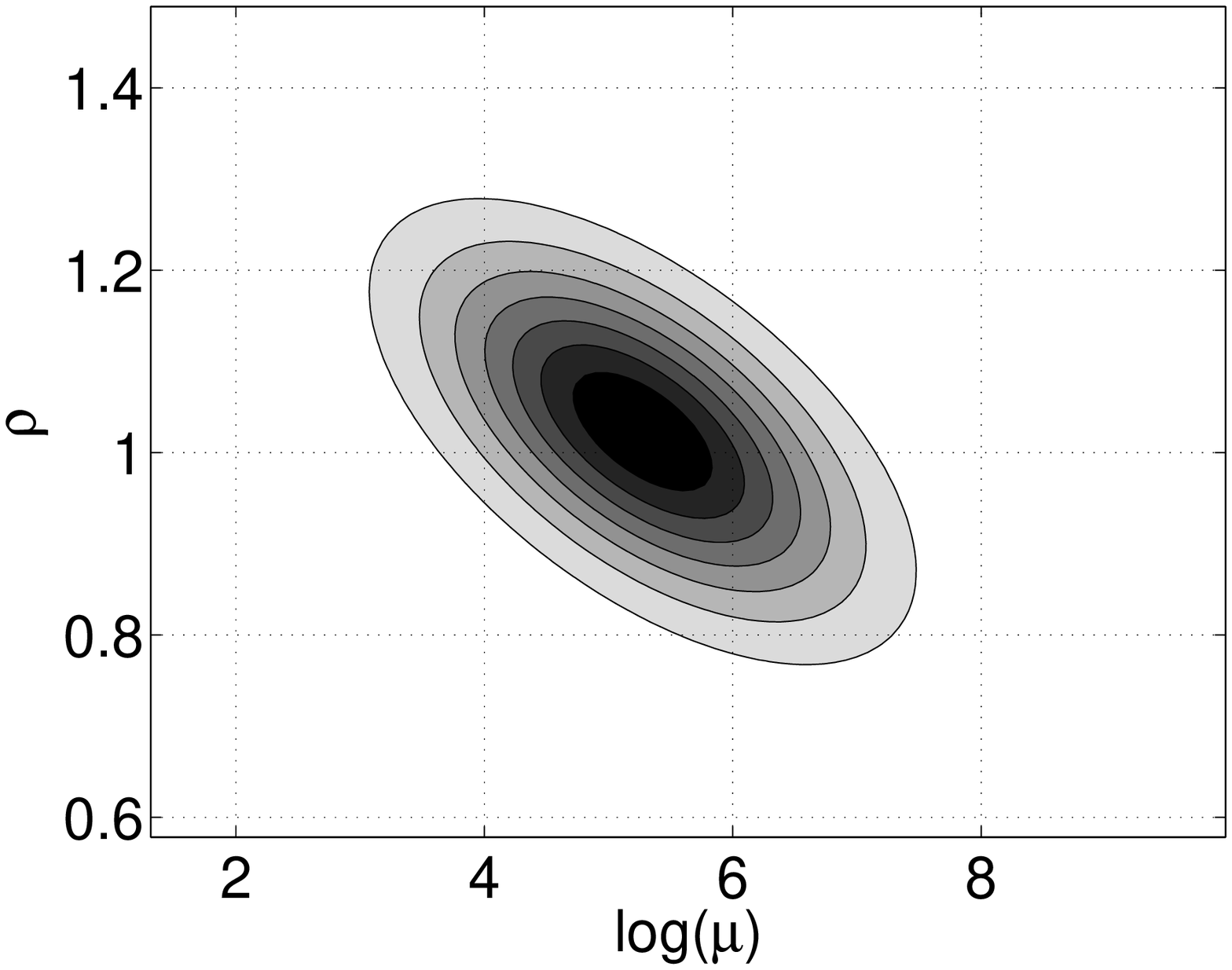}}  
\subfigure[Youtube: Synthetic]
  {\includegraphics[width=.23\textwidth]{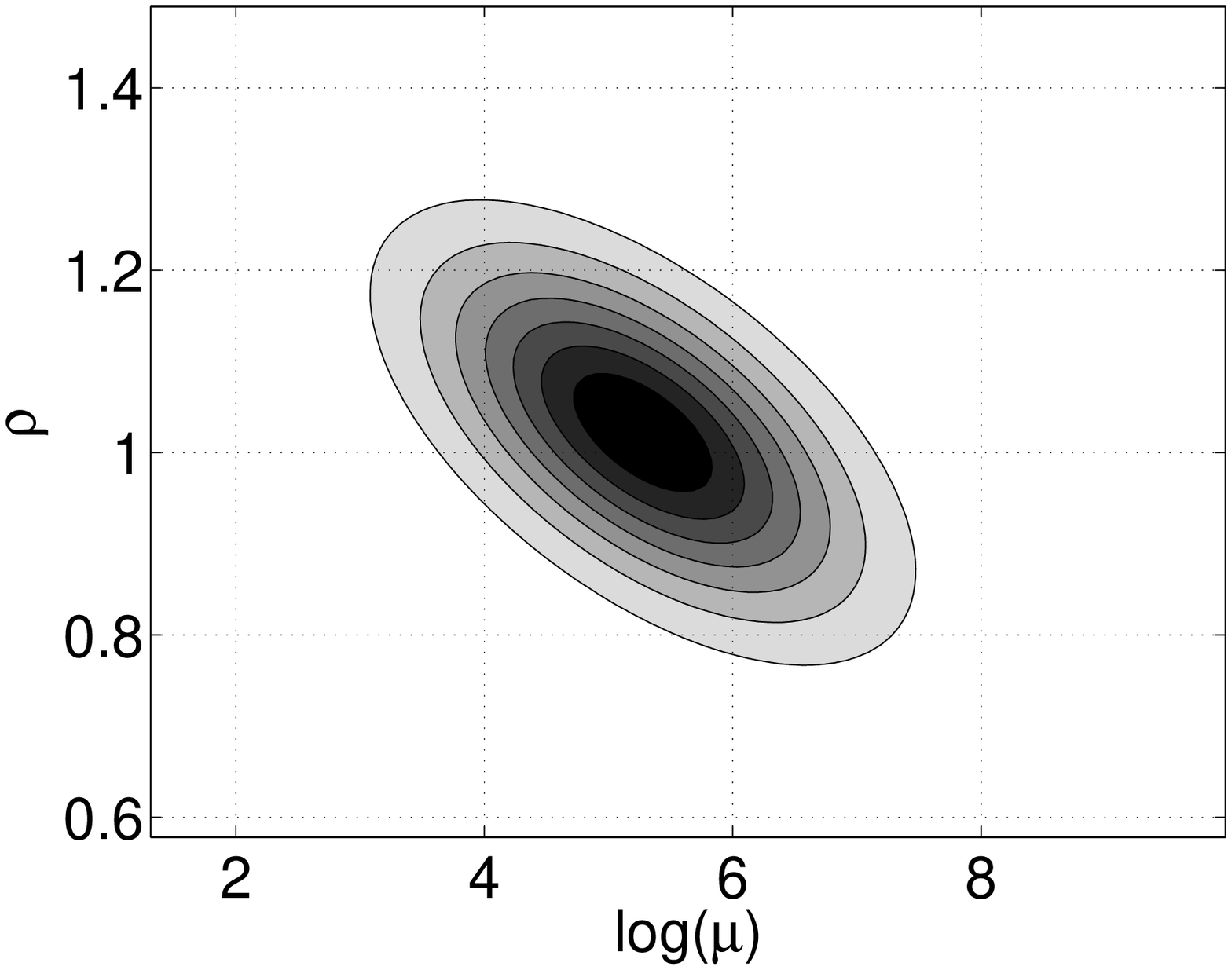}}  
\caption{Comparison between the real and synthetic datasets generated from bivariate Gaussians.}
  \label{fig:metadists}
\end{figure*}

\begin{table}[!htbp]%
\tbl{Parameters of the bivariate Gaussian distributions.\label{tab:gausstable}}{%
	\begin{tabular}{c c c c c c c c c}
\hline
System & AskMe	& Digg &	Enron &	Meta &	Mefi &	Phone	& SMS	& Youtube \\
\hline
$E(\rho_i)$  & 0.927 & 0.930  & 0.830 & 1.004 & 1.033 & 1.388  & 0.920 & 1.023 \\
$E(\log(\mu_i))$ & 5.625 & 5.126 & 8.251 & 5.455 & 5.748 & 5.714 & 5.672 & 5.274 \\
Var$(\rho_i)$ & 0.016 & 0.013 & 0.006 & 0.012 & 0.014 & 0.007 & 0.006 & 0.015 \\
Var$(\log(\mu_i))$ & 0.470 & 0.291 & 0.417 & 0.317 & 0.487 & 0.041 & 0.301 & 1.163 \\
Cov$(\rho_i, \log(\mu_i))$ & -0.028 & -0.010 & -0.018 & 0.000263 & -0.021 & -0.002 & -0.025 & -0.080 \\
\hline
\end{tabular}}
\end{table}

\section{Anomalies}
\label{sec:anomaly}

In the previous section, we showed that the majority of individuals of our eight datasets is well modeled by the \dpp. Moreover, we showed that the collective behavior of the individual \ieds of these datasets is well modeled by a bivariate Gaussian distribution. A natural application of these findings would be for anomaly detection. An individual that does not have a \ied that can be explained by the \dpp is a potential individual to be observed, since it has a distinct communication behavior from the majority of the other users. Moreover, an individual $i$ whose $\rho_i$ and $\log(\mu_i)$ values are significantly different from the typical individual is also a likely target to investigate.  

Thus, we define anomalies those individuals that fall into one of the following three criteria: 
\begin{itemize}
	\item \textbf{A1}: the \ied is well modeled by \dpp but its $\rho_i$ or $\log(\mu_i)$ values are significantly distant from the bivariate Gaussian that describes the population;
	\item \textbf{A2}: the \ied is not well modeled by the \dpp but its $\rho_i$ and $\log(\mu_i)$ values are inside the bivariate Gaussian that describes the population;
	\item \textbf{A3}: the \ied is not well modeled by the \dpp and its $\rho_i$ or $\log(\mu_i)$ values are significantly distant from the bivariate Gaussian that describes the population.
\end{itemize}
   
The distance between an individual and the expected typical behavior modeled by the bivariate Gaussian distribution is calculated by the Mahalanobis distance $D^2$ that is commonly used to measure the distance between an individual sample point $y$ and its expected value $m$~\cite{johnson:2007}. If $y$ follows a bivariate Gaussian distribution with covariance matrix $\Sigma$ then $D$ is given by 
\[ D^2 = (y-m)^t ~ \Sigma^{-1} (y-m) \]
which follows a chi-squared distribution with 2 degrees of freedom.  Individuals with $D^2$ greater than 25 occurs at a rate of 4 per million and hence they are considered outliers.  
 These anomalies are represented in Figure \ref{fig:anomalyFigs} by the colors green (A1), red (A2) and blue (A3), respectively.  White individuals have or are similar to the typical behavior.

\begin{figure*}[!hbtp]
\centering
\subfigure[Askme]
  {\includegraphics[width=.23\textwidth]{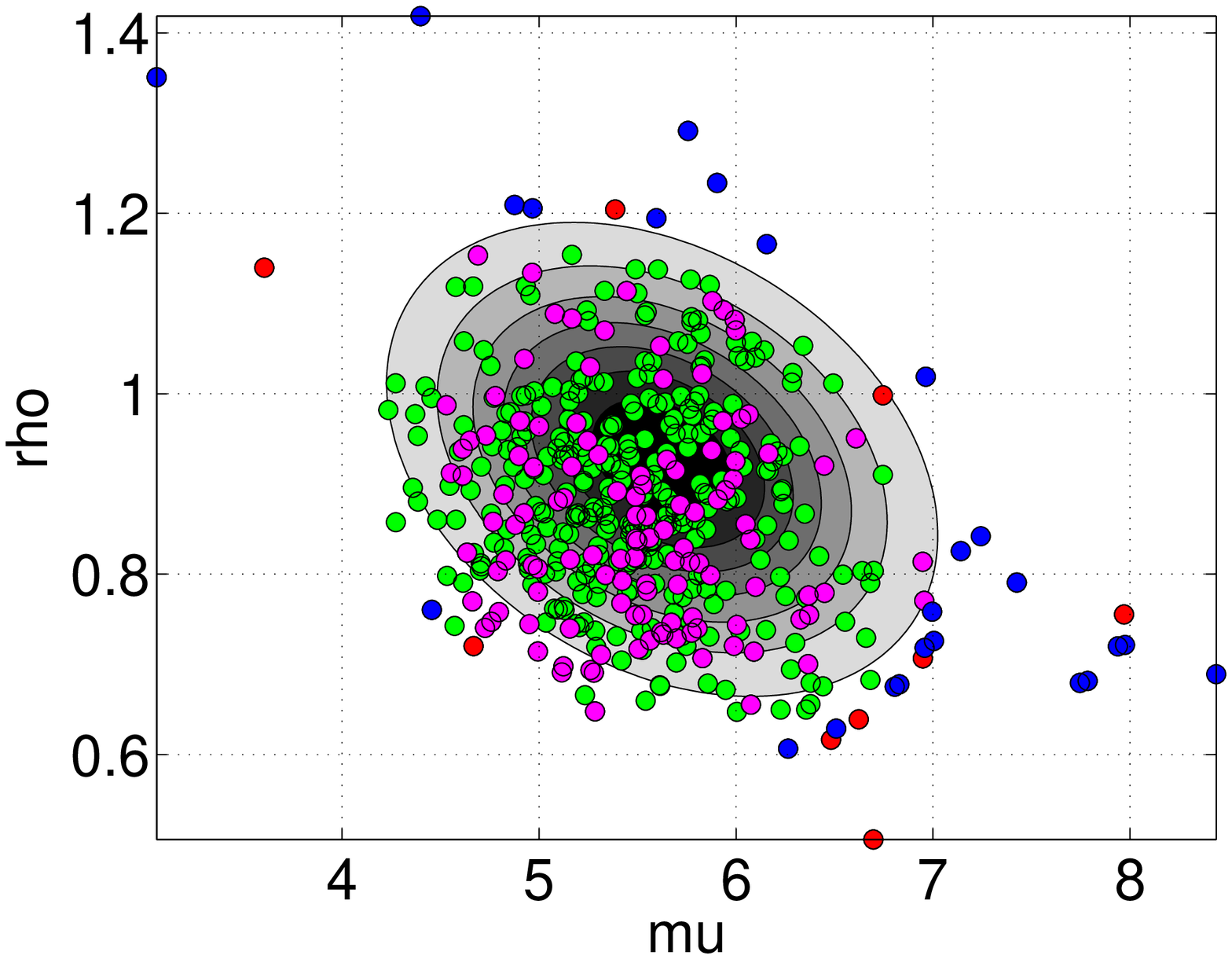}}
\subfigure[Digg]
  {\includegraphics[width=.23\textwidth]{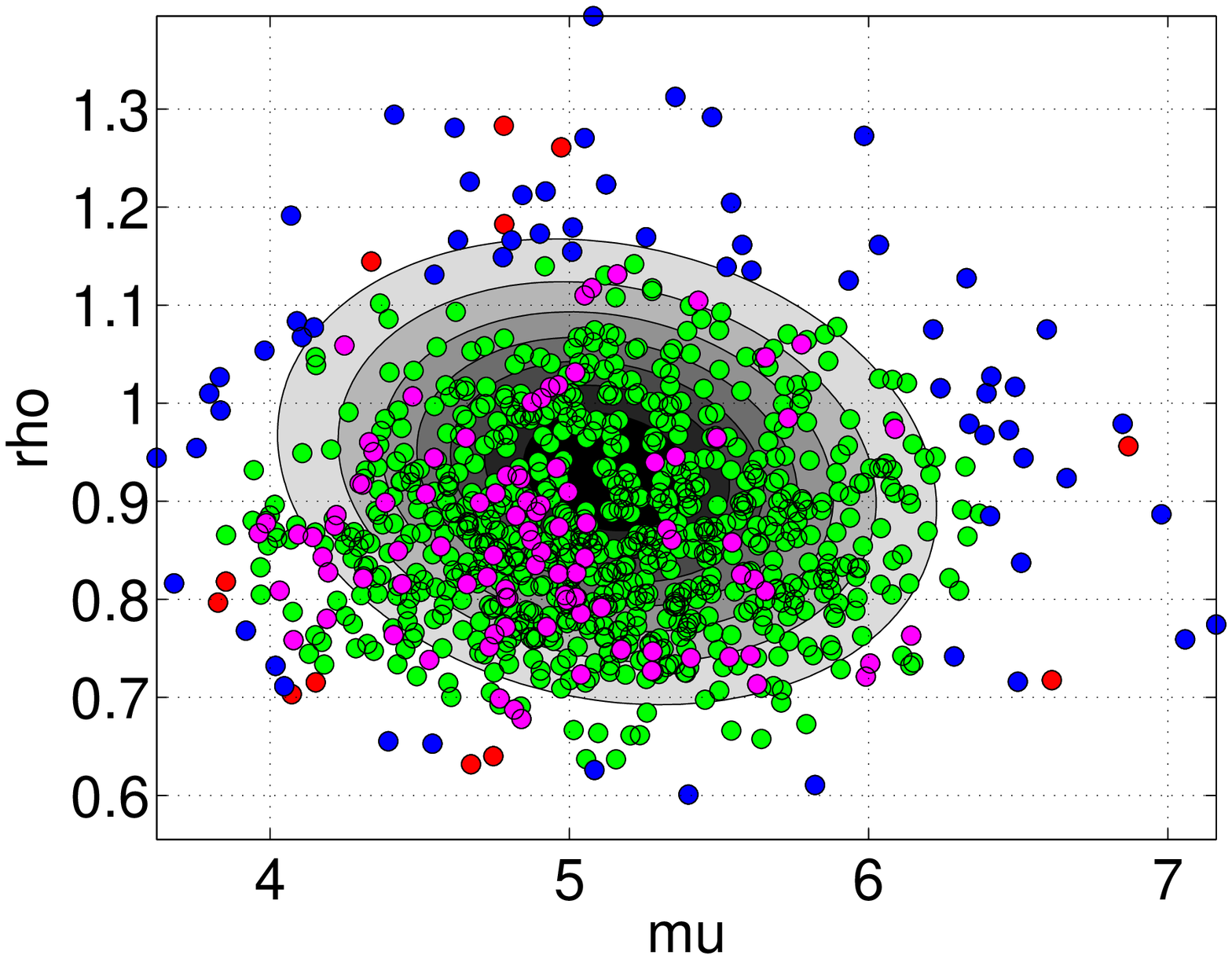}}
\subfigure[Enron]
  {\includegraphics[width=.23\textwidth]{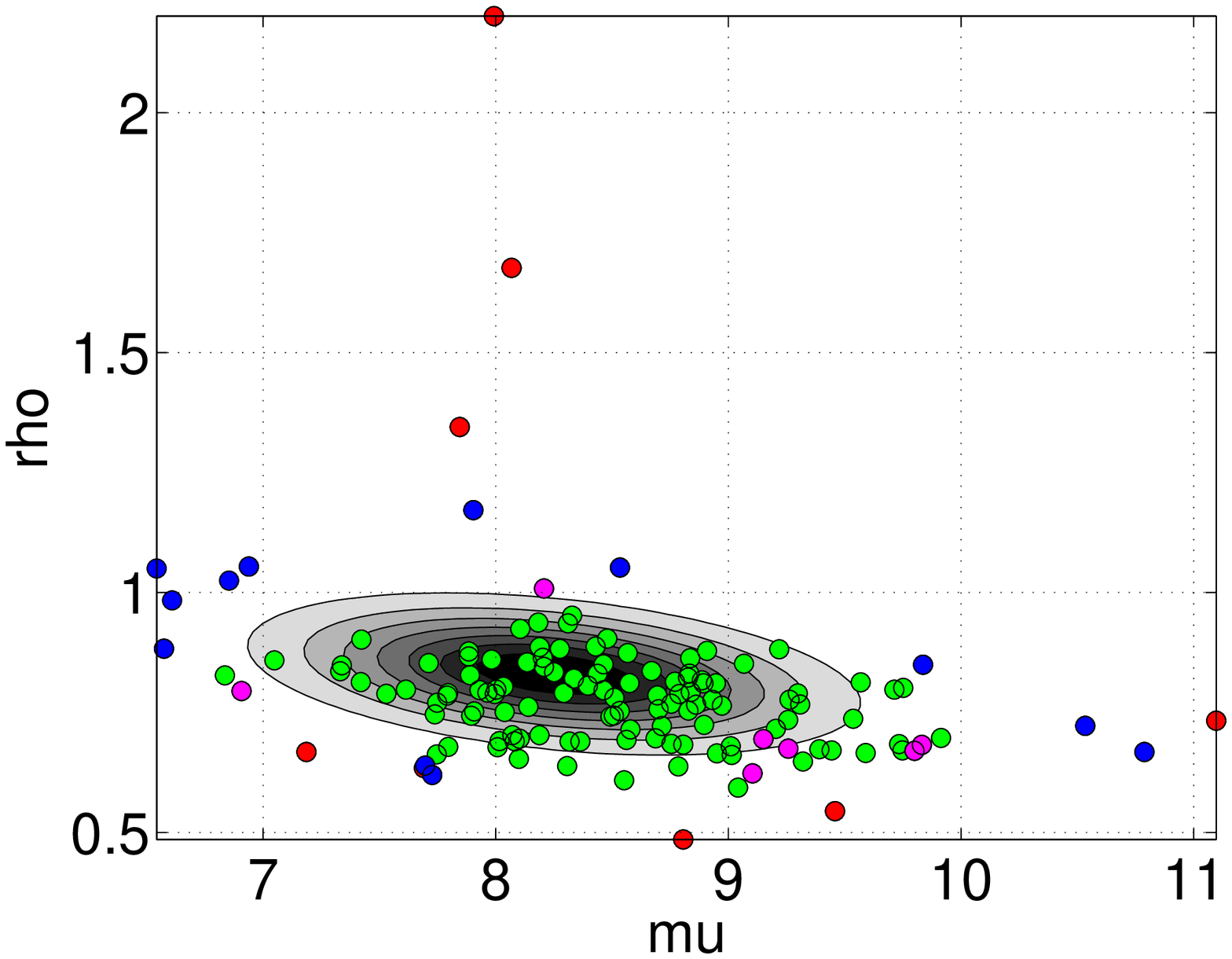}}
\subfigure[MetaFilter]
  {\includegraphics[width=.23\textwidth]{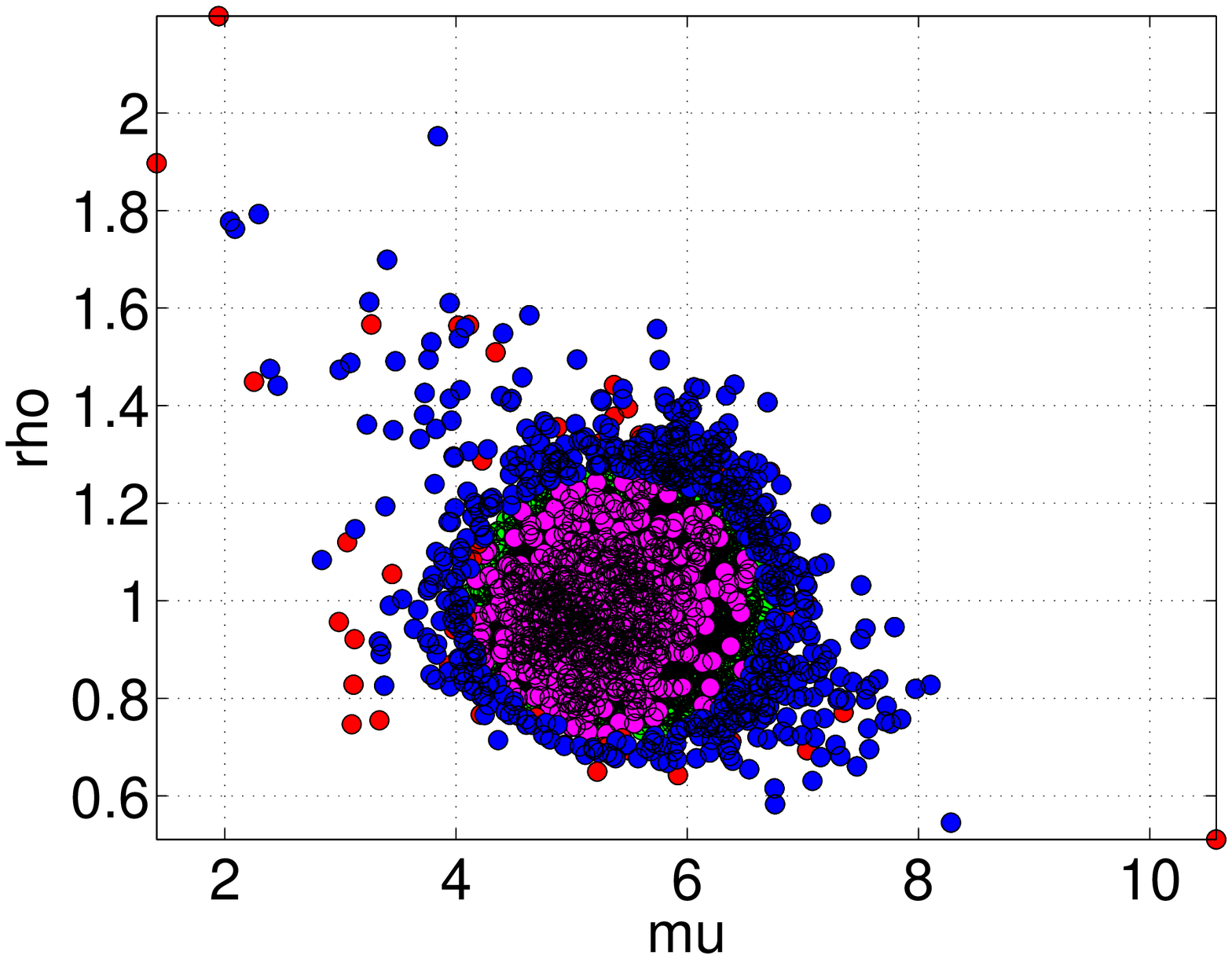}}
\subfigure[MetaTalk]
  {\includegraphics[width=.23\textwidth]{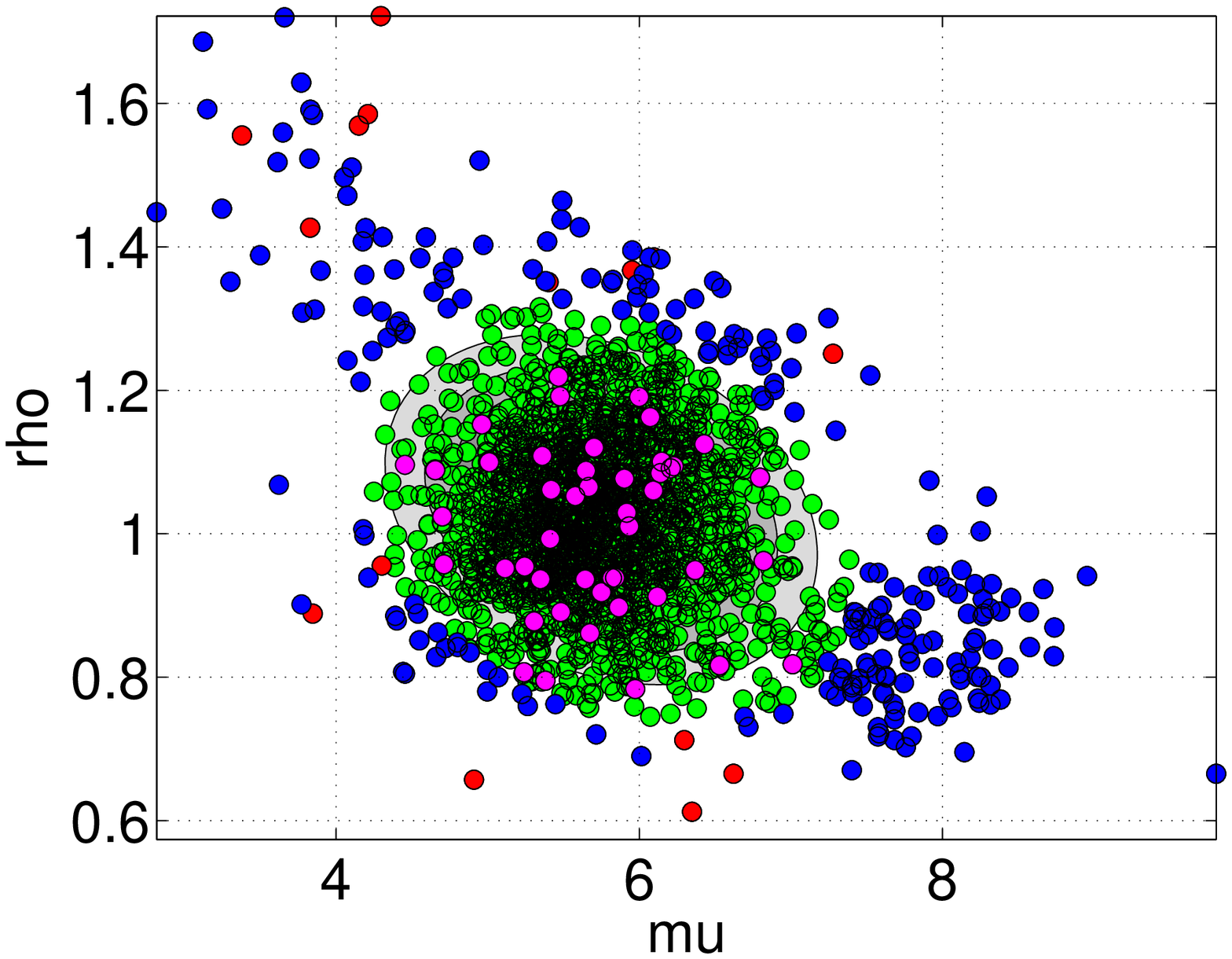}}
\subfigure[Phone]
  {\includegraphics[width=.23\textwidth]{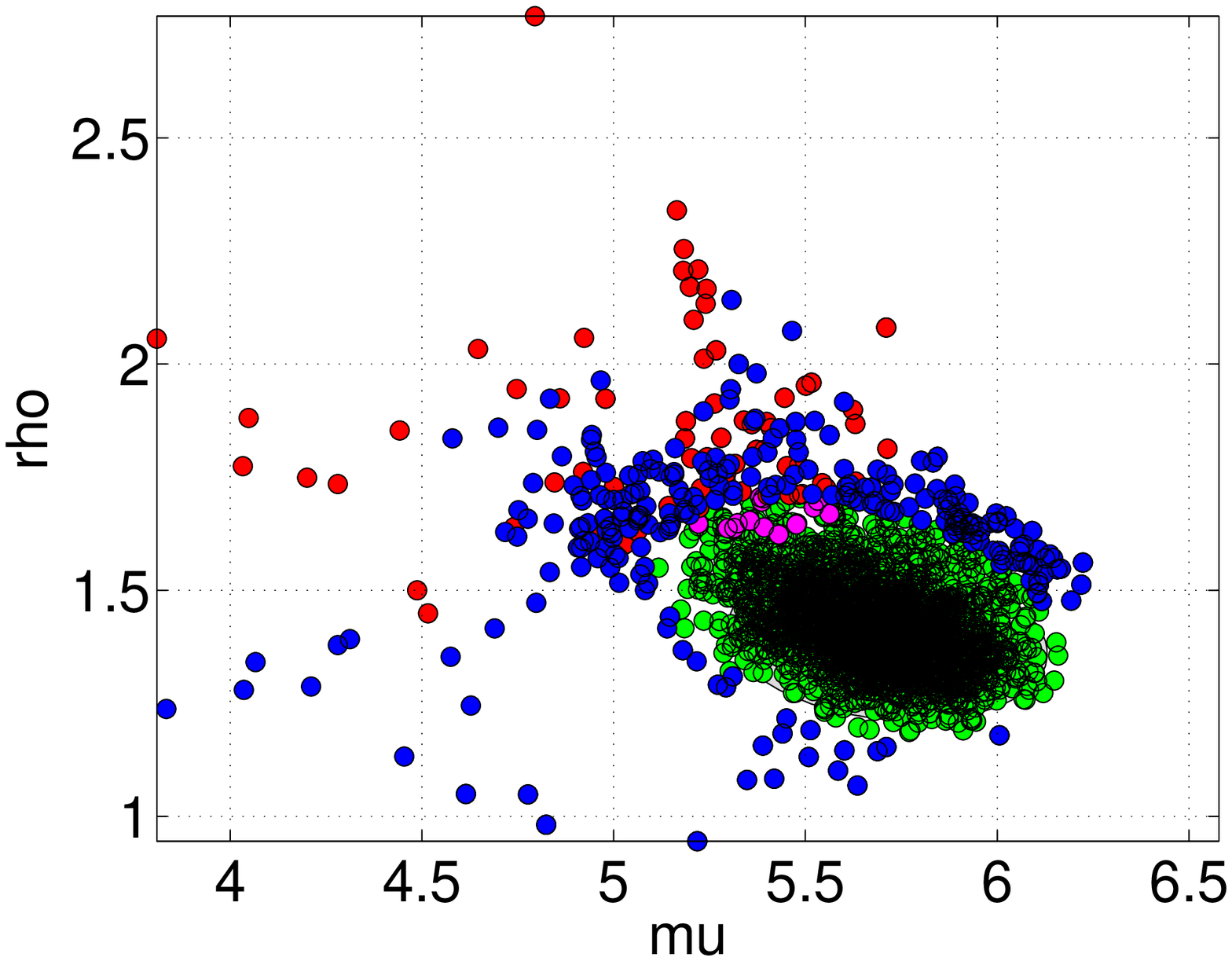}}
\subfigure[SMS]
  {\includegraphics[width=.23\textwidth]{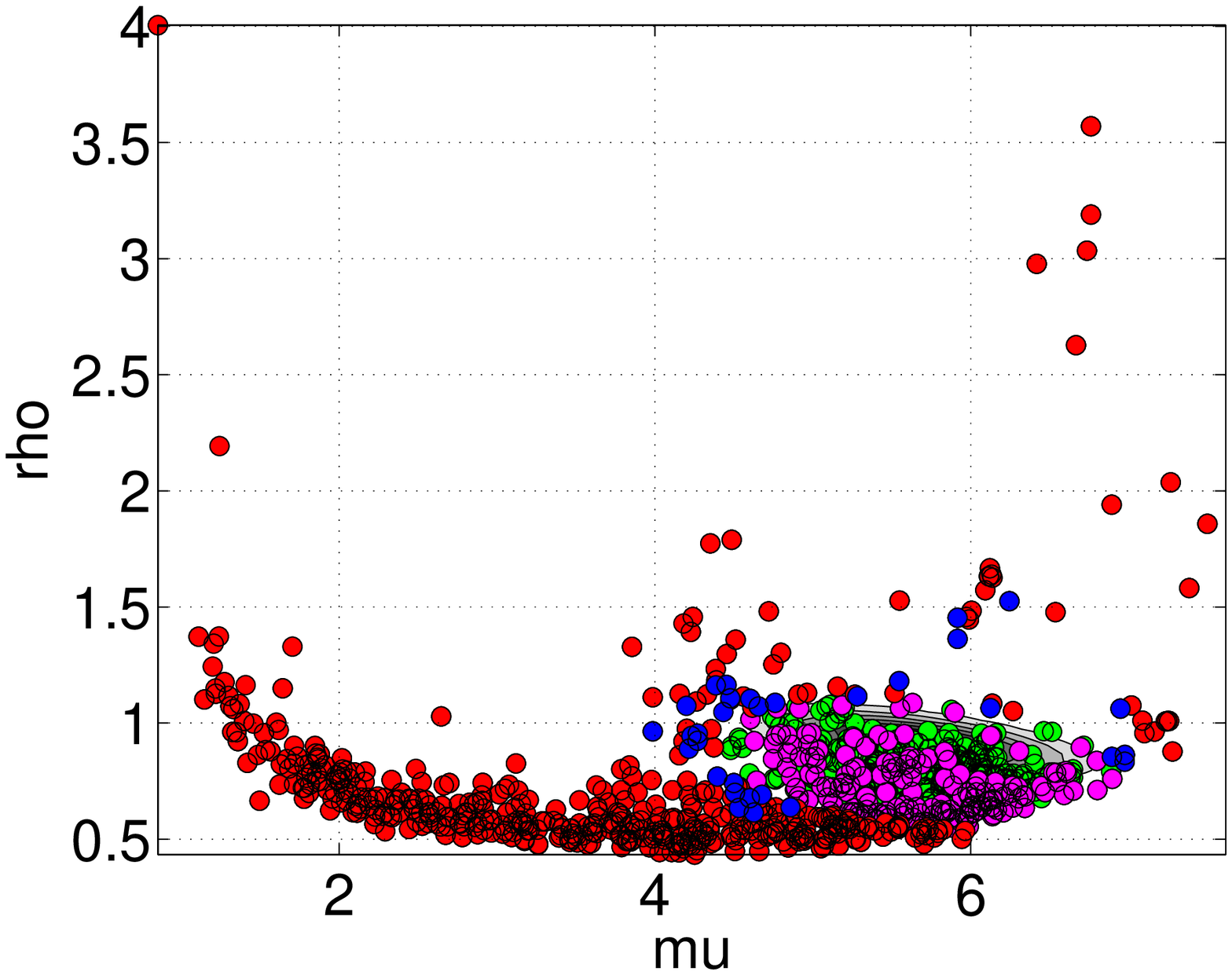}}
\subfigure[YouTube]
  {\includegraphics[width=.23\textwidth]{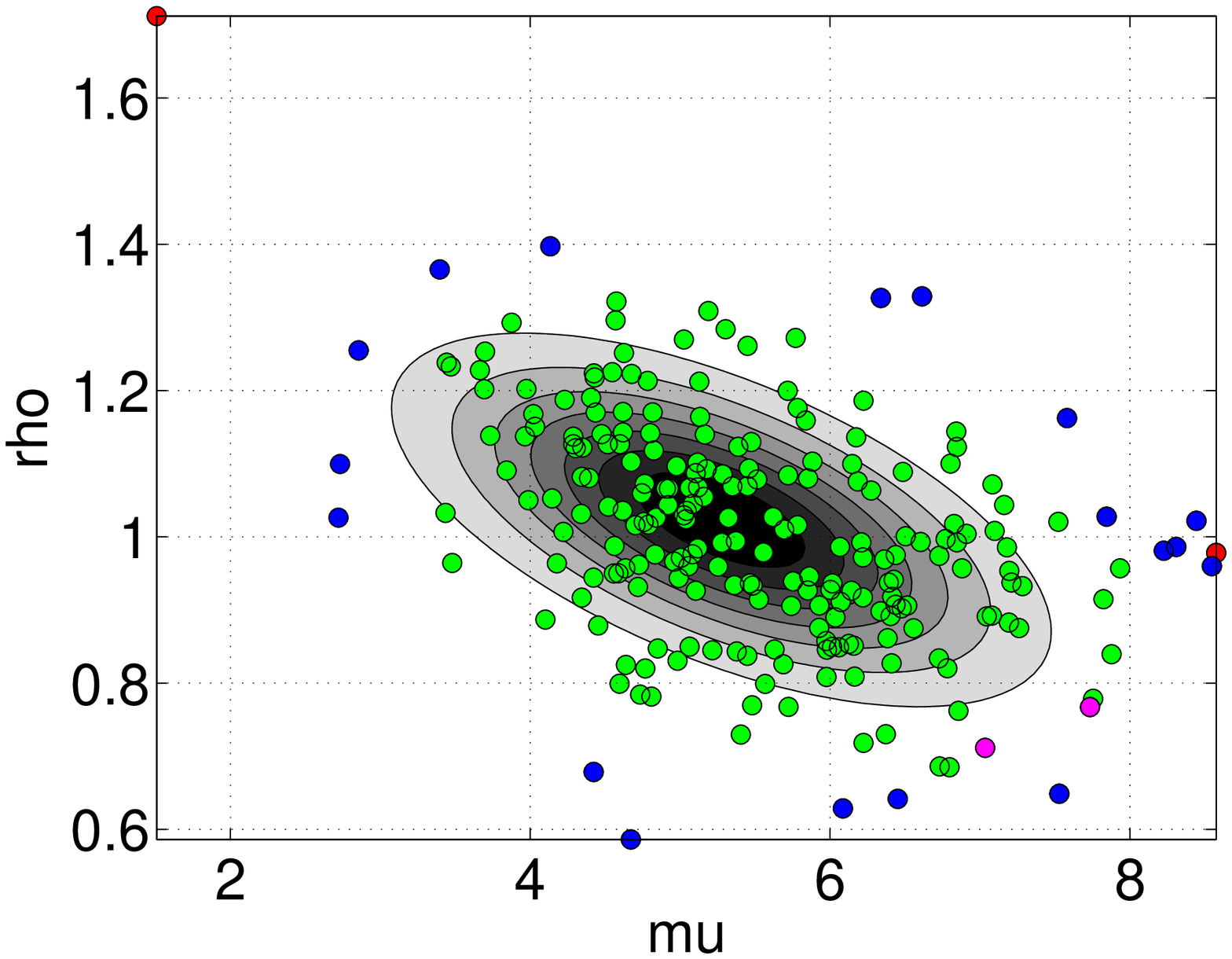}}

	\caption{Types of anomalies detected}
  \label{fig:anomalyFigs}
\end{figure*}


By checking the top anomalies according to our criteria, we found interesting examples. In AskMe dataset, for example, a point out of the bivariate Gaussian (type A3) was detected by the AskMe staff as a topic containing inappropriate content for this kind of service. When we accessed the post (see Figure~\ref{fig:exAnomalies}) we found the following message: ``This post was deleted for the following reason: Historical outliers notwithstanding, this is not what askme is for.''. As another example, we spotted a type A1 anomaly in the MetaFilter dataset  that was also detected manually by the administrator and deleted from the community. The post was changed (see Figure~\ref{fig:exAnomalies}) and the following message was posted: ``This post was deleted for the following reason: if you are going to be subtle, you need to be clearer.''. Considering the anomalies found in the MetaTalk dataset, we detected a type A2 anomaly that is a post from the developers talking about improvements and asking comments and suggestions to the community was spotted. This type of post is not the goal of MetaTalk. Finally, a YouTube video with the subject ``Bill Nye: Creationism Is Not Appropriate For Children''  was identified as a type A3 anomaly by our framework. First, while the typical number of views of the videos posted by the owner of this video is around dozens of thousands, this particular video has close to five million views. Moreover, by analyzing the content of the comments posted on this video, we could identify constant flaming, i.e., insulting interactions among the comments.

\begin{figure*}[!hbt]
\centering
  {\includegraphics[width=.9\textwidth]{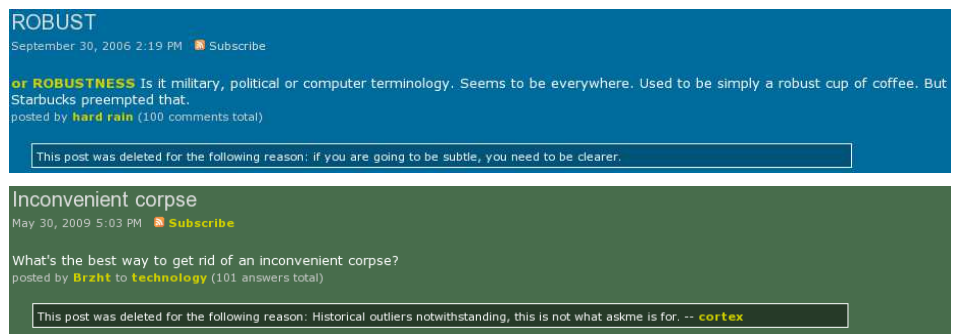}}  
  \caption{Examples of deleted posts detected as anomalies for the Askme (top) and MetaFilter (bottom) datasets.}
  \label{fig:exAnomalies}
\end{figure*}

Concerning the SMS dataset, we found several interesting anomalies of the type A3. These anomalous behavior is derived from the fact that it is common that users subscribe to automated applications which periodically send messages to them about a given topic, e.g. news, movies, sports etc. These messages are counted in the \ied of the user as a regular SMS or e-mail message, but they do not represent a social interaction. Thus, when the percentage of these messages is high, the \ied shape is modeled by two random variables, the one which represents the social interactions and the one which represents the incoming messages from automated services. If the amount of messages of the second type is significant, then the distribution deviates significantly from the one generated by the \dpp, as we observe in Figure~\ref{fig:anomalysms}. 

\begin{figure}[!hbt]
\centering
	\subfigure[Odds Ratio]
  {\includegraphics[width=.40\textwidth]{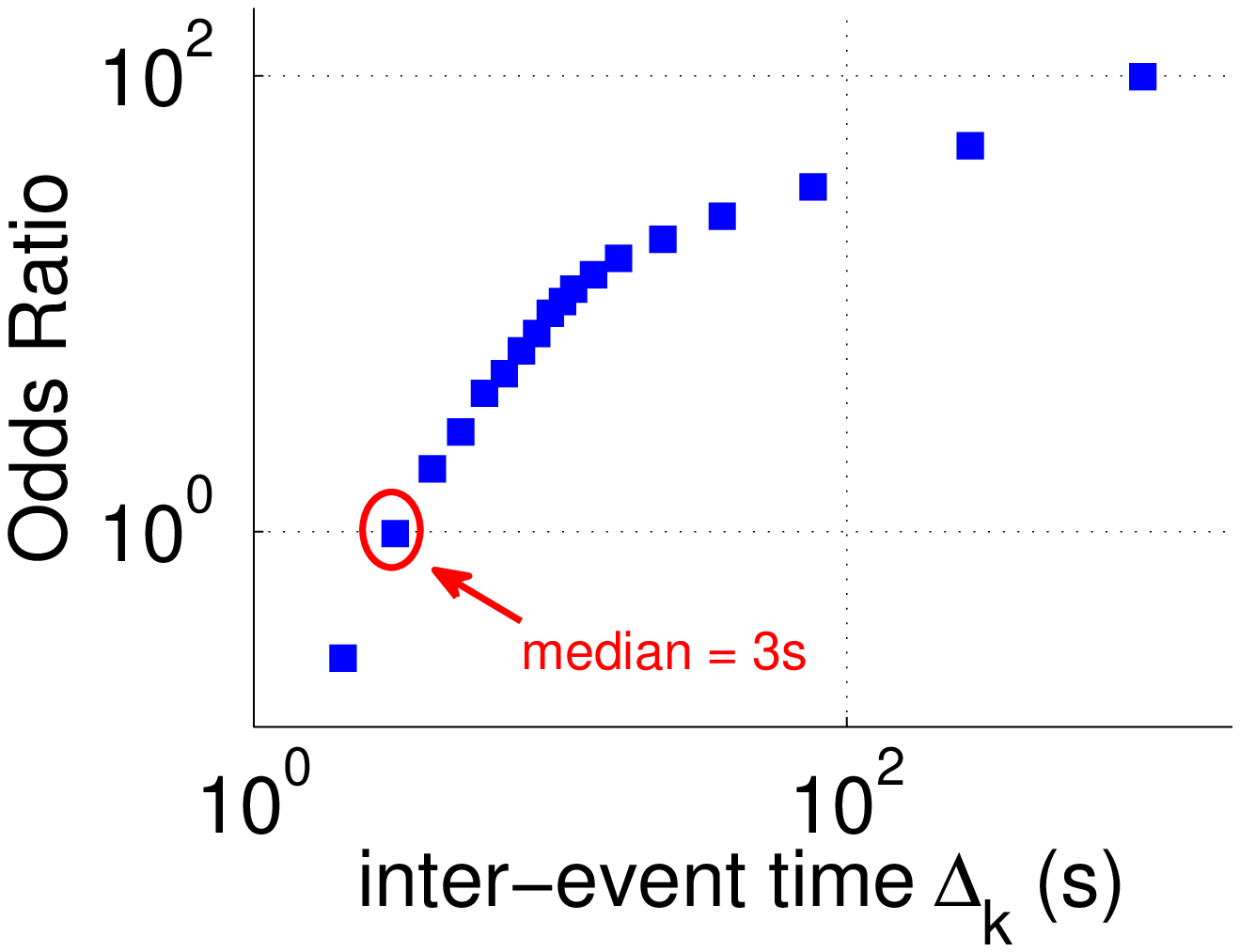}}  
  \subfigure[Histogram]
  {\includegraphics[width=.40\textwidth]{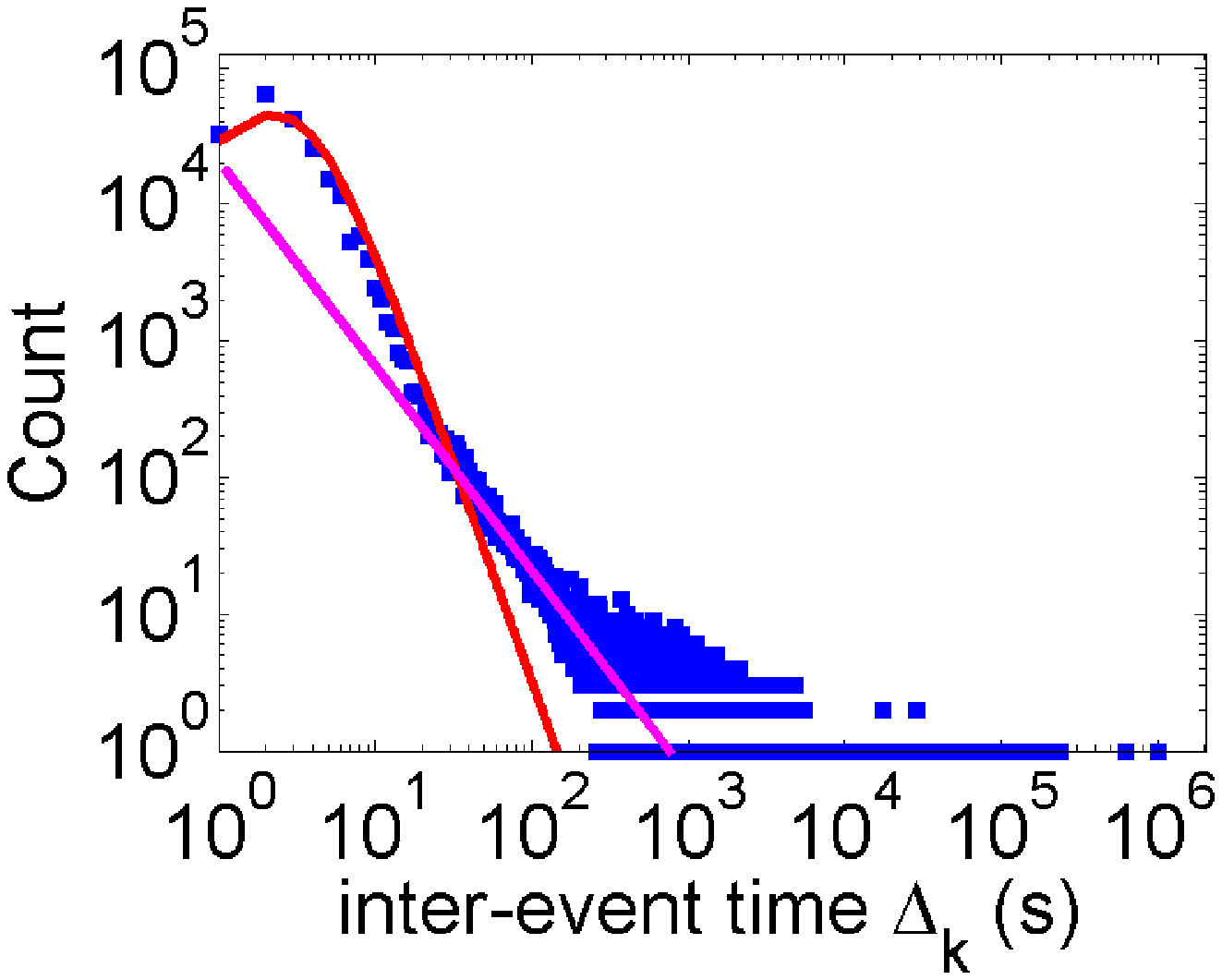}}  
  \caption{The \ied of the most active user of our 8 datasets, with 229590 communication events. Observe how the behavior differs from the \dpp or any other well known distribution, probably due to the significant number of automatic messages sent, that represents $97\%$ of the total number of communication events.}
  \label{fig:anomalysms}
\end{figure}

\section{Conclusions}
\label{sec:conclusion}

In this paper, we showed that eight different systems share four common properties in their communication dynamics. These universal properties are:

\begin{enumerate}
	\item the marginal distribution of the time intervals between communications follows an odds ratio power law;
	\item the slope of this power law is typically $1$;
	\item individual sequences of communications tend to show a high dependence between consecutive inter-arrival times;
	\item the collection of individual \ieds is very well modeled by a Bivariate Gaussian Distribution.
\end{enumerate}

Moreover, we proposed the \dpp model, which reconciles previous approaches for human communication dynamics and also is able to generate communication events that match all the four universal properties listed above. Finally, we showed that from the knowledge presented in this paper is possible to generate realistic synthetic datasets of communications and to spot anomalies.

\appendix
\section*{APPENDIX}

\section{Parameters}
\label{sec:parameters}

Before reaching the generalized \dpp model described in the paper, we had a simpler version of it, relying on a different parametrization scheme:

\begin{model}
Self-Feeding Process \dpp(C,a).
\begin{equation}
\boxed{
\begin{array}{rcl}
			 \delta_1 &\leftarrow& C\\
       \delta_t &\leftarrow& \mbox{Poisson Process}(\beta = \delta_{t-1} + C) \\
       \Delta_k &\leftarrow& \delta_t^a,   
\end{array}
}
\end{equation}
\label{eq:dppC}
\end{model}

where $C$ is the location parameter and $a$ is the shape parameter that defines the odds ratio slope $\rho$. An easy and direct way to define the relationships between this model's parameters and the distribution properties $\mu$ and $\rho$ is through simulations. 

Thus, the first point we consider is the median $\mu$ of the inter-event times generated by the \dppzero{} model when $a=1$. When $OR(x)=1$, $x$ is the median $\mu$ of the distribution. Thus, in Figure~\ref{fig:cvsmu}-a, we plot the OR for different values of $C$. We observe that changing the value of $C$ changes $\mu$ and, consequently, the location of the distribution, but maintains its slope. We also see that $\mu$ is close but different than the value of $C$. 

In order to investigate the relationship between $C$ and $\mu$, we run simulations of the model for all integer values of $C$ between [1,10000]. As we observe in Figure~\ref{fig:cvsmu}-b, the median $\mu$ of the inter-event times distribution (\ied) varies linearly with $C$ according to a slope of $\approx 2.72$, that can be approximated by Euler's number $e$, in a way that $\mu \propto e\times C$. This allows us to generate inter-event times with a determined $\mu$ when the slope $\rho=1$. We ignore the constant factor $3.8$ because its $95\%$ confidence interval is $(-8.596, 16.3)$, which contains zero.

\begin{figure}[!hbtp]
\centering
\subfigure[The OR of the \ied for different values of $C$]
  {\includegraphics[width=.40\textwidth]{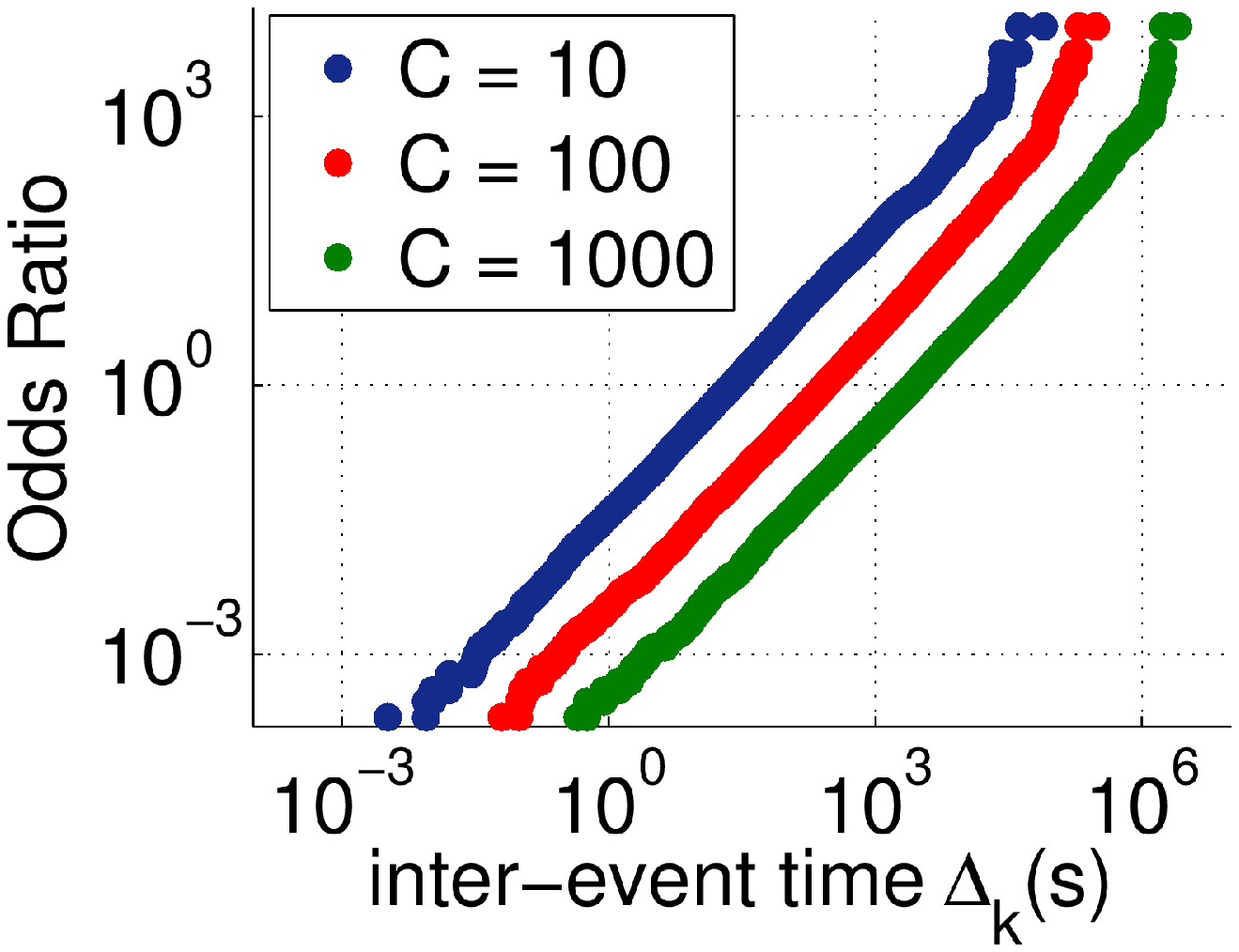}}
\subfigure[$\mu$ as a function of $C$]
  {\includegraphics[width=.40\textwidth]{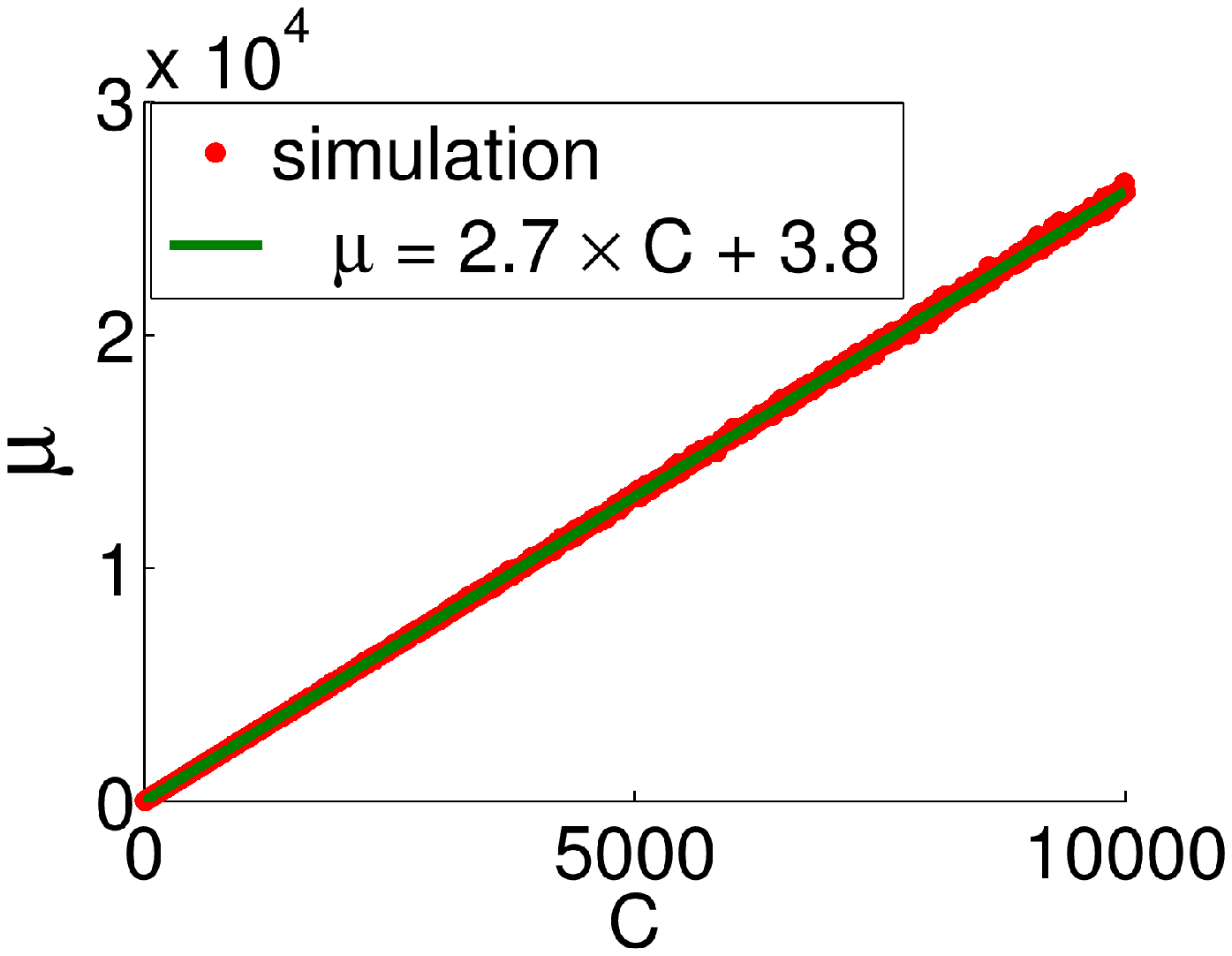}}
  \caption{Changing the value of $C$ changes the location of the distribution. The median of the distribution $\mu$ varies linearly with $C$, $\mu = a\times C + b$, with $a=2.719$ and $b=3.8$. The $95\%$ confidence interval for $a$ is $(2.715, 2.723)$ and for $b$ is $(-8.60, 16.3)$. Since the confidence interval for $b$ contains $0$, $b$ is not significant.}
  \label{fig:cvsmu}
\end{figure}

Now we know how to generate inter-event times with different medians $\mu$ using the parameter $C=\mu/e$ of \dppzero. The next step is to verify how the \dpp{} model can generate IEDs with a desired slope $\rho \neq 1$. Considering that up to this point the \dpp{} model generates a set of inter-event times $I_1$ with a slope $1$, the idea is to use an exponent $a$ to transform $I_1$ into $I_{\rho}$, which is an IED with a different slope $\rho$. When we elevate each $\Delta_k \in I_1$ to the power of $a \neq 1$, the resulting slope ${\rho}$ becomes different from $1$, as we see in Figure~\ref{fig:rhovsa}-a. In the same way we did for $C$, we run simulations of the model for 1000 different values of $a \in [0.1,2]$. As we observe in Figure~\ref{fig:rhovsa}-b, there is an inverse relationship between $a$ and $\rho$, i.e., $\rho = a^{-1}$. Moreover, since the median of the distribution is also elevated to the power of $a$, we have to elevate the parameter $\mu$ to the power of $\rho = a^{-1}$ to preserve the median.

\begin{figure}[!hbtp]
\centering
\subfigure[The OR of the \ied for different values of $a$]
  {\includegraphics[width=.40\textwidth]{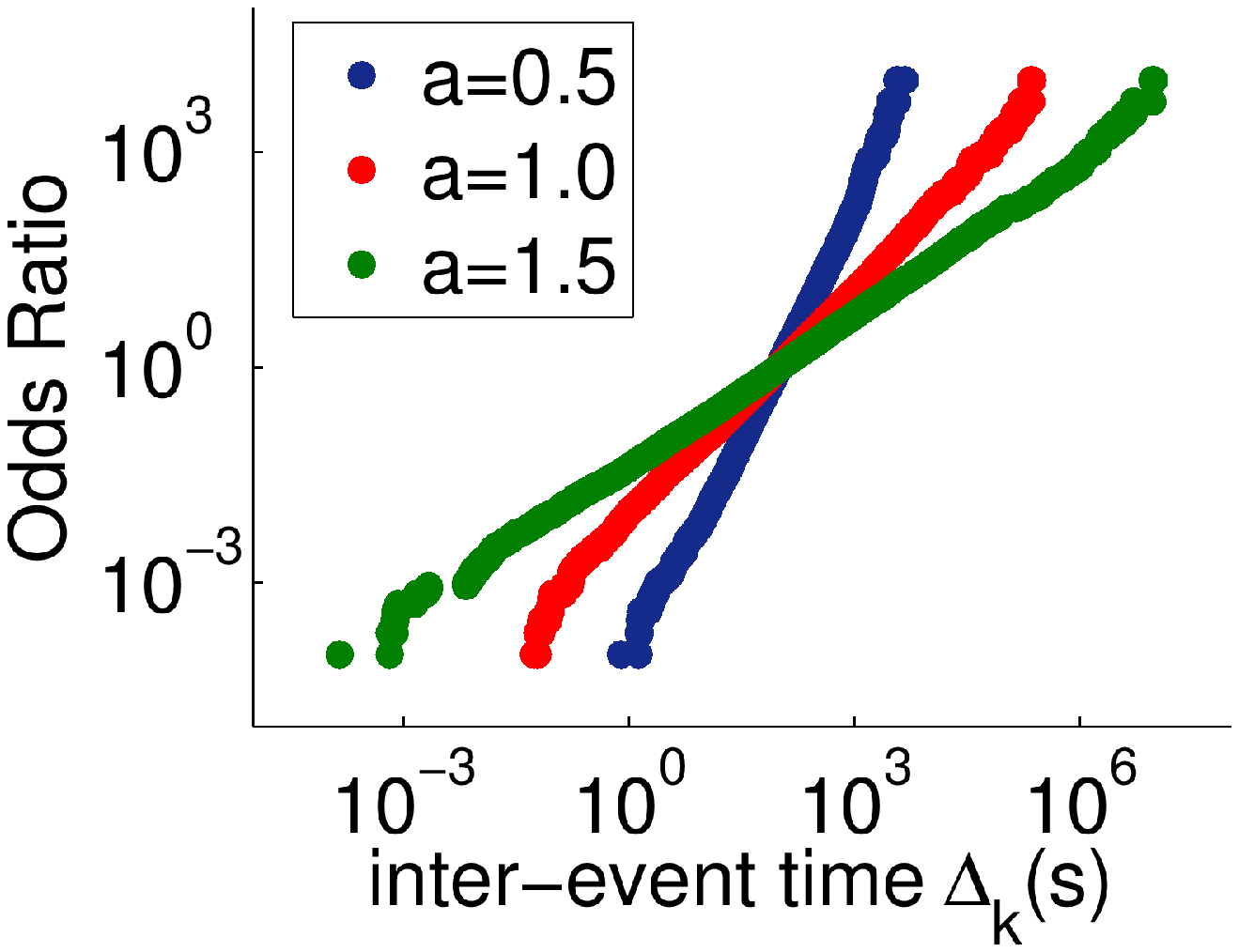}}
\subfigure[$\rho$ as a function of $a$]
  {\includegraphics[width=.40\textwidth]{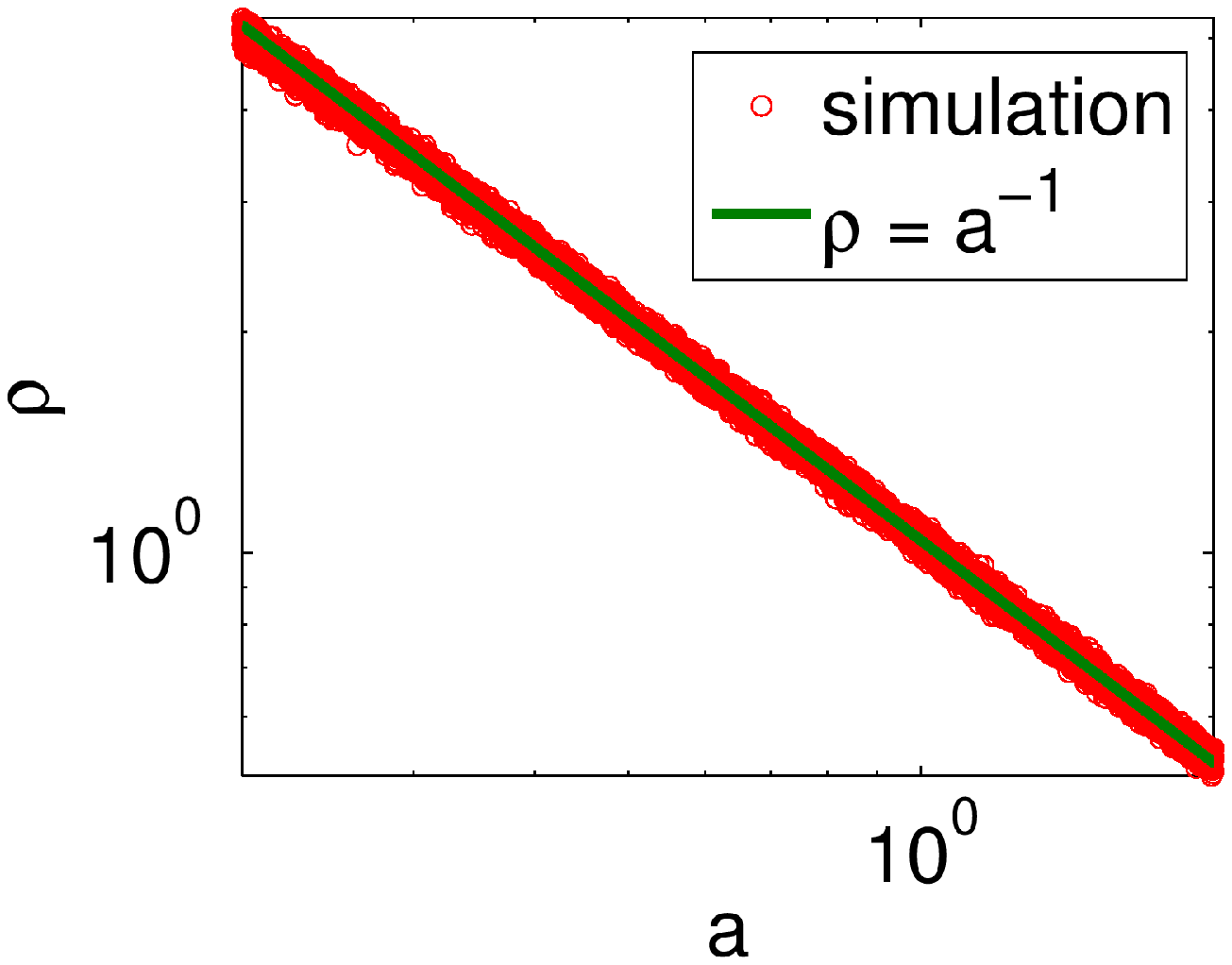}}
  \caption{Changing the value of $a$ changes the slope $\rho$ of the distribution in a way that $\rho = a^{-1}$.}
  \label{fig:rhovsa}
\end{figure}

Regarding the need for the constant $\mu/e$ in \dpp, we propose the following lemma:

\begin{lemma}
The constant $C = \mu/e > 0$ of Model~\ref{eq:dpp} is needed to assure that the inter-event times generated by the \dpp{} model will not converge to zero.
\end{lemma}
\begin{proof}
If we remove the constant $C$ from Model~\ref{eq:dpp}, $\Delta_k = (\Delta_{k-1})\times(-\ln(U(0,1)))$, or $\Delta_k$ will be equal to $\Delta_{k-1}$ multiplied by a random number $X$ extracted from the exponential distribution with parameter $\beta=\lambda=1$. If ($X=\frac{1}{k}\ |\ k>1$), then $\Delta_k$ will be equal to $\Delta_{k-1}$ divided by $k$. The probability of $X$ to be $\frac{1}{k}$ is
 $P(X=\frac{1}{k}) = e^{-\frac{1}{k}} = \frac{1}{\sqrt[k]{e}}$. On the other hand, the probability of multiplying $\Delta_{k}$ by $k$ and, therefore, return $\Delta_{k+1}$ to $\Delta_{k-1}$ value is $P(X=k) = e^{-k} = \frac{1}{e^k}$. Given these probabilities, observe that $P(X=\frac{1}{k}) = \frac{1}{\sqrt[k]{e}} > P(X=k) = \frac{1}{e^k}, \forall k>1$. From this, we conclude that the expected value of $\Delta_k$ when $t\rightarrow\infty$ is $0$. With $C$ in the equation, even when $\Delta_{k-1} = 0$, $\Delta_k = -C\times \ln(U(0,1)$, that is a classic Poisson process with $\beta = C$, and, obviously, does not converge to $0$.
\end{proof}

\section{The \dpp Stationary Distribution}
\label{sec:dppdist}

In this section, we analyze the properties of the stationary distribution generated by the \dpp. All the results shown in this section are coherent with the following conjecture:
\begin{conjecture}
The \dppzero{} model generates a log-logistic distribution with $\rho=1$, \label{conj:dpp}
\end{conjecture}
where $\rho = 1/\sigma$ and $\sigma$ is the shape parameter of the log-logistic distribution.

As we show in this section, we have several and significant evidences that the \dpp{} generates a log-logistic distribution, but at this moment we do not have a formal analytical proof that this is true. 

\subsection{Log-logistic Distribution}
\label{sec:loglogistic}

The log-logistic distribution was first proposed by Fisk~\cite{fisk:1961} to model income distribution, after observing that the OR plot of real data in log-log scales follows a power law $OR(x) = cx^{\rho}$. In summary, a random variable is log-logistically distributed if the logarithm of the random variable is logistically distributed. The logistic distribution is very similar to the normal distribution, but it has heavier tails. In the literature, there are examples of the use of the log-logistic distribution in survival analysis~\cite{bennet:1983,mahmood:2000}, distribution of wealth~\cite{fisk:1961}, flood frequency analysis~\cite{ahmad:1988}, software reliability~\cite{gokhale:1998} and phone calls duration~\cite{vazdemelo:2010}.  A commonly used log-logistic parametrization is~\cite{wiley:1982}:

\begin{equation}
\begin{array}{rcl}
PDF_{\tll}(x) &=& \frac{e^z}{\sigma x (1+e^z)^2}, \\
CDF_{\tll}(x) &=& \frac{1}{1 + e^{-z}}, \\
z &=& (\ln(x) - \ln(\mu)) / \sigma,\end{array}\label{eqn:llg}\end{equation}
where $\sigma = 1/\rho$, the slope of our \dpp{} model, and $\mu$ is the same. Moreover, when $\sigma = 1$, it is the same distribution as the Generalized Pareto distribution~\cite{Lorenz:1905} with shape parameter $\kappa = 1$, scale parameter $\mu$ and threshold parameter $\theta = 0$.

\subsection{Analytical Result}

Wold processes \cite{wold:1948,cox:1955} are stochastic processes where the inter-events intervals have a dependence following 
a Markovian property. That is, the probability distribution law of the $t$-th inter-event time  $\Delta_t$
depends only on the previous inter-event time $\Delta_{t-1}$. Our \dpp model falls within this Wold processes 
class as we assume that, conditionally on the entire previous inter-event times, the distribution of 
$\Delta_k = \delta_t^{1/\rho}$ is an exponential distribution with expected value given by $\delta_{t-1}+\mu^{\rho}/e$.
Wold processes are not well understood due to the mathematical difficulties
in deriving their probabilistic properties. 

Consider the existence of a stationary distribution for the generalized SFP model.
A stationary PDF $f(x)$ of the Markov chain $\delta_t$ must satisfy
\begin{eqnarray*}
 f(x) &=& \int_0^{\infty} f(y \rightarrow x) f(y) dy \\
      &=& \int_0^{\infty} \frac{1}{y+\mu^{\rho}/e}\exp(-x/(y + \mu^{\rho}/e)) f(y) dy \\
\end{eqnarray*}
This integral equation has no obvious analytical solution but in the next sections we show via
simulations of the point process that $f(x)$ is very well approximated
by a log-logistic density.
This mathematical difficulty is common in the previous attempts to
model data with Wold processes. Even if a consistent density
$f(x)$ and a transition kernel $f(y \rightarrow x)$ are given, properties
are, in general, difficult to obtain~\cite{cox:1980}.

Let the Markovian distribution of $\delta_t$ conditional on $\delta_{t-1} = x$ be given by an 
exponential distribution with mean $\alpha x + c$ where $0 < \alpha \leq 1$ and $c > 0$ are constants.  We have 
\[ \mu_t = \mathbb{E}(\delta_t) = \mathbb{E}\left\{ \mathbb{E}(\delta_t | \delta_{t-1}) \right\} = 
\mathbb{E}\left\{ \alpha \delta_{t-1} + c \right\} = \alpha \mu_{t-1} + c  \]
Applying recursively, we find 
\[ \mu_t = \alpha^t \mu_0 + c \sum_{k=0}^{t-1} \alpha^k \]

If $\alpha = 1$, $\mu_t = \mu_0 + tc$. Assuming that this process is stationary 
implies that $\mu_t = \mu_0$ is constant and the only solution is
to take $\mu_t = \mu_0 = \infty$. Hence the process has infinite mean, 
as is the case of the log-logistic distribution with shape parameter equal or smaller than 1.

If $\alpha < 1$, then 
\[ \mu_t = \alpha^t \mu_0 + c \frac{1 - \alpha^{t}}{1-\alpha} \rightarrow \frac{c}{1-\alpha} \]
Obviously, $\mu_t = \mu = c/(1-\alpha)$ is a solution to the recursive equation $\mu_t = \alpha \mu_{t-1} + c$.

When we have a finite expectation for $\delta_t$ we can calculate the variance $\mathbb{V}(\delta_t)=\sigma^2_t$:
\begin{eqnarray*}
  \sigma^2_t &=& \mathbb{E} \left\{ \mathbb{V}(\delta_t|\delta_{t-1}) \right\} + \mathbb{V} \left\{ \mathbb{E}(\delta_t|\delta_{t-1})  \right\} \\
   &=& \mathbb{E} \left\{ (\alpha \delta_{t-1}+ c)^2 \right\} + \mathbb{V} \left\{  \alpha \delta_{t-1}+ c \right\} \\
   &=& \alpha^2 \left( \sigma_{t-1}^2 + \mu_{t-1}^2 \right) + 2c\alpha \mu_{t-1} + c^2 + \alpha^2 \sigma_{t-1}^2  \\
   &=& 2 \alpha^2 \sigma_{t-1}^2 + (\alpha \mu_{t-1} + c)^2  
\end{eqnarray*}
Assuming that the process is stationary, we have $\mu_t = c/(1-\alpha)$ and $\sigma^2_t = \sigma^2$ constant, which implies
into 
\[ 0 < \sigma^2 = 2 \alpha^2 \sigma^2 + (c/(1-\alpha))^2 \]
If $\alpha < 1/\sqrt{2}$, this has a solution as

\begin{eqnarray*}
	\sigma^2 &=& \left( \frac{c}{1-\alpha} \right)^2 \frac{1}{(1-\sqrt{2}\alpha)(1+\sqrt{2}\alpha)} \\
	&=& \mu^2  \frac{1}{(1-\sqrt{2}\alpha)(1+\sqrt{2}\alpha)}
\end{eqnarray*}


Therefore, if $\alpha = 1$, the process has infinite mean. If $\alpha < 1$, the expected value is finite and equal to $\mu = c/(1-\alpha)$.
Concerning the variance, it exists only if  $\alpha < 1/\sqrt{2} \approx 0.70$ and, 
in this case, we have $\sigma^2 = \mu^2/(1-2\alpha^2)$. 

\subsection{Fitting Synthetic Data}

In Figure~\ref{fig:model}-a, we plot the histogram of 100,000 time intervals $\Delta_k$ generated by the \dppzero{} model with $\mu = e$. Moreover, in Figure~\ref{fig:model}-b, we plot the OR for the same time intervals. While a classic PP generates an exponential distribution, we observe that the generated data by the \dppzero{} perfectly fits a distribution with an Odds Ratio function that is a power law with slope $\rho = 1$. This is also coherent with Conjecture~\ref{conj:dpp}.

\begin{figure}[!hbtp]
\centering
\subfigure[Histogram (marginal).]
  {\includegraphics[width=.40\textwidth]{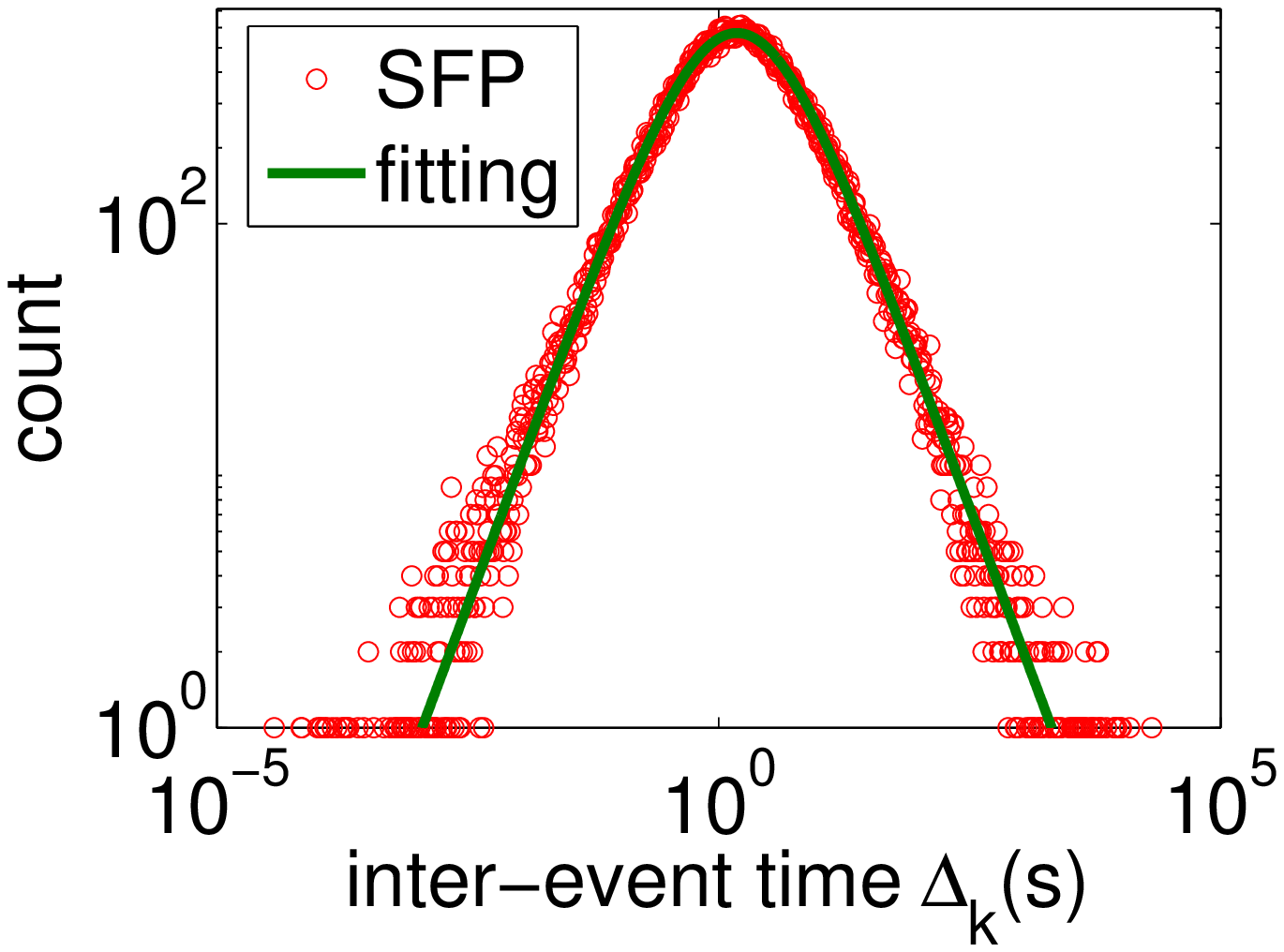}}
\subfigure[OR (marginal).]
  {\includegraphics[width=.40\textwidth]{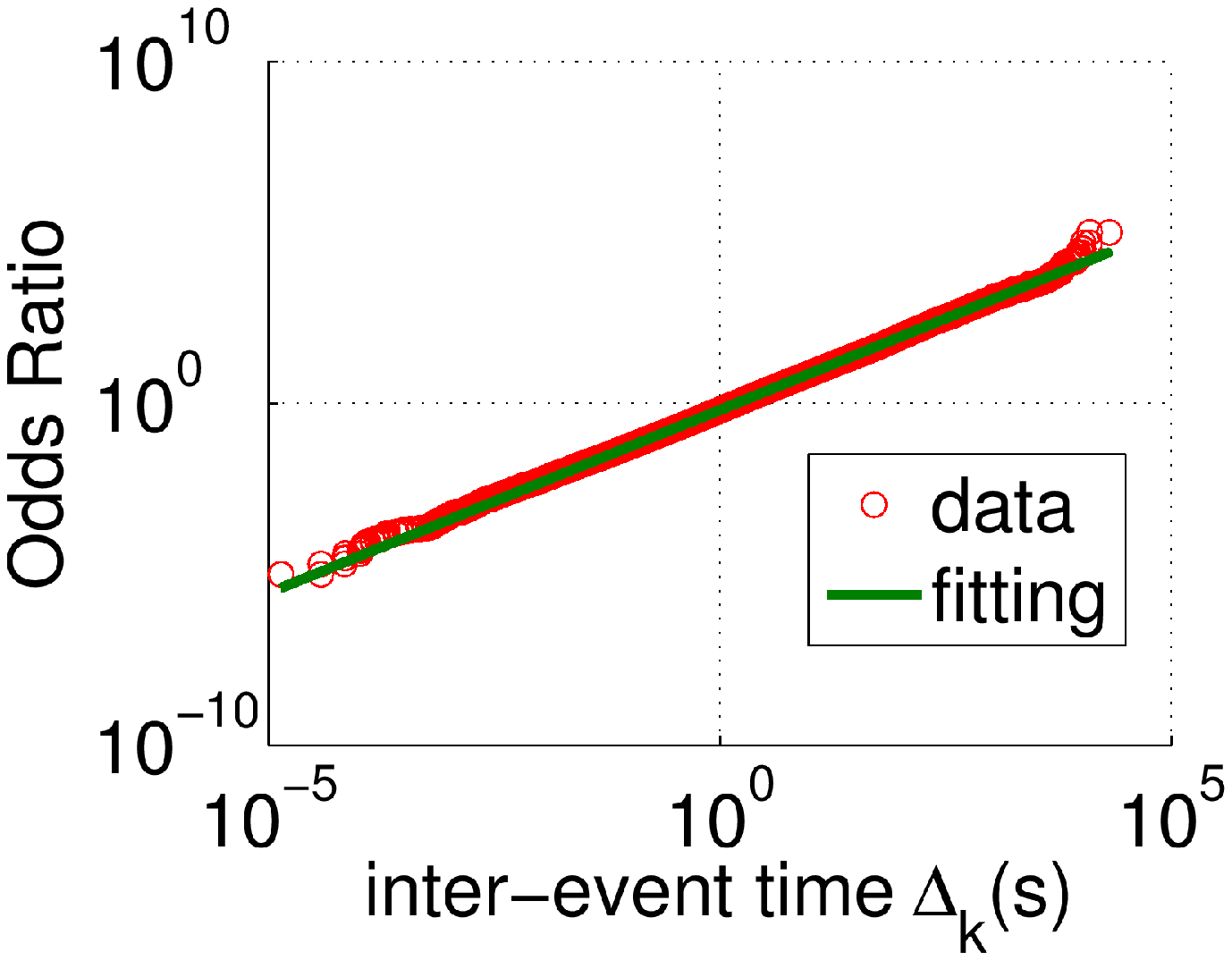}}
  \caption{Inter-event times $\Delta_k$ generated by the \dppzero{}. The generated $\Delta_k$s are perfectly fitted by a log-logistic distribution with the slope $\rho = 1$.}
  \label{fig:model}
\end{figure}

\subsection{The \dpp Markov Chain}
\label{sec:markov}

The \dpp can be naturally considered as a Markov Chain (MC), since it is a sequence of random variables $\Delta_1, \Delta_2, \Delta_3, ...$ with the Markov property, namely that, given the present inter-event time, or state, the future and past inter-event times, or states, are independent. Thus, here we model the \dpp as a time-homogeneous Markov chain with a finite state space to give another evidence that the \dpp has a stationary distribution and that is very likely that this distribution is the log-logistic.

Originally, the \dpp can be considered as a continuous-time MC, but for simplicity, we build a discrete-time Markov chain in a way that each state $i = \{1, 2, 3, ...\}$ is associated with an inter-event time $\Delta_i = \{\Delta_1, \Delta_2, \Delta_3, ...\}$ with values within the interval $(i-1,i]$. For instance, considering the granularity in seconds, if the current inter-event time is $3.8$ seconds, then the MC is in the state $4$. Also for simplicity, we build a finite-state MC with a maximum number of states $n$, i.e., the states go from $1$ to $n$. The MC will be in state $n$ every time the current inter-event time is within the interval $(n,\infty)$.

Thus, considering a $n$-state MC build from the \dpp model, the transitions probabilities $p_{i,j}$ of going from state $i$ to $j$ are given in the following way:

\scriptsize
\[
  p_{i,j} = \left\{ 
  \begin{array}{l l}
    CDF_{exp}(\text{x=j,$\beta$=i+C}) - CDF_{exp}(\text{x=j-1,$\beta$=i+C}) & \quad \text{if j$<$n}\\
    1-CDF_{exp}(\text{x=j,$\beta$=i+C}) & \quad \text{if j=n},\\
  \end{array} \right.
\]
\normalsize
where $CDF_{exp}(x,\beta)$ is the cumulative distribution function of the exponential distribution on $x$ with mean $\beta$ and $C=\mu/e$, given in the \dpp (Equation \ref{eq:dppzero}). Observe in Figure~\ref{fig:mc} that the  probability density function of the log-logistic is virtually identical to the one of the stationary distribution of the \dpp Markov Chain. This is another strong indication that the \dpp generates log-logistically distributed data.

\begin{figure}[!hbt]
\centering
  {\includegraphics[width=.40\textwidth]{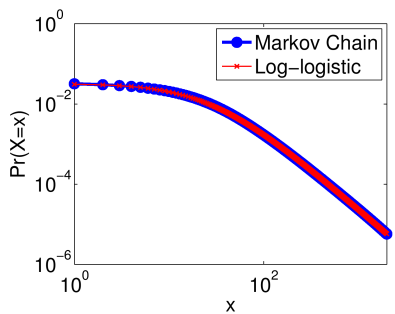}}
\caption{The probability density function of the log-logistic distribution and the stationary distribution of the \dpp MC.}
  \label{fig:mc}
\end{figure}

It is important to point out that in~\cite{Chierichetti:www:2012} the authors showed that behavior of the Web users is not Markovian, i.e., a user's next action does not depends only on her/his current state. Our assumption differs from this one because we assume that users have Markovian behavior in communications, while~\cite{Chierichetti:www:2012} studied whether users have Markovian behavior while navigating on the Web.

\subsection{ Solving by Computation}

If the \dppzero{} model generates a log-logistic distribution with slope $\rho=1$, then we can write the PDF of the inter-event times generated by the \dppzero{} model as

\begin{equation}
f(x) = \int_{0}^{\infty}{PDF_{LLG}(y; \rho=1,\mu) \times PDF_{EXP}(x;\beta=y+C)dy}, \label{eq:dppzeroint}
\end{equation}

or

\begin{equation}
f(x) = \int_{0}^{\infty}{(\frac{1}{\mu} \times \frac{1}{(1 + \frac{y}{\mu})^2}) \times \frac{e^{\frac{-x}{y+c}}}{y+c}dy}. \label{eq:dppzerofull}
\end{equation}

%
%
%

In this way, we verify if Equation~\ref{eq:dppzerofull} is correct by numerically evaluating the integral using the adaptive Gauss-Kronrod quadrature method for every $x \in ]0:10^4]$ and comparing the result with the $PDF_{LLG}(x)$. As we can see in Figure~\ref{fig:sfpVsLLG2}, the proposed PDF and the $PDF_{LLG}$ match perfectly for different values of $\mu$.

\begin{figure}[!hbtp]
\centering
  {\includegraphics[width=.40\textwidth]{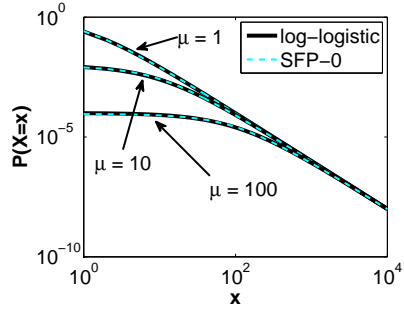}}
  \caption{The PDF of the log-logistic distribution VS. the PDF of Equation~\ref{eq:dppzerofull}, that is derived from the \dppzero{}.  Observe that the PDFs are identical for different values of $\mu$. We numerically evaluated the integral using the adaptive Gauss-Kronrod quadrature method.}
  \label{fig:sfpVsLLG2}
\end{figure}


\subsection{\dpp{} has Power Law Tail}
\label{sec:powerlaw}

The universality class model proposed by Barab\'{a}si~\cite{barabasi:2005} states that the IED has a power law tail. The proposed \dpp{} model agrees with this model in a way that:
\begin{lemma}
If Conjecture~\ref{conj:dpp} is correct, then the \dpp{} model generates an IED that converges to a power law when $x\rightarrow\infty$, i.e., $\lim_{x\to\infty}\frac{PDF_{\tll}(x)}{x^{-\alpha}} = k$, where $k$ is a constant greater than $0$. 
\end{lemma}
\begin{proof}
Considering the Probability Density Function of the log-logistic distribution showed in Equation~\ref{eqn:llg}, if we set the location parameter $\mu=1$ for simplicity, $e^z = x^{1/\sigma}$. Then, $PDF_{\tll}(x)$ can be simplified to
\begin{equation*}
PDF_{\tll}(x) = \frac{x^{\frac{1}{\sigma}-1}}{\sigma (1+x^{\frac{1}{\sigma}})^2}.
\end{equation*}
When $x\to\infty$, the addition of $1$ in the denominator can be disregarded, resulting in the following simplification:
\begin{equation*}
PDF_{\tll}(x) = \frac{x^{-(1+1/\sigma)}}{\sigma}, {x\to\infty}.
\end{equation*}
Thus, when $x\rightarrow\infty$, the \ied generated by the \dpp{} model is a power law with slope
\begin{equation}
\alpha = -(1+1/\sigma) = -(1+\rho).
\end{equation}
Observe again Figure~\ref{fig:model}-a and note the power law tail.
\end{proof}

\section{Generating Realistic Inter-event Times}
\label{sec:deviations}

Time intervals between communications are very hard to model. We showed in this work that even today there is no consensus about the more appropriate model to represent them. One of the main reasons for that is that real data is significantly noisier, what may harms deeply the statistical analysis. In this section, we show some of these noisy behaviors seen in real data and, more important, we show how to reproduce them in order to generate more realistic data.

First, in Section~\ref{sec:minordev}, we describe minor deviations seen in data that we prefer not to approach via changes in the \dpp model. Then, in the following sections, we make slight alterations in the \dpp in order to explain the most significant deviations seen in data. In Section~\ref{sec:ltc}, we show how to generate inter-event times with more realistic correlation between consecutive inter-event times. Then, in Section~\ref{sec:mult}, we show how to mimic the inter-event times between SMS messages sent to multiple recipients. Finally, in Section~\ref{sec:phoneoverhead}, we add an overhead parameter to the \dpp model in order to explain the apparent lower bound for inter-event times seen in the phone dataset. 

\subsection{Minor Deviations}
\label{sec:minordev}

In most of the \ied{s}, there are small deviations from the OR power law for $\Delta_k$ values close to 10 hours. This is explained by the regular sleep intervals between two communication events. Since everyone has to go to sleep and it is not usual to call during the sleeping hours, it is common to have a time interval of approximately 10 hours at every 24 hours, more than what is expected by the fitting. For more details about sleep intervals, refer to~\cite{vazdemelo:2011b}. 

Additionally, it is also common to see a high amount of communication events at round times, such as 1 minute or 1 hour, because the recorded time of the event is rounded by the server, e.g. 8:54 is rounded to 9:00. We overcome such deviations by computing the odds ratio for the percentiles of the distribution. Since the typical individual has thousands of communication events, then these deviations are not shown in the odds ratio plot.

Finally, it is important to consider the deviations that appear specifically in the SMS dataset, that is the dataset that presented the worst goodness of fit result. Probably the main reason for that is the fact that a significant amount of messages arrive at their destinations with a considerable delay, such as human delay and other noisy non-regular delays caused by the mobile network infrastructure or personal issues, e.g., a customer left his mobile phone unattended and the battery died, delaying all the incoming SMS messages for when the mobile phone is recharged again. Imagine, for instance, that Smith had sent a message to John at time $t_1$ and, due to a transmission delay $d_1$, the message arrived only at $t_2$. In his turn, Smith saw the message at $t_2$ and immediately replied, but again, due to a transmission delay $d_2$, the message arrived to John only at $t_3$. Thus, for John, the inter-event time between sending the message and receiving the reply is $\Delta = t_3 - t_1 = (t_2 + d_2) - (t_1 + d_1)$, with \textit{two} transmission delays embedded in the registered inter-event time.

\subsection{Lower Temporal Correlation}
\label{sec:ltc}

The \dpp model is build upon a direct dependence between consecutive inter-event times. Because of that, the correlation between consecutive inter-event times is significantly higher than real data. While the average Pearson's correlation coefficient for real data is approximately $0.4$, for synthetic data generated by the \dpp model is approximately $0.7$. In order to generate more realistic data, we suggest a slight modification in the \dpp process. Instead of generating the next inter-event time ($\Delta_k$) based on the immediate previous one ($\Delta_{k-1}$), we propose that it should be generated from a $\epsilon$-th previous one ($\Delta_{k-\epsilon}$). This can be done by extracting $\epsilon$ from an exponential distribution with mean $\beta=1$ and making its ceiling, so the lower bound for $\epsilon$ is $1$. In summary, the \dpp model is changed as follows:

\begin{model}
Self-Feeding Process* \dppzero($\mu$). 

//$\mu$ is the desired median of the marginal PDF
\begin{equation*}
\boxed{
\begin{array}{rcl}
\Delta_1 &\leftarrow& \mu \\
\epsilon &\leftarrow& \lceil  \mbox{Exponential  }( \mbox{mean~} \beta = 1) \rceil \\
\Delta_k &\leftarrow & \mbox{Exponential  }( \mbox{mean~} \beta = \Delta_{k-\epsilon} + \mu/e) \\
\end{array}}
\end{equation*}
\label{eq:dppcorr}
\end {model}

Observe in Figure~\ref{fig:ltc} that the synthetic data generate by the SFP* has a lower correlation ($0.43$) between consecutive inter-event times than the original one ($0.70$). Despite of that, the odds ratio generated by the SFP* is still a power law with slope $\rho \approx 1$.

\begin{figure*}[!hbt]
\centering
\subfigure[Correlations in \dpp]
  {\includegraphics[width=.23\textwidth]{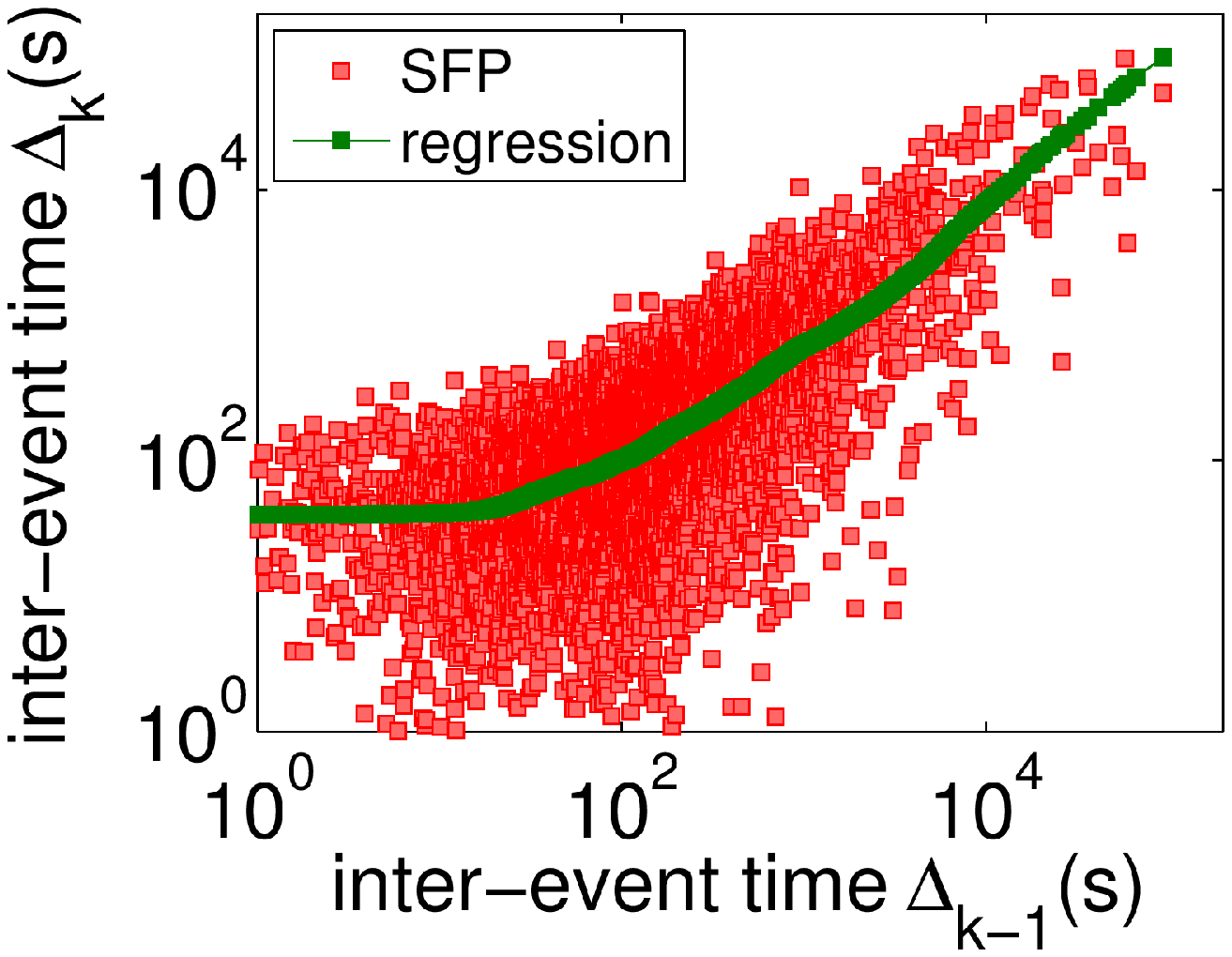}}  
\subfigure[Correlations in \dpp*]
  {\includegraphics[width=.23\textwidth]{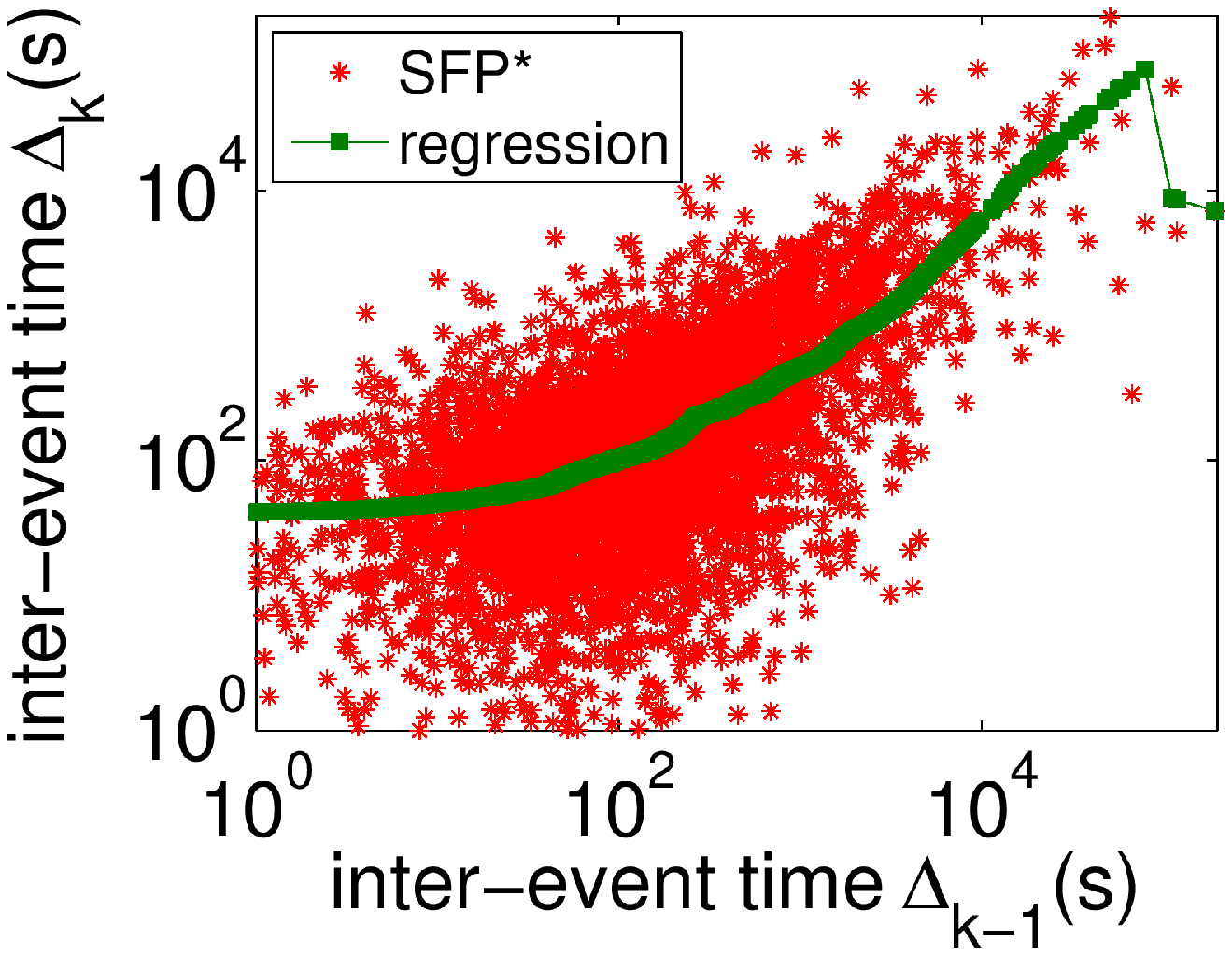}}     
\subfigure[Odds Ratio for \dpp]
  {\includegraphics[width=.23\textwidth]{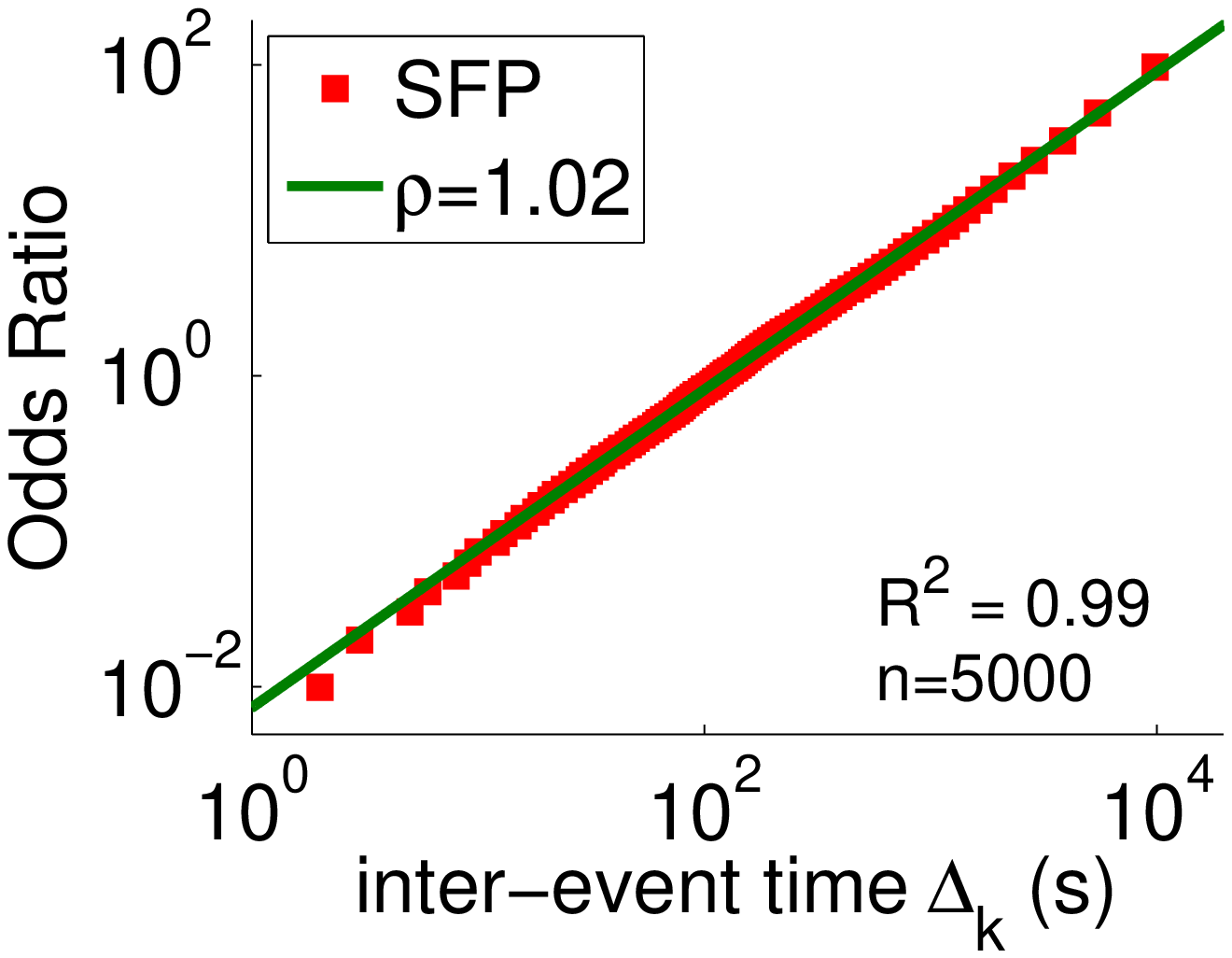}}  
\subfigure[Odds Ratio for \dpp*]
  {\includegraphics[width=.23\textwidth]{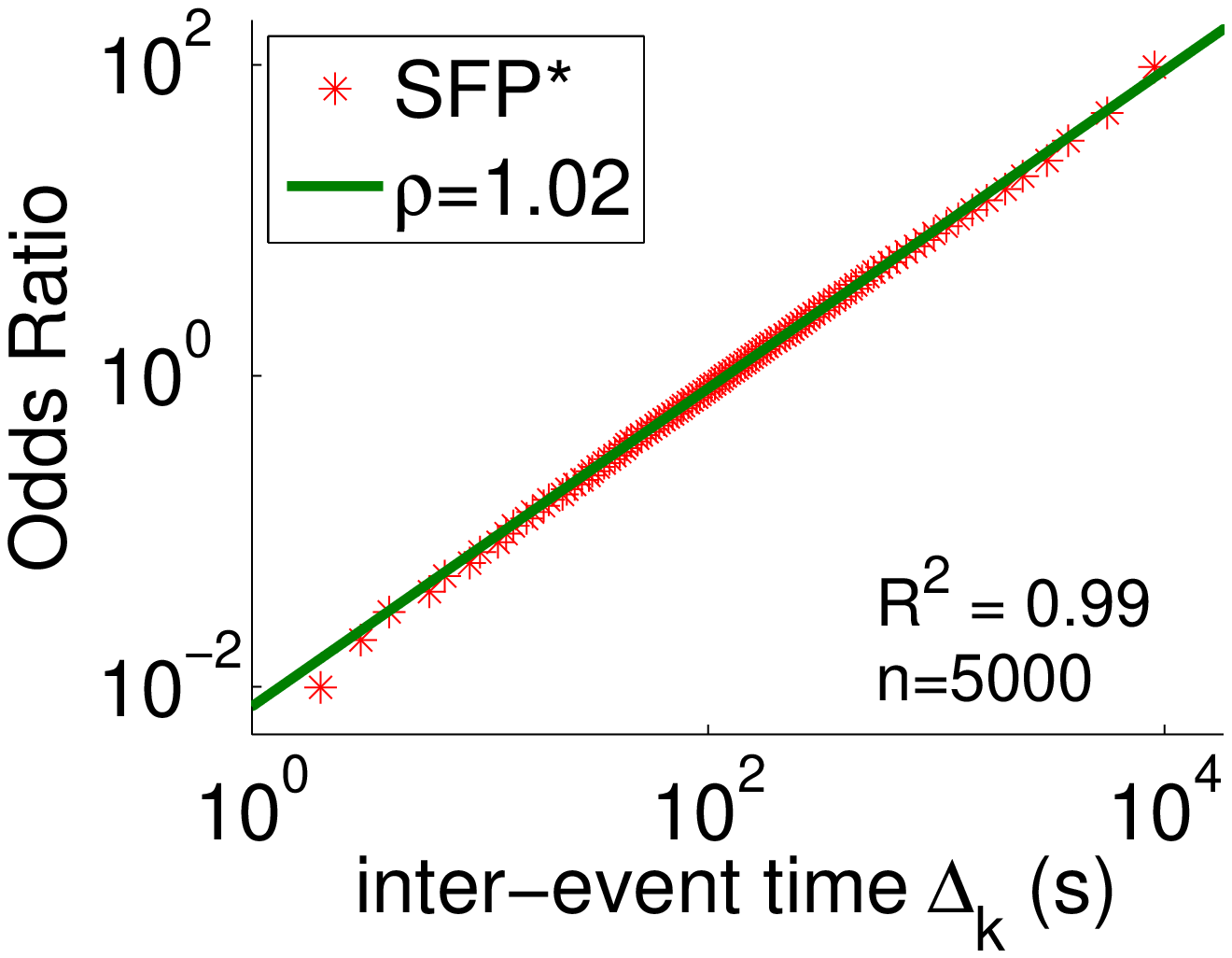}}     
  \caption{Comparison between the synthetic data generated by the \dpp and the \dpp*. Observe that the synthetic data generate by the \dpp* has a lower correlation ($0.43$) between consecutive inter-event times than the \dpp ($0.70$). Despite of that, the odds ratio generated by the \dpp* is still a power law with slope $\rho \approx 1$.}
  \label{fig:ltc}
\end{figure*}

\subsection{Multiple Recipients and Sleep Intervals}
\label{sec:mult}

When sending SMS messages it is common to copy the message to multiple recipients. Because of that, several \ied{s} of the SMS dataset show a deviation from the OR power law in the first seconds of the distribution. In order to mimic this behavior, we propose two small changes in the \dpp. First, we generate the number of recipients of a communication event from an exponential distribution (e.g.: mean $\beta = 1$). Then, we send the communication event to every recipient with a delay also extracted from an exponential distribution (e.g.: mean $\beta = 1$). As we observe in Figure~\ref{fig:deviation1}, these small changes are able to represent the inter-event times sent to multiple recipients.

\begin{figure}[!hbt]
\centering
  {\includegraphics[width=.40\textwidth]{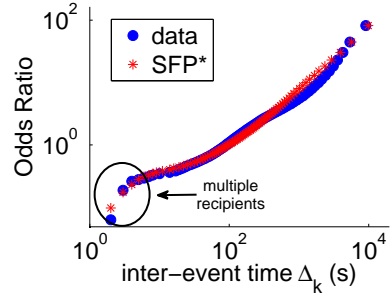}}      
  \caption{
  The comparison of the between the \ied of a typical SMS user with the one synthetically generated by the \dpp* model described in this section.}
  \label{fig:deviation1}
\end{figure}

\subsection{Phone Dialing Speed Limit}
\label{sec:phoneoverhead}

For the majority of phone users, the number of $\Delta_k$ values close to 10 seconds is underestimated by the odds ratio power law fitting. This happens because the IED of phone data is usually lower bounded by the setup time $\Delta_k^0$ of making a phone call, that involves dialing the numbers, waiting for the signal, and waiting for the other part to answer the call. We can mimic this behavior by simply adding a overhead constant $\theta$ to every inter-event time generated by the \dpp model, representing the time it takes for a individual to dial and wait for the reply. We change the \dpp model in the following way:

\begin{model}
Self-Feeding Process* \dppzero($\mu$, $\theta$). 

//$\mu$ is the desired median of the marginal PDF
//$\theta$ is the usual time it takes for a individual to dial and get the reply
\begin{equation*}
\boxed{
\begin{array}{rcl}
\Delta_1 &\leftarrow& \mu + \theta\\
\Delta_k &\leftarrow & \mbox{Exponential  }( \mbox{mean~} \beta = \Delta_{k-\epsilon} + \mu/e) + \theta\\
\end{array}}
\end{equation*}
\label{eq:dppspeed}
\end {model}

Observe in Figure~\ref{fig:deviation2} that this simple modification can accurately mimic the phone dialing speed limit seen in real data.

\begin{figure}[!hbt]
\centering
  {\includegraphics[width=.45\textwidth]{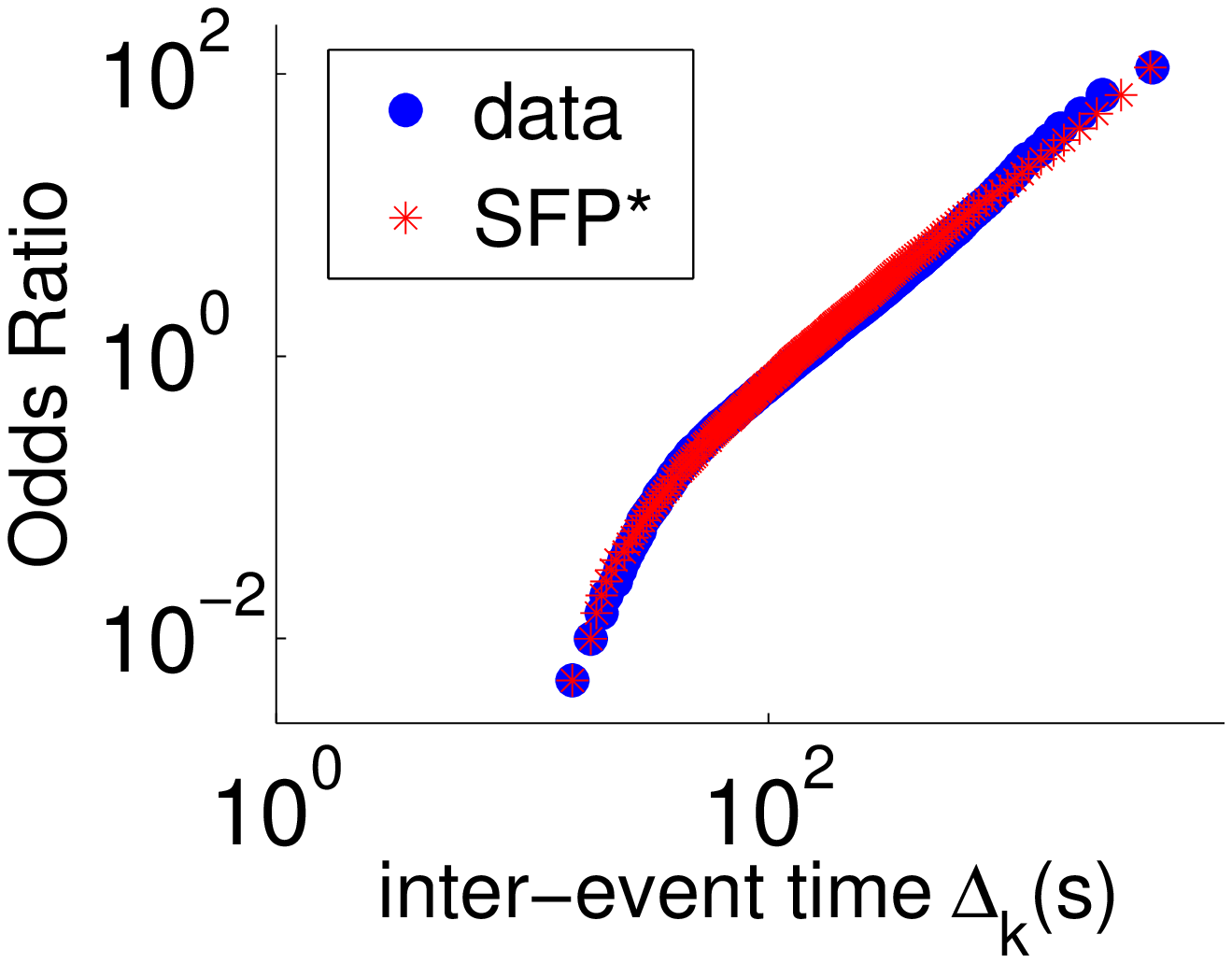}}  
  \caption{\dpp* vs. real data}
  \label{fig:deviation2}
\end{figure}

\newpage

\section{\dpp Code}

Below we show the \textit{Python} code for the \dpp generator.

\scriptsize
\begin{verbatim}
def SFP(n, mu, rho=1):
    #first inter-event time
    deltat = mu
    #list of inter-event times
    Deltat = []
    for i in range(1, n):
        #Poisson Process which Beta=deltat+mu/e
        deltat = -(deltat+(mu**rho)/math.e) 
        deltat = deltat * math.log(random.random()) 
        Deltat.append(deltat**(1/rho))
    return Deltat
\end{verbatim}
\normalsize

\section{Data Sample}
\label{sec:sample}

In this section we show the \ied of 8 typical talkative individuals of each dataset. In each figure we show the identification of the individual (when possible), the slope $\rho$, the determination coefficient $R^2$ and the number of communication events $n$. 

\begin{figure*}[htpb]
\centering
{\includegraphics[width=.23\textwidth]{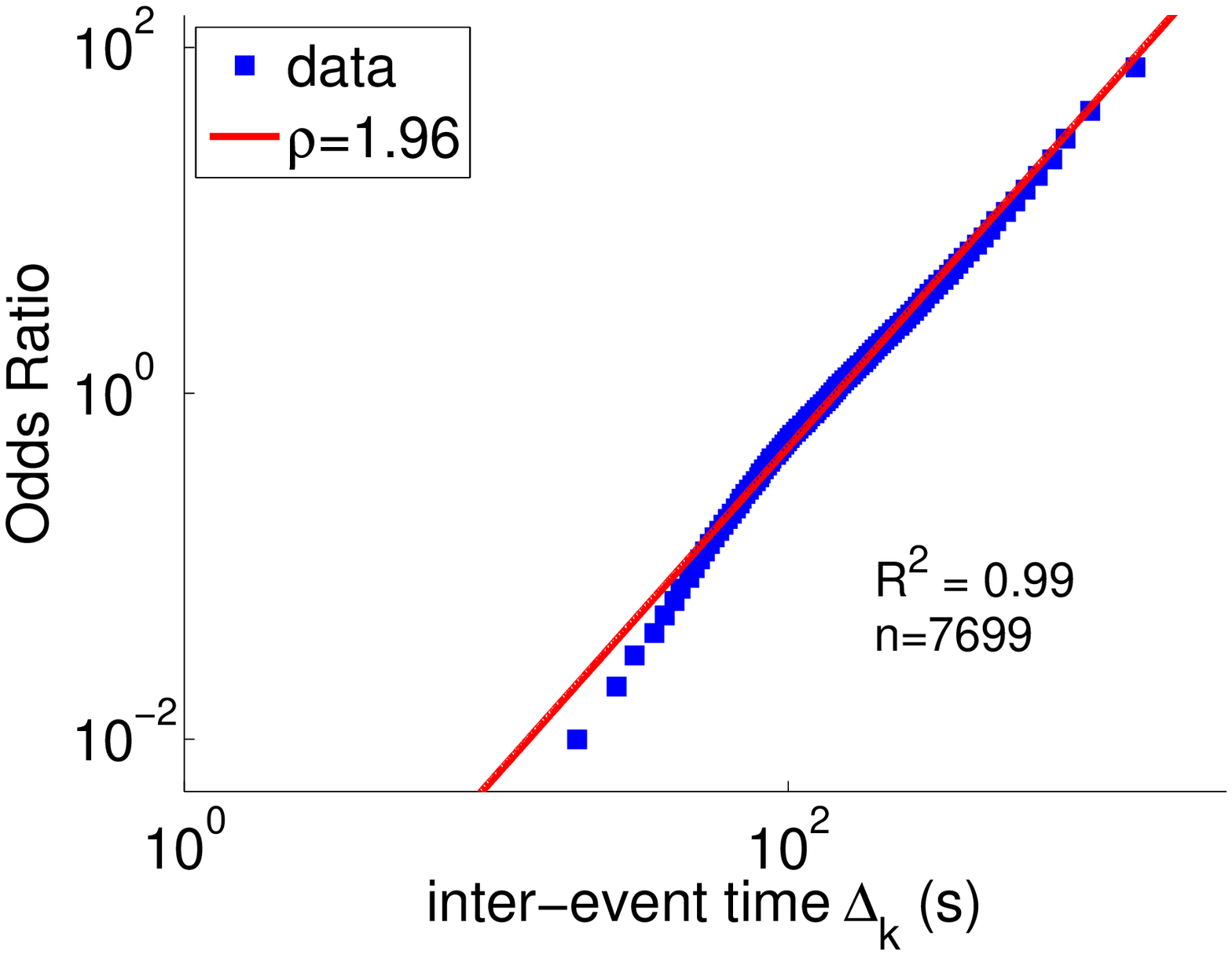}}
{\includegraphics[width=.23\textwidth]{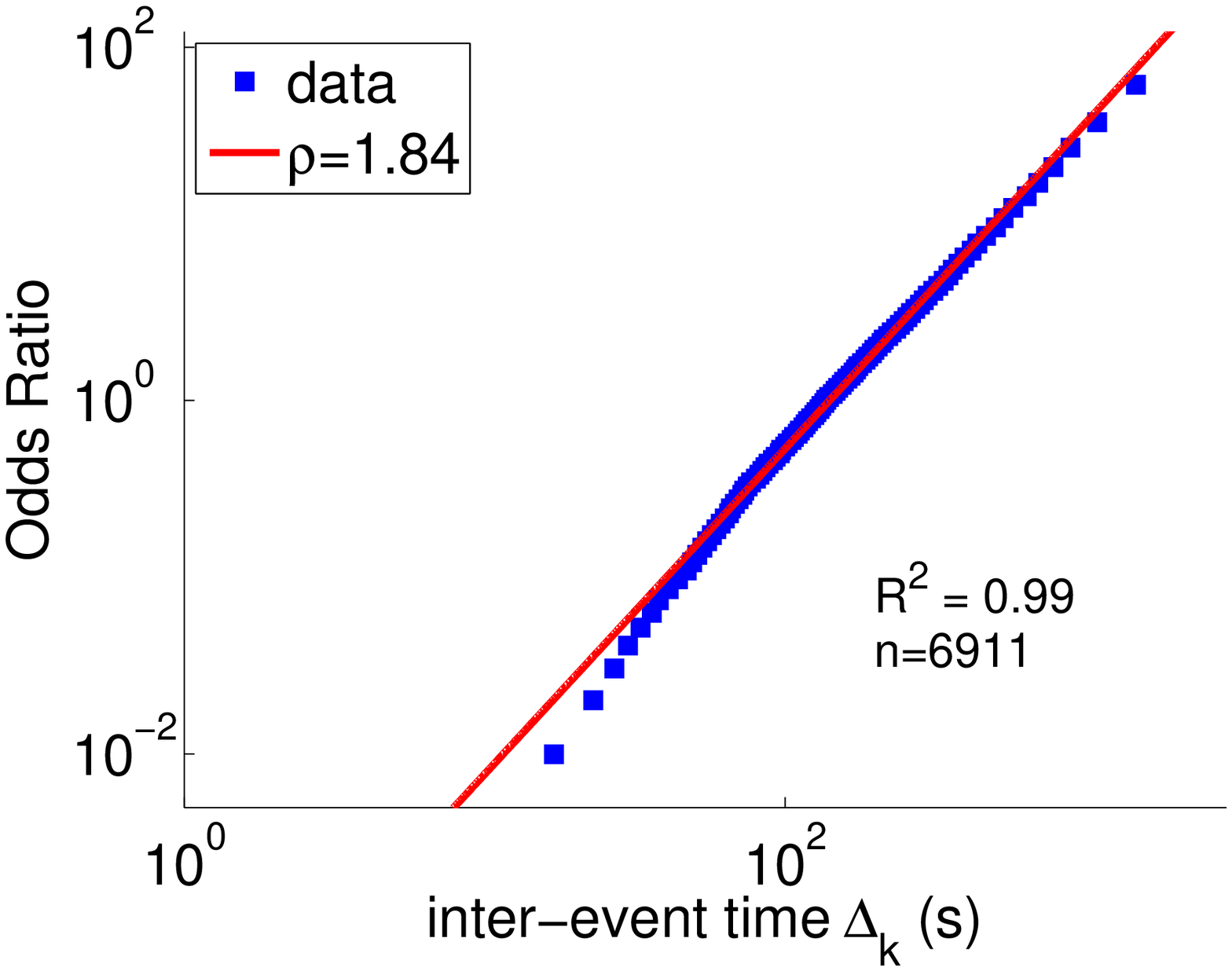}}
{\includegraphics[width=.23\textwidth]{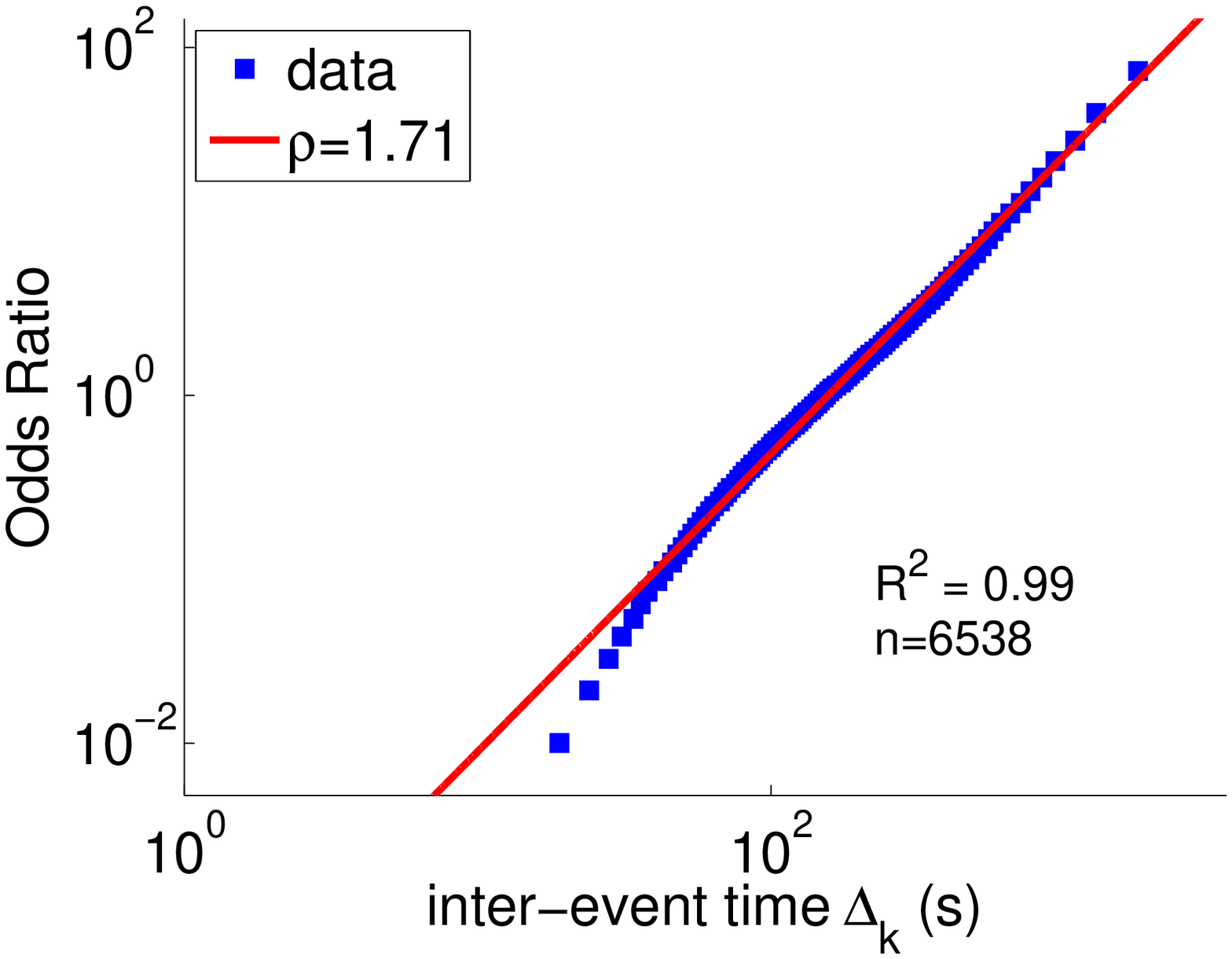}}
{\includegraphics[width=.23\textwidth]{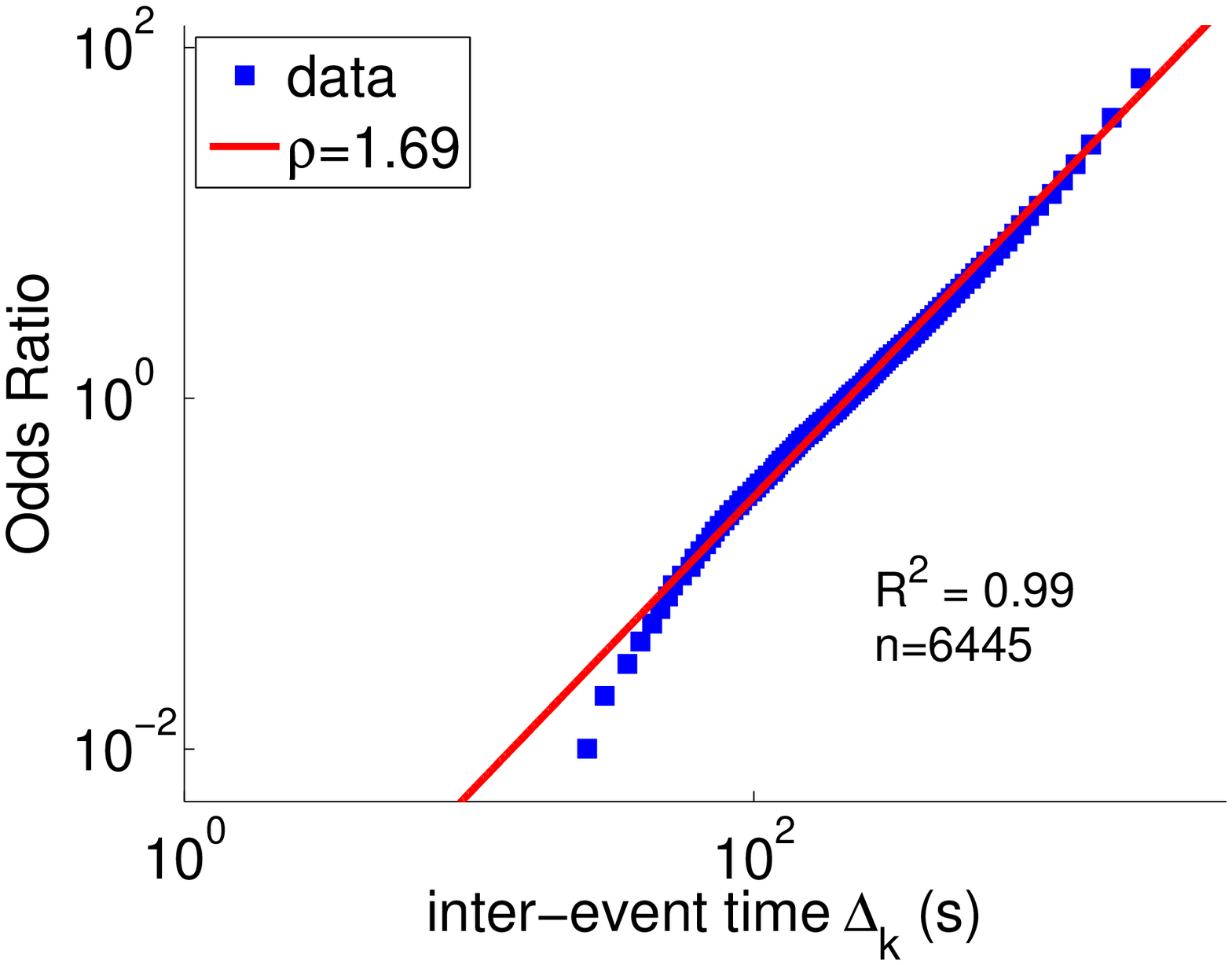}}
{\includegraphics[width=.23\textwidth]{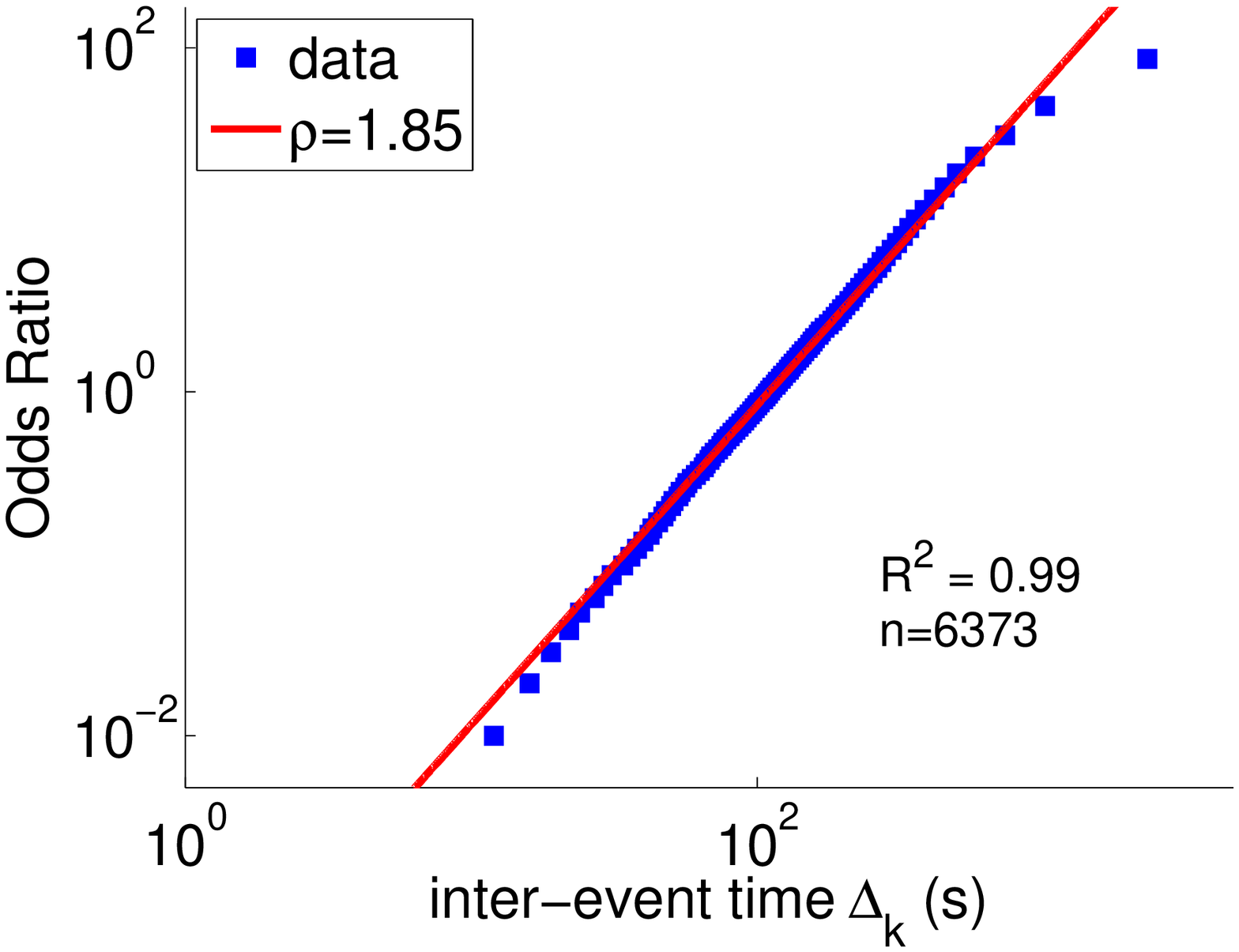}}
{\includegraphics[width=.23\textwidth]{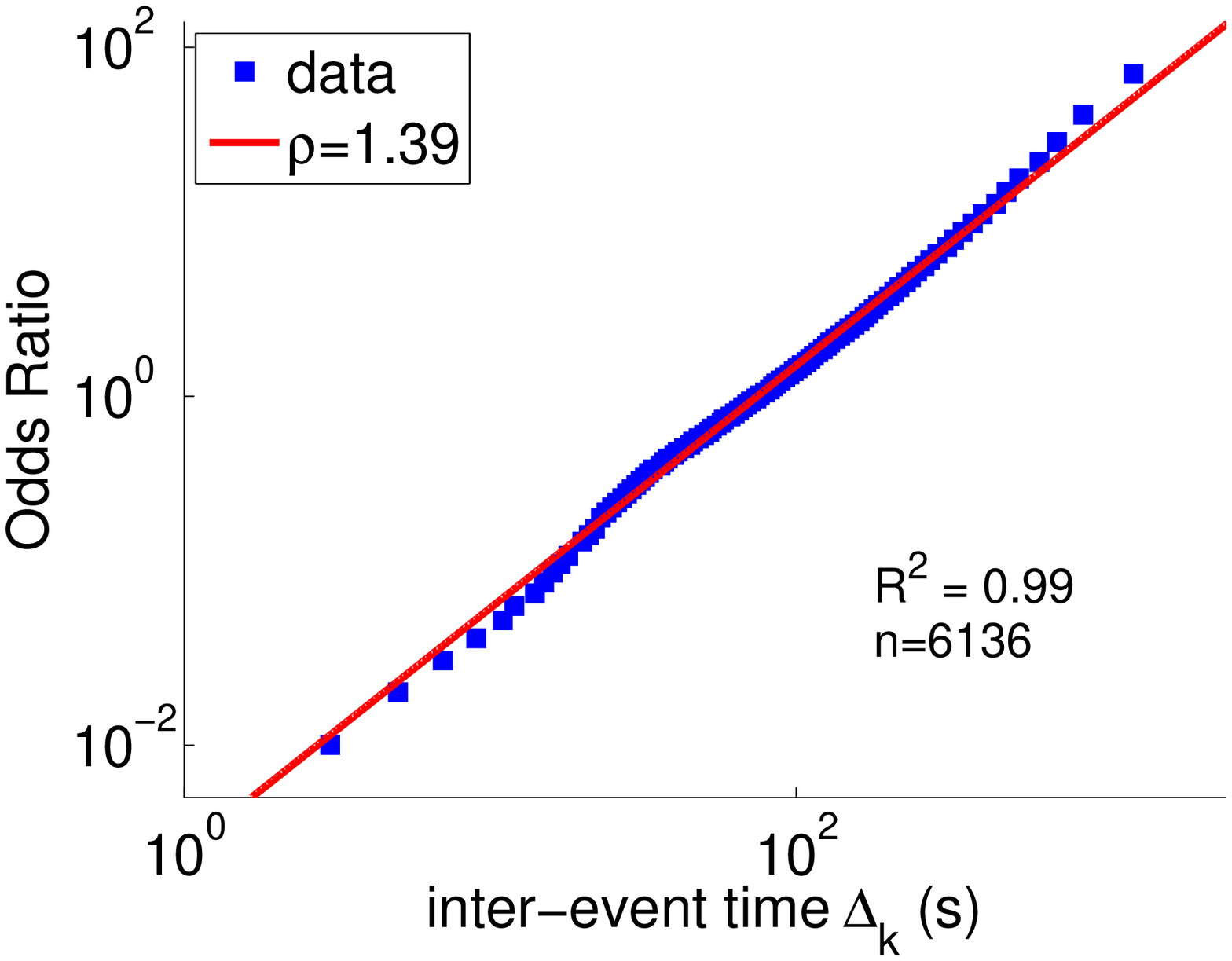}}
{\includegraphics[width=.23\textwidth]{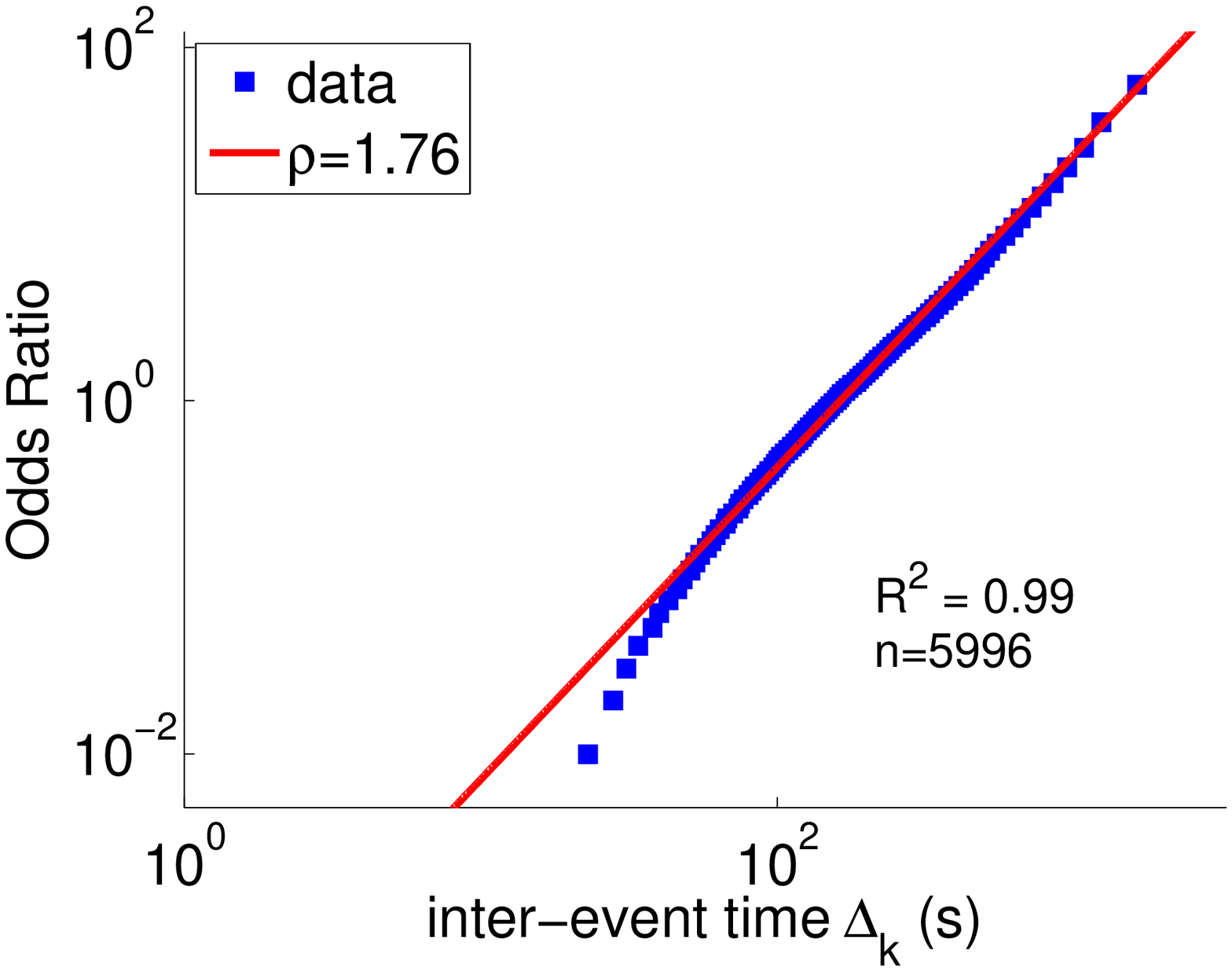}}
{\includegraphics[width=.23\textwidth]{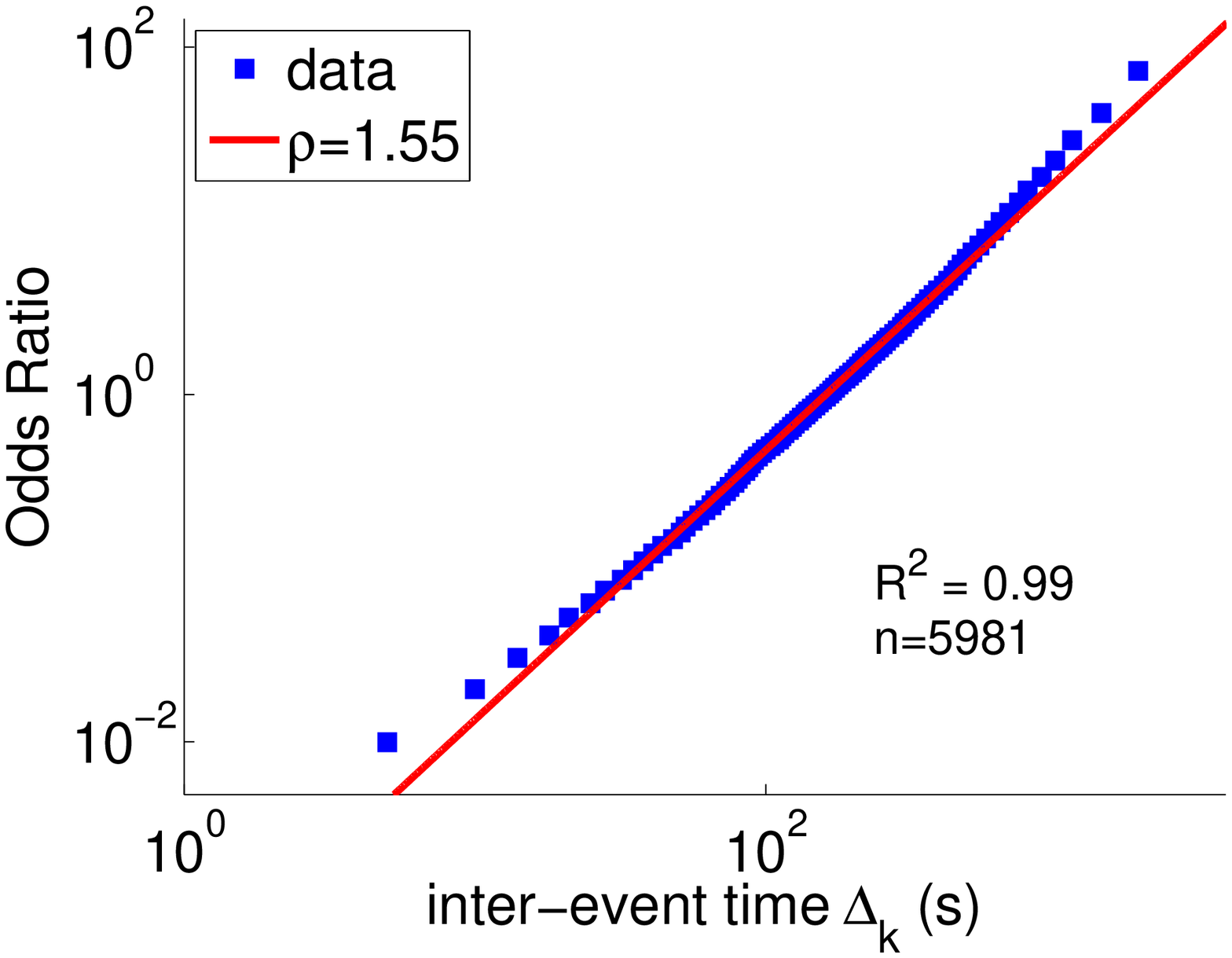}}
\caption{Sample from the Phone dataset.}
\end{figure*}

\begin{figure*}[htpb]
\centering
{\includegraphics[width=.23\textwidth]{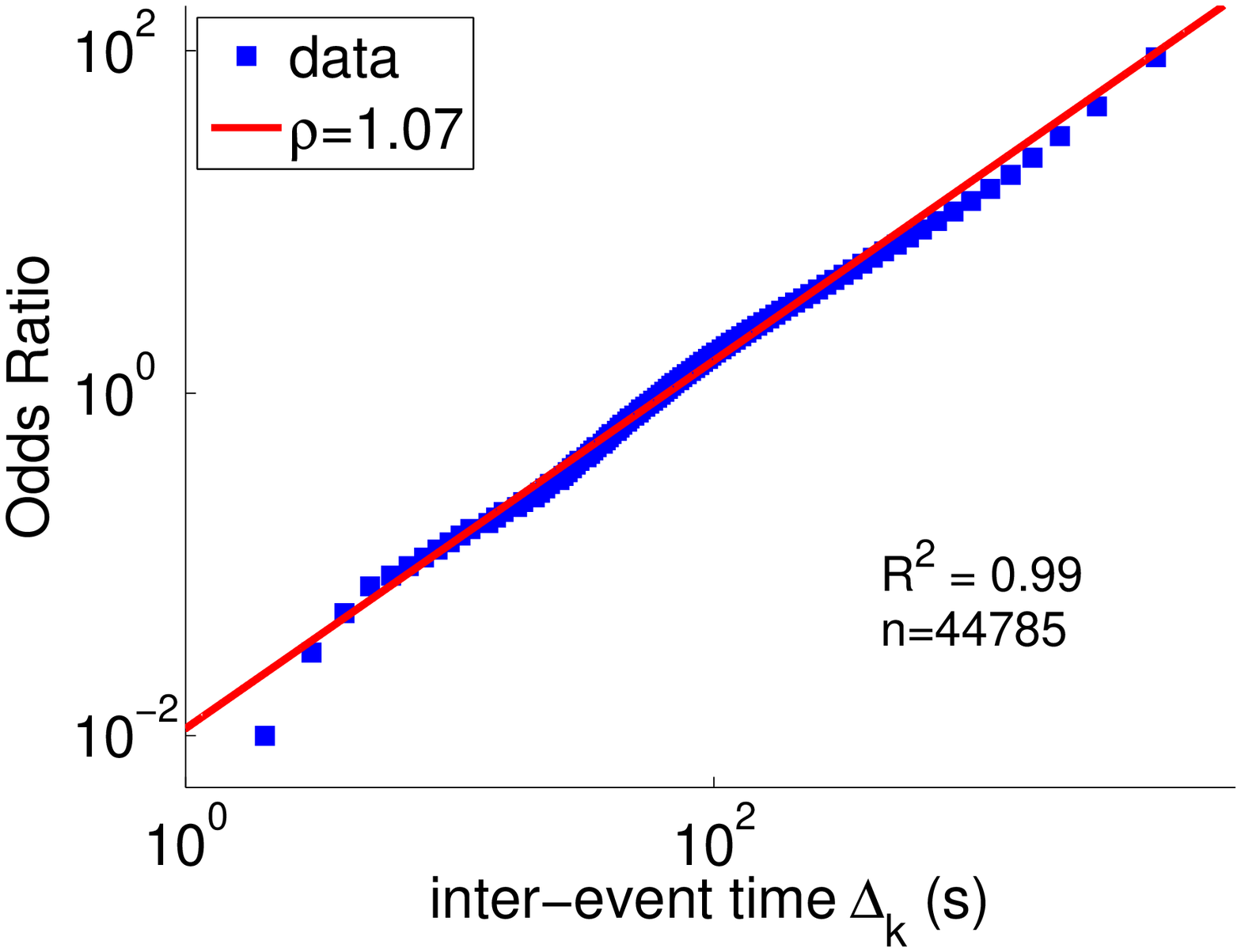}}
{\includegraphics[width=.23\textwidth]{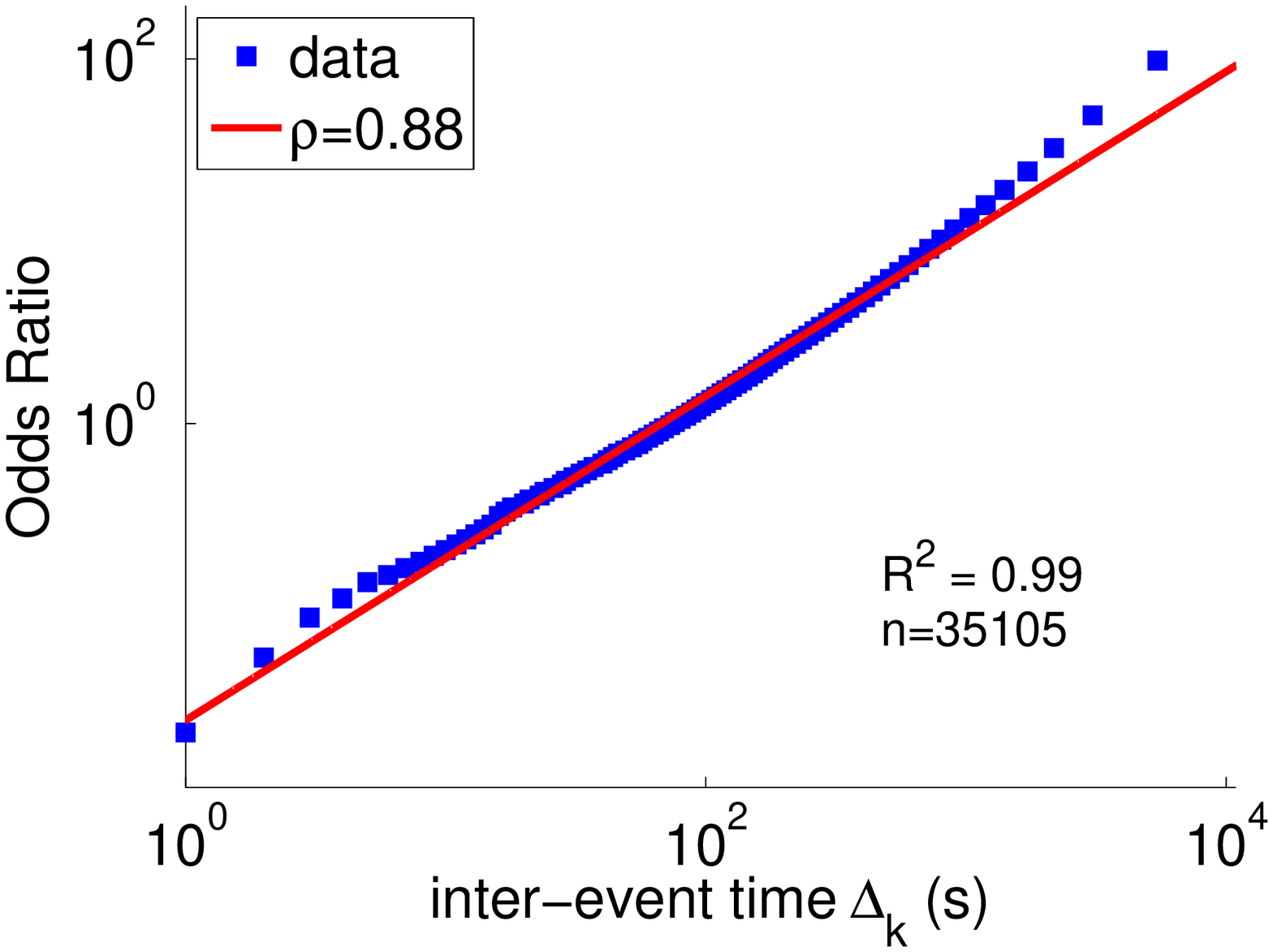}}
{\includegraphics[width=.23\textwidth]{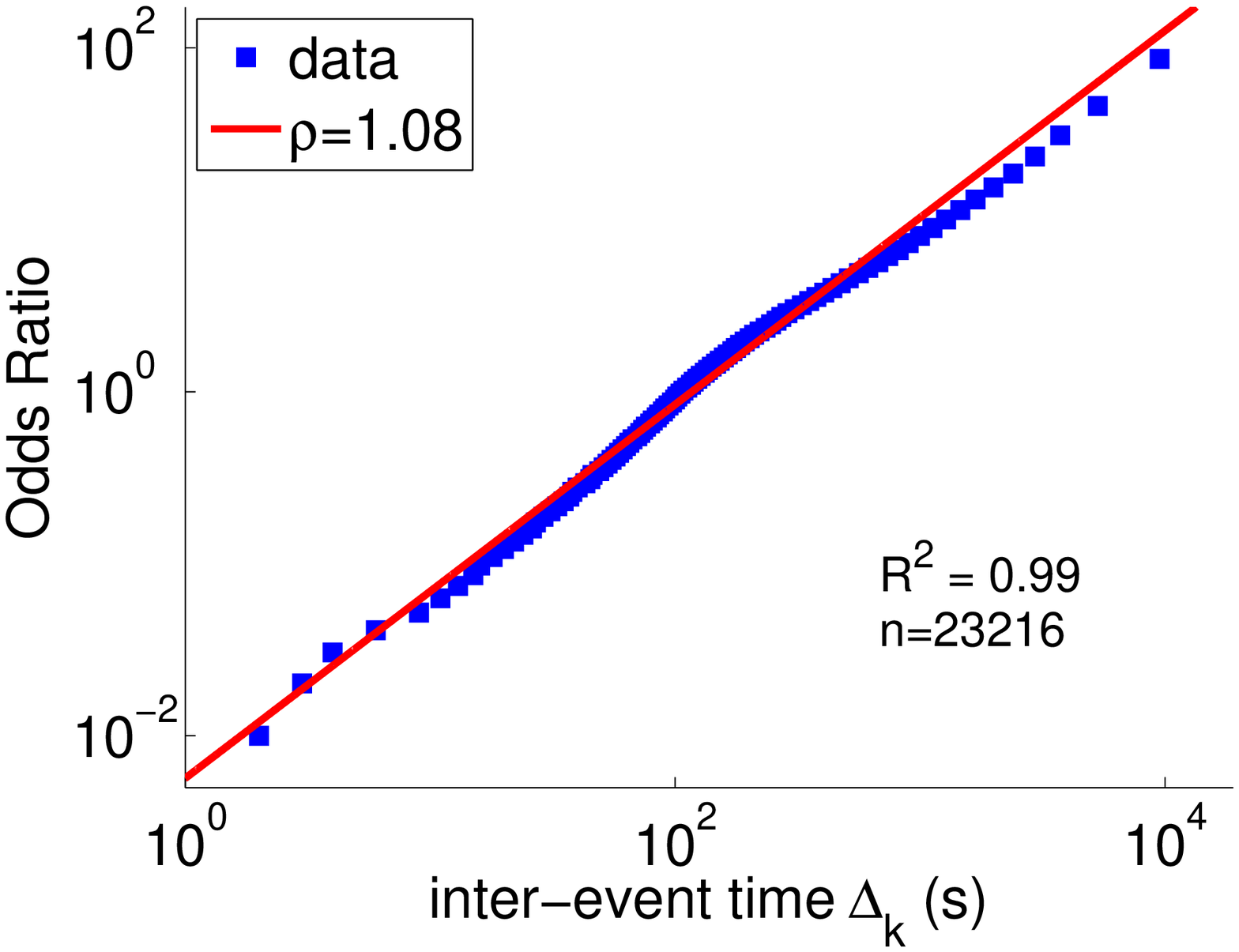}}
{\includegraphics[width=.23\textwidth]{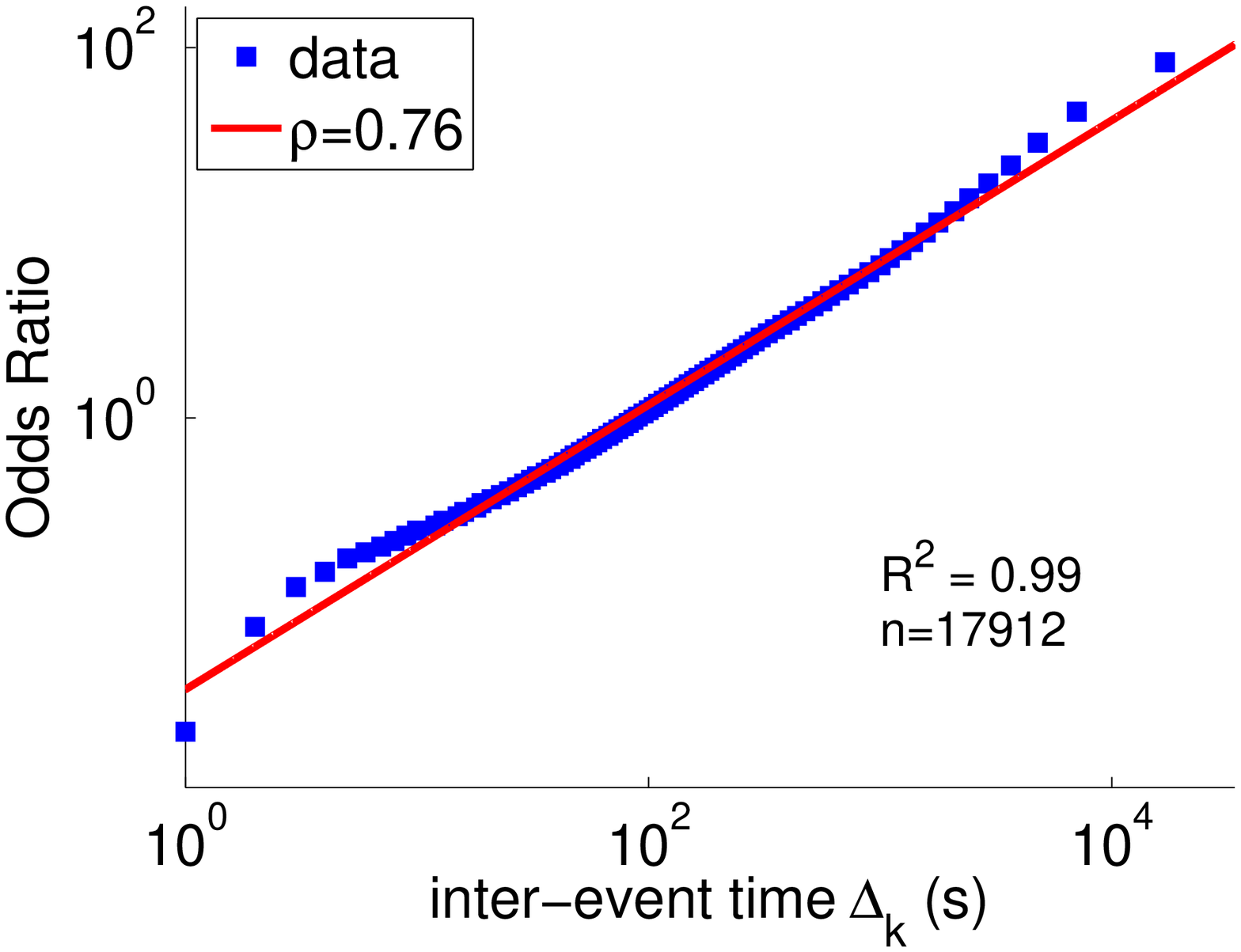}}
{\includegraphics[width=.23\textwidth]{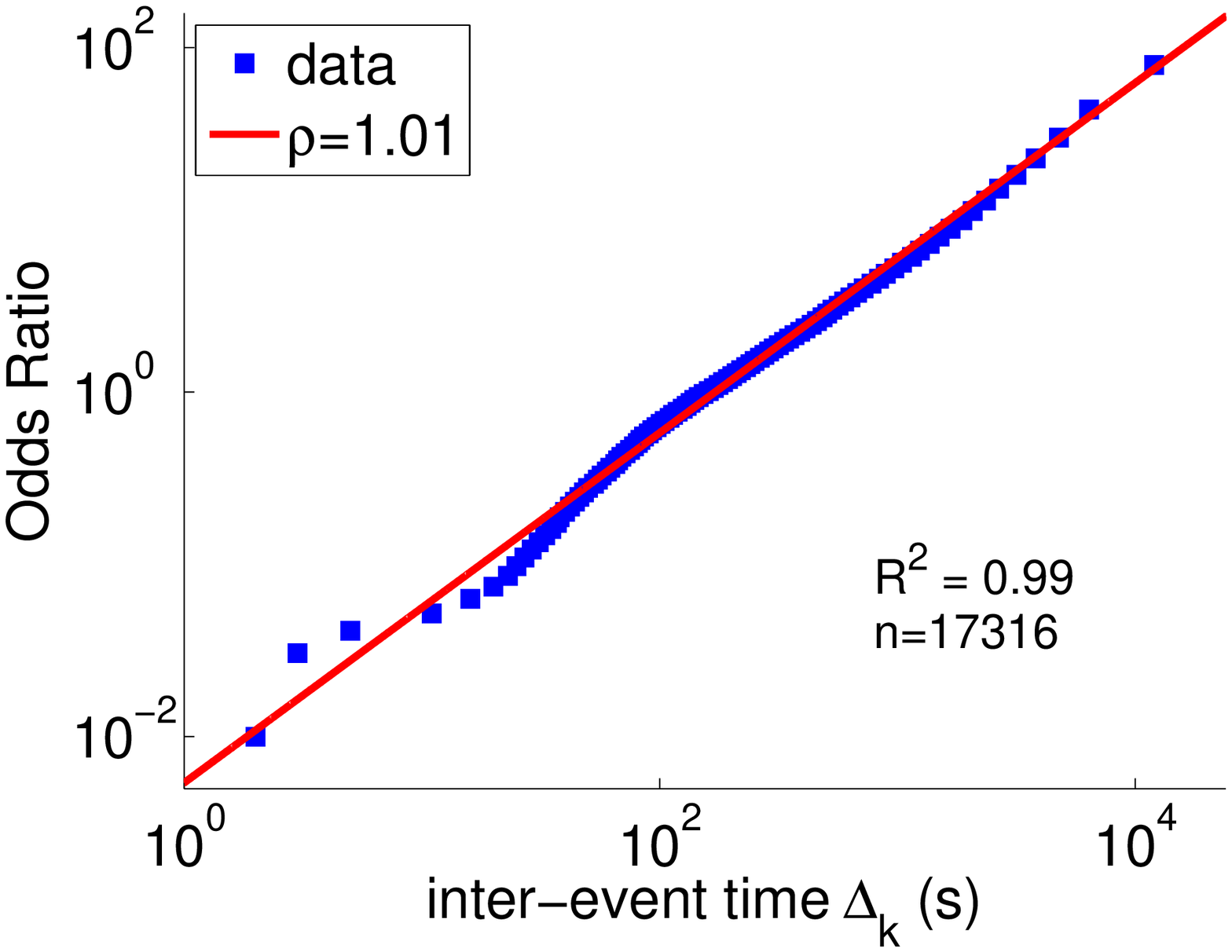}}
{\includegraphics[width=.23\textwidth]{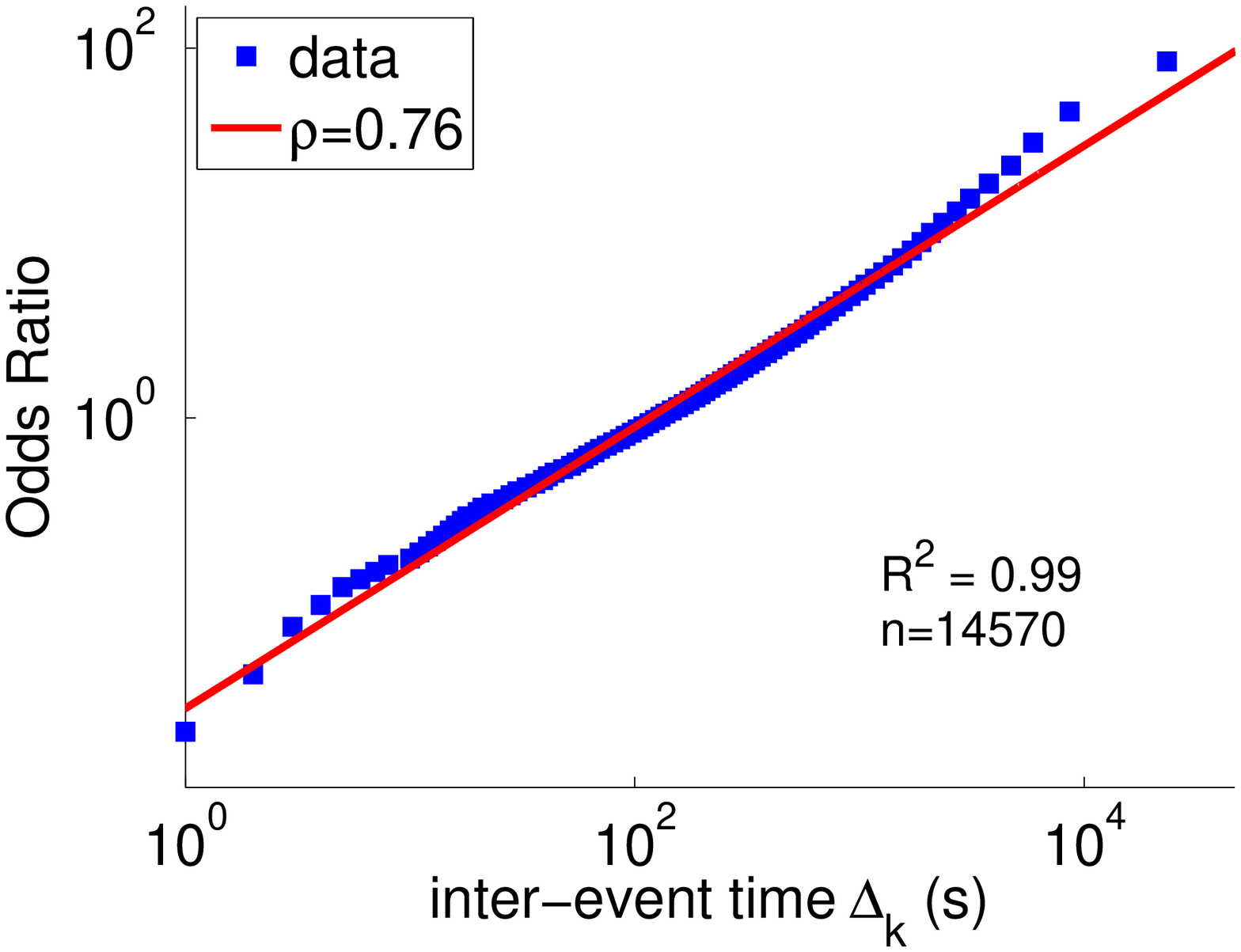}}
{\includegraphics[width=.23\textwidth]{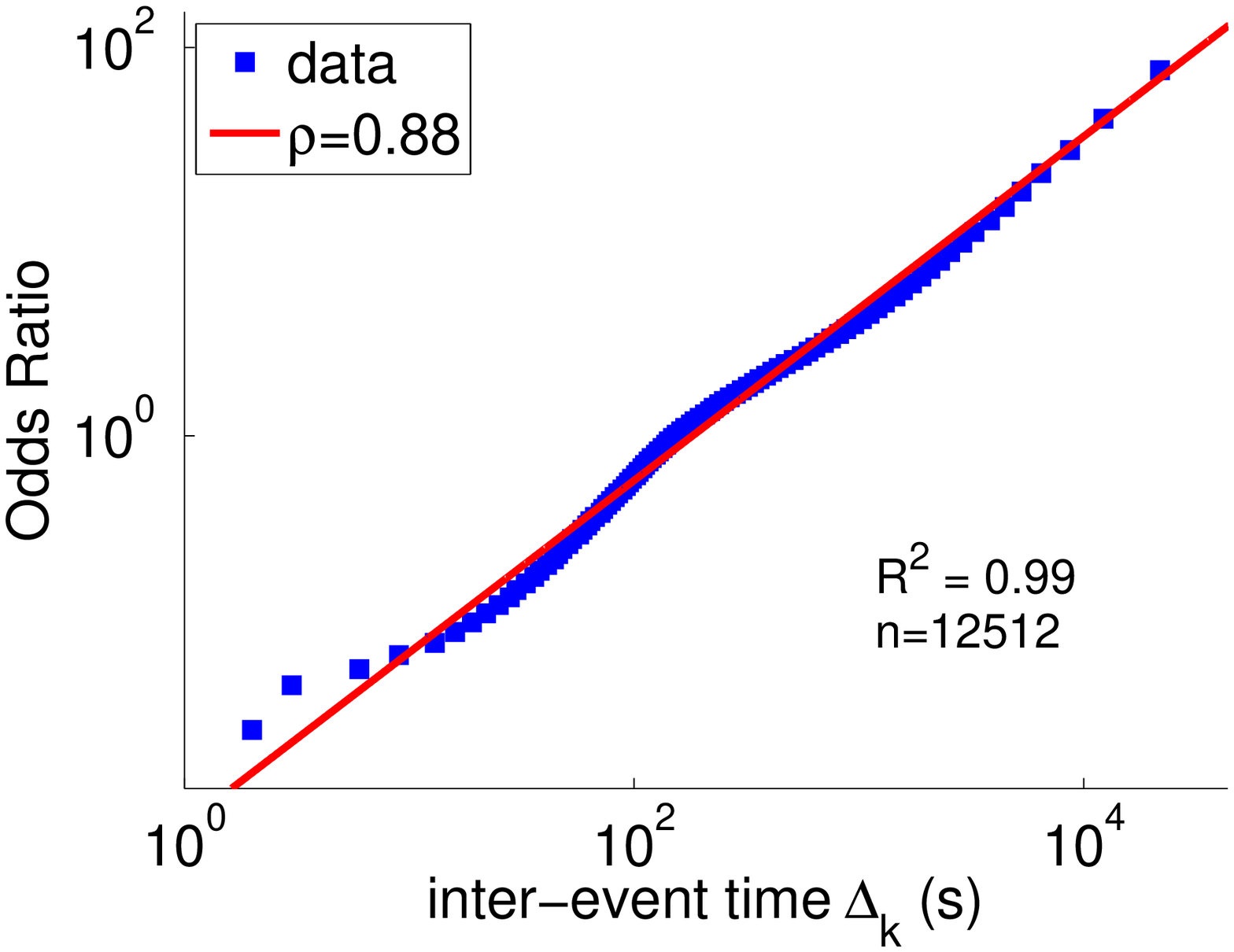}}
{\includegraphics[width=.23\textwidth]{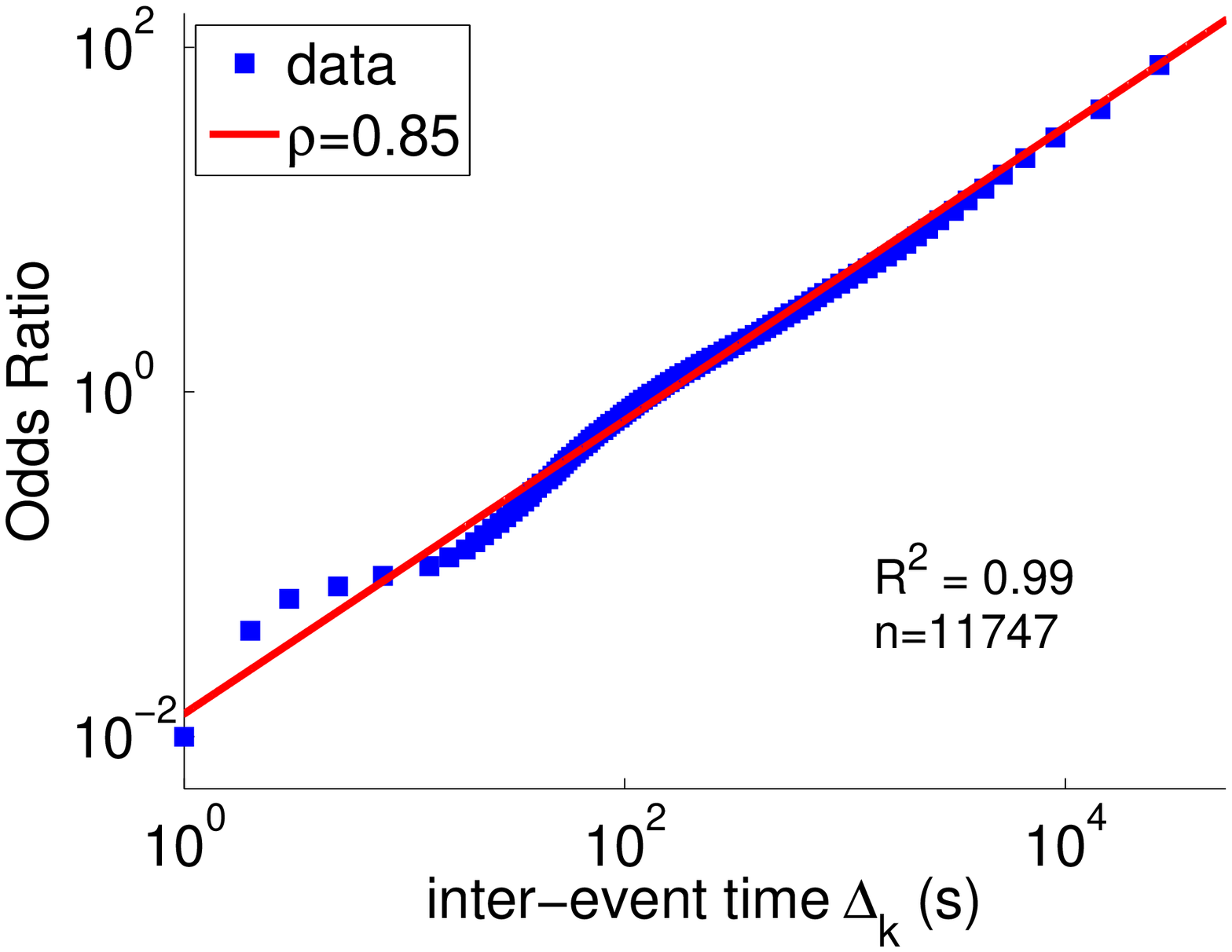}}
\caption{Sample from the SMS dataset.}
\end{figure*}

\begin{figure*}[htpb]
\centering
{\includegraphics[width=.24\textwidth]{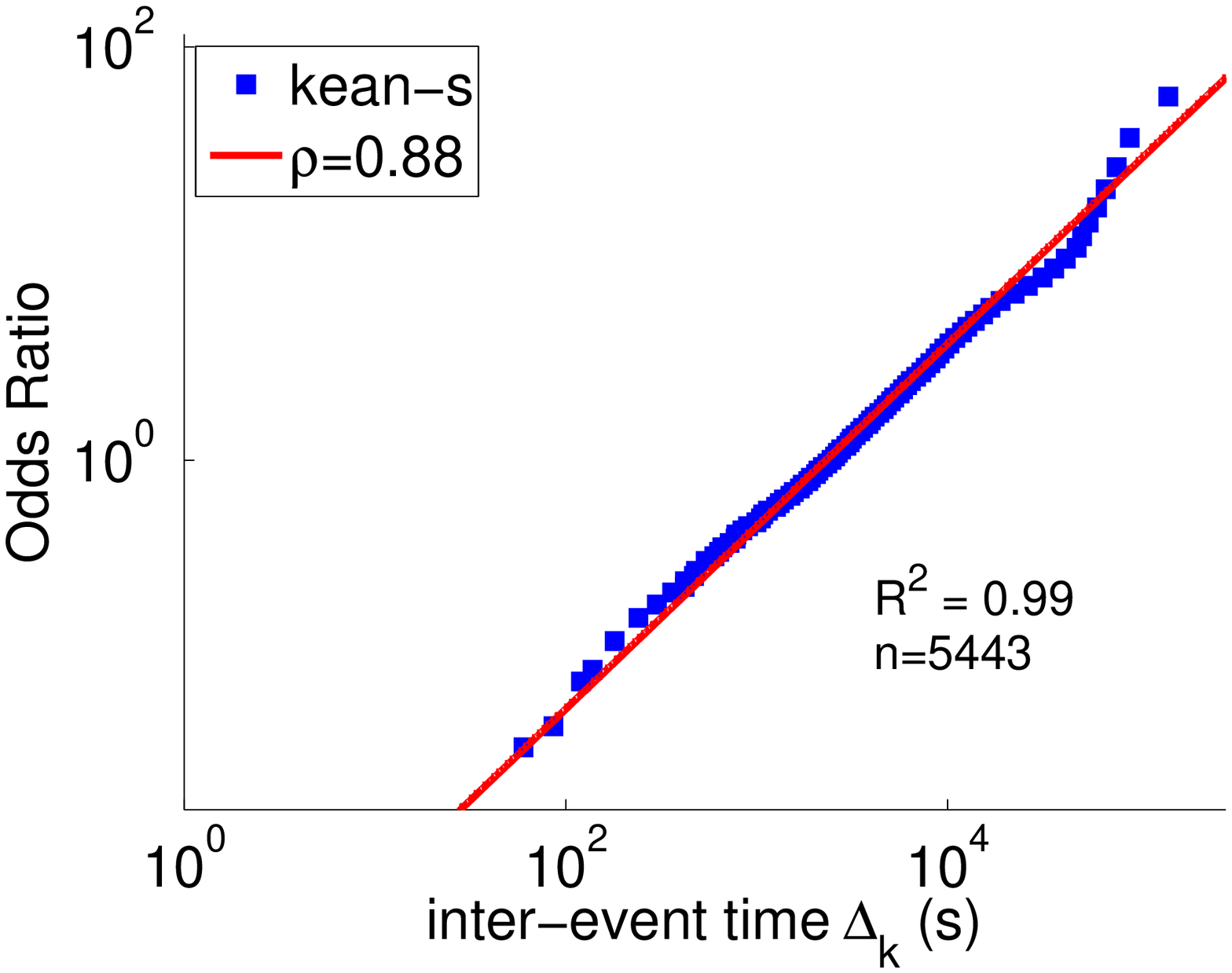}}
{\includegraphics[width=.24\textwidth]{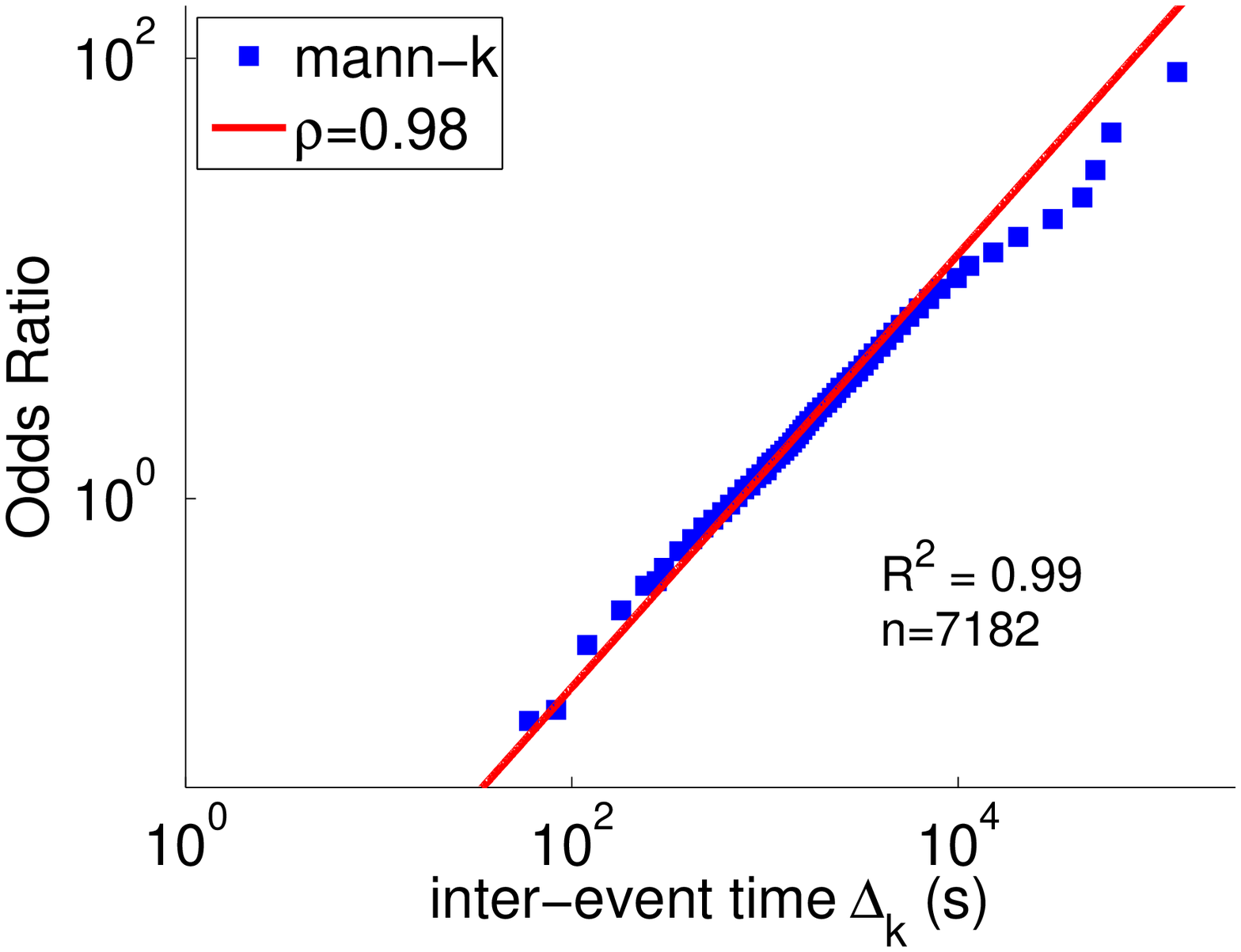}}
{\includegraphics[width=.24\textwidth]{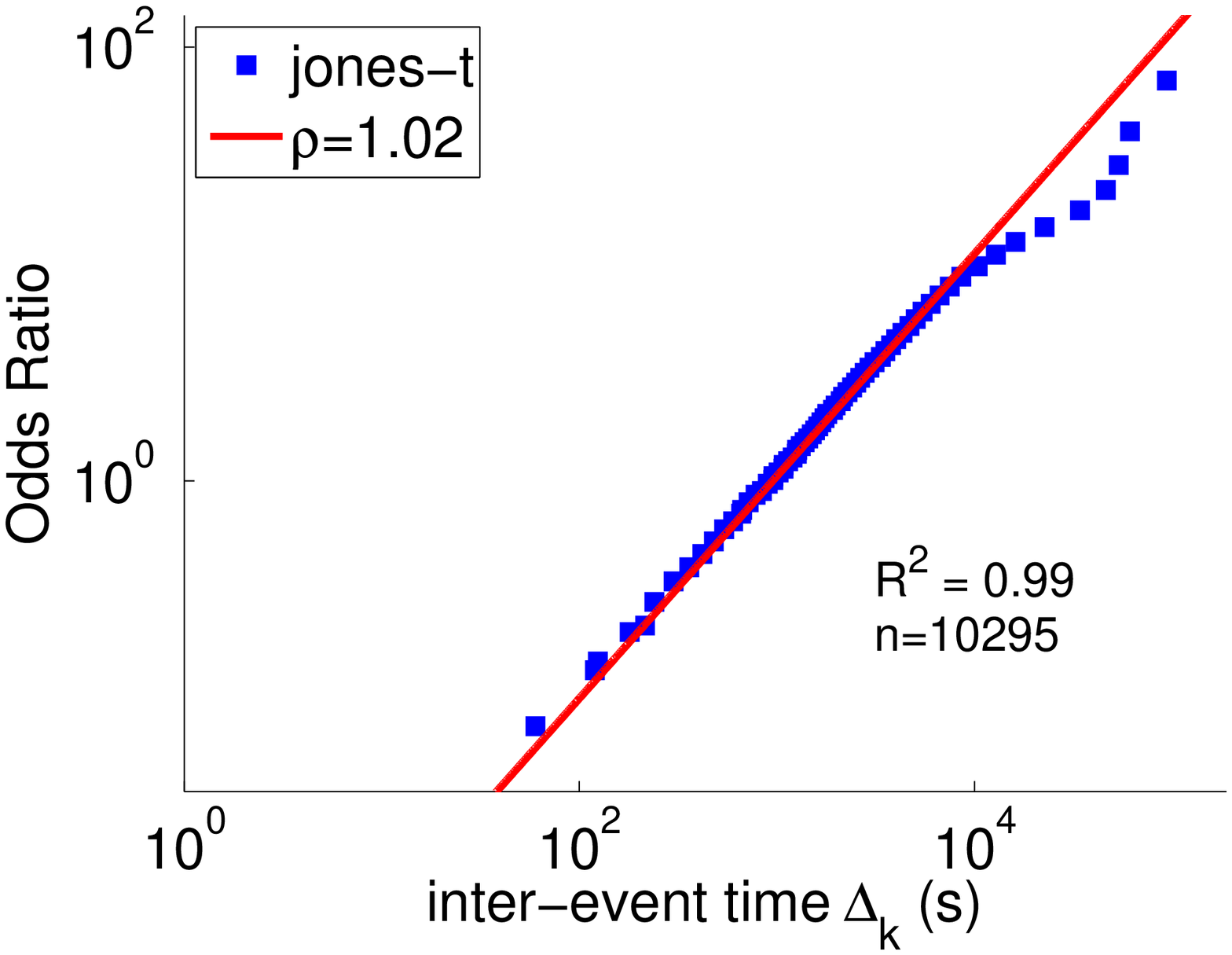}}
{\includegraphics[width=.24\textwidth]{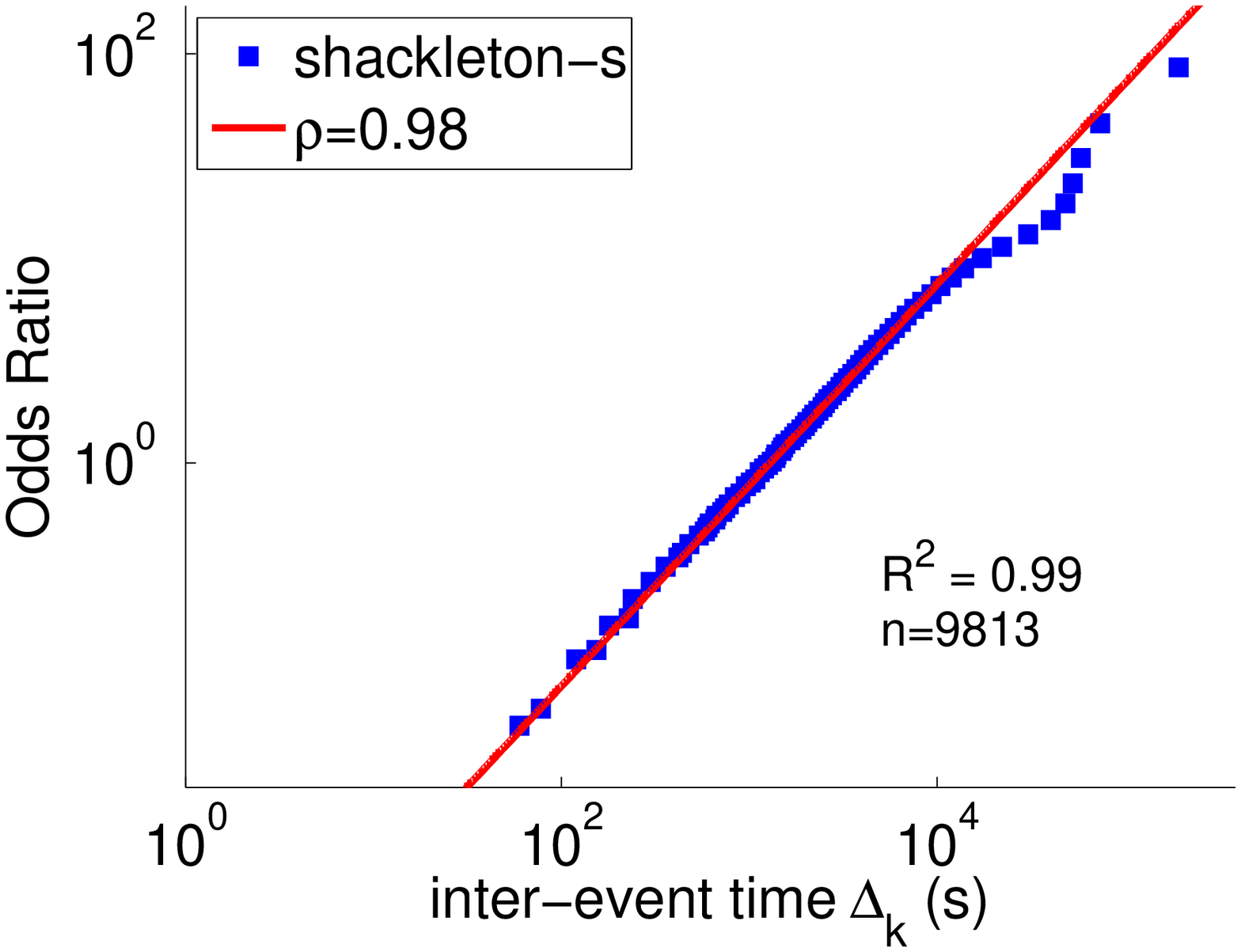}}
{\includegraphics[width=.24\textwidth]{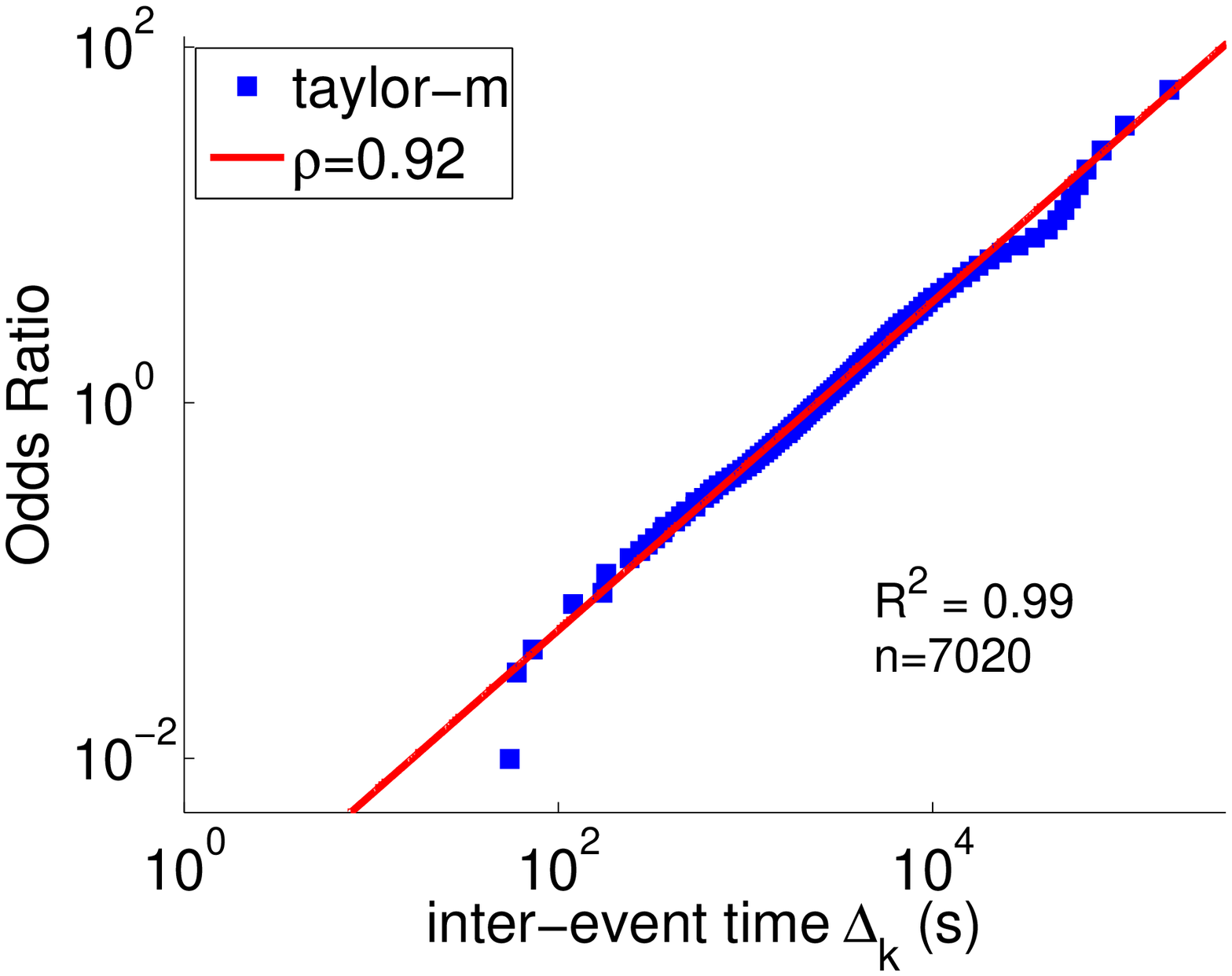}}
{\includegraphics[width=.24\textwidth]{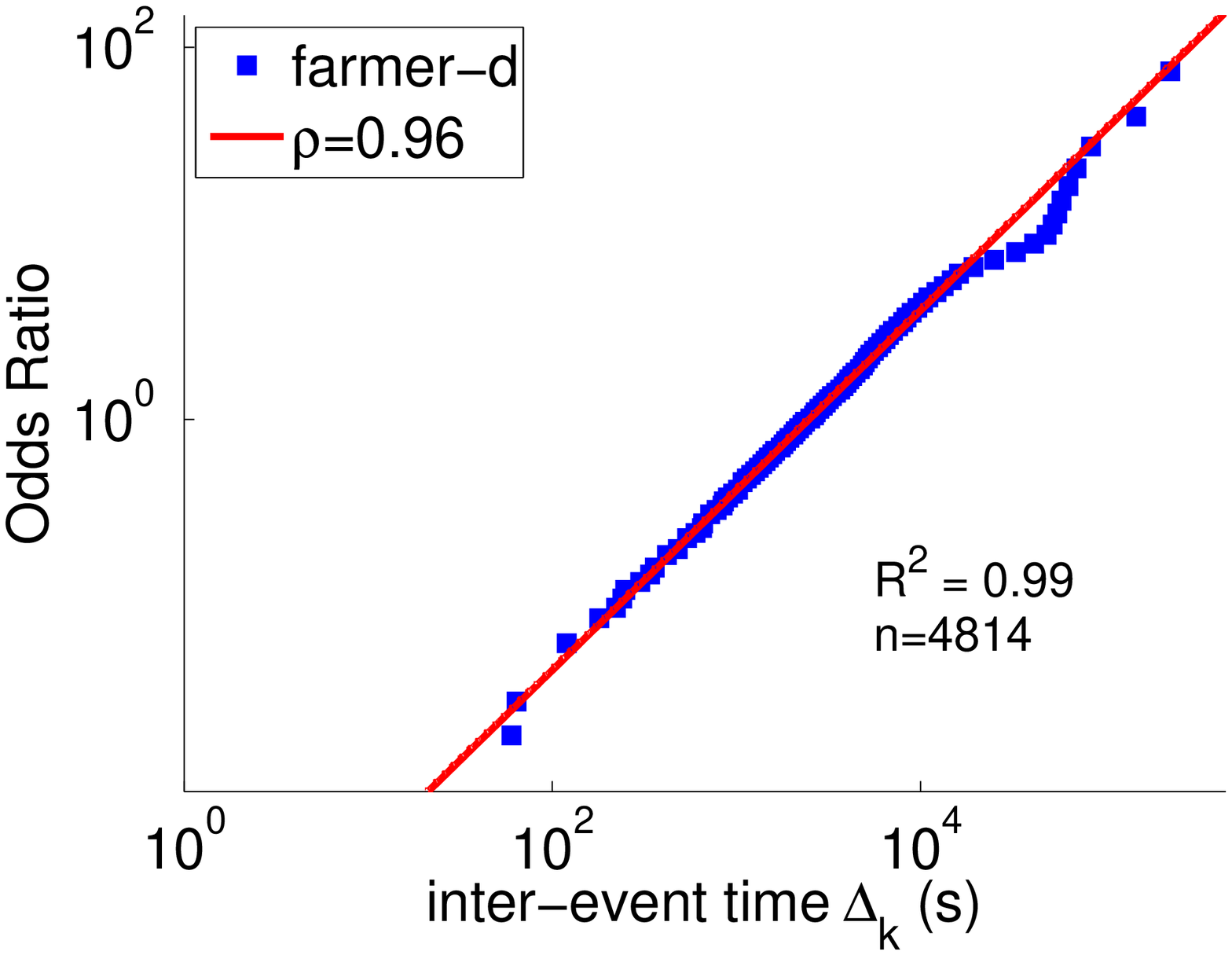}}
{\includegraphics[width=.24\textwidth]{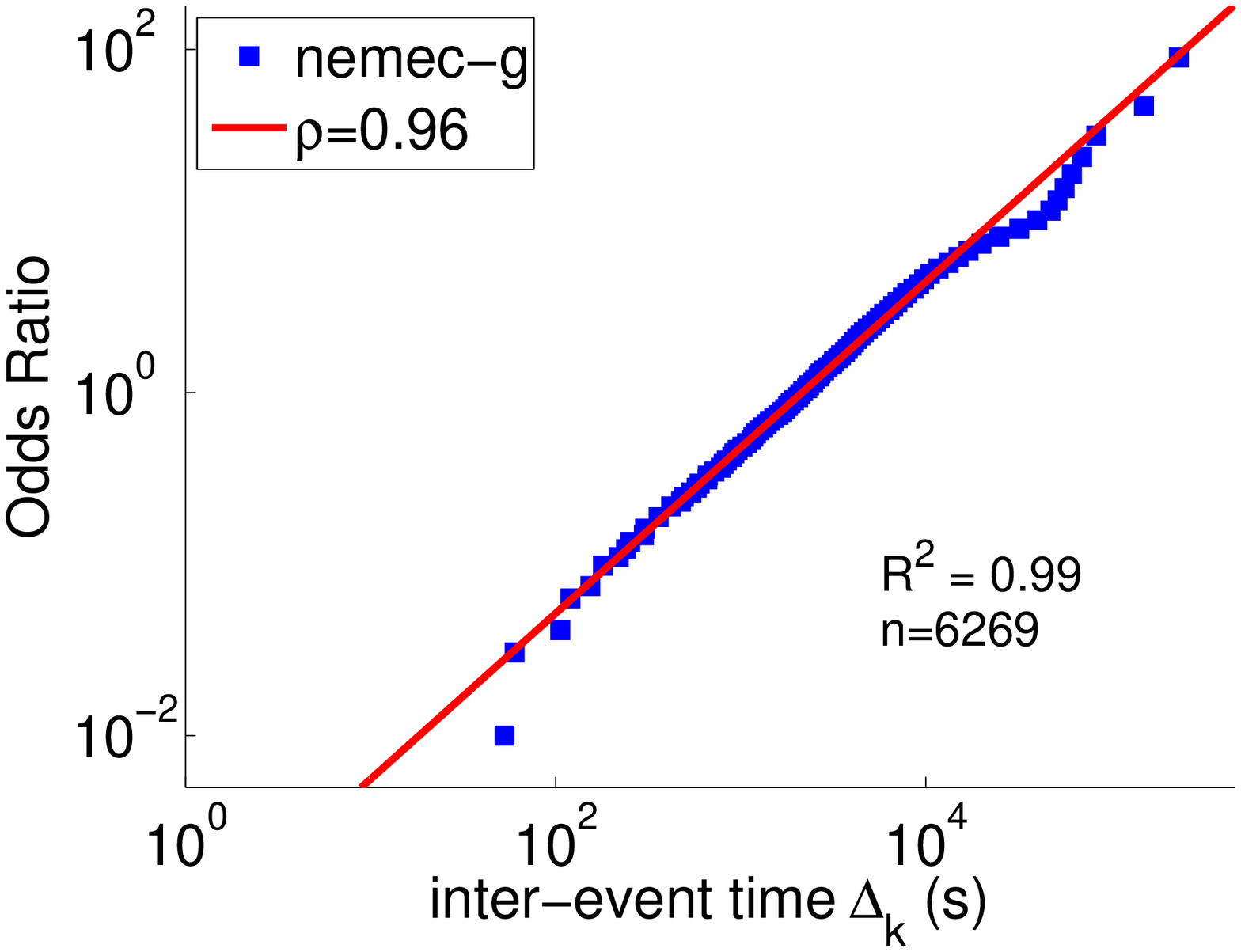}}
{\includegraphics[width=.24\textwidth]{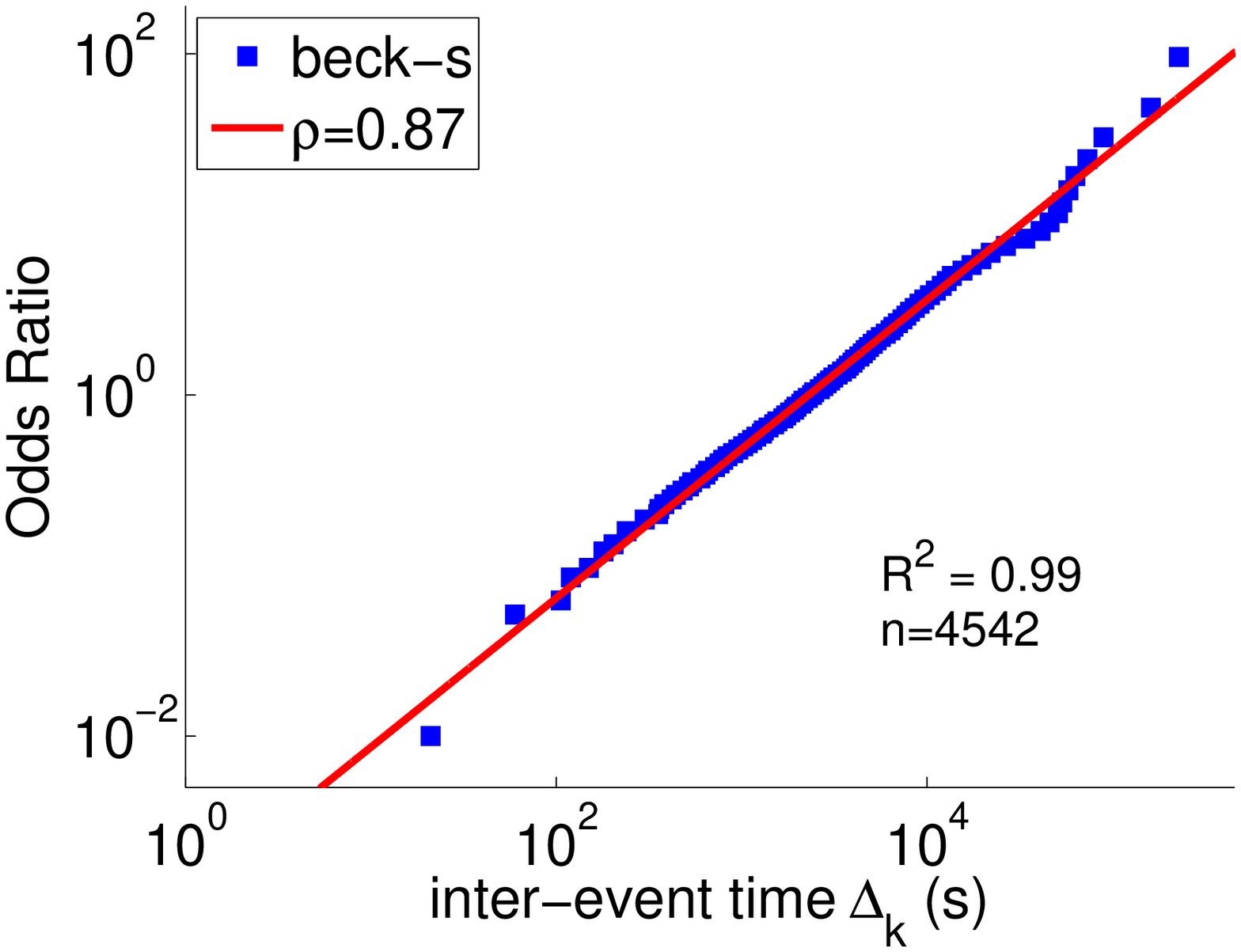}}
\caption{Sample from the Enron dataset.}
\end{figure*}

\begin{figure*}[htpb]
\centering
{\includegraphics[width=.24\textwidth]{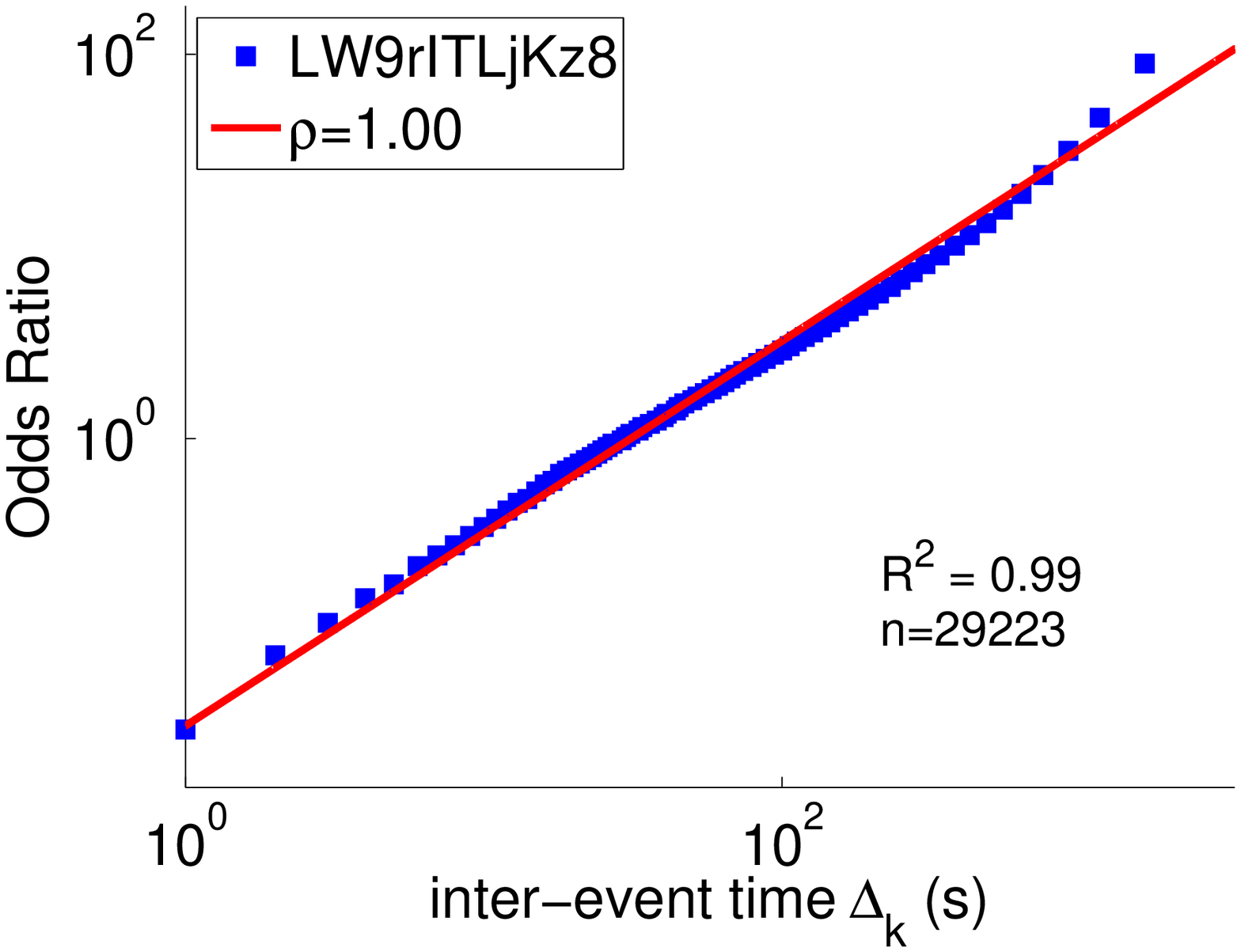}}
{\includegraphics[width=.24\textwidth]{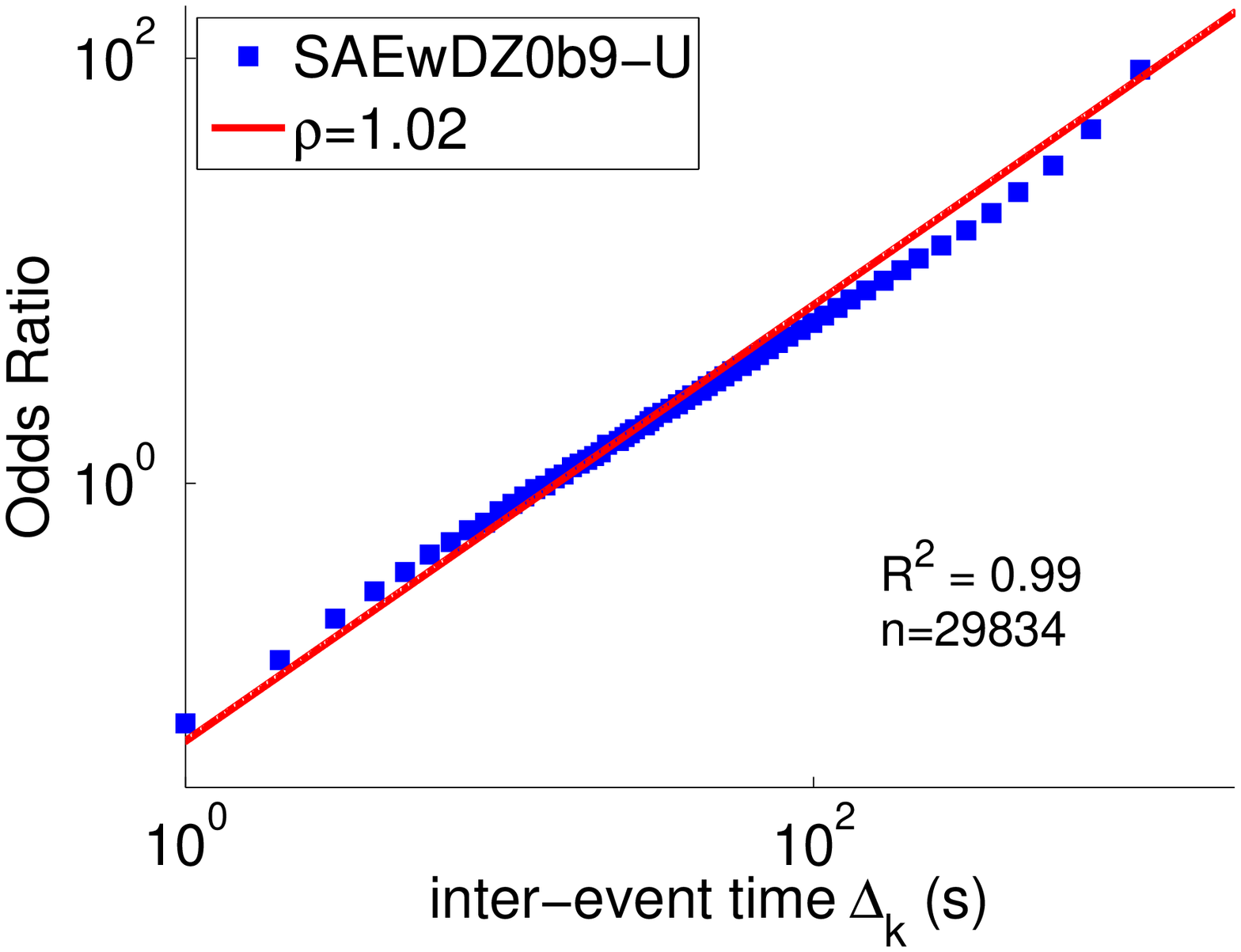}}
{\includegraphics[width=.24\textwidth]{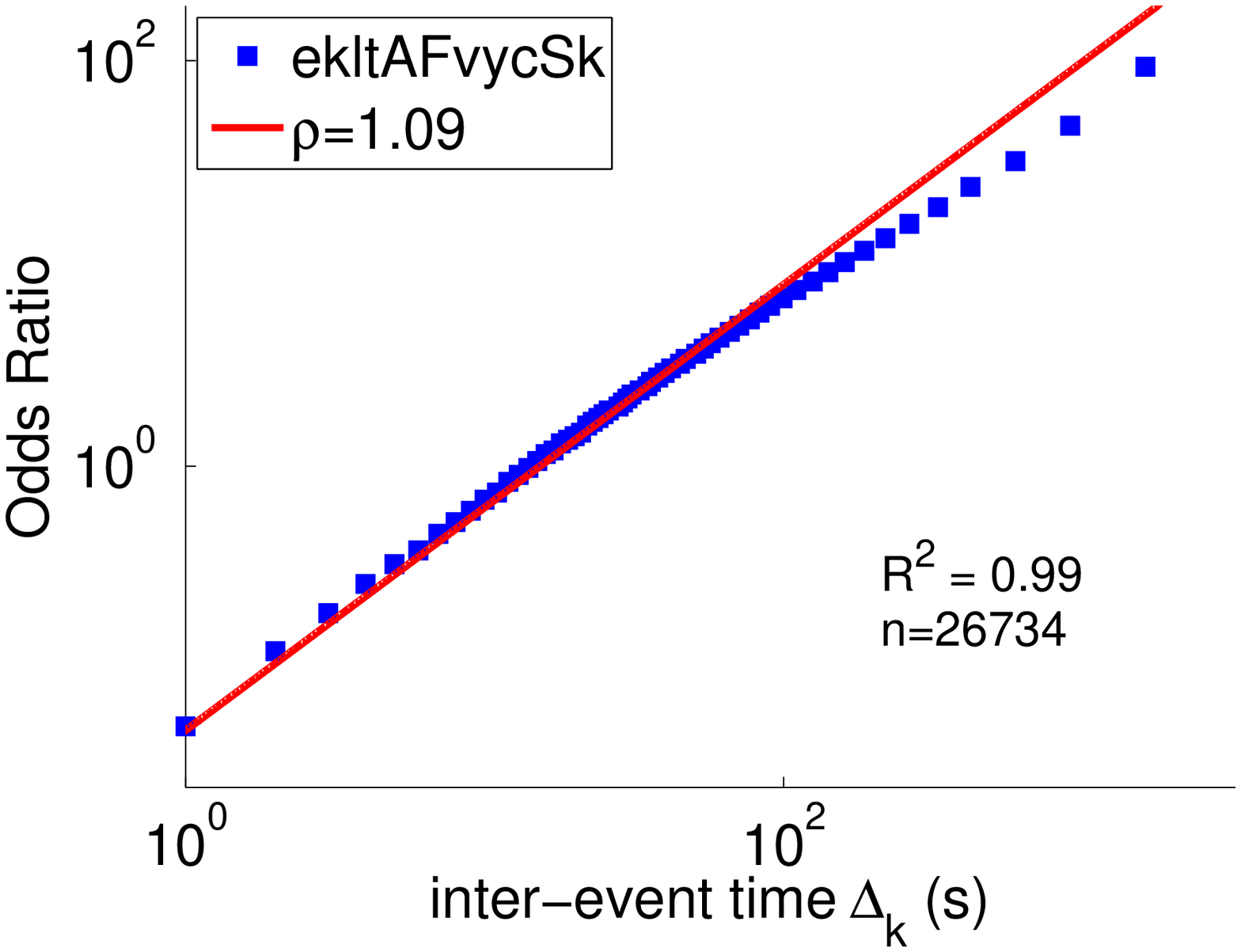}}
{\includegraphics[width=.24\textwidth]{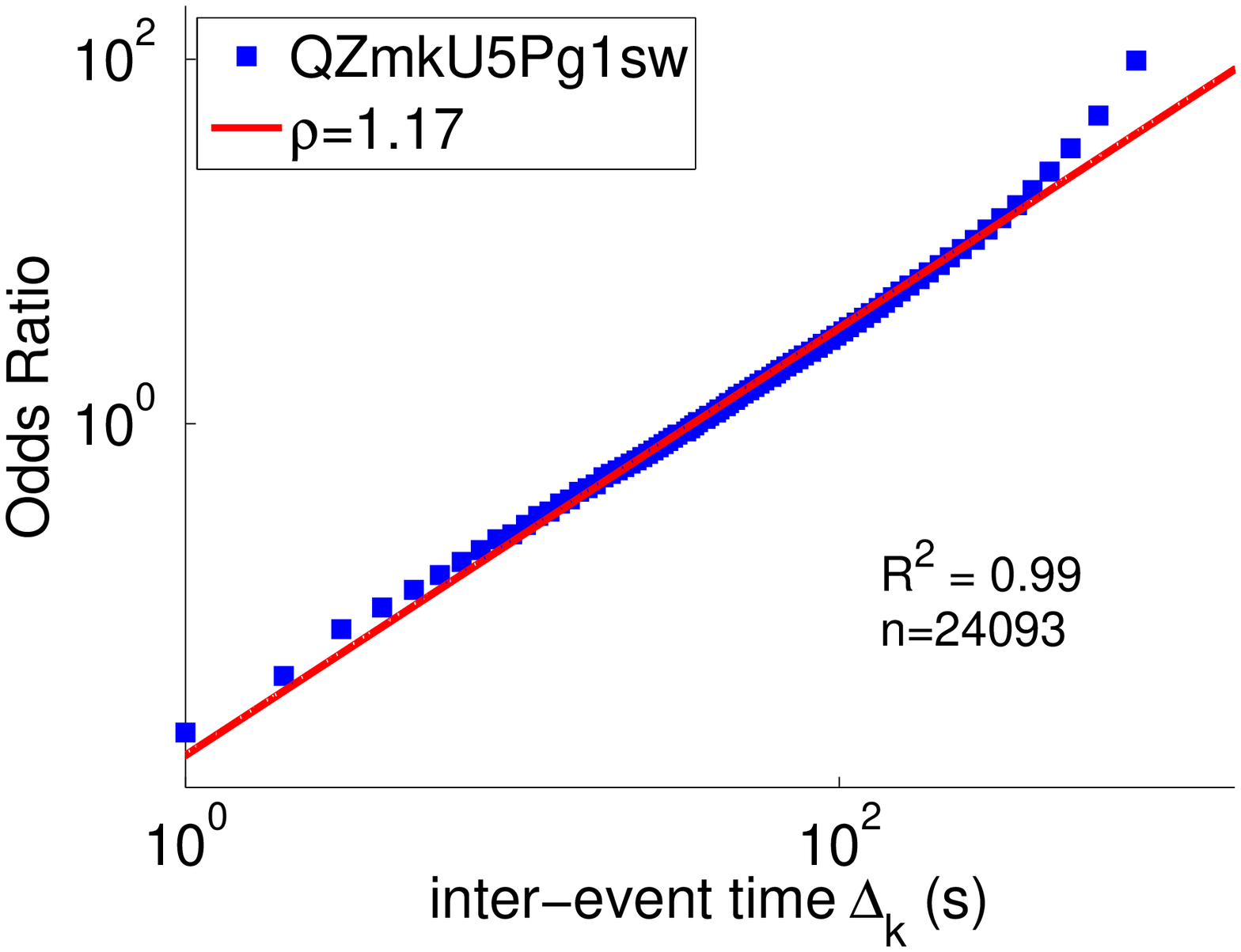}}
{\includegraphics[width=.24\textwidth]{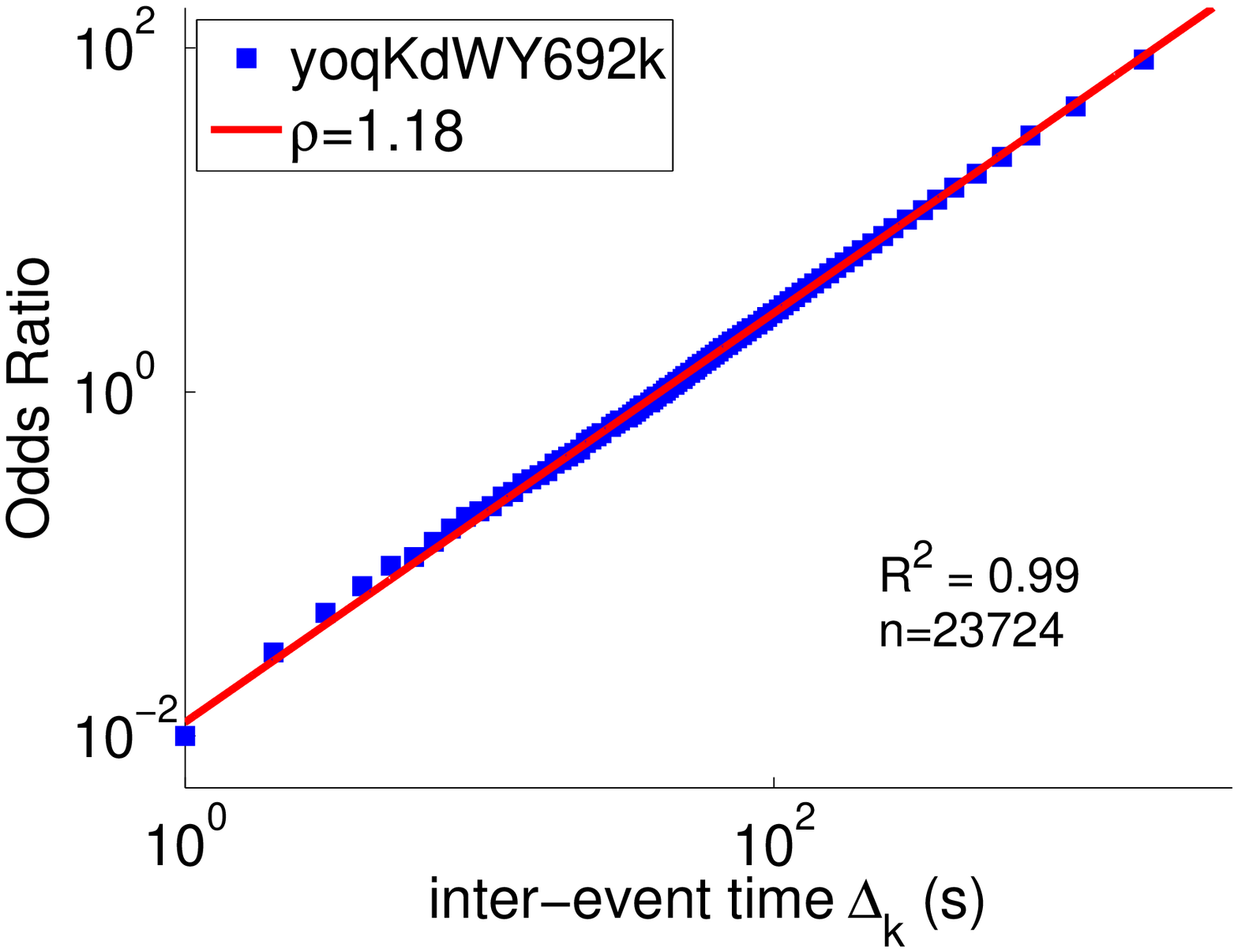}}
{\includegraphics[width=.24\textwidth]{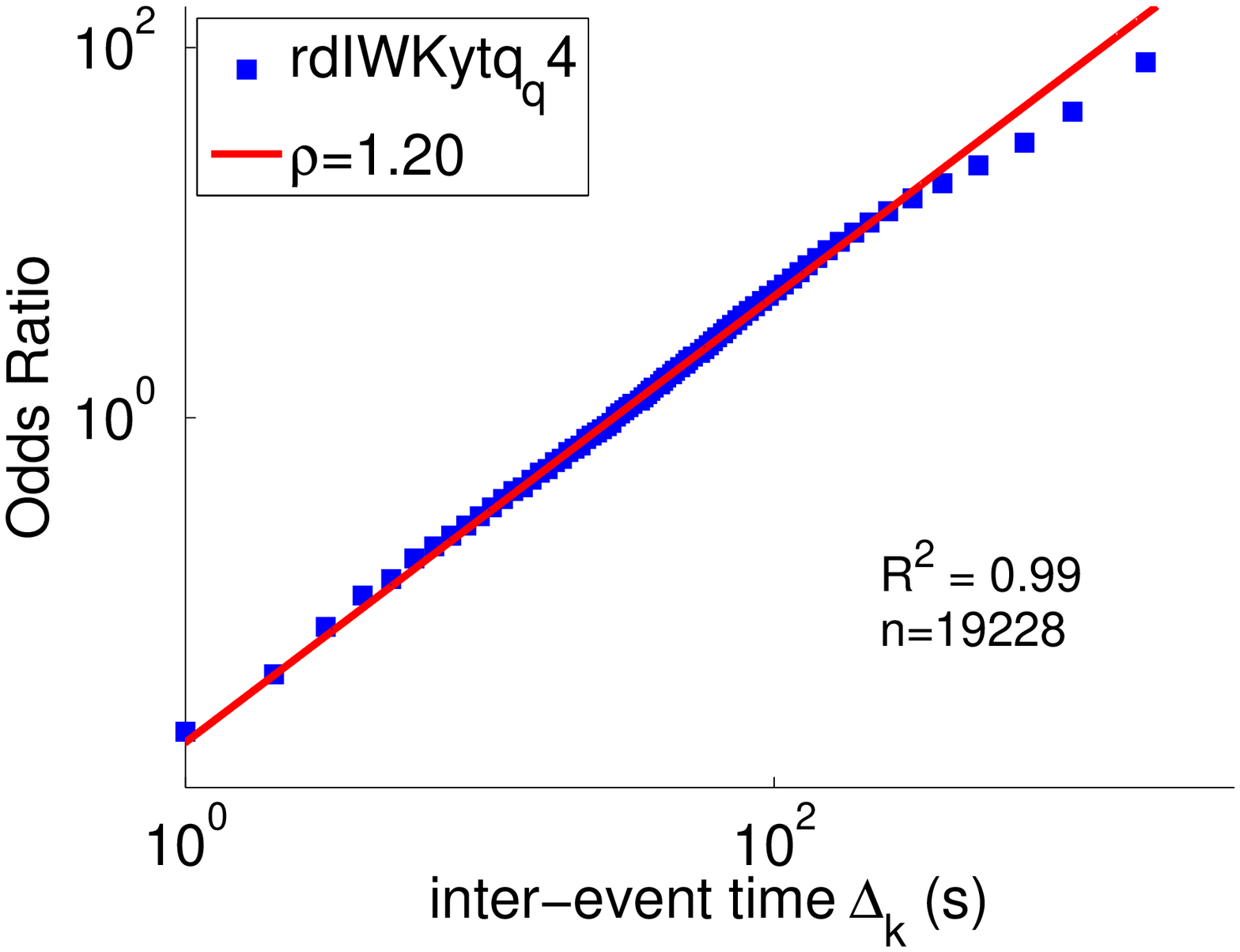}}
{\includegraphics[width=.24\textwidth]{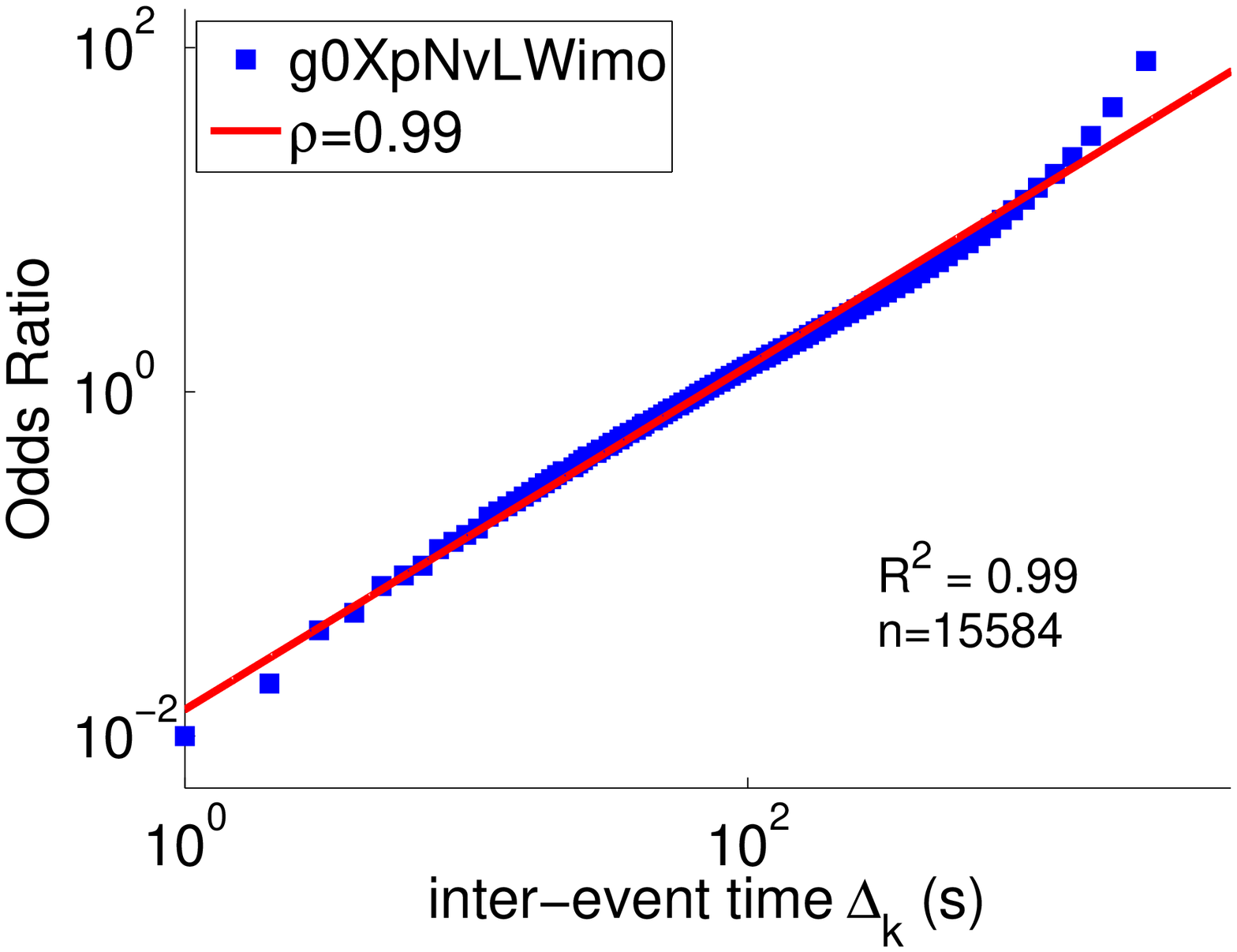}}
{\includegraphics[width=.24\textwidth]{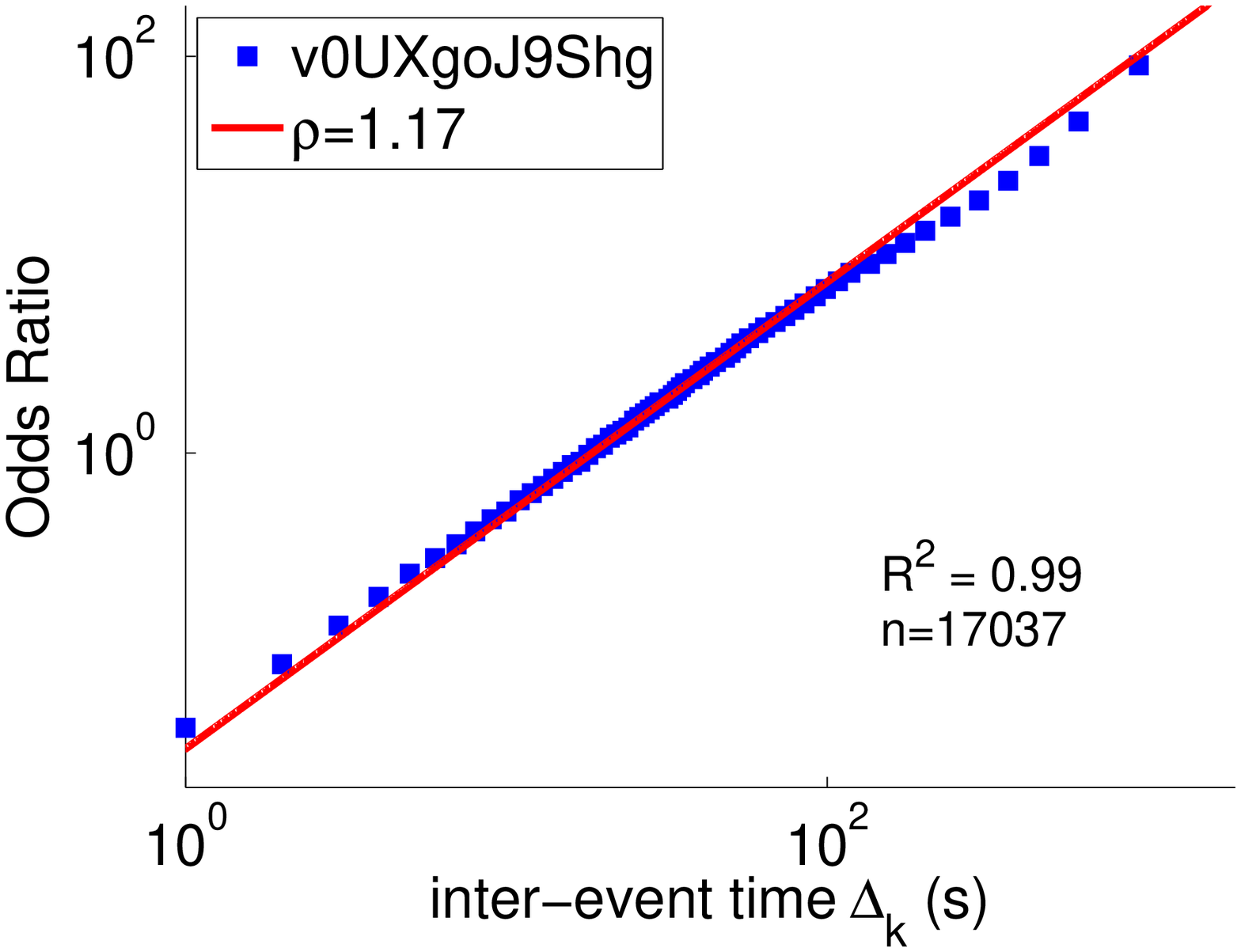}}
\caption{Sample from the Youtube dataset.}
\end{figure*}

\begin{figure*}[htpb]
\centering
{\includegraphics[width=.24\textwidth]{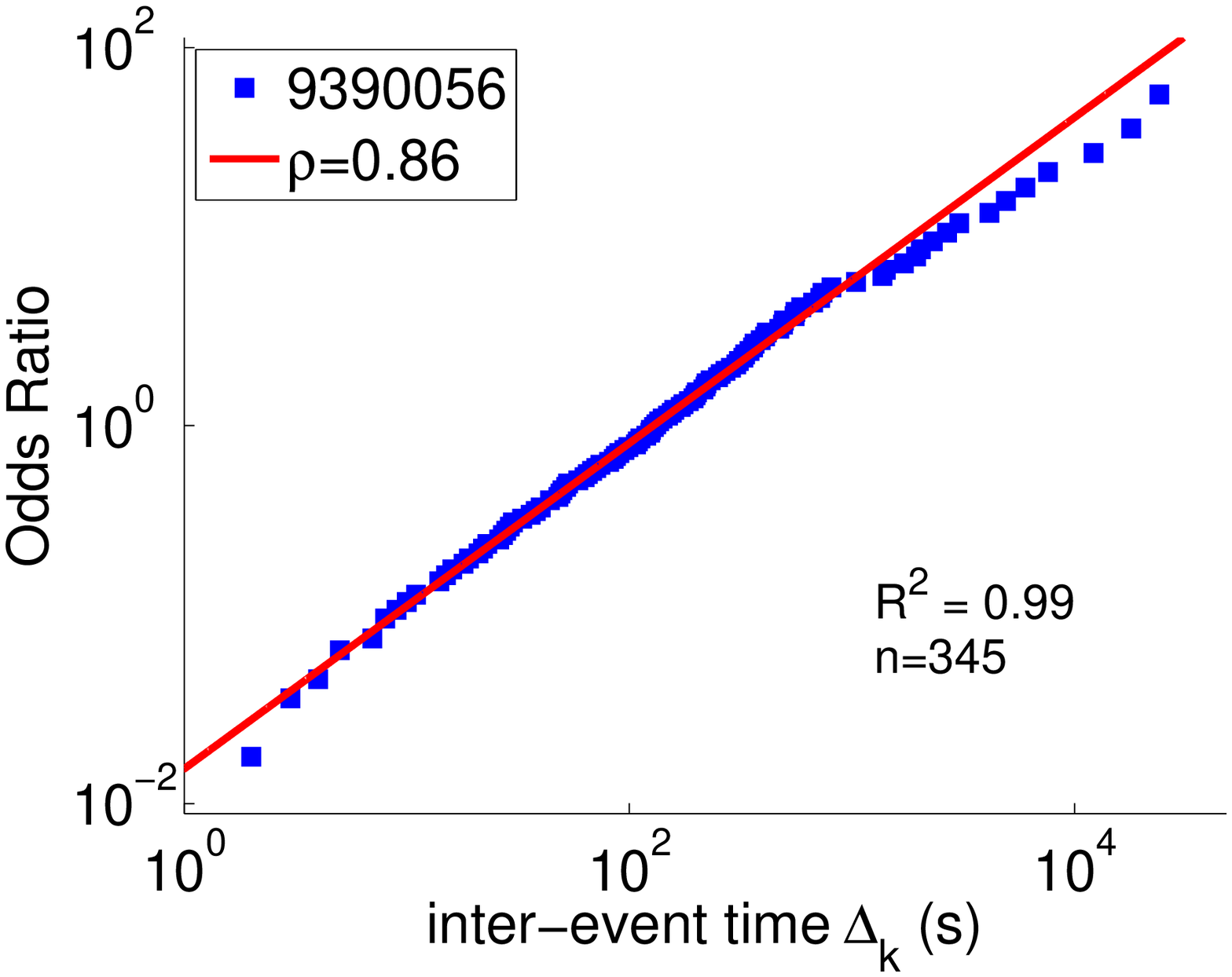}}
{\includegraphics[width=.24\textwidth]{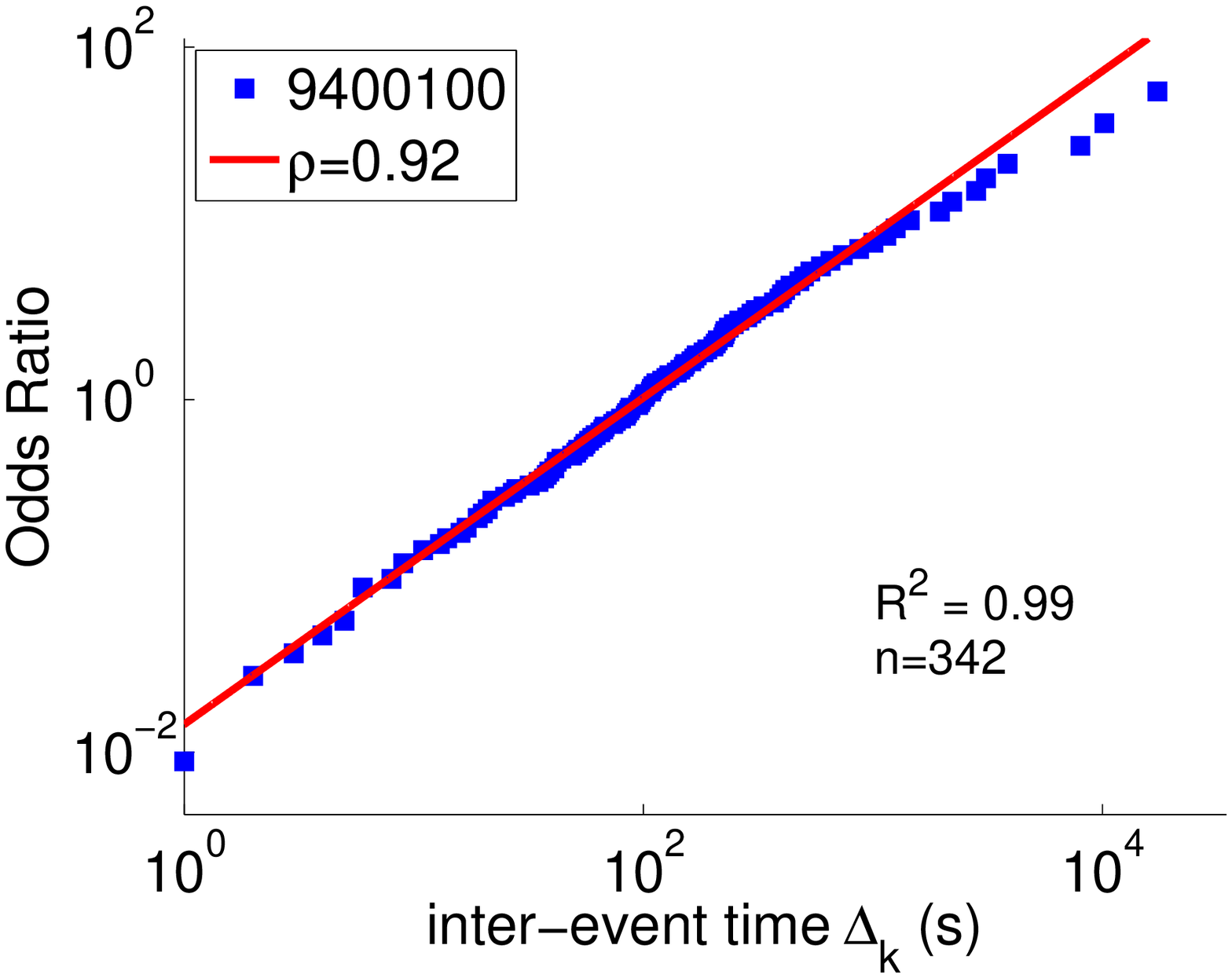}}
{\includegraphics[width=.24\textwidth]{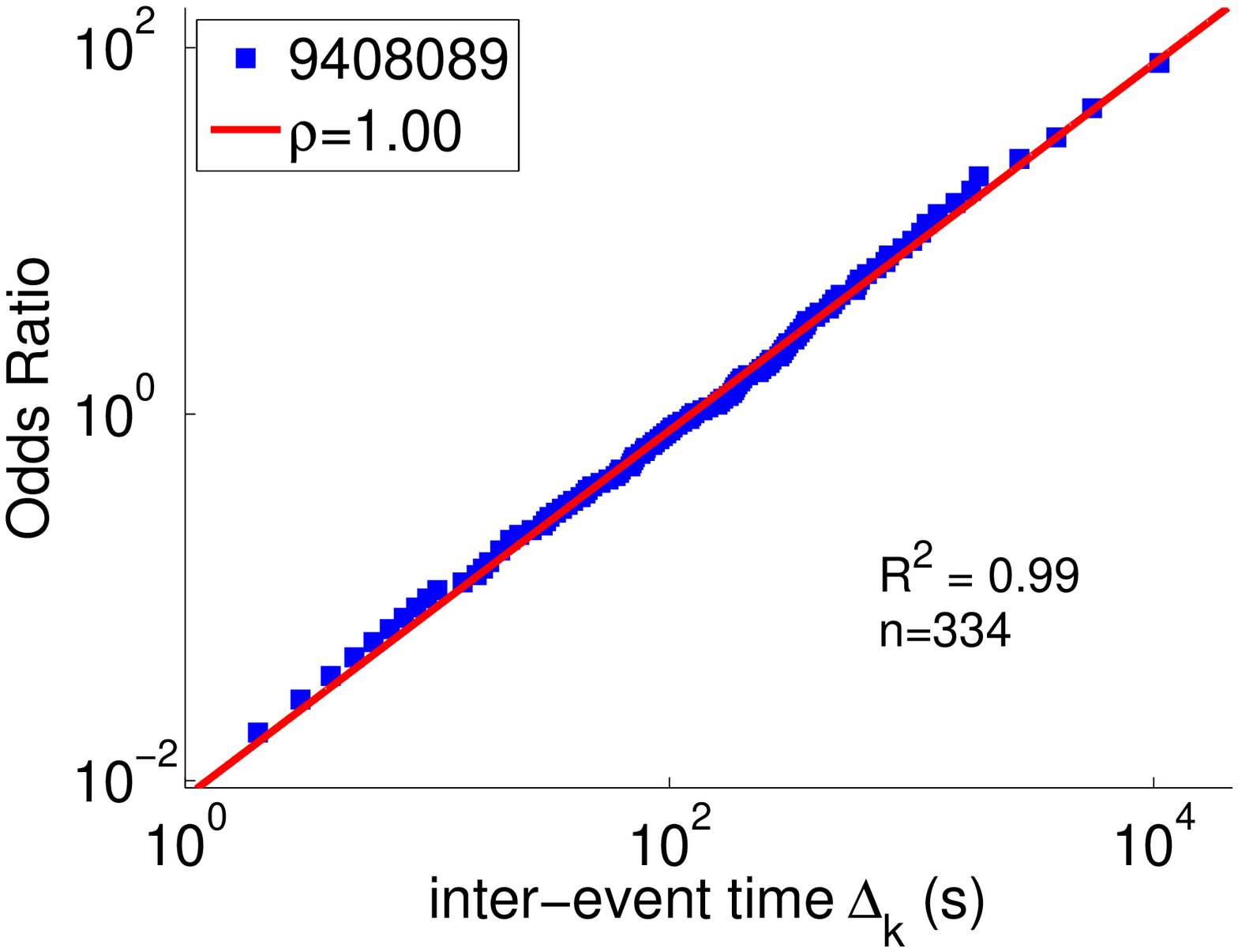}}
{\includegraphics[width=.24\textwidth]{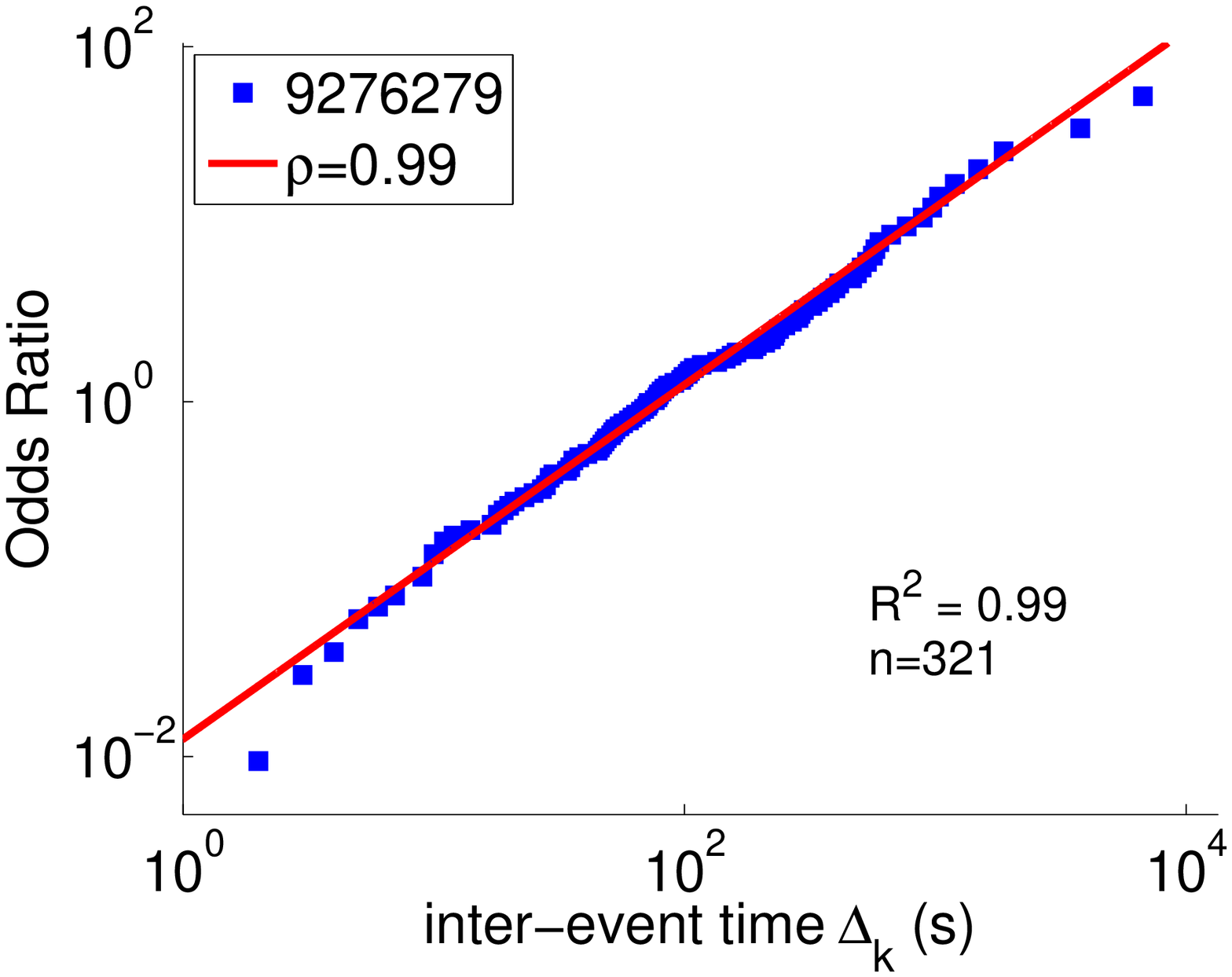}}
{\includegraphics[width=.24\textwidth]{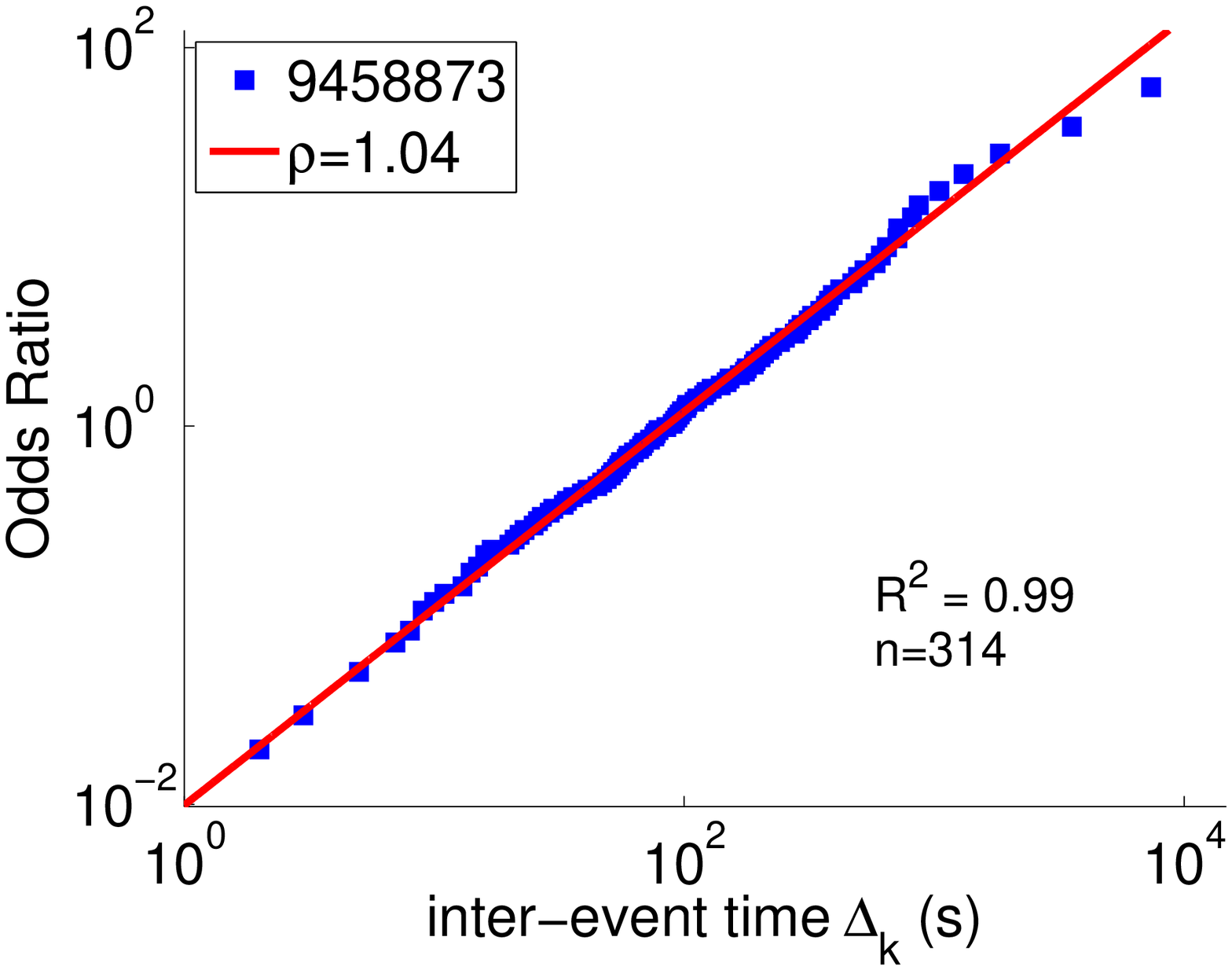}}
{\includegraphics[width=.24\textwidth]{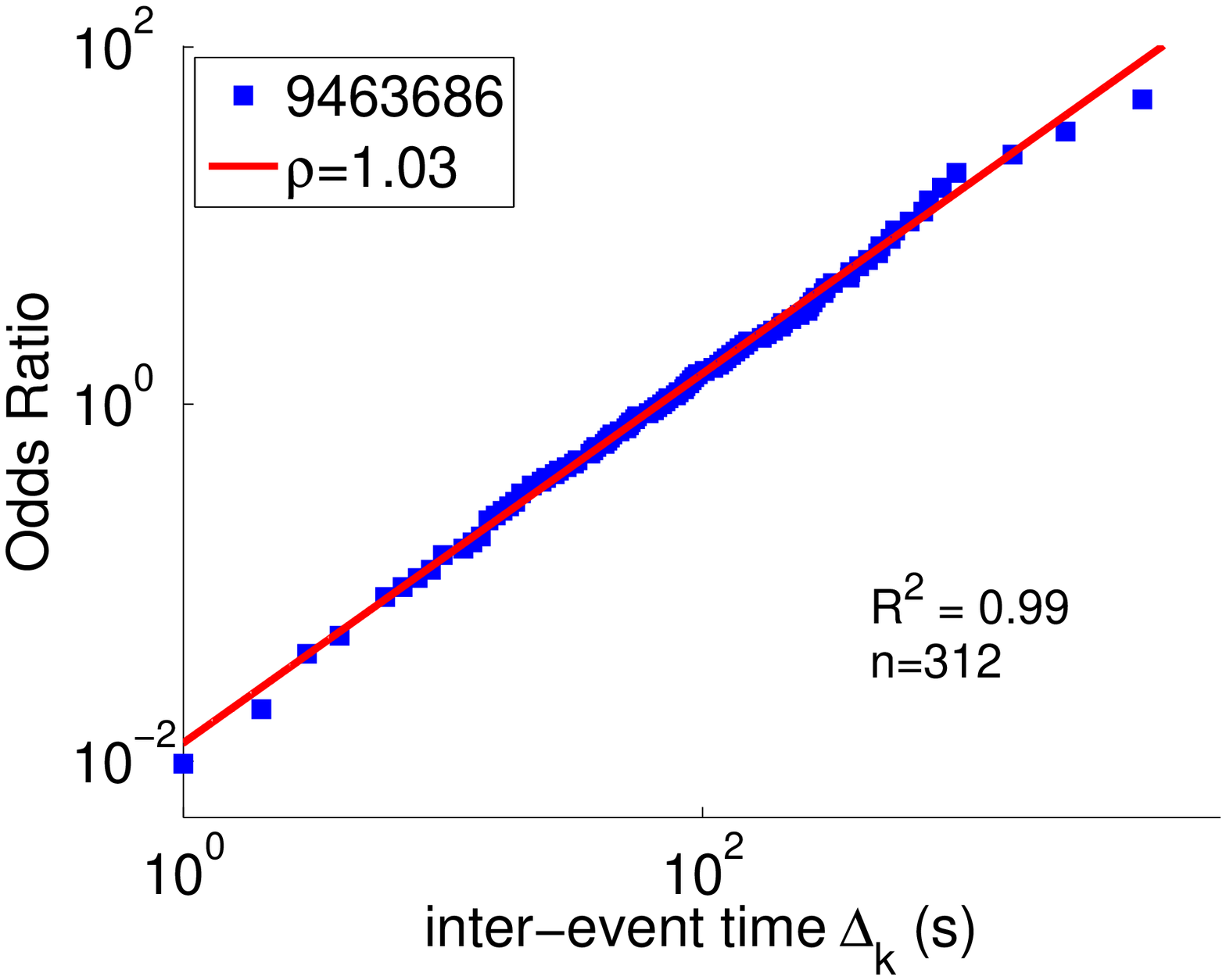}}
{\includegraphics[width=.24\textwidth]{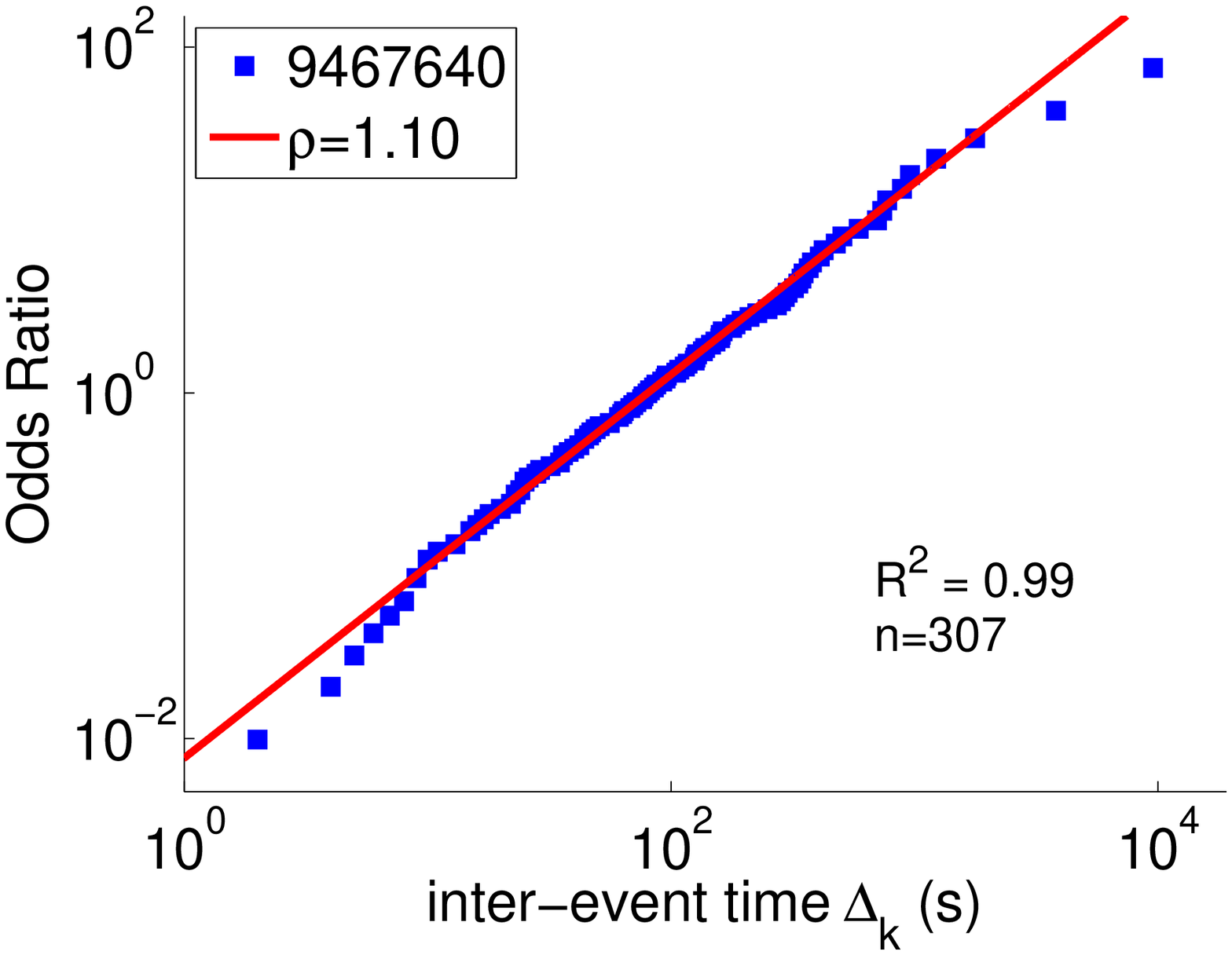}}
{\includegraphics[width=.24\textwidth]{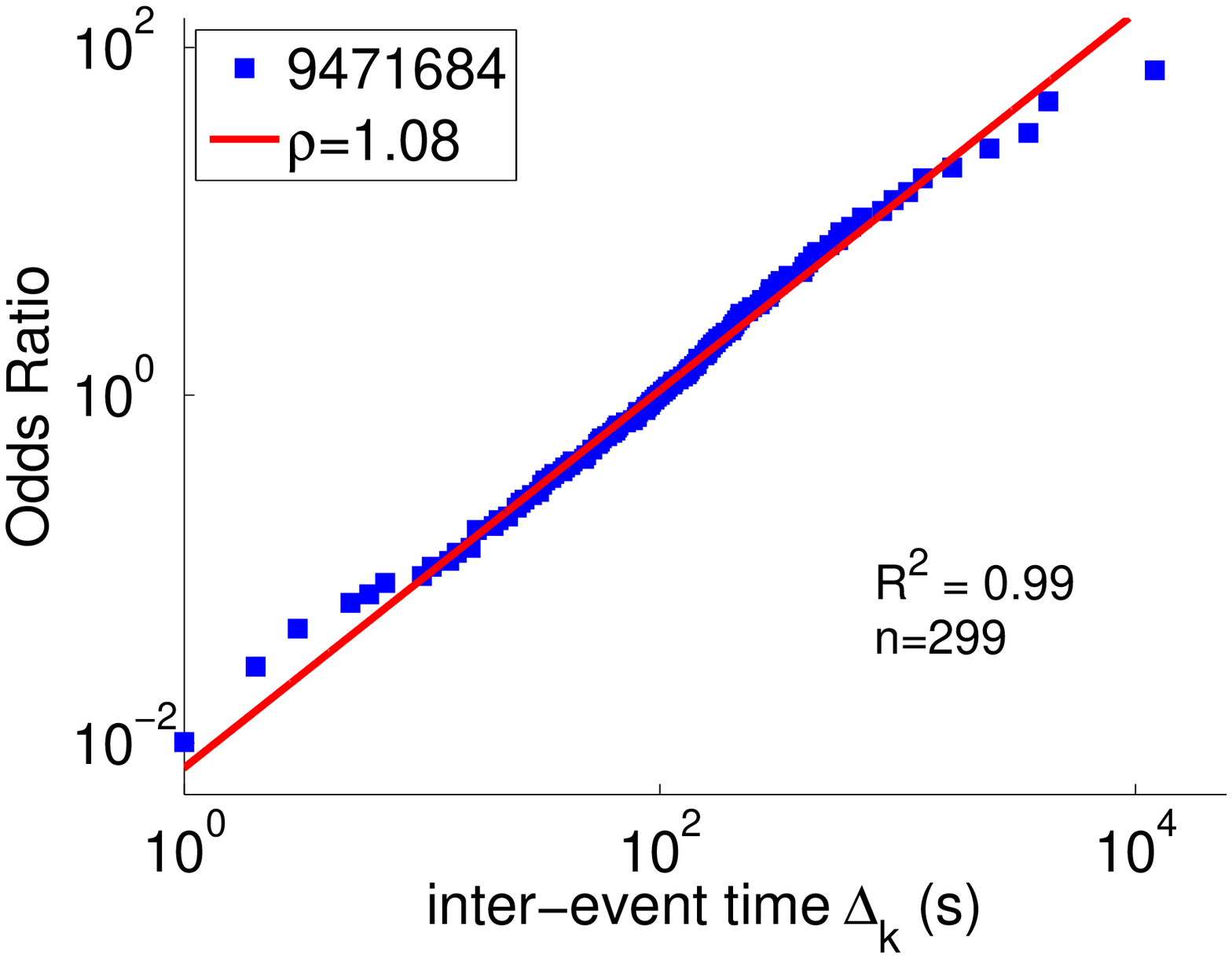}}
\caption{Sample from the Digg dataset.}
\end{figure*}

\begin{figure*}[htpb]
\centering
{\includegraphics[width=.24\textwidth]{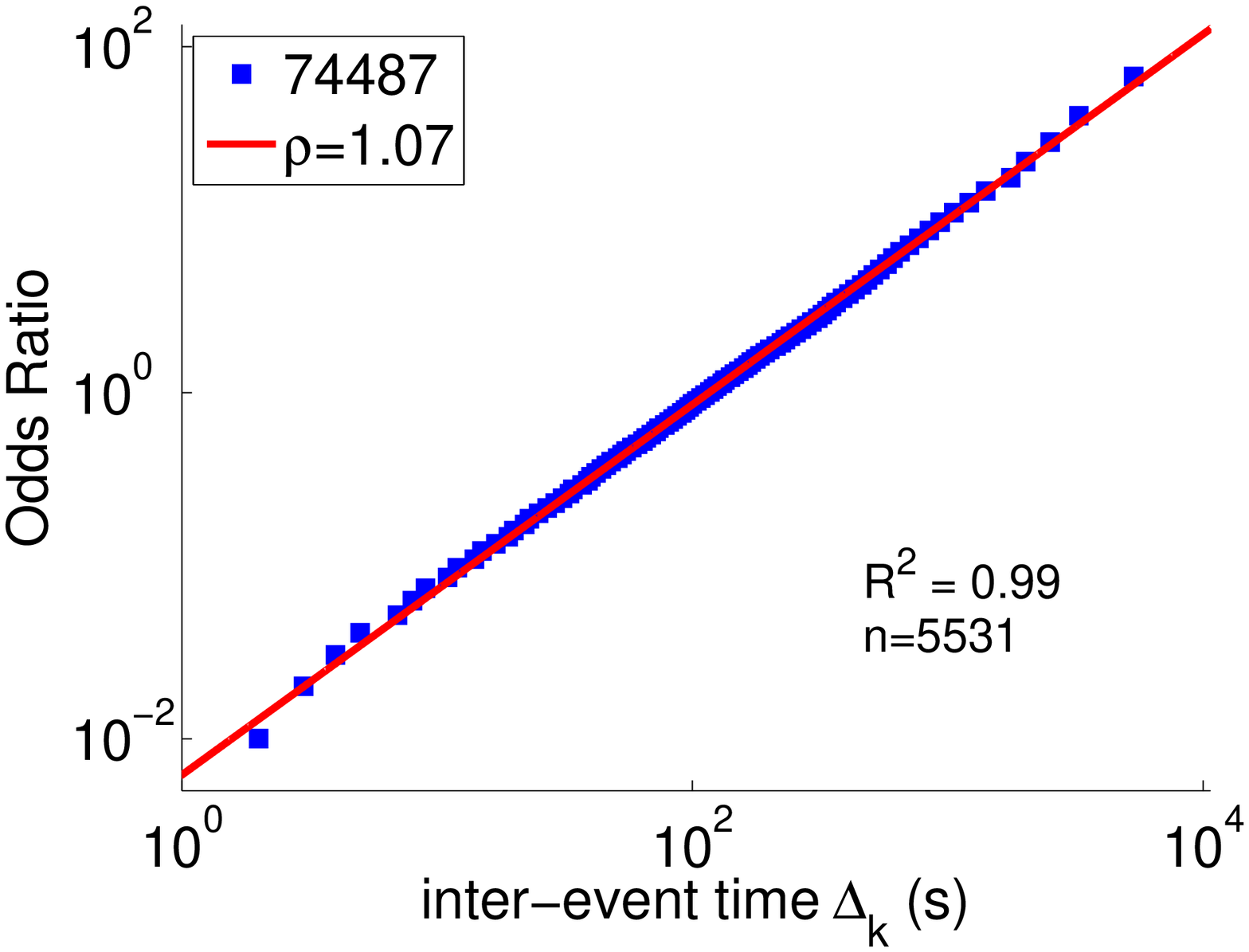}}
{\includegraphics[width=.24\textwidth]{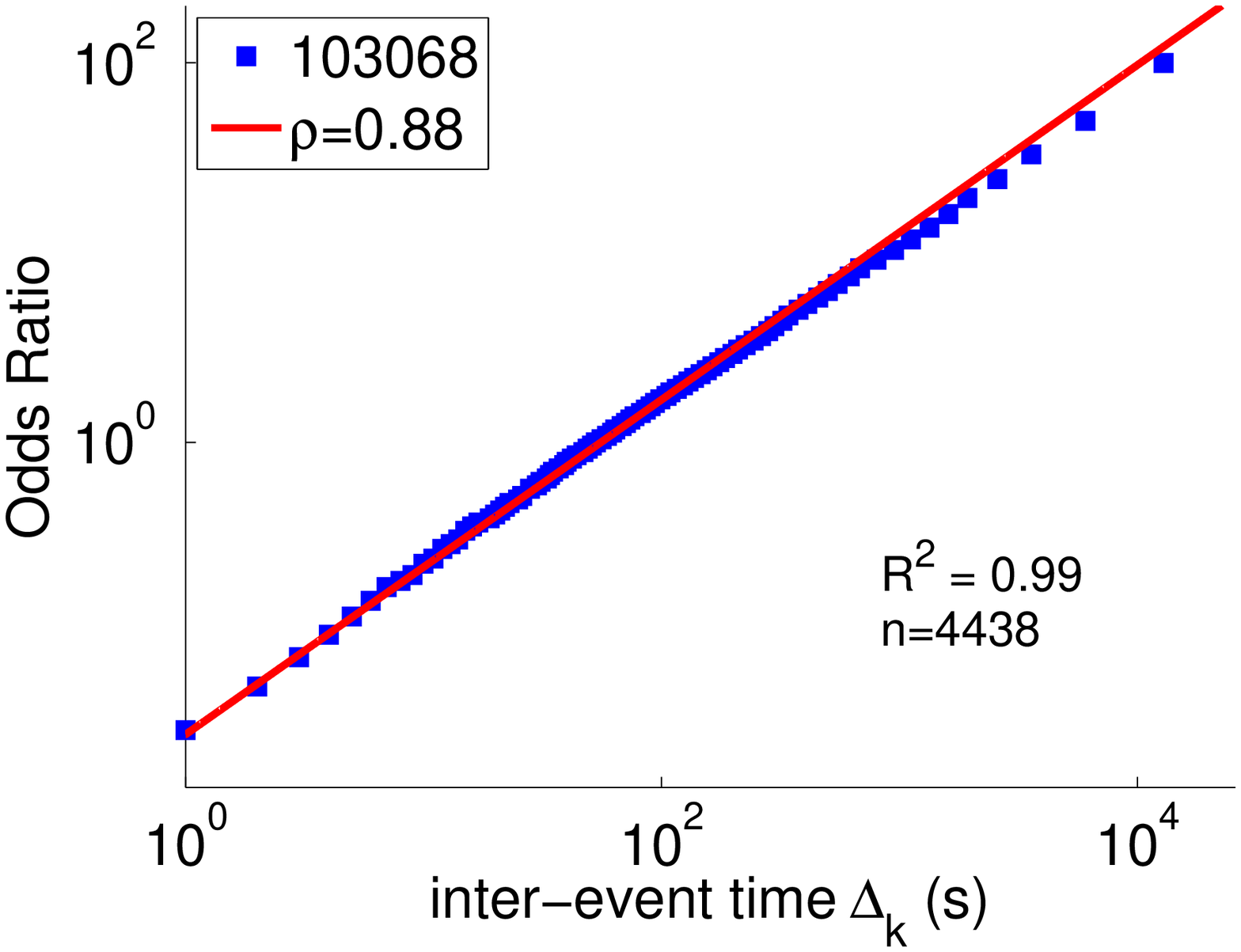}}
{\includegraphics[width=.24\textwidth]{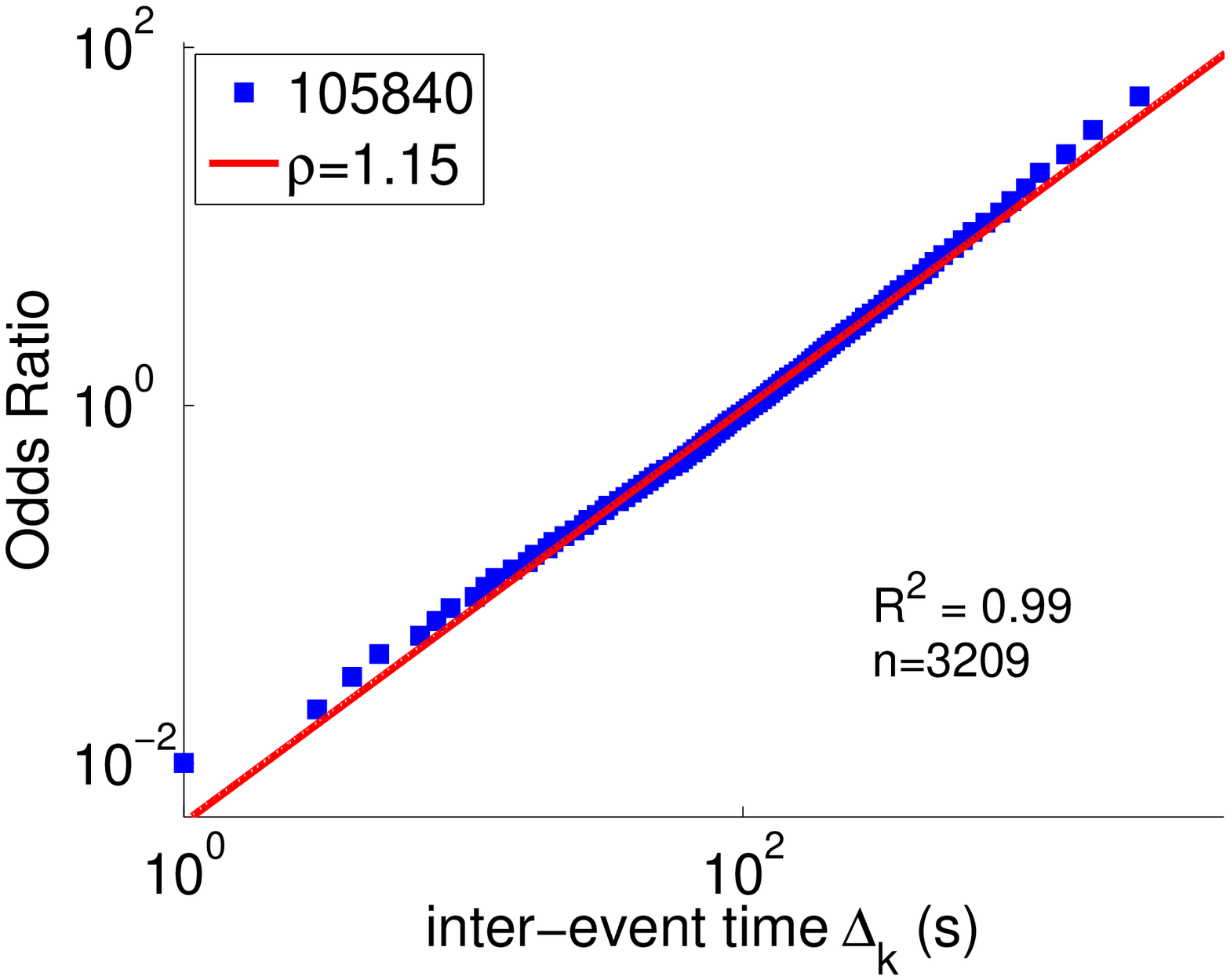}}
{\includegraphics[width=.24\textwidth]{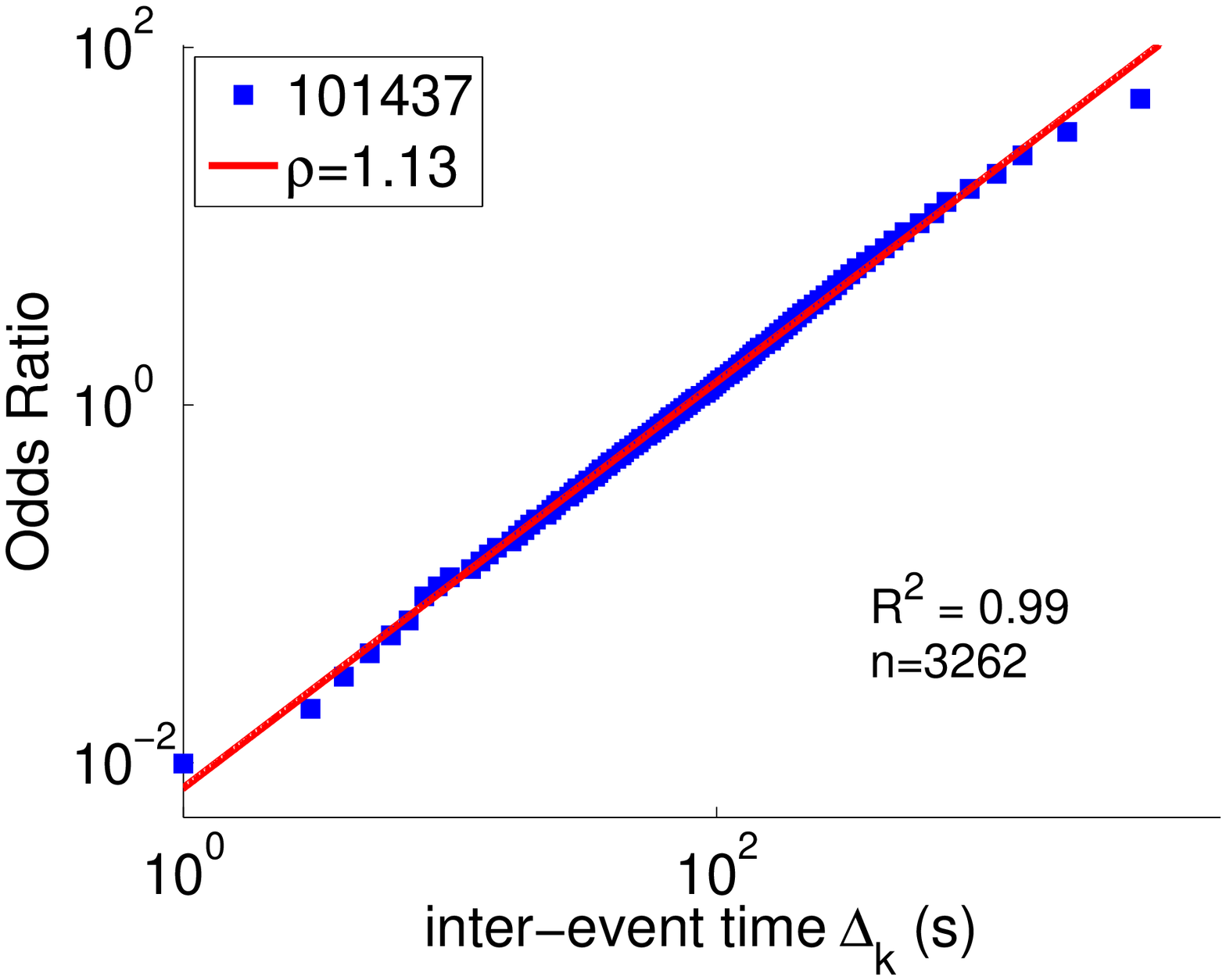}}
{\includegraphics[width=.24\textwidth]{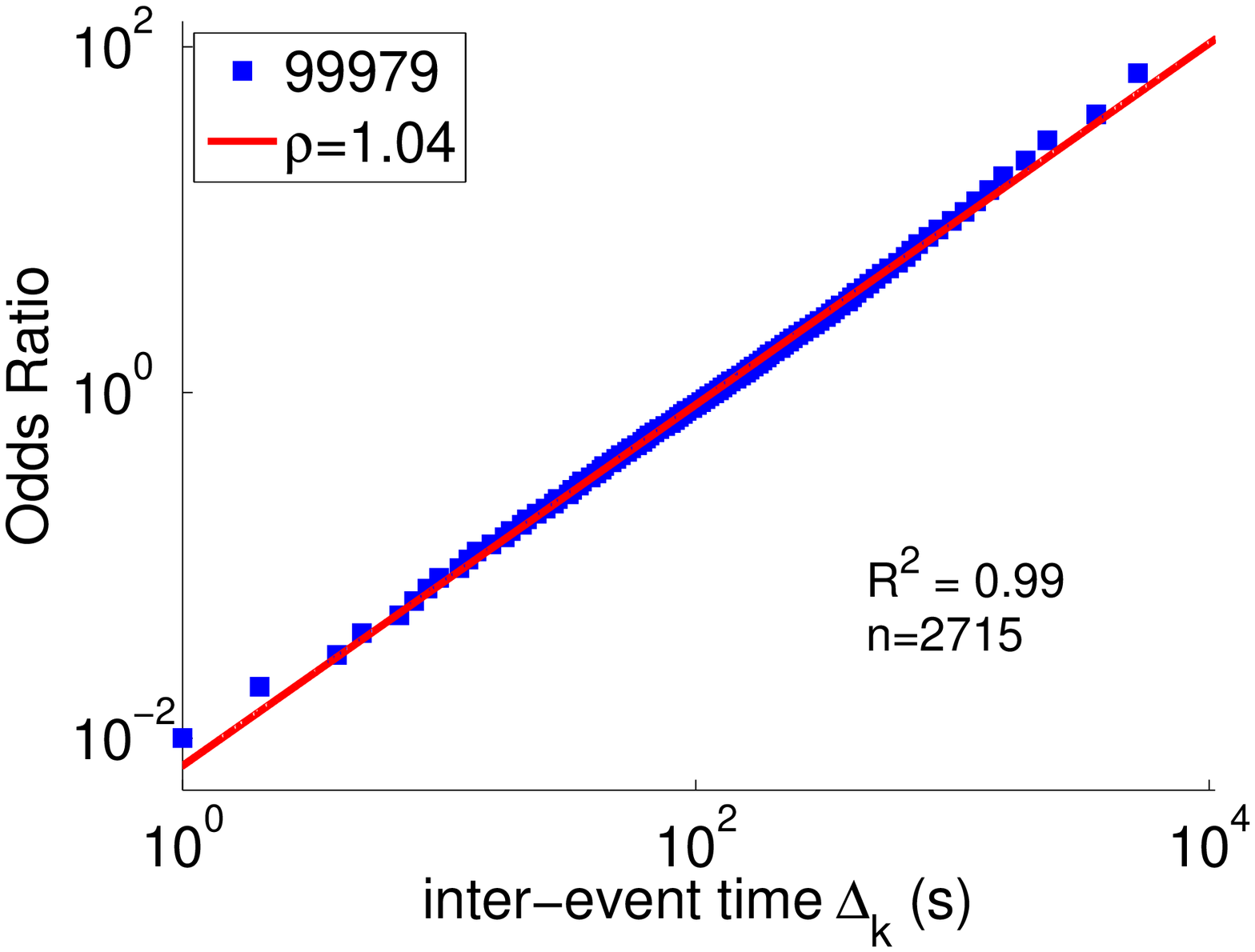}}
{\includegraphics[width=.24\textwidth]{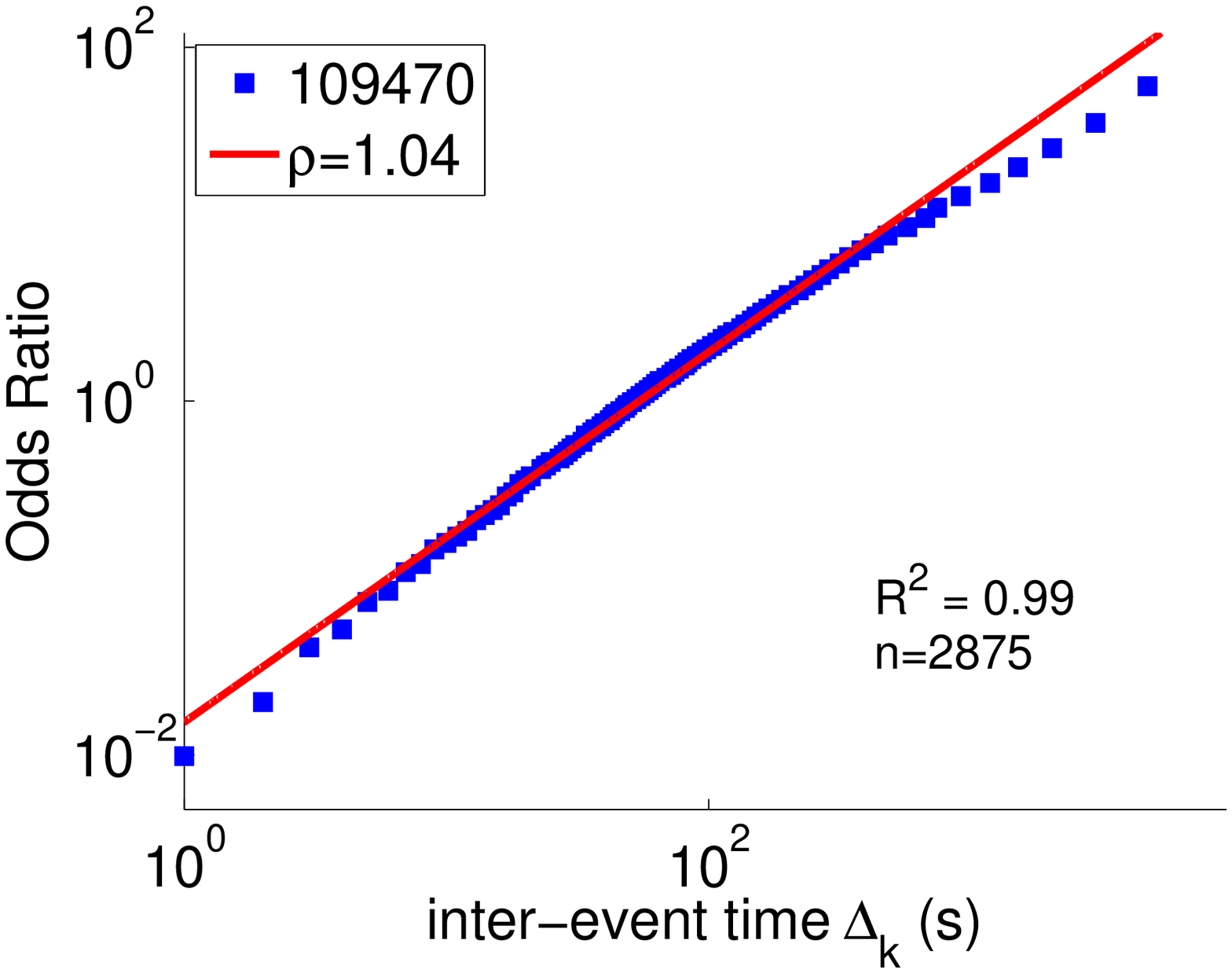}}
{\includegraphics[width=.24\textwidth]{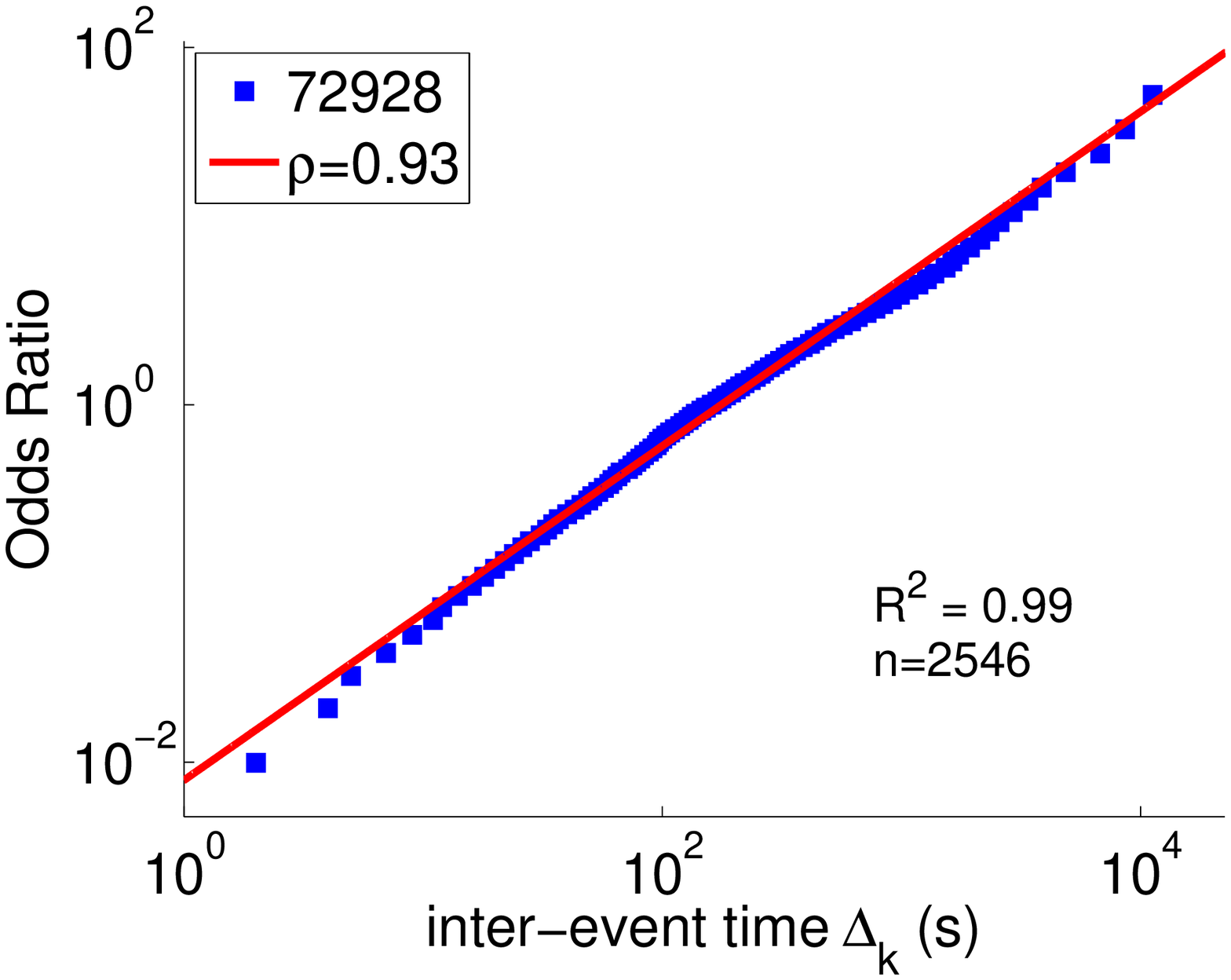}}
{\includegraphics[width=.24\textwidth]{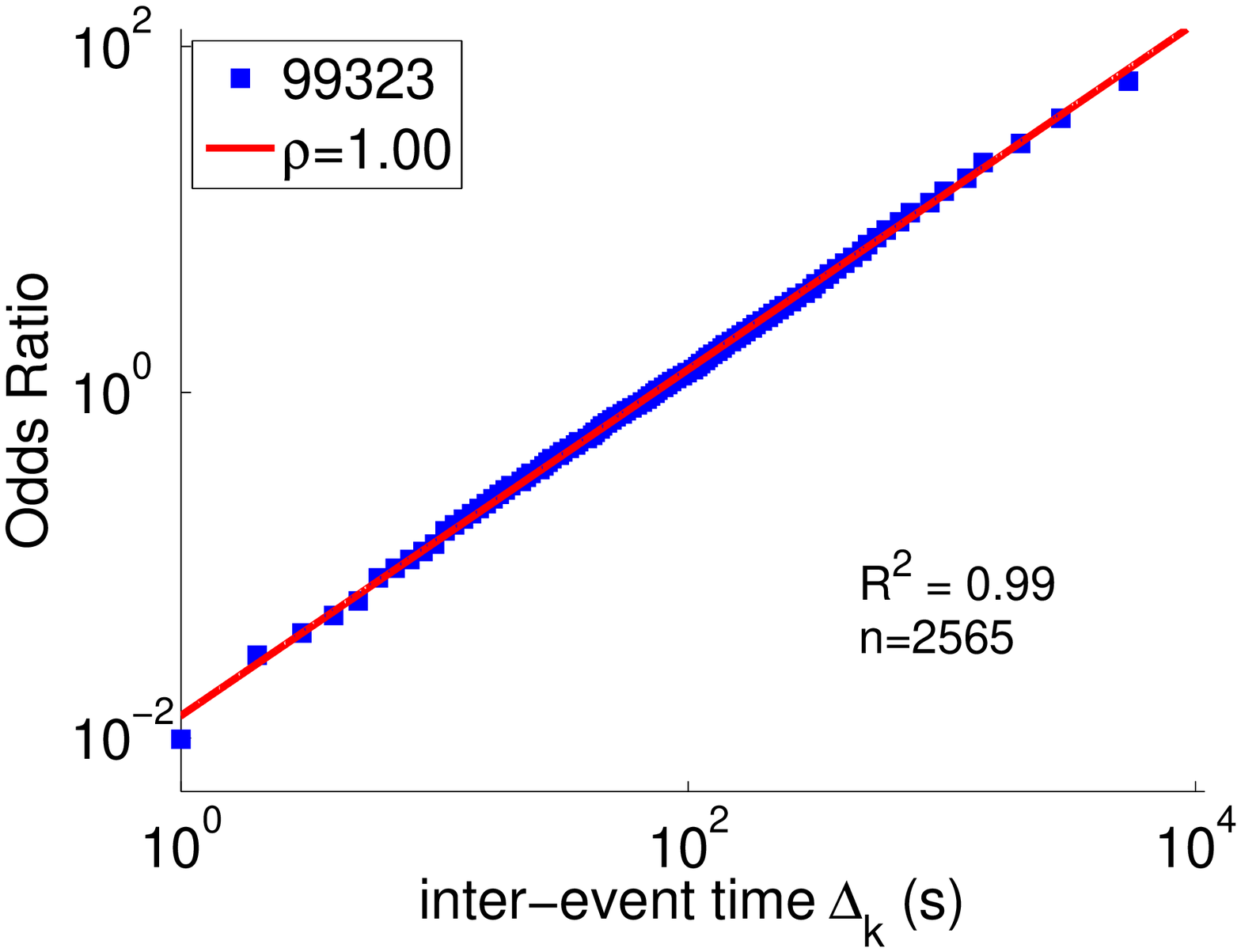}}
\caption{Sample from the MetaFilter dataset.}
\end{figure*}

\begin{figure*}[htpb]
\centering
{\includegraphics[width=.24\textwidth]{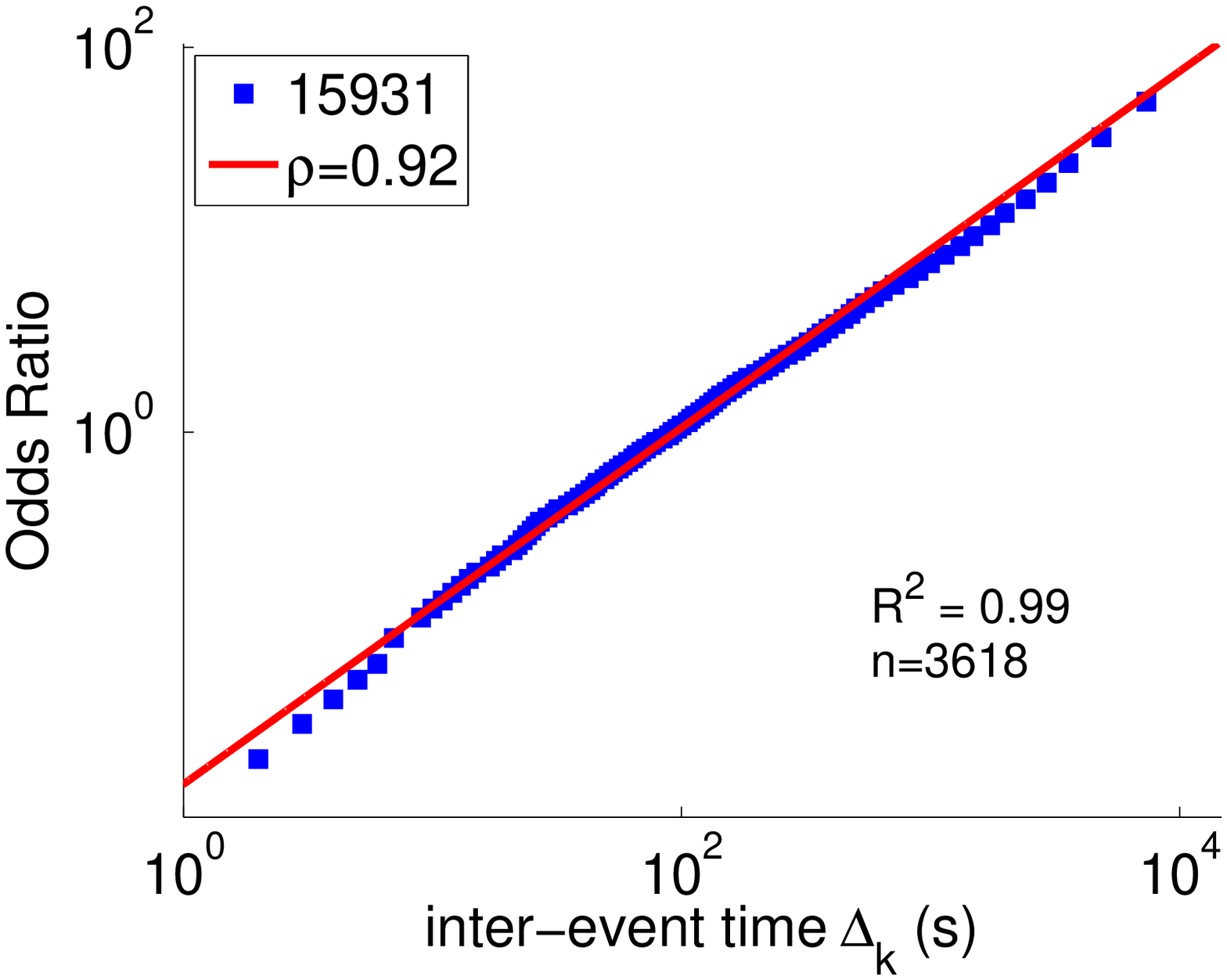}}
{\includegraphics[width=.24\textwidth]{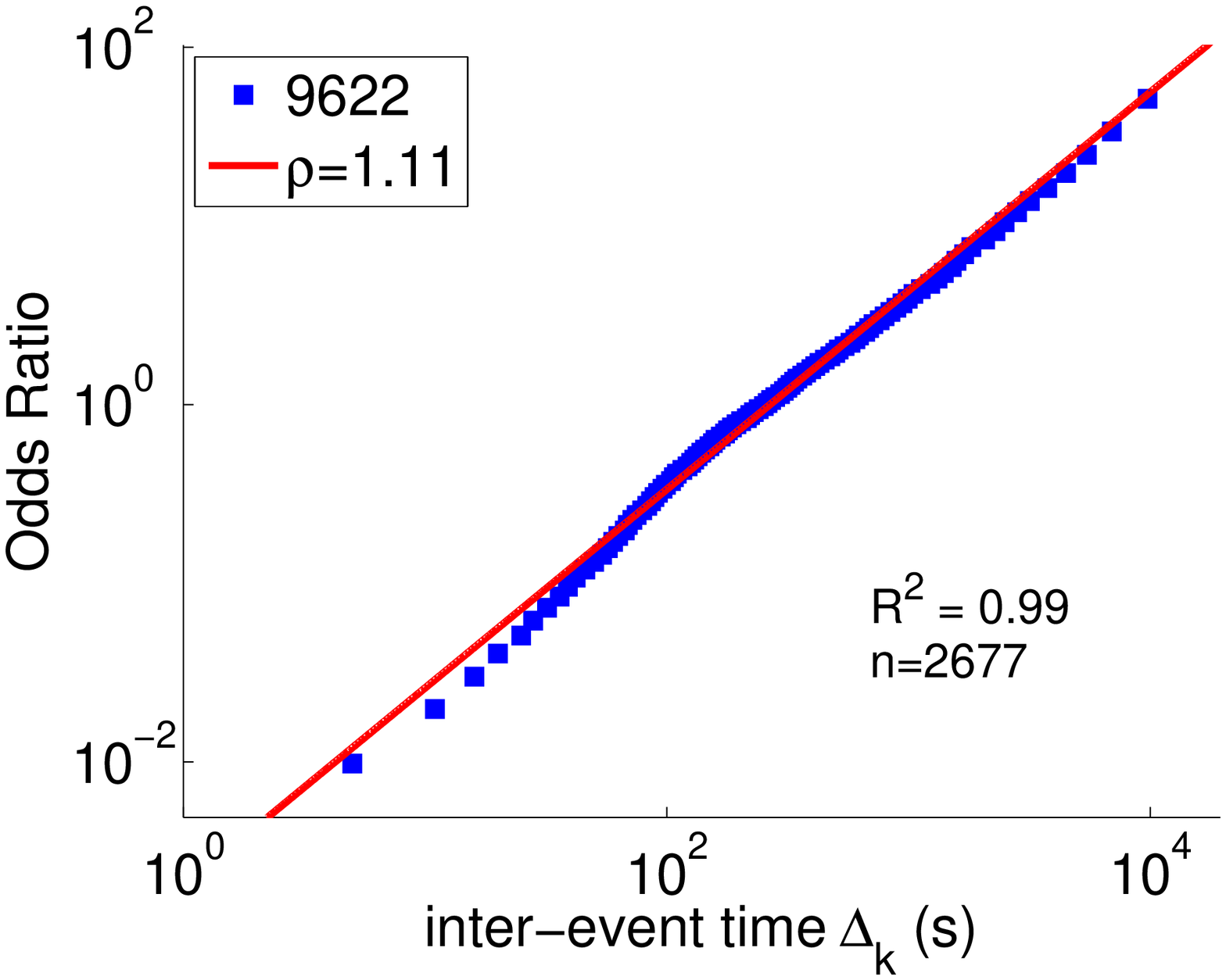}}
{\includegraphics[width=.24\textwidth]{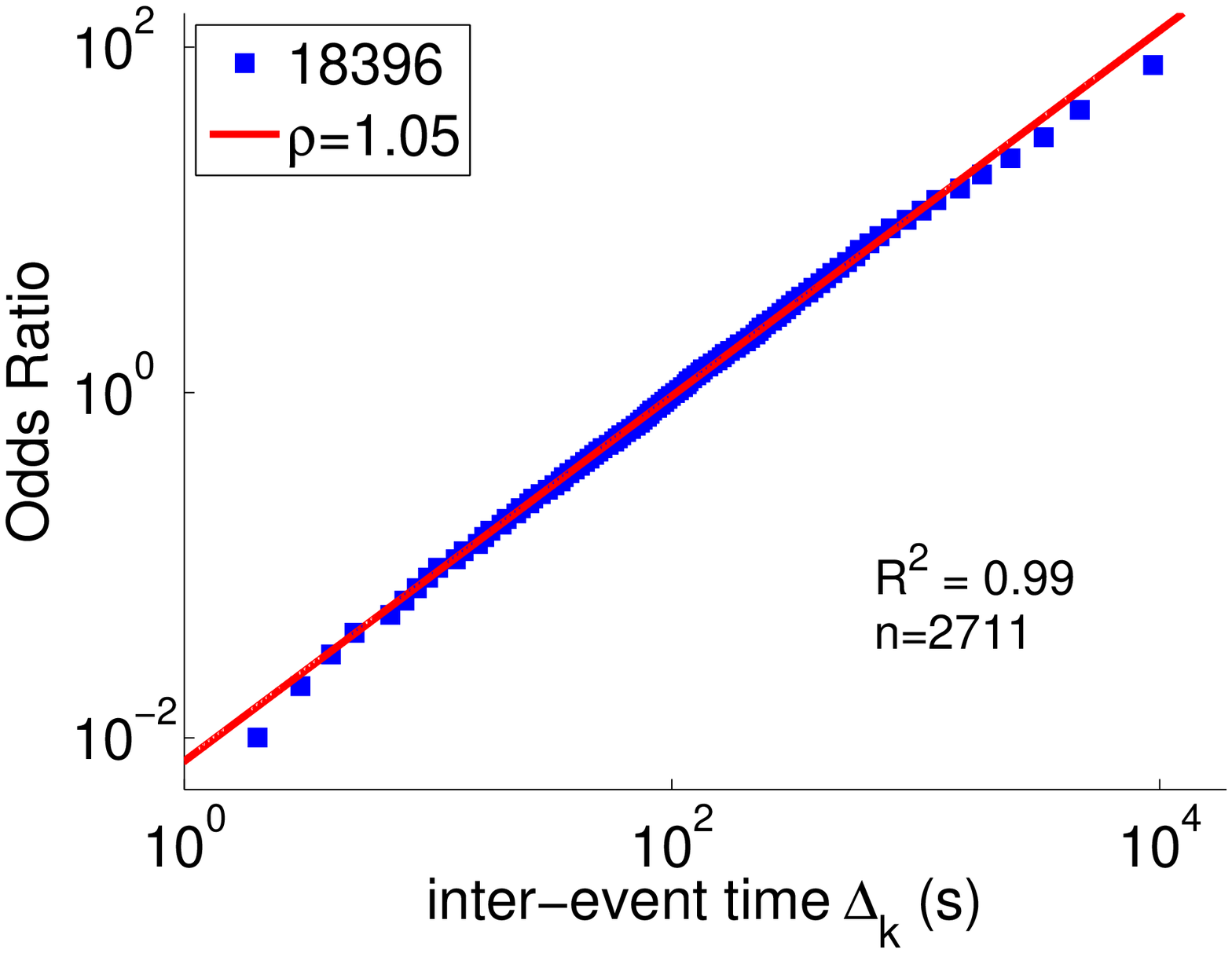}}
{\includegraphics[width=.24\textwidth]{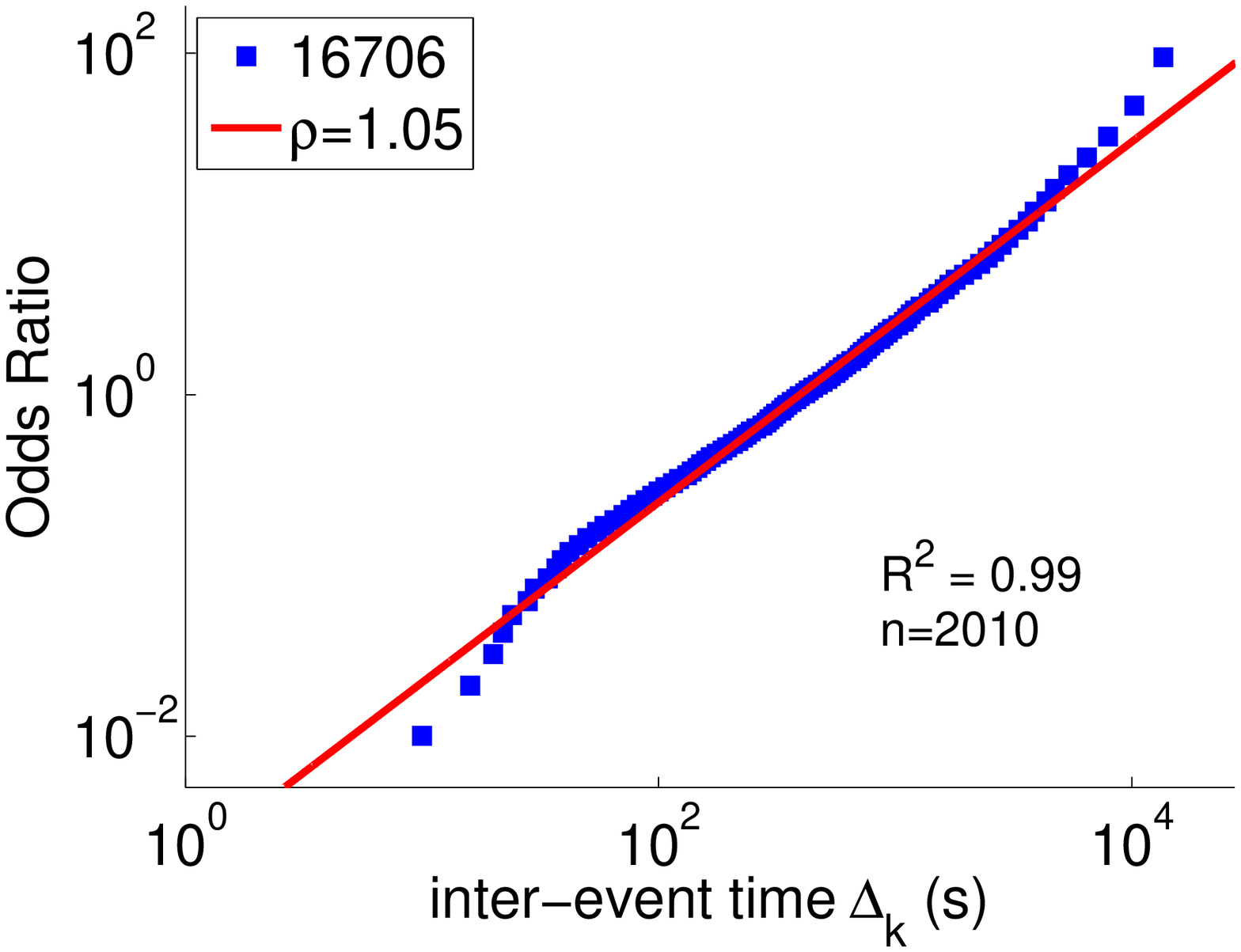}}
{\includegraphics[width=.24\textwidth]{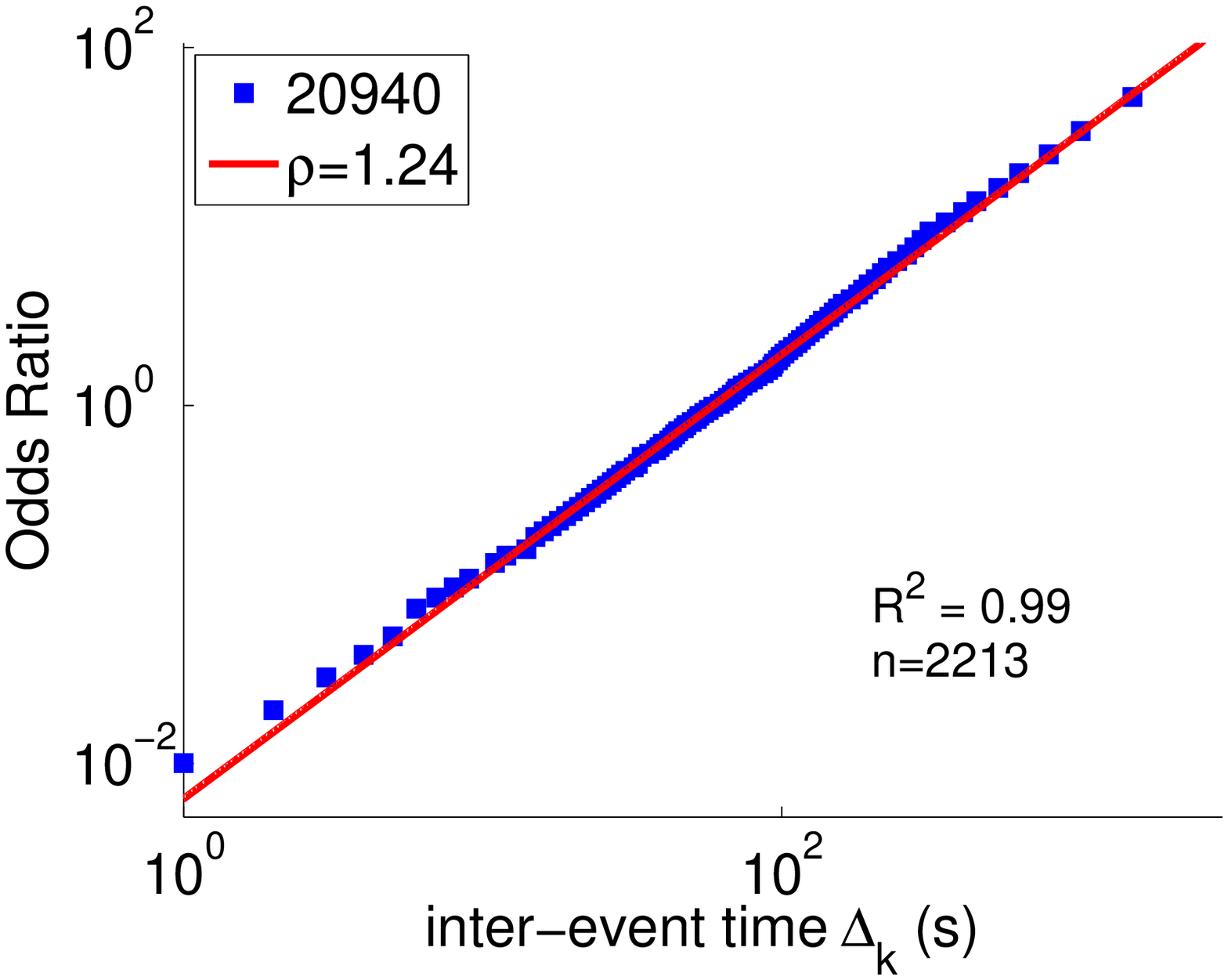}}
{\includegraphics[width=.24\textwidth]{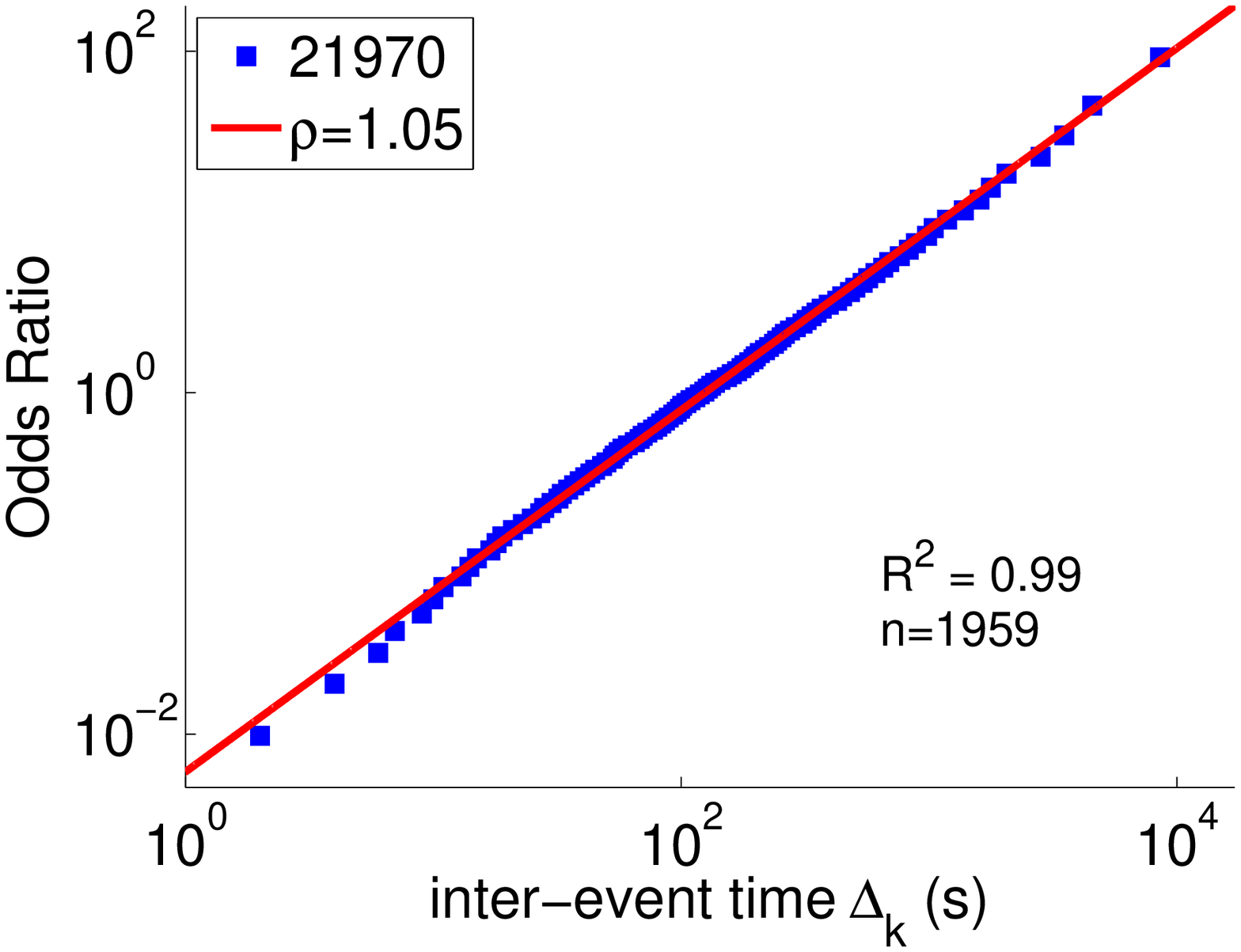}}
{\includegraphics[width=.24\textwidth]{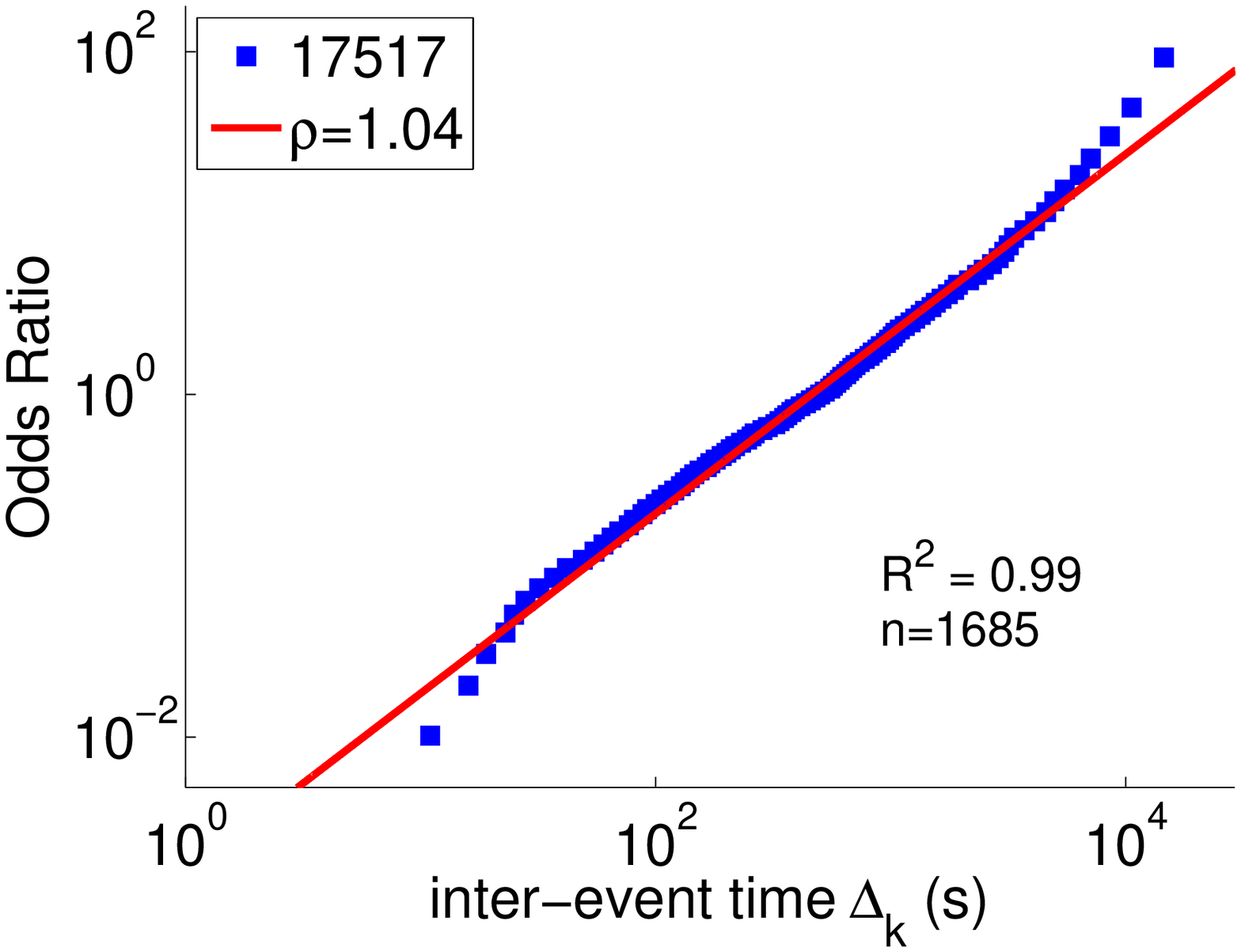}}
{\includegraphics[width=.24\textwidth]{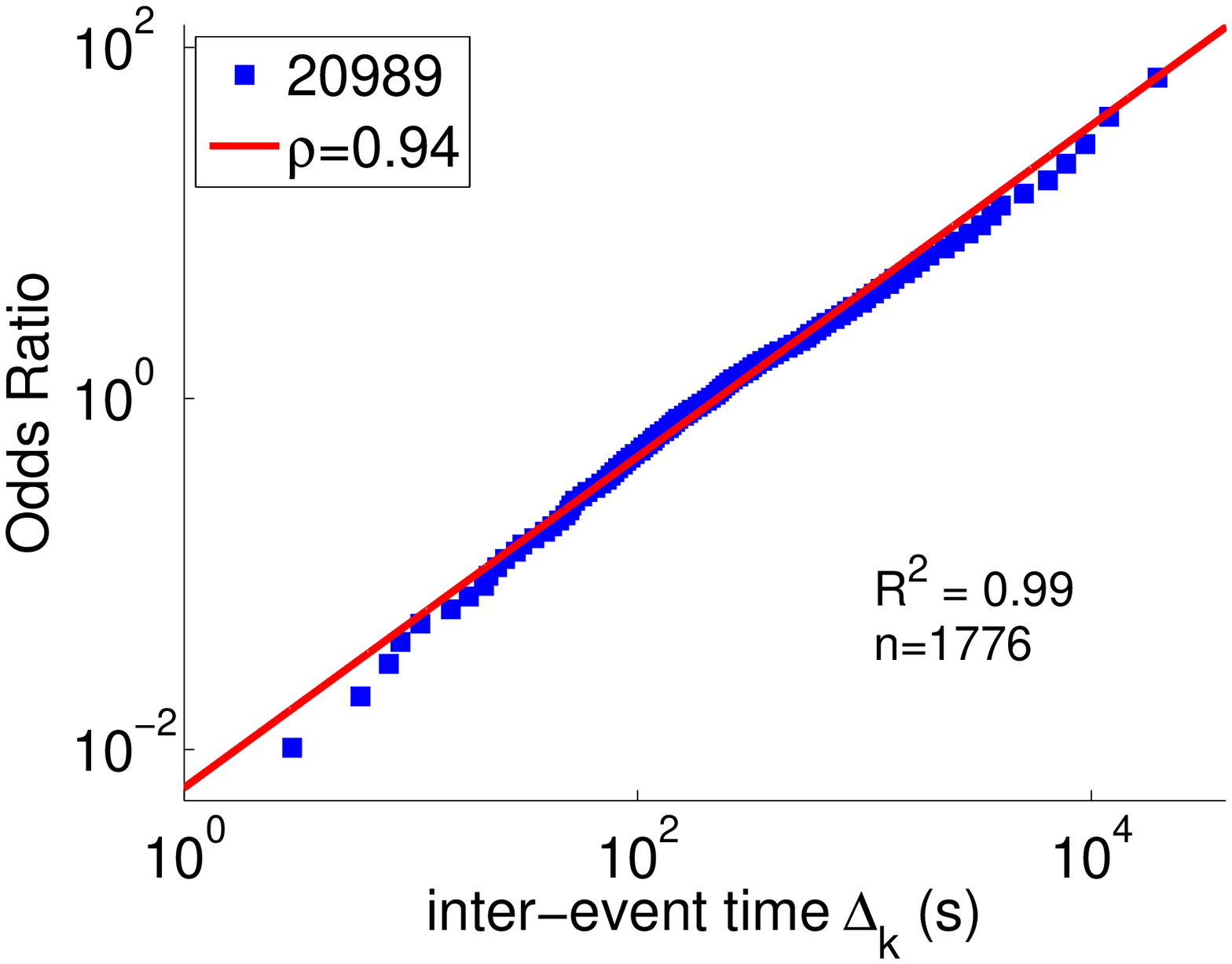}}
\caption{Sample from the MetaTalk dataset.}
\end{figure*}

\begin{figure*}[htpb]
\centering
{\includegraphics[width=.24\textwidth]{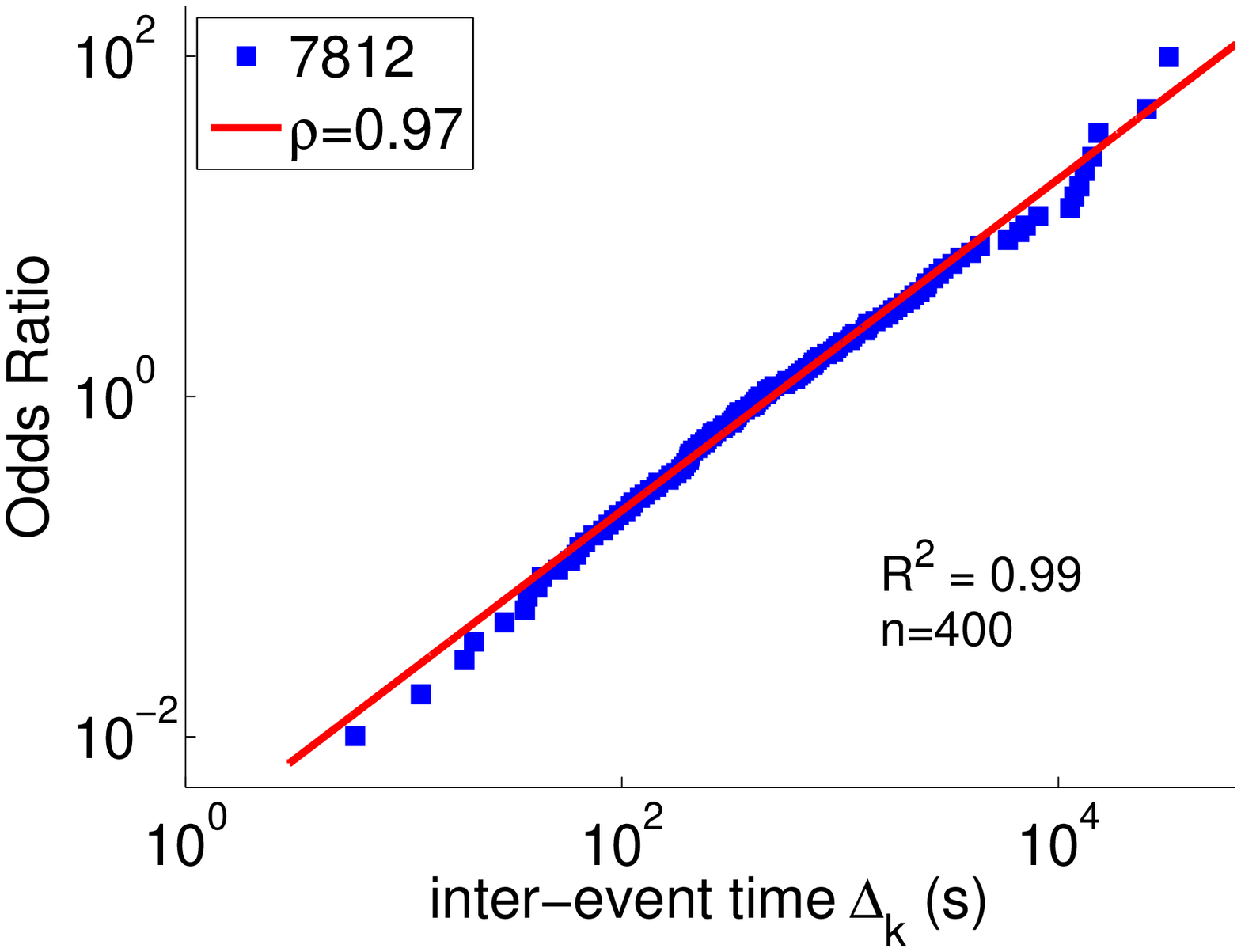}}
{\includegraphics[width=.24\textwidth]{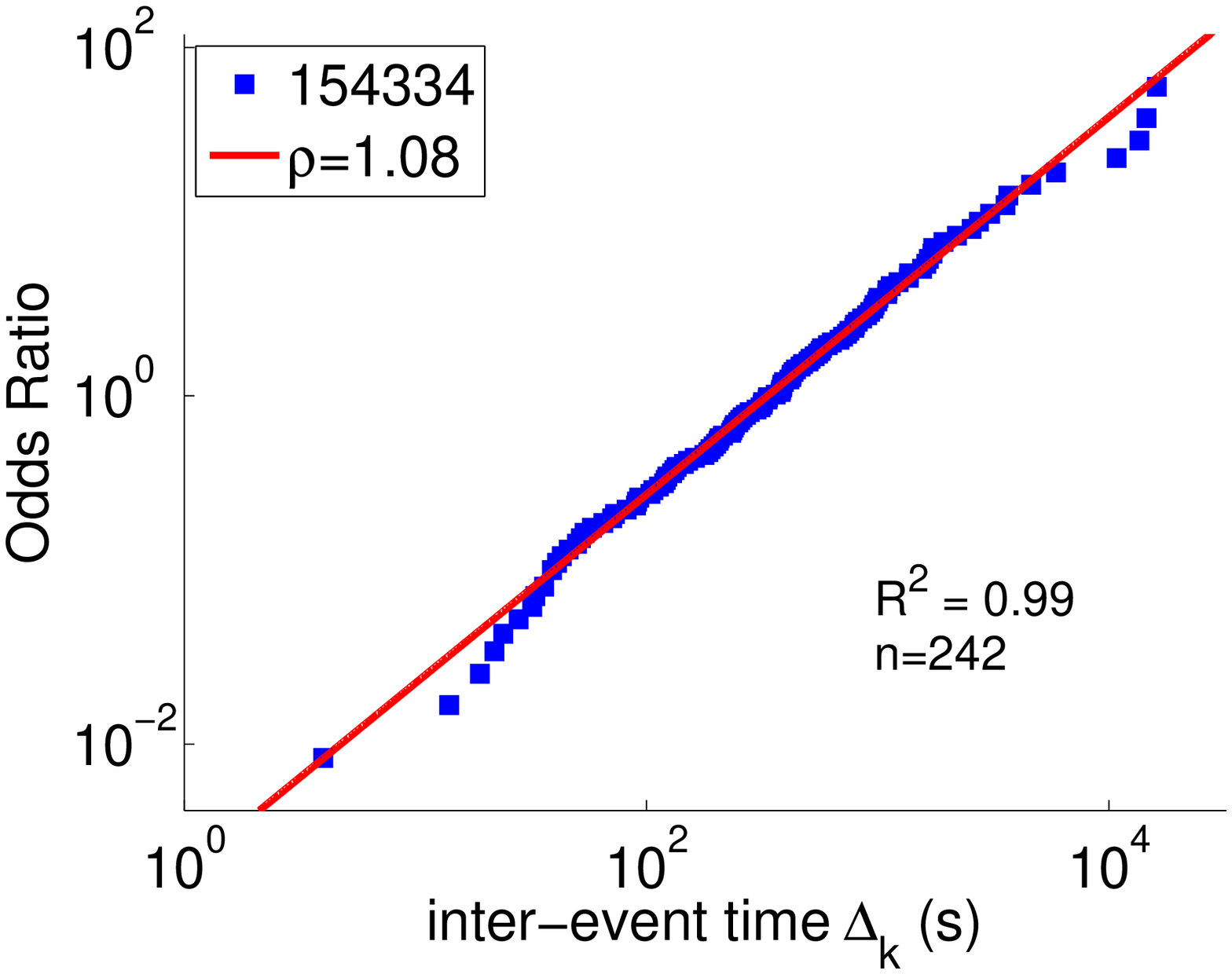}}
{\includegraphics[width=.24\textwidth]{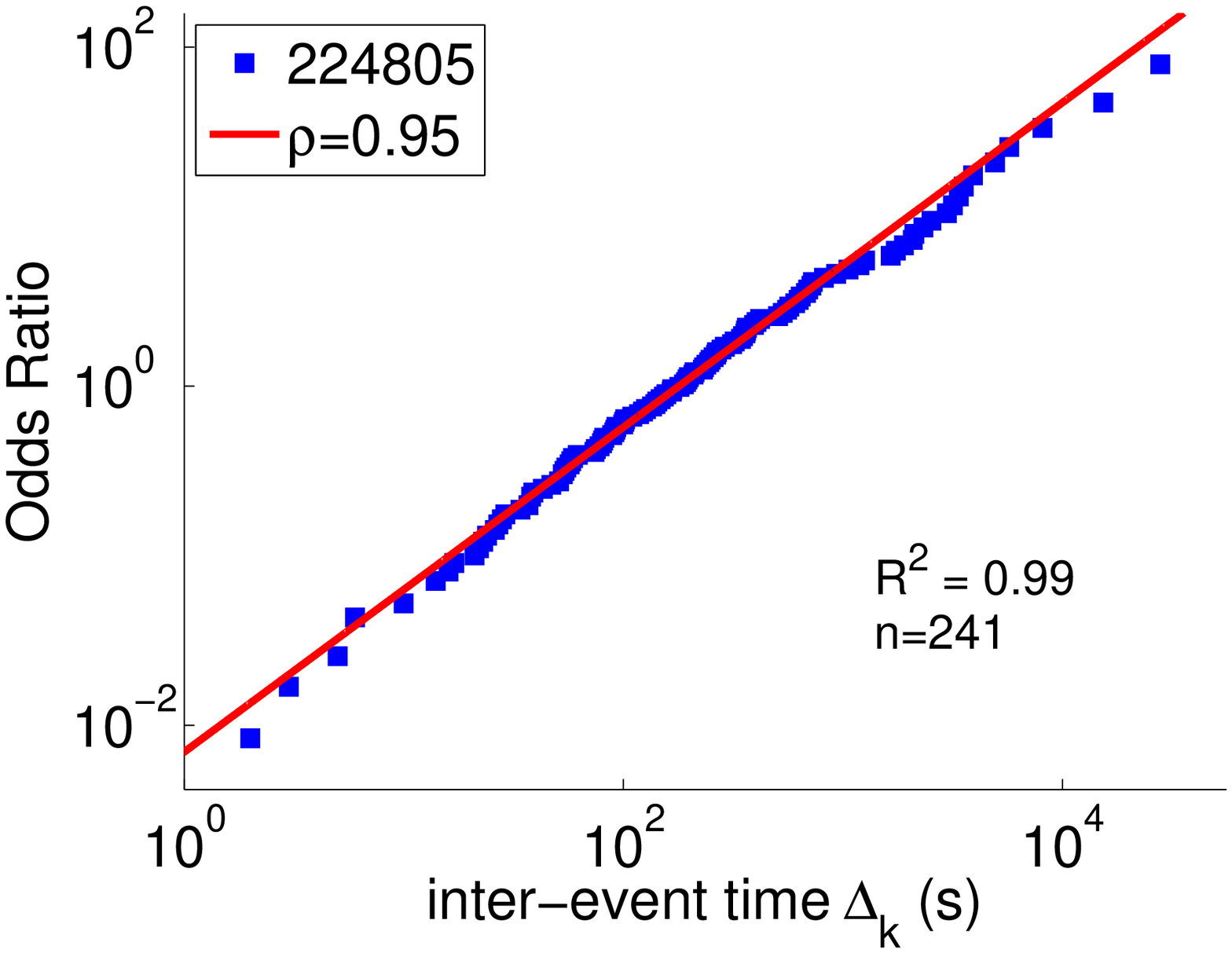}}
{\includegraphics[width=.24\textwidth]{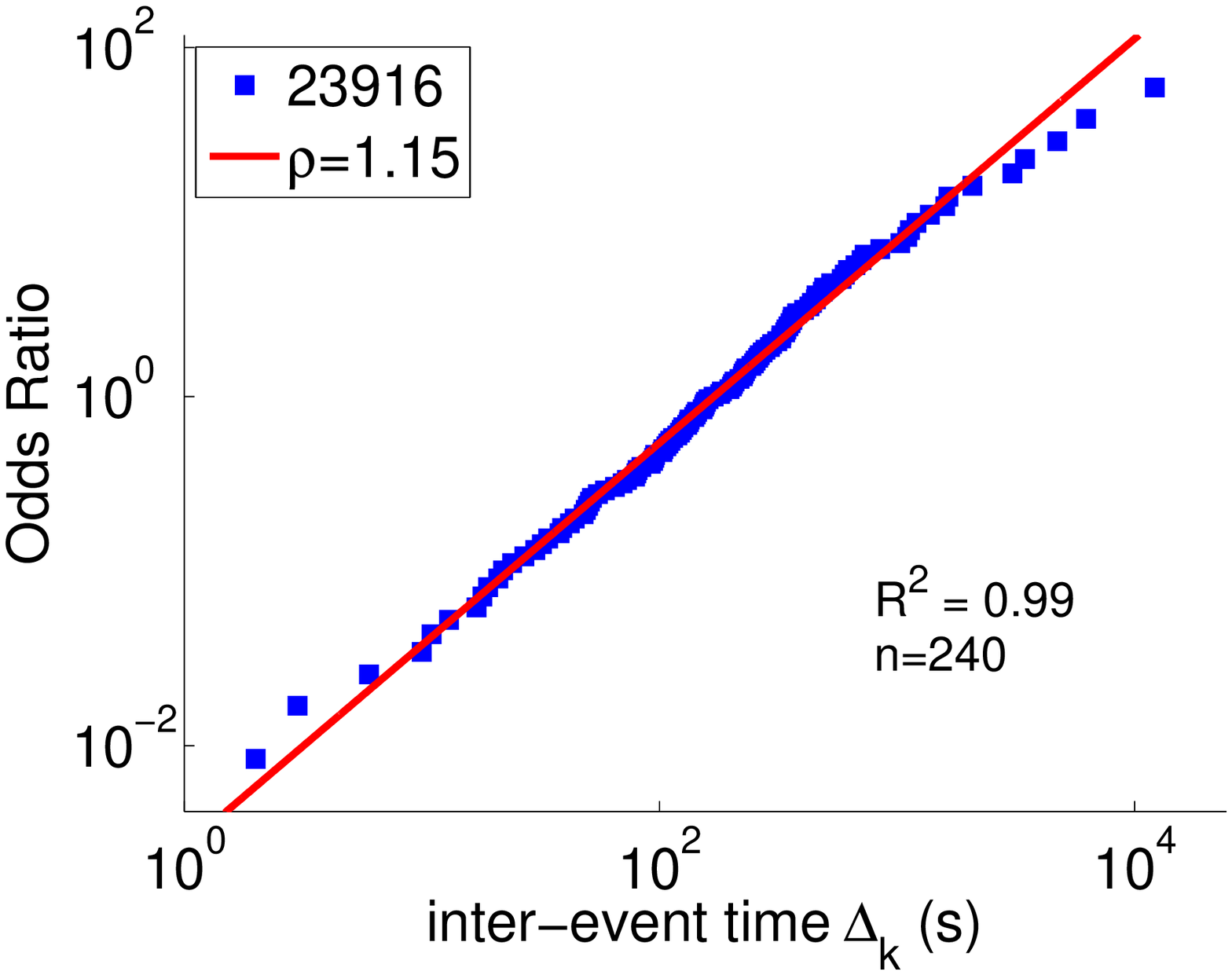}}
{\includegraphics[width=.24\textwidth]{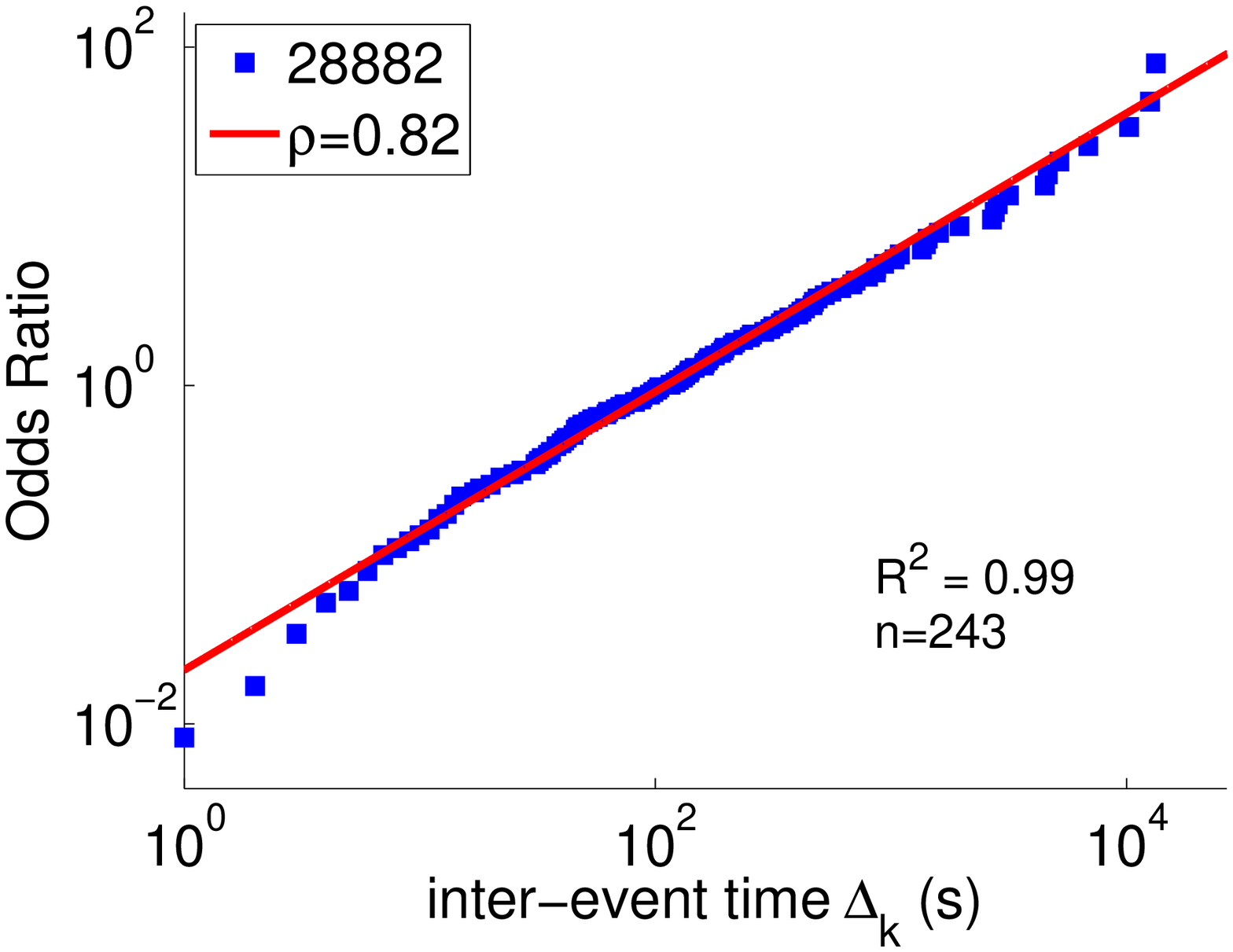}}
{\includegraphics[width=.24\textwidth]{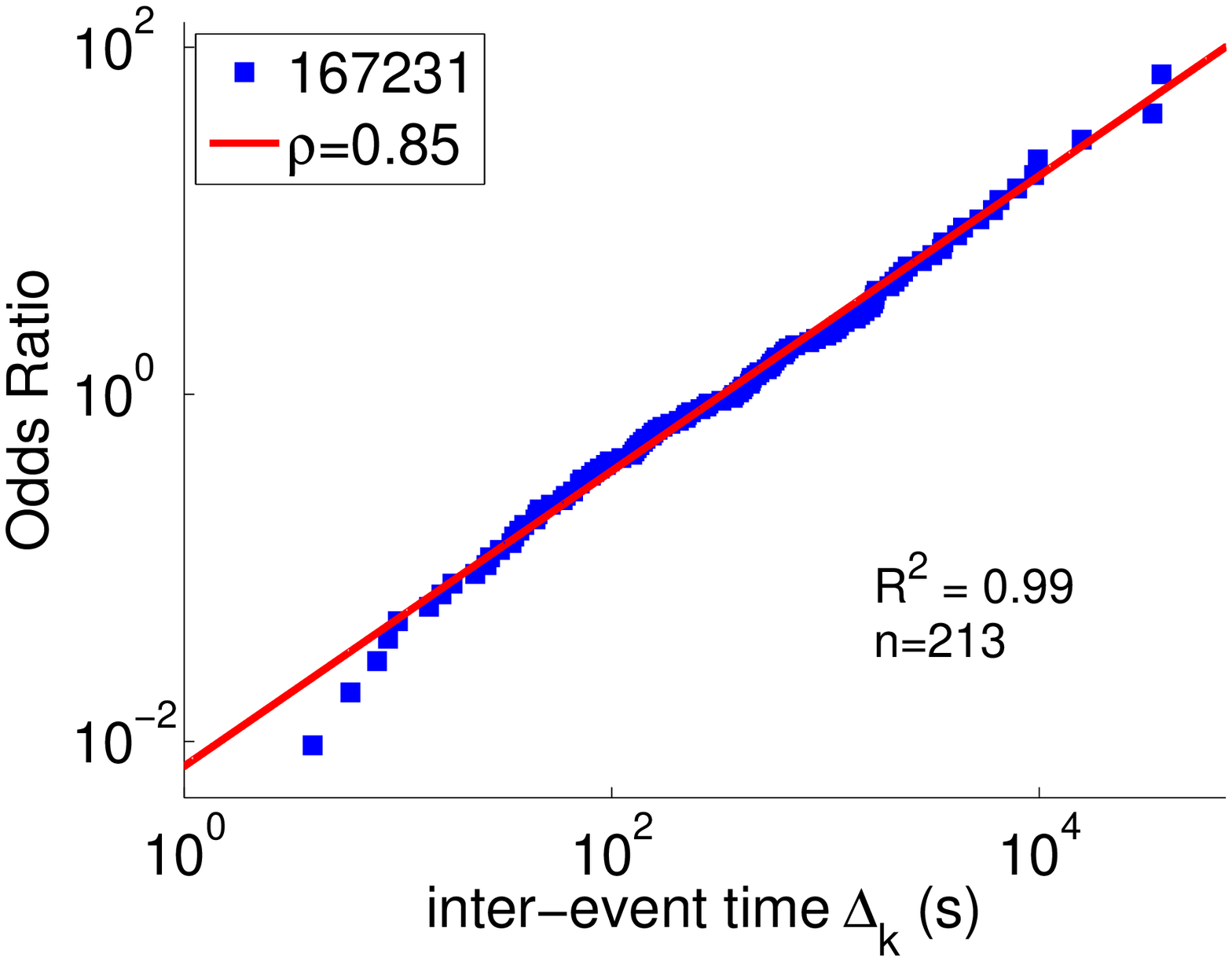}}
{\includegraphics[width=.24\textwidth]{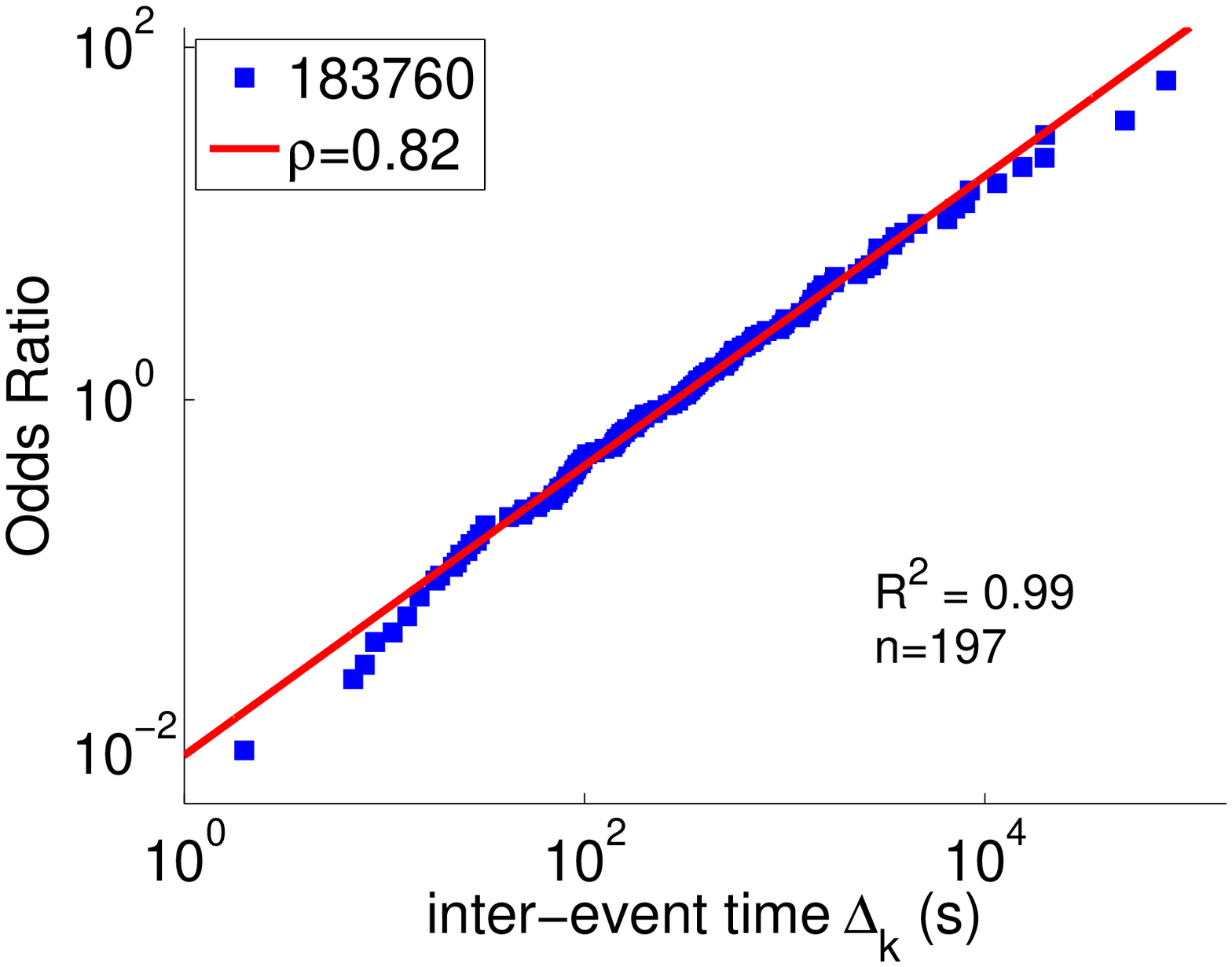}}
{\includegraphics[width=.24\textwidth]{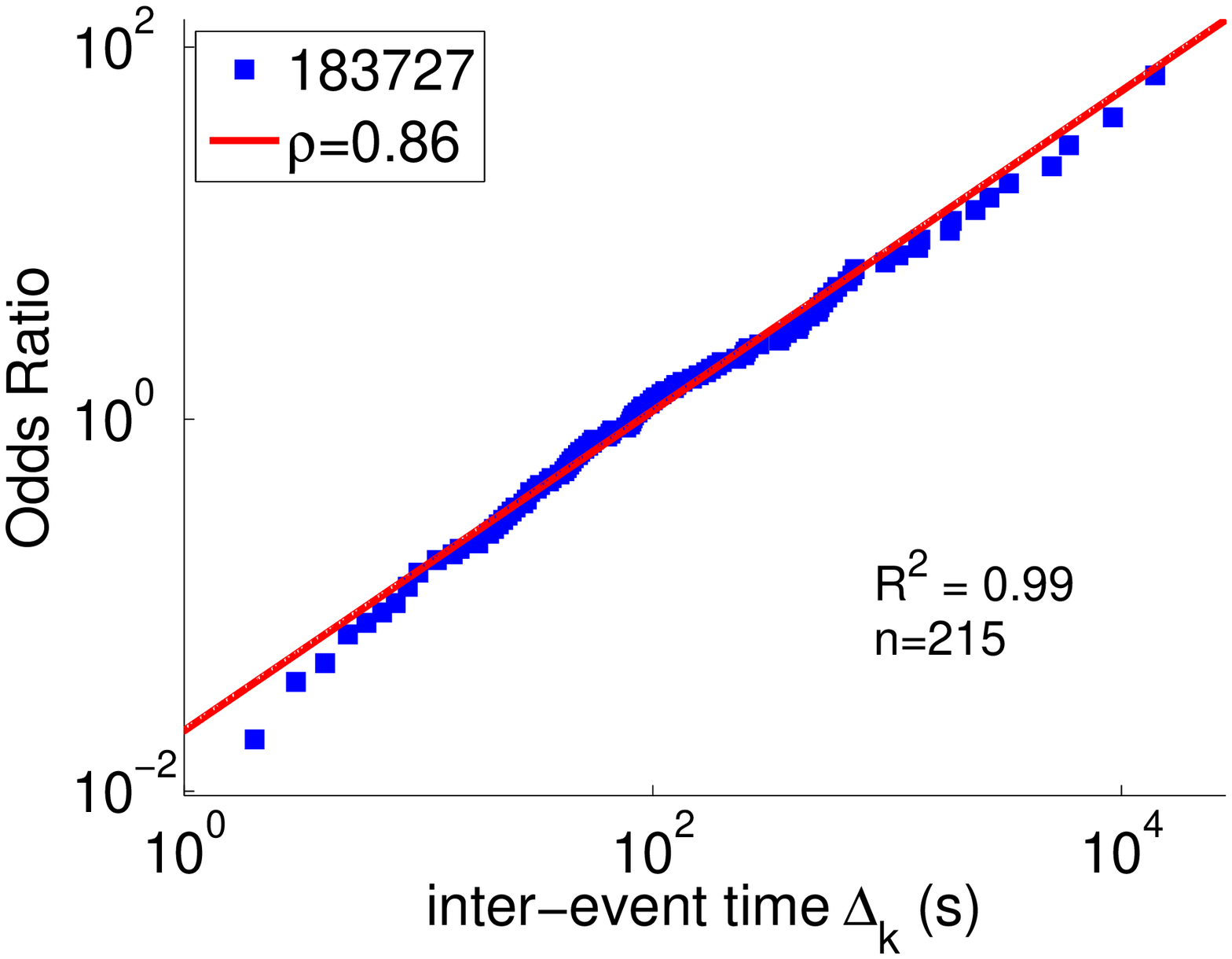}}
\caption{Sample from the Ask MetaFilter dataset.}
\end{figure*}





%
%
%

\end{document}